\documentclass{lmcs}
\pdfoutput=1

\usepackage{lastpage}
\renewcommand{\dOi}{14(4:21)2018}
\lmcsheading{}{1--\pageref{LastPage}}{}{}%
{Dec.~14,~2017}{Dec.~07,~2018}{}

\usepackage{amsthm}
\usepackage{amsmath}
\usepackage{amsfonts}
\usepackage{amssymb}
\usepackage{stmaryrd}
\usepackage{epic}
\usepackage{amscd}
\usepackage[dvips]{epsfig}
\usepackage{wrapfig}
\usepackage[caption=false]{subfig}
\usepackage{pst-tree}
\usepackage{bussproofs}
\usepackage{cmll}

\theoremstyle{plain}
\newtheorem*{notations}{Notations}

\newcommand{\tens}{\otimes}
\newcommand{\lpar}{\parr}
\newcommand{\one}{1}
\newcommand{\bottom}{\bot}
\newcommand{\cod}{\oc}
\newcommand{\contr}{\wn}
\newcommand{\target}[1]{t_{#1}}
\newcommand{\targetAnyPort}[1]{\mathfrak{t}_{#1}}
\newcommand{\basis}[1]{\beta(#1)}
\newcommand{\numberInvisibleComponents}[1]{\Theta(#1)}

\newcommand{\TaylorExpansion}[1]{\mathcal{T}_{#1}[0]}
\newcommand{\pseudoExperiments}[1]{\mathfrak{E}(#1)}
\newcommand{\subsets}[1]{\mathfrak{P}(#1)}
\newcommand{\subsetsofsizetwo}[1]{\mathfrak{P}_2(#1)}
\newcommand{\Nat}{\ensuremath{\mathbb{N}}}
\newcommand{\atoms}[1]{\textit{At}(#1)}
\newcommand{\nontrivialconnected}[3]{\mathcal{S}_{#1}^{#3}(#2)}
\newcommand{\connectedcomponents}[1]{\mathcal{H}(#1)}
\newcommand{\connectedcomponentsContaining}[1]{H_{#1}}
\newcommand{\size}[1]{\textit{size}(#1)}
\newcommand{\height}[1]{\textit{height}(#1)}

\newcommand{\leftwires}[1]{\mathcal{L}(#1)}
\newcommand{\leftwiresatzero}[1]{\mathcal{L}_0(#1)}

\newcommand{\labelofport}[1]{l_{#1}}

\newcommand{\termofTaylor}[3]{\mathcal{T}_{#1}[#3](#2)}
\newcommand{\isaCopyfrom}[3]{\kappa_{#1}[#3](#2)}
\newcommand{\criticalports}[3]{\mathcal{K}_{#2, #3}(#1)}
\newcommand{\groundof}[1]{\mathcal{G}(#1)}
\newcommand{\finitemultisets}[1]{\mathcal{M}_\textit{fin}(#1)}
\newcommand{\finitesubsets}[1]{\mathfrak{P}_\textit{fin}(#1)}

\newcommand{\cosize}[1]{\textit{co-size}(#1)}
\newcommand{\depthof}[1]{\textit{depth}(#1)}
\psset{treemode=U,levelsep=5mm}

\def\restriction#1#2{\mathchoice
              {\setbox1\hbox{${\displaystyle #1}_{\scriptstyle #2}$}
              \restrictionaux{#1}{#2}}
              {\setbox1\hbox{${\textstyle #1}_{\scriptstyle #2}$}
              \restrictionaux{#1}{#2}}
              {\setbox1\hbox{${\scriptstyle #1}_{\scriptscriptstyle #2}$}
              \restrictionaux{#1}{#2}}
              {\setbox1\hbox{${\scriptscriptstyle #1}_{\scriptscriptstyle #2}$}
              \restrictionaux{#1}{#2}}}
\def\restrictionaux#1#2{{#1\,\smash{\vrule height .8\ht1 depth .85\dp1}}_{\,#2}} 

\newcommand{\dom}[1]{\mathsf{dom}(#1)}

\newcommand{\im}[1]{\mathsf{im}(#1)}
\newcommand{\emptysequence}{\varepsilon}
\newcommand{\typesoflinks}{\mathfrak{T}}

\newcommand{\cuts}[1]{\mathcal{C}(#1)}
\newcommand{\cutsatzero}[1]{\mathcal{C}_0(#1)}
\newcommand{\portsatzero}[1]{\mathcal{P}_0(#1)}
\newcommand{\portsatdepthgreater}[2]{\mathcal{P}_{> #2}(#1)}
\newcommand{\portsatdepthleq}[2]{\mathcal{P}_{\leq #2}(#1)}
\newcommand{\wiresatzero}[1]{\mathcal{W}_0(#1)}
\newcommand{\axiomsatzero}[1]{\mathcal{A}_0(#1)}
\newcommand{\arity}[1]{{\textit{a}}_{#1}}
\newcommand{\ports}[1]{\mathcal{P}(#1)}
\newcommand{\nonshallowConclusions}[1]{\mathcal{P}_{>0}^{\mathsf{f}}(#1)}
\newcommand{\conclusions}[1]{\mathcal{P}^{\mathsf{f}}(#1)}
\newcommand{\temporaryConclusions}[2]{\mathcal{P}_{#1}^{\mathsf{f}}(#2)}
\newcommand{\aboveBang}[2]{{\cod}_{#1}(#2)}
\newcommand{\axioms}[1]{\mathcal{A}(#1)}
\newcommand{\multiplicativeports}[1]{\mathcal{P}^\textit{m}(#1)}
\newcommand{\multiplicativeportsatzero}[1]{\mathcal{P}_0^\textit{m}(#1)}

\newcommand{\sm}[1]{\llbracket #1 \rrbracket}
\newcommand{\Card}[1]{\mathsf{Card}\left( #1 \right)}
\newcommand{\portsoftype}[2]{\mathcal{P}^{#1}(#2)}
\newcommand{\portsatzerooftype}[2]{\mathcal{P}_0^{#1}(#2)}
\newcommand{\exponentialportsatzero}[1]{\mathcal{P}_0^\textit{e}(#1)}
\newcommand{\boxes}[1]{\mathcal{B}(#1)}
\newcommand{\boxesatzero}[1]{\mathcal{B}_{0}(#1)}

\newcommand{\exactboxes}[2]{\mathcal{B}^{=#2}(#1)}
\newcommand{\exactboxesatzero}[2]{\mathcal{B}_{0}^{=#2}(#1)}
\newcommand{\boxesgeq}[2]{\mathcal{B}^{\geq #2}(#1)}
\newcommand{\boxesatzerogeq}[2]{\mathcal{B}_0^{\geq #2}(#1)}
\newcommand{\boxesatzerosmaller}[2]{\mathcal{B}_0^{< #2}(#1)}
\newcommand{\boxessmaller}[2]{\mathcal{B}^{< #2}(#1)}
\newcommand{\exponentialports}[1]{\mathcal{P}^{\textit{e}}(#1)}
\newcommand{\wires}[1]{\mathcal{W}(#1)}
\newcommand{\supp}[1]{\textit{Supp}(#1)}
\newcommand{\multi}[1]{[#1]}
\newcommand{\scalefactten}{1}

\newcommand{\scalefactter}{0.4}

\newcommand{\scalefactR}{0.6}

\newcommand{\propvar}[0]{\mathcal{X}}

\newcommand{\contractionsUnder}[2]{\mathcal{P}_{#1}^{\contr}(#2)}

\begin{document}

\title{Taylor expansion in linear logic is invertible}

\author[D.~de Carvalho]{Daniel de Carvalho}
\address{Innopolis University, Universitetskaya St, 1, Innopolis, Respublika Tatarstan, 420500, Russia}	\email{d.carvalho@innopolis.ru} 
\keywords{Linear Logic, Denotational semantics, Taylor expansion}

\begin{abstract}
Each Multiplicative Exponential Linear Logic (MELL) proof-net can be expanded into a differential net, which is its Taylor expansion. We prove that two different MELL proof-nets have two different Taylor expansions. As a corollary, we prove a completeness result for MELL: We show that the relational model is injective for MELL proof-nets, i.e. the equality between MELL proof-nets in the relational model is exactly axiomatized by cut-elimination.
\end{abstract}

\maketitle

In the seminal paper by Harvey Friedman \cite{Friedman}, it has been shown that equality between simply-typed lambda-terms in the full typed structure $\mathcal{M}_X$ over an infinite set $X$ is completely axiomatized by $\beta$ and $\eta$: for any simply-typed lambda-terms $v$ and $u$, we have $(\mathcal{M}_X \vDash v = u \Leftrightarrow v \simeq_{\beta \eta} u)$. Some variations, refinements and generalizations of this result have been provided by Gordon Plotkin \cite{Plotkin:Lambda_Definability} and Alex Simpson \cite{Simpson:TLCA95}. A natural problem is to know whether a similar result could be obtained for Linear Logic. 

Such a result can be seen as a ``separation'' theorem. To obtain such separation theorems, it is a prerequisite to have a ``canonical'' syntax. 
When Jean-Yves Girard introduced Linear Logic (LL) \cite{ll}, he not only introduced a sequent calculus system but also ``proof-nets''. Indeed, as for LJ and LK (sequent calculus systems for intuitionnistic and classical logic, respectively), different proofs in LL sequent calculus can represent ``morally'' the same proof: proof-nets were introduced to find a unique representative for these proofs. 

The technology of proof-nets was completely satisfactory for the multiplicative fragment without units.\footnote{For the multiplicative fragment with units, it has been recently shown \cite{MLLwithunits, HeijltjesH15} that, in some sense, no satisfactory notion of proof-net can exist. Our proof-nets have no jump, so they identify too many sequent calculus proofs, but not more than the relational semantics.} For proof-nets having additives, contractions or weakenings, it was easy to exhibit different proof-nets that should be identified. Despite some flaws, the discovery of proof-nets was striking. In particular, Vincent Danos proved by syntactical means in \cite{phddanos} the confluence of these proof-nets for the Multiplicative Exponential Linear Logic fragment (MELL). For additives, the problem to have a satisfactory notion of proof-net has been addressed in \cite{mallpn}. For MELL, a ``new syntax'' was introduced in \cite{hilbert}. In the original syntax, the following properties of the weakening and of the contraction did not hold:
\begin{itemize}
\item the associativity of the contraction;
\item the neutrality of the weakening for the contraction;
\item the contraction and the weakening as morphisms of co-algebras.
\end{itemize}
But they hold in the new syntax; at least for MELL, we got a syntax that was a good candidate to deserve being considered ``canonical''. Then trying to prove that any two ($\eta$-expanded) MELL proof-nets that are equal in some denotational semantics are $\beta$-joinable has become sensible and had at least the two following motivations:
\begin{itemize}
\item to prove the canonicity of the ``new syntax'' (if we quotient more normal proof-nets, then we would identify proof-nets having different semantics);
\item to prove the confluence by semantic means (if a proof-net reduces to two cut-free proof-nets, then they have the same semantics, so they would be $\beta$-joinable, hence equal).
\end{itemize}
The problem of \emph{injectivity}\footnote{The tradition of the lambda-calculus community rather suggests the word ``completeness'' and the terminology of category theory rather suggests the word ``faithfulness'', but we follow here the tradition of the Linear Logic community.} of the denotational semantics for MELL, which is the question whether equality in the denotational semantics between ($\eta$-expanded) MELL proof-nets is exactly axiomatized by cut-elimination or not, can be seen as a study of the separation property with a semantic approach. The first work on the study of this property in the framework of proof-nets is \cite{Marco} where the
authors deal with the translation into LL of the pure $\lambda$-calculus; it has been studied more recently for the intuitionistic multiplicative
fragment of LL \cite{typedbohm} and for differential nets \cite{separationdiff}. For Parigot's $\lambda \mu$-calculus, see \cite{lmbohm} and \cite{separationsaurin}.

Finally the precise problem of  injectivity for MELL has been addressed by Lorenzo Tortora de Falco in his PhD thesis \cite{phdtortora} and in \cite{injectcoh} for the (multiset based) coherence semantics and the multiset based relational semantics. 
He gave partial results and counter-examples for the coherence semantics: the (multiset based) coherence semantics is not injective for MELL.  
Also, it was conjectured that the relational model is injective for MELL. It is worth mentioning that the injectivity of the relational model trivially entails the injectivity of other denotational semantics: non-uniform coherence semantics (\cite{bucciarelliehrhard01} and \cite{Boudes11}), finiteness spaces \cite{Ehrhard:Finiteness}, weighted sets \cite{AminiEhrhard}... 

We presented an abstract of a proof of this conjecture in \cite{Carvalho:CSL}. This result can be seen as
\begin{itemize}
\item a semantic separation property in the sense of \cite{Friedman};
\item a semantic proof of the confluence property;
\item a proof of the ``canonicity'' of the new syntax of MELL proof-nets;
\item a proof of the fact that if the Taylor expansions of two cut-free MELL proof-nets into differential nets coincide, then the two proof-nets coincide.
\end{itemize}
Differential proof-nets \cite{EhrhardRegnier:DiffNets} are linear approximations of proof-nets that are meant to allow the expression of the Taylor expansion of proof-nets as infinite series of their
linear approximations, which can be seen as a syntactic counterpart of quantitative semantics of Linear Logic (see \cite{Ehrhard:IntroToDiff} for an introduction to the topic). Now, in the present paper, we not only provide a fully detailed proof of this result, we also prove a more general result: We show that if the Taylor expansions of \emph{any} two MELL proof-nets into differential nets coincide, then the two proof-nets coincide, i.e. we removed the assumption of the absence of cuts. Then the injectivity of the relational semantics becomes a corollary of this new result. By the way, the proof is essentially the same as before.\footnote{The two main differences are the following ones: \begin{enumerate} \item The \emph{pseudo-experiments} we consider are not necessarily induced by \emph{experiments} any more, which means that we consider simple differential nets that might reduce to $0$ and have no counterpart in the denotational semantics. \item The constraints on the basis $k$ of the \emph{$k$-heterogeneous pseudo-experiments} we consider are stronger. \end{enumerate}}

In \cite{injectcoh}, a proof of the injectivity of the relational model is given for a weak fragment. 
But despite many efforts (\cite{phdtortora}, \cite{injectcoh}, \cite{boudesunifying}, \cite{pagani06a}, \cite{separationdiff}, \cite{taylorexpansioninverse}...), all
the attempts to prove the conjecture failed up to now. 
New progress was made in \cite{LPSinjectivity}, where it has been proved that the relational semantics is injective for ``connected'' MELL proof-nets. Even though, there, ``connected'' is understood as a very strong assumption, the set of ``connected'' MELL proof-nets contains the fragment of MELL defined by removing weakenings and units. Actually \cite{LPSinjectivity} proved a much stronger result: in the full MELL fragment, two cut-free proof-nets $R$ and $R'$ with the same
interpretation are the same up to the map associating auxiliary doors with their box (we say that they have the same LPS\footnote{The LPS of a cut-free proof-net is the graph obtained by forgetting the outline of the boxes but keeping the trace of the auxiliary doors. The acronym LPS originally stands for ``Linear Proof-Structure''; this terminology might be misleading since the LPS is much more informative than the result of an injective $1$-experiment but is well-established, so we keep the acronym forgetting what it stood for.} - for instance, there are exactly four different proof-nets whose LPS is the LPS depicted in Figure~\ref{fig: LPS},~p.~\ref{fig: LPS}: These four proof-nets are the ones depicted in Figure~\ref{fig: R_1}, Figure~\ref{fig: R_2}, Figure~\ref{fig: R_3} and Figure~\ref{fig: R_4},~p.~\pageref{fig: R_1}). We wrote: ``This result can be expressed in terms of differential nets: two cut-free proof-nets with different LPS have different Taylor expansions. We also believe this work is an essential step
towards the proof of the full conjecture.'' Despite the fact we obtained a very interesting result about \emph{all} the proof-nets (i.e. also for non-``connected'' proof-nets\footnote{and even adding the MIX rule}), the last sentence was a bit too optimistic, since, in the present paper, which presents a proof of the full conjecture, we could not use any previous result nor any previous technique/idea.

Let us give one more interpretation of its significance. First, notice that a proof of this result should consist in showing that, given two non $\beta$-equivalent proof-nets $R$ and $R'$, their respective semantics $\sm{R}$ and $\sm{R'}$ are not equal, i.e. $\sm{R} \setminus \sm{R'} \not= \emptyset$ or $\sm{R'} \setminus \sm{R} \not= \emptyset$.\footnote{The converse, i.e. two $\beta$-equivalent proof-nets have the same semantics, holds by soundness.} But, actually, we prove something much stronger: We prove that, given a proof-net $R$, there exist two points $\alpha$ and $\beta$ such that, for any proof-net $R'$, we have $(\{ \alpha, \beta \} \subseteq \sm{R'} \Leftrightarrow R \simeq_\beta R')$. 

Now, the points of the relational model can be seen as non-idempotent intersection types\footnote{Idempotency of intersection ($\alpha \cap \alpha = \alpha$) does not hold.}  (see \cite{phddecarvalho} and \cite{Carvalhoexecution} for a correspondence between points of the relational model and System R - System R has also been studied recently in \cite{inhabitation}). And the proof given in the present paper uses MELL types only to derive the normalization property; actually we prove the injectivity for cut-free proof-nets in an untyped framework:\footnote{Our proof even works for ``non-correct'' proof-structures (correctness is the property characterizing nets corresponding in a typed framework with proofs in sequent calculus): we could expect that if the injectivity of the relational semantics holds for proof-nets corresponding with MELL sequent calculus, then it still holds for proof-nets corresponding with MELL+MIX sequent calculus, since the category \textbf{Rel} of sets and relations is a compact closed category. The paper \cite{k=2} assuming correctness substituted in the proof the ``bridges'' of \cite{LPSinjectivity}, which are essentially connected components (in the strong sense of the term since the notion of \emph{bridge} ignores boxes - we will consider other ``connected components'' in our proof), by ``empires'', which, in contrast, discriminate between the connectors $\tens$ and $\parr$.} Substituting the assumption that proof-nets are typed by the assumption that proof-nets are normalizable does not change anything to the proof.\footnote{Except that we have to consider the \emph{atomic} subset of the interpretation instead of the full interpretation (see Definition~\ref{definition: atomic point}).} In \cite{CarvPagTdF10}, we gave a semantic characterization of normalizable untyped proof-nets and we characterized ``head-normalizable'' proof-nets as proof-nets having a non-empty interpretation in the relational semantics, while \cite{CarvalhoF16} gave a characterization of strongly normalizable untyped proof-nets \emph{via} non-idempotent intersection types. Principal typings in untyped $\lambda$-calculus are intersection types which allow to recover all the intersection types of some term. If, for instance, we consider the System $R$ of \cite{phddecarvalho} and \cite{Carvalhoexecution}, it is enough to consider some \emph{injective $1$-point}\footnote{An \emph{injective $k$-point} is a point in which all the positive multisets have cardinality $k$ and in which each atom occurring in it occurs exactly twice.} to obtain the principal typing of an untyped $\lambda$-term. But, generally, for normalizable MELL proof-nets, \emph{injective $k$-points}, for any $k$, are not principal typings; indeed, two cut-free MELL proof-nets having the same LPS have the same injective $k$-points for any $k \in \Nat$. In the current paper we show that a $1$-point and a \emph{$k$-heterogeneous point}\footnote{\emph{$k$-heterogeneous points} are points in which every positive multiset has cardinality $k^j$ for some $j> 0$ and, for any $j > 0$, there is at most one occurrence of a positive multiset having cardinality $k^j$ (see our Definition~\ref{definition: k-heterogeneous point}).} together allow to recover the interpretation of any normalizable MELL proof-net; by the way, our result cannot be improved in such a way that one point would be already enough for any MELL proof-net (see our Proposition~\ref{prop: non-principal typing}). So, our work can also be seen as a first attempt to find a right notion of ``principal typing'' of intersection types in Linear Logic. As a consequence, the introduction of technologies allowing to compute directly by semantic means this principal typing should make possible normalization by evaluation, as in \cite{Rocca88} for $\lambda$-calculus; that said, the complexity of such a computation is still unclear.

Section~\ref{section: Syntax} formalizes untyped proof-structures (PS's) and typed proof-structures (typed PS's). Taylor expansion is defined in Section~\ref{section: Taylor expansion}. Section~\ref{section: rebuilding} presents our algorithm leading from the Taylor expansion of $R$ to the rebuilding of $R$ and proves its correctness, which shows the invertibility of Taylor expansion (Corollary~\ref{cor: invertibility of Taylor expansion}). Section~\ref{section: injectivity} is devoted to show the completeness (the injectivity) of the relational semantics: for any typed PS's $R$ and $R'$, we have $(\llbracket R \rrbracket = \llbracket R' \rrbracket \Leftrightarrow R \simeq_\beta R')$ (Corollary~\ref{cor: injectivity}), where $\simeq_\beta$ is the reflexive symmetric transitive closure of the cut-elimination relation, by showing first that cut-free PS's are characterized by their relational interpretation (Theorem~\ref{thm: injectivity in the untyped framework}).

\begin{notations} 
We denote by $\emptysequence$ the empty sequence. 

For any $n \geq 2$, for any $\alpha_1, \ldots, \alpha_{n+1}$, we define, by induction on $n$, the $(n+1)$-tuple $(\alpha_1, \ldots, \alpha_{n+1})$ by setting $(\alpha_1, \ldots, \alpha_{n+1}) = (\alpha_1, (\alpha_2, \alpha_3, \ldots, \alpha_{n+1}))$.

For any set $E$, we denote by $\subsets{E}$ the set of subsets of $E$, by $\finitesubsets{E}$ the set of finite subsets of $E$ and by $\subsetsofsizetwo{E}$ the set $\{ \mathcal{E}_0 \in \subsets{E} ; \Card{\mathcal{E}_0} = 2 \}$.

A multiset $f$ of elements of some set $\mathcal{E}$ is a function $\mathcal{E} \to \Nat$; we denote by $\supp{f}$ \emph{the support of $f$} i.e. the set $\{ e \in \mathcal{E} ; f(e) \not= 0 \}$. A multiset $f$ is said to be \emph{finite} if $\supp{f}$ is finite. The set of finite multisets of elements of some set $\mathcal{E}$ is denoted by $\finitemultisets{\mathcal{E}}$.

If $f$ is a function $\mathcal{E} \to \mathcal{E}'$, $x_0 \in \mathcal{E}$ and $y \in \mathcal{E}'$, then we denote by $f[x_0 \mapsto y]$ the function $\mathcal{E} \to \mathcal{E}'$ defined by $f[x_0 \mapsto y](x) = \left\lbrace \begin{array}{ll} f(x) & \textit{if $x \not= x_0$;}\\ y & \textit{if $ x = x_0$.} \end{array} \right.$ If $f$ is a function $\mathcal{E} \to \mathcal{E}'$ and $\mathcal{E}_0 \subseteq \dom{f} = \mathcal{E}$, then we denote by $f[\mathcal{E}_0]$ the set $\{ f(x) ; x \in \mathcal{E}_0 \}$ and by $f_\ast$ the function $\subsets{E} \to \subsets{E}'$ that associates with every $\mathcal{E}_0 \in \subsets{E}$ the set $f[\mathcal{E}_0]$.
\end{notations}

\section{Syntax}\label{section: Syntax}

\subsection{Differential proof-structures}

We introduce the syntactical objects we are interested in. As recalled in the introduction, correctness does not play any role, that is why we do not restrict our nets to be correct and we rather consider proof-structures (\emph{PS's}). Since it is convenient to represent formally our proof using differential nets possibly with boxes (\emph{differential PS's}), we define PS's as differential PS's satisfying some conditions (Definition~\ref{defin: differential PS}). More generally, \emph{differential in-PS's} are defined by induction on the depth, which is the maximum level of box nesting: Definition~\ref{defin: diff ground-structure},  Definition~\ref{defin: ground-structure} and Definition~\ref{defin: typed simple differential net} concern what happens at depth $0$, i.e. whenever there is no box; in particular, typed ground-structures allow to represent proofs of the multiplicative fragment (MLL).

We set $\typesoflinks = $ $\{ \tens, $ $\parr, $ $\one, \bottom, \cod, \contr, \textit{ax} \}$.

\begin{defi}\label{defin: diff ground-structure}
A \emph{pre-net} is a $7$-tuple $\mathcal{G} = (\mathcal{P}, l, \mathcal{W}, \mathcal{A}, $ $\mathcal{C}, $ $t, \mathcal{L})$, where
\begin{itemize}
\item $\mathcal{P}$ is a finite set; the elements of $\mathcal{P}$ are the \emph{ports of $\mathcal{G}$};
\item $l$ is a function $\mathcal{P} \to \typesoflinks$; the element $l(p)$ of $\typesoflinks$ is the \emph{label of $p$ in $\mathcal{G}$};
\item $\mathcal{W}$ is a subset of $\mathcal{P}$; the elements of $\mathcal{W}$ are the \emph{wires of $\mathcal{G}$};\footnote{We identify a wire with its source port.}
\item $\mathcal{A} \subseteq \subsetsofsizetwo{\mathcal{P}}$ is a partition of $\{ p \in \mathcal{P} ; l(p) = \textit{ax}\}$; the elements of $\mathcal{A}$ are the \emph{axioms of $\mathcal{G}$};
\item $\mathcal{C}$ is a subset of $\subsetsofsizetwo{\mathcal{P} \setminus \mathcal{W}}$ such that $(\forall c, c' \in \mathcal{C}) (c \cap c' \not= \emptyset \Rightarrow c=c')$; the elements of $\mathcal{C}$ are the \emph{cuts of $\mathcal{G}$};
\item $t$ is a function $\mathcal{W} \to \{ p \in \mathcal{P} ; l(p) \notin \{ \one, \bottom, \textit{ax} \} \}$ such that, for any $p \in \mathcal{P}$, we have $(l(p) \in$ $\{ \tens, $ $\parr \}$ $\Rightarrow$ $\Card{\{ w \in \mathcal{W} ; t(w) = p \}} = 2)$; 
if $t(w) = p$, then $w$ is a \emph{premise of $p$}; the \emph{arity $\arity{\mathcal{G}}(p)$ of $p$} is the number of its premises;
\item and $\mathcal{L}$ is a subset of $\{ w \in \mathcal{W}; l(t(w)) \in \{ \tens, \parr \} \}$ such that $(\forall p \in \mathcal{P})$ $(l(p) \in \{ \tens, \parr \}$ $\Rightarrow$ $\Card{{\{ w \in \mathcal{L} ; t(w) = p \}}} = 1)$; if $w \in \mathcal{L}$ such that $t(w) = p$, then $w$ is \emph{the left premise of $p$}; if $w \in \mathcal{W} \setminus \mathcal{L}$ such that $l(t(w)) \in \{ \tens, \parr \}$, then $w$ is \emph{the right premise of $t(w)$}.
\end{itemize}

We set $\wires{\mathcal{G}} = \mathcal{W}$, $\ports{\mathcal{G}} = \mathcal{P}$, $\labelofport{\mathcal{G}} = l$, $\target{\mathcal{G}} = t$, $\leftwires{\mathcal{G}} = \mathcal{L}$, $\axioms{\mathcal{G}} = \mathcal{A}$ and $\cuts{\mathcal{G}} = \mathcal{C}$. The set $\conclusions{\mathcal{G}} = \mathcal{P} \setminus (\mathcal{W} \cup \bigcup \mathcal{C})$ is the set of \emph{conclusions of $\mathcal{G}$}. For any $t \in \typesoflinks$, we set $\portsoftype{t}{\mathcal{G}} = \{ p \in \mathcal{P}; l(p) = t \}$; we set $\multiplicativeports{\mathcal{G}} = \portsoftype{\tens}{\mathcal{G}} \cup \portsoftype{\parr}{\mathcal{G}}$; the set $\exponentialports{\mathcal{G}}$ of \emph{exponential ports of $\mathcal{G}$} is $\portsoftype{\cod}{\mathcal{G}} \cup \portsoftype{\contr}{\mathcal{G}}$. 

A \emph{pre-ground-structure} is a pre-net $\mathcal{G}$ such that $\im{\target{\mathcal{G}}} \cap \portsoftype{\cod}{\mathcal{G}} = \emptyset$.
\end{defi}

Notice that, although we depict cuts as wires\footnote{like wires between principal ports in the formalism of interaction nets \cite{Lafont:fromproofnetstointeractionnets} (but, in contrast with interaction nets, Definition~\ref{defin: diff ground-structure} allows axiom-cuts)} (see the content of the box $o_3$ of the PS $R$ - the third leftmost box at depth $0$ of Figure~\ref{fig: new_example},~p.~\pageref{fig: new_example} - for an example of a cut), cuts are not elements of the set $\mathcal{W}$. A wire $p \in \mathcal{W}$ goes from a port that has the same name $p$ to its target $t(p)$; instead of using arrows in our figures to indicate the direction, we will use the following convention: Unless $l(p) = \textit{ax}$ (but in this case there is no ambiguity since such a port $p$ can never be the target of any wire), whenever a wire goes from $p$ to some port $q$, it will be depicted by an edge reaching underneath the vertice corresponding to $p$.

\begin{defi}\label{defin: ground-structure}
Given a pre-net $\mathcal{G}$, we denote by $\le_{\mathcal{G}}$ the reflexive transitive closure of the binary relation $P_{\mathcal{G}}$ on $\ports{\mathcal{G}}$ defined by $(P_{\mathcal{G}}(q, p) \Leftrightarrow \target{\mathcal{G}}(p) = q)$.

A \emph{simple differential net} (resp. a  \emph{ground-structure}) is a pre-net $\mathcal{G}$ (resp. a pre-ground-structure) such that the relation $P_{\mathcal{G}}$ is irreflexive and the relation $\leq_{\mathcal{G}}$ is antisymmetric.
\end{defi}

\begin{exa}\label{example: ground-structure}
The ground-structure $\mathcal{G}$ of the content of the box $o_1$ of the PS $R$ (the leftmost box of Figure~\ref{fig: new_example}) is defined by: $\ports{\mathcal{G}} = \{ p_1, p_2, p_3, p_4 \}$, $\wires{\mathcal{G}} = \{ p_2 \}$, $\labelofport{\mathcal{G}}(p_1) = \bot = \labelofport{\mathcal{G}}(p_2)$, $\labelofport{\mathcal{G}}(p_3) = \contr$, $\labelofport{\mathcal{G}}(p_4) = \one$, $\target{\mathcal{G}}(p_2) = p_3$ and $\cuts{\mathcal{G}} = \emptyset = \axioms{\mathcal{G}}$.
\end{exa}

Types will be used only in Subsection~\ref{subsection: Typed}. We introduce them right now in order to help the reader to see how ground-structures can represent MLL proofs.

We are given a set $\propvar$ of propositional variables. We set $\propvar^{\perp} = \{ X^\perp ; X \in \propvar \}$. 
We define the set $\mathbb{T}$ of \emph{MELL types} as follows: $\mathbb{T} ::= \propvar \: \vert \: \propvar^\perp \: \vert \: \one \: \vert \: \bot \vert \: (\mathbb{T} \tens \mathbb{T}) \: \vert \: (\mathbb{T} \parr \mathbb{T}) \: \vert \: \cod \mathbb{T} \: \vert \: \contr \mathbb{T}$. We extend the operator $-^{\perp}$ from the set $\propvar$ to the set $\mathbb{T}$ by defining $T^\perp \in \mathbb{T}$ by induction on $T$, for any $T \in \mathbb{T} \setminus \propvar$, as follows: $(X^{\perp})^{\perp} = X$ if $X \in \propvar$; $\one^\perp = \bot$; $\bot^\perp = \one$; $(A \tens B)^\perp = (A^\perp \parr B^\perp)$; $(A \parr B)^\perp = (A^\perp \tens B^\perp)$; $(\cod A)^\perp = \contr A^\perp$; $(\contr A)^\perp = \cod A^\perp$.

\begin{defi}\label{defin: typed simple differential net}
A \emph{typed simple differential net}  (resp. a \emph{typed ground-structure}) is a pair $(\mathcal{G}, \mathsf{T})$ such that $\mathcal{G}$ is a pre-net (resp. a pre-ground-structure) and $\mathsf{T}$ is a function $\ports{\mathcal{G}} \to \mathbb{T}$ such that 
\begin{itemize}
\item for any axiom $a$ of $\mathcal{G}$, there exists a propositional variable $C$ such that $\mathsf{T}[a] = \{ C, C^\perp \}$;\footnote{Our typed proof-structures are $\eta$-expanded.}
\item for any cut $c$ of $\mathcal{G}$, there exists a MELL type $T$ such that $\mathsf{T}[c] = \{ T, T^\perp \}$;
\end{itemize}
and, for any $p \in \ports{\mathcal{G}}$, the following properties hold:
\begin{itemize}
\item $(\labelofport{\mathcal{G}}(p) \in \{ \one, \bottom \} \Rightarrow \mathsf{T}(p) = \labelofport{\mathcal{G}}(p))$;
\item if $p \in \portsoftype{\tens}{\mathcal{G}}$, then $\mathsf{T}(p) = (\mathsf{T}(w_1) \tens \mathsf{T}(w_2))$, where $w_1$ (resp. $w_2$) is the left premise of $p$ (resp. the right premise of $p$);
\item if $p \in \portsoftype{\parr}{\mathcal{G}}$, then $\mathsf{T}(p) = (\mathsf{T}(w_1) \parr \mathsf{T}(w_2))$, where $w_1$ (resp. $w_2$) is the left premise of $p$ (resp. the right premise of $p$);
\item and, if $p$ is an exponential port of $\mathcal{G}$, then there exists a MELL type $C$ such that $(\mathsf{T}(p) = \labelofport{\mathcal{G}}(p) C \wedge (\forall w \in \mathcal{W}) (\target{\mathcal{G}}(w) = p \Rightarrow \mathsf{T}(w) = C))$.
\end{itemize}
\end{defi}

\begin{figure}
\begin{minipage}{0.4\textwidth}
\centering
\scalebox{\scalefactten}{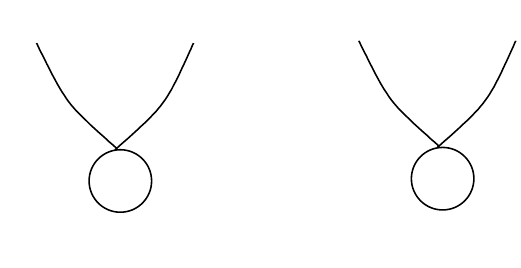}
\caption{Typing of exponential ports}
\label{fig: Typing of exponential ports}
\end{minipage}\hfill
\begin{minipage}{0.4\textwidth}
\centering
\scalebox{\scalefactR}{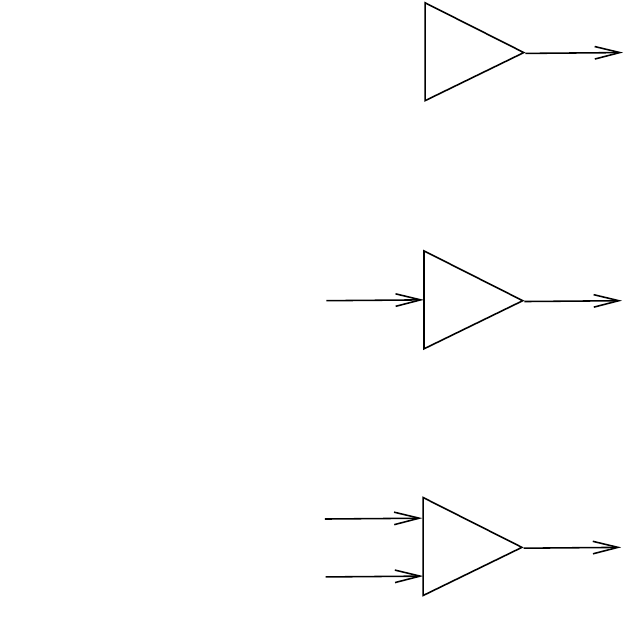}
\caption{Original cells of differential nets}
\label{fig: Original cells}
\end{minipage}\hfill
\end{figure}

Notice that the ports labelled by ``$\cod$'' are completely symmetric to the ports labelled by ``$\contr$'': They can have any number of premises and the typing rule systematically introduces the connector $\cod$ (see Figure~\ref{fig: Typing of exponential ports}, while in \cite{EhrhardRegnier:DiffNets}, there were three different kinds of \emph{cells}: co-weakenings (of arity $0$) and co-derelictions (of arity $1$) that introduce the connector $\cod$, and co-contractions (of arity $2$) that do not modify the type (see Figure~\ref{fig: Original cells}).

\begin{fact}
Let $(\mathcal{G}, \mathsf{T})$ be a typed ground-structure (resp. a typed simple differential net). Then $\mathcal{G}$ is a ground-structure (resp. a simple differential net).
\end{fact}

\begin{proof}
It is enough to notice that, for any $p \in \wires{\mathcal{G}}$, the size of $\mathsf{T}(\target{\mathcal{G}}(p))$ is greater than the size of $\mathsf{T}(p)$.
\end{proof}

A ground-structure $\mathcal{G}$ such that $\portsoftype{\cod}{\mathcal{G}} = \emptyset$ is essentially a PS of depth $0$, so MLL proofs can be represented by typed ground-structures.

\begin{figure}
\begin{minipage}{0.4\textwidth}
\centering
\resizebox{\textwidth}{!}{
\AxiomC{\tiny{$\vdash A, A^\perp$}}
\AxiomC{\tiny{$\vdash B, B^\perp$}}
\RightLabel{\tiny{$\otimes$}}\BinaryInfC{\tiny{$\vdash (A \otimes B), A^\perp, B^\perp$}}
\RightLabel{\tiny{$\parr$}}\UnaryInfC{\tiny{$\vdash (A \otimes B), (A^\perp \parr B^\perp)$}}
\AxiomC{\tiny{$\vdash \underline{A}, \underline{A}^\perp$}}
\RightLabel{\tiny{$\otimes$}}\BinaryInfC{\tiny{$\vdash ((A \otimes B) \otimes \underline{A}), (A^\perp \parr B^\perp), \underline{A}^\perp$}}
\DisplayProof}
\caption{Proof $\pi_1$}\label{fig: sequent calculus 1}
\end{minipage}\hfill
\begin{minipage}{0.4\textwidth}
\centering
\resizebox{\textwidth}{!}{
\AxiomC{\tiny{$\vdash A, A^\perp$}}
\AxiomC{\tiny{$\vdash B, B^\perp$}}
\RightLabel{\tiny{$\otimes$}}\BinaryInfC{\tiny{$\vdash (A \otimes B), A^\perp, B^\perp$}}
\AxiomC{\tiny{$\vdash \underline{A}, \underline{A}^\perp$}}
\RightLabel{\tiny{$\otimes$}}\BinaryInfC{\tiny{$\vdash ((A \otimes B) \otimes \underline{A}), A^\perp, B^\perp, \underline{A}^\perp$}}
\RightLabel{\tiny{$\parr$}}\UnaryInfC{\tiny{$\vdash ((A \otimes B) \otimes \underline{A}), (A^\perp \parr B^\perp), \underline{A}^\perp$}}
\DisplayProof}
\caption{Proof $\pi_2$}\label{fig: sequent calculus 2}
\end{minipage}\hfill
\end{figure}
\begin{figure}
\begin{minipage}{0.4\textwidth}
\centering
\resizebox{\textwidth}{!}{
\AxiomC{\tiny{$\vdash A, A^\perp$}}
\AxiomC{\tiny{$\vdash B, B^\perp$}}
\RightLabel{\tiny{$\otimes$}}\BinaryInfC{\tiny{$\vdash (A \otimes B), A^\perp, B^\perp$}}
\AxiomC{\tiny{$\vdash \underline{A}, \underline{A}^\perp$}}
\RightLabel{\tiny{$\otimes$}}\BinaryInfC{\tiny{$\vdash ((A \otimes B) \otimes \underline{A}), A^\perp, B^\perp, \underline{A}^\perp$}}
\RightLabel{\tiny{$\parr$}}\UnaryInfC{\tiny{$\vdash ((A \otimes B) \otimes \underline{A}), (\underline{A}^\perp \parr B^\perp), A^\perp$}}
\DisplayProof}
\caption{Proof $\pi_3$}\label{fig: sequent calculus 3}
\end{minipage}
\end{figure}
\begin{figure}
\centering
\begin{minipage}{0.45\textwidth}
\centering
\resizebox{.7 \textwidth}{!}{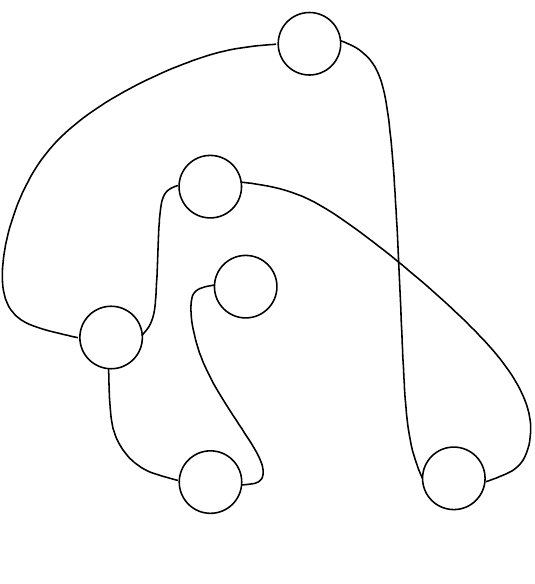}
\captionsetup{width=.9 \textwidth}
\caption{The typed proof-net $R'$}
\label{fig: Proof-net R'}
\end{minipage}\hfill
\begin{minipage}{0.45\textwidth}
\centering
\resizebox{.7 \textwidth}{!}{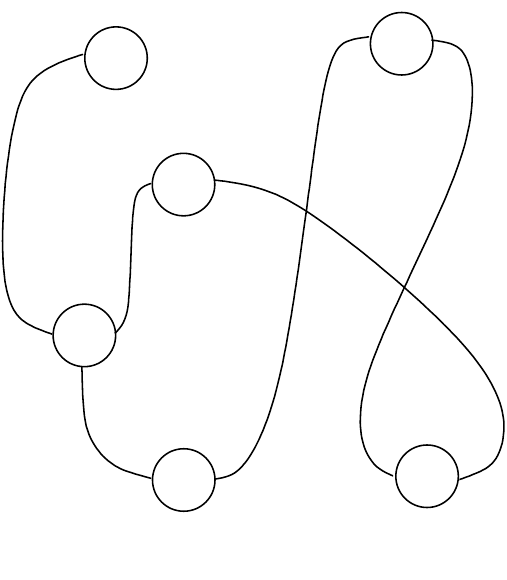}
\captionsetup{width=.9 \textwidth}
\caption{The typed proof-net $R''$}
\label{fig: Proof-net R''}
\end{minipage}\hfill
\end{figure}

\begin{exa}
As we wrote in the introduction, the motivation for proof-nets was to have a canonical object to represent different sequent calculus proofs that should be identified. For instance, Figure~\ref{fig: sequent calculus 1}, Figure~\ref{fig: sequent calculus 2} and Figure~\ref{fig: sequent calculus 3} are three different sequent calculus proofs of the same sequent,\footnote{We underline some occurrences of propositional variables in order to distinguish between different occurrences of the same propositional variable instead of using explicitly the exchange rule.} but the two first proofs are two different sequentializations of the same typed proof-net $(R', \textsf{T}')$ depicted in Figure~\ref{fig: Proof-net R'}, while the third proof is a sequentialization of the typed proof-net $(R'', \textsf{T}'')$ depicted in Figure~\ref{fig: Proof-net R''}. Let $\mathcal{G}'$ (resp. $\mathcal{G}''$) be the ground-structure that corresponds to the proof-net $R'$ (resp. $R''$). 

We can define $\mathcal{G}'$ and $\mathcal{G}''$ as follows: $\ports{\mathcal{G}'} = \{ p_1, p_2, p_3, p_4, p_5, p_6, p_7, p_8, p_9 \} = \ports{\mathcal{G}''}$; $\axioms{\mathcal{G}'} = \{ \{ p_3, p_4 \}, $ $\{ p_5, p_6 \},$ $\{ p_7, p_8 \} \} = $ $\axioms{\mathcal{G}''}$; $\labelofport{\mathcal{G}'} = \labelofport{\mathcal{G}''}$ with $\labelofport{\mathcal{G}'}(p_i) = \left\lbrace \begin{array}{ll} \textit{ax} & \textit{if $3 \leq i \leq 8$;}\\ \tens & \textit{if $i \in \{ 1, 9 \}$;}\\ \parr & \textit{if $i = 2$;} \end{array}\right.$ $\wires{\mathcal{G}'} = \{ p_3, $ $p_5, $ $p_6, $ $p_7, $ $p_8,$ $p_9 \}$ and $\wires{\mathcal{G}''} = \{ p_3, $ $p_5, $ $p_6, $ $p_7, $ $p_4,$ $p_9 \}$; $\leftwires{\mathcal{G}'} = \{ p_7, p_8, p_9 \}$ and $\leftwires{\mathcal{G}''} = \{ p_7, p_4, p_9 \}$; and $\target{\mathcal{G}'}(p_3) = p_1 = \target{\mathcal{G}''}(p_3)$, $\target{\mathcal{G}'}(p_6) = p_2 = \target{\mathcal{G}''}(p_6)$,  $\target{\mathcal{G}'}(p_7) = p_9 = \target{\mathcal{G}''}(p_7)$, $\target{\mathcal{G}'}(p_5) = p_9 = \target{\mathcal{G}''}(p_5)$, $\target{\mathcal{G}'}(p_9) = p_1 = \target{\mathcal{G}''}(p_9)$, $\target{\mathcal{G}'}(p_8) = p_2$ and $\target{\mathcal{G}''}(p_4) = p_2$.
\end{exa}

We can now define our notion of \emph{PS}: We recall that this notion formalizes Danos \& Regnier's new syntax, and not Girard's original syntax. Figure~\ref{fig: Girard_syntax1}, Figure~\ref{fig: Girard_syntax2} and Figure~\ref{fig: new_syntax} illustrate some differences between the two syntaxes: Figure~\ref{fig: Girard_syntax1} and Figure~\ref{fig: Girard_syntax2} are two different objects in the original syntax, both of them are represented in the new syntax by the PS that is depicted in Figure~\ref{fig: new_syntax}. In particular, in the new syntax, auxiliary doors of boxes are always premises of contractions. Since between auxiliary doors and contractions several box boundaries might be crossed, we need the auxiliary notion of \emph{(differential) in-PS}. Concerning \emph{differential PS's}, it is worth noticing that the content of each of their boxes is an in-PS, in particular every $\cod$-port inside is always the main door of a box.

\begin{figure}
\centering
\begin{minipage}{0.3\textwidth}
\centering
\resizebox{\textwidth}{!}{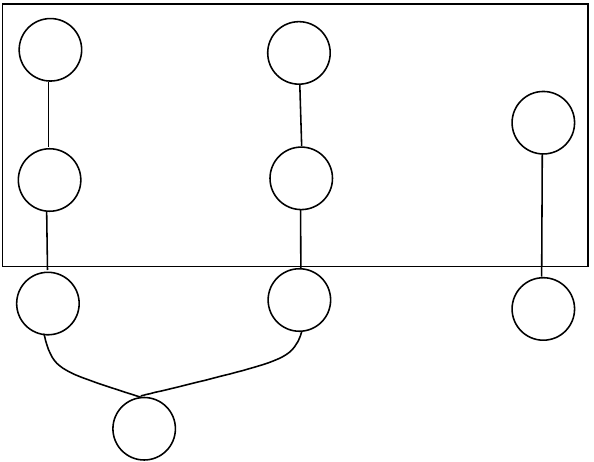}
\captionsetup{width=.8 \textwidth}
\caption{Proof-net in Girard's original syntax}
\label{fig: Girard_syntax1}
\end{minipage}\hfill
\begin{minipage}{0.3\textwidth}
\centering
\resizebox{\textwidth}{!}{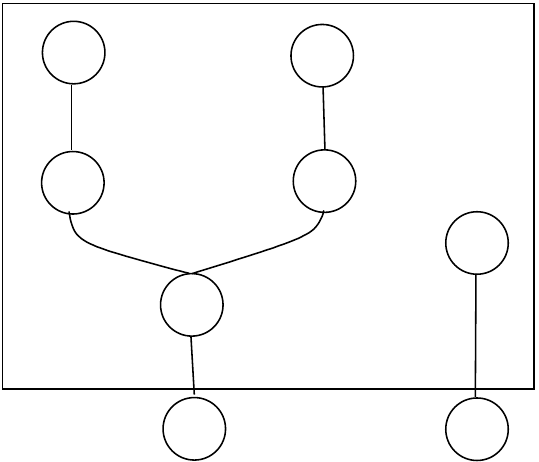}
\captionsetup{width=.8 \textwidth}
\caption{Proof-net in Girard's original syntax}
\label{fig: Girard_syntax2}
\end{minipage}\hfill
\begin{minipage}{0.3\textwidth}
\centering
\resizebox{\textwidth}{!}{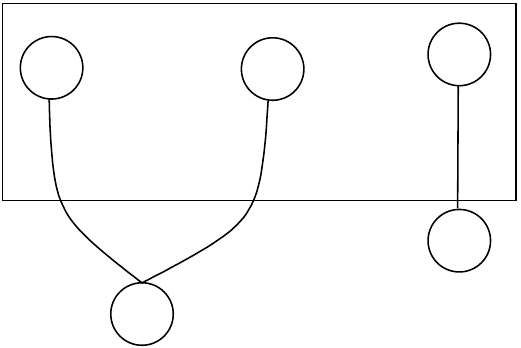}
\captionsetup{width=.8 \textwidth}
\caption{PS (Danos \& Regnier's new syntax)}
\label{fig: new_syntax}
\end{minipage}\hfill
\end{figure}

\begin{defi}\label{defin: differential PS}
For any $d \in \Nat$, we define, by induction on $d$, the set of \emph{differential in-PS's of depth $d$} (resp. the set of \emph{in-PS's of depth $d$}) and, for any differential in-PS $S$ of depth $d$, the sets $\ports{S}$ and $\conclusions{S} \subseteq \ports{S}$. A \emph{differential $in$-PS of depth $d$} (resp. an \emph{$in$-PS of depth $d$}) is a 4-tuple $S = (\mathcal{G}, \mathcal{B}_0, B_0, t)$ such that
\begin{itemize}
\item $\mathcal{G}$ is a simple differential net (resp. a ground-structure); we set $\groundof{S} = \mathcal{G}$;
\item $\mathcal{B}_0 \subseteq \{ p \in \portsoftype{\cod}{\mathcal{G}} ; \arity{\mathcal{G}}(p) = 0 \}$ (resp. $\mathcal{B}_0 = \portsoftype{\cod}{\mathcal{G}}$) such that $\emptysequence \notin \mathcal{B}_0$ and, for any pair $(p_1, p_2) \in \ports{\mathcal{G}}$, we have $p_1 \notin \mathcal{B}_0$ and, if $p_1$ is a pair $(p, p')$ too, then $p \notin \mathcal{B}_0$;\footnote{We cannot simply disallow pairs in $\ports{\mathcal{G}}$ since in the definition of the differential in-PS $\termofTaylor{R}{e}{i}$ (Definition~\ref{defin: Taylor}) we will use pairs to denote copies of ports of the contents of the boxes that have been expanded.} the elements of $\mathcal{B}_0$ are the \emph{boxes of $S$ at depth $0$};\footnote{We identify a box with its main door.}
\item $B_0$ is a function that associates with every $o \in \mathcal{B}_0$ an in-PS of depth $< d$ that enjoys the following property: if $d > 0$, then there exists $o \in \mathcal{B}_0$ such that $B_0(o)$ is an in-PS of depth $d-1$;\footnote{The function $B_0$ maps boxes at depth $0$ to their contents.} we set $\ports{S} = \ports{\mathcal{G}} \cup \bigcup_{o \in \mathcal{B}_0} (\{ o \} \times \ports{B_0(o)})$; the elements of $\ports{S}$ are the \emph{ports of $S$};
\item $t$ is a partial function $\bigcup_{o \in \mathcal{B}_0} (\{ o \} \times \conclusions{B_0(o)}) \rightharpoonup \portsoftype{\contr}{\mathcal{G}} \cup \mathcal{B}_0$ such that, for any $o \in \mathcal{B}_0$, there is a unique $q_o \in \conclusions{\groundof{B_0(o)}}$, which we will denote by $\aboveBang{S}{o}$, such that $\{ (o, q_o) \} = $ $\{ q \in \dom{t} ; t(q) = o \}$;\footnote{The function $t$ maps to exponential ports at depth $0$ their premises that are doors of boxes.} 
we set $\conclusions{S} = \conclusions{\groundof{S}} \cup \bigcup_{o \in \mathcal{B}_0} \{ (o, q) ; (q \in \conclusions{B_0(o)} \wedge (o, q) \notin \dom{t}) \}$ and $\nonshallowConclusions{S} = \conclusions{S} \setminus \conclusions{\groundof{S}}$; the elements of $\conclusions{S}$ (resp. of $\conclusions{\groundof{S}}$, resp. of $\nonshallowConclusions{S}$) are the (resp. \emph{shallow}, resp. \emph{non-shallow}) \emph{conclusions of $S$}.
\end{itemize}

We set $\portsatzero{S} = \ports{\groundof{S}}$ (the elements of $\portsatzero{S}$ are the \emph{ports of $S$ at depth $0$}) and, for any $l \in \typesoflinks \cup \{ \textit{m}, \textit{e} \}$, we set $\portsatzerooftype{l}{S} = \portsoftype{l}{\groundof{S}}$. We set $\wiresatzero{S} = \wires{\groundof{S}}$, $\leftwiresatzero{S} = \leftwires{\groundof{S}}$, $\axiomsatzero{S} = \axioms{\groundof{S}}$ and $\cutsatzero{S} = \cuts{\groundof{S}}$. The function $\arity{S}: \portsatzero{S} \to \Nat$ is defined by setting $\arity{S}(p) = \arity{\groundof{S}}(p) + \Card{\{ q \in \dom{t} ; t(q) = p \}}$ for any $p \in \portsatzero{S}$. The integer $\cosize{S}$ is defined by induction on $\depthof{S}$:\footnote{The supremum is taken in $\Nat$, hence, if $S$ is the empty PS, then $\cosize{S} = 0$.} $$\cosize{S} = \sup{(\{ \arity{S}(p) ; p \in \portsatzero{S} \} \cup \{ \cosize{B_0(o)} ; o \in \mathcal{B}_0(S) \})}$$

We set $\boxesatzero{S} = \mathcal{B}_0$, $B_S = B_0$ and $\target{S} = t$. For any $o \in \boxesatzero{S}$, we set $\temporaryConclusions{S}{o} = \{ p \in \conclusions{B_S(o)} ; (o, p) \in \dom{\target{S}} \}$.\footnote{Equivalently, $\temporaryConclusions{S}{o} = \{ p \in \conclusions{B_S(o)} ; (o, p) \notin \conclusions{S}\}$.} We denote by $\targetAnyPort{S}$ the function $\wiresatzero{S} \cup \bigcup_{o \in \boxesatzero{S}} (\{ o \} \times \temporaryConclusions{S}{o}) \to \portsatzero{S}$ that associates with every $p \in \wiresatzero{S}$ the port $\target{\groundof{S}}(p)$ of $\groundof{S}$ and with every $(o, p)$, where $o \in \boxesatzero{S}$ and $p \in \temporaryConclusions{S}{o}$, the port $\target{S}(o, p)$ of $\groundof{S}$. We set $\contractionsUnder{S}{o} = \target{S}[\{ o \} \times \temporaryConclusions{S}{o}] \setminus \{ o \}$ for any $o \in \boxesatzero{S}$. The set $\boxes{S}$ of \emph{boxes of $S$} is defined by induction on $\depthof{S}$: $\boxes{S} = \boxesatzero{S} \cup \bigcup_{o \in \boxesatzero{S}} \{ (o, o') ; o' \in \boxes{B_0(o)} \}$. For any binary relation $P \in \{ \geq, =, < \}$ on $\Nat$, for any $i \in \Nat$, we set $\mathcal{B}_0^{P i}(S) = \{ o \in \boxesatzero{S} ; P(\depthof{B_S(o)}, i) \}$ and we define, by induction on $\depthof{S}$, the set $\mathcal{B}^{P i}(S) \subseteq \boxes{S}$ as follows: $\mathcal{B}^{P i}(S) = \mathcal{B}_0^{P i}(S) \cup \bigcup_{o \in \boxesatzero{S}} \{ (o, o') ; o'  \in \mathcal{B}^{P i}(B_S(o)) \}$. We set $\portsatdepthgreater{S}{i} = \bigcup_{o \in \boxesatzerogeq{S}{i}} \{ (o, q) ; q \in \ports{B_S(o)} \}$ and $\mathcal{P}_{\leq i}(S) = \ports{S} \setminus \portsatdepthgreater{S}{i}$. 

A \emph{differential PS} (resp. a \emph{PS}) is a differential in-PS (resp. an in-PS) $S$ such that $\conclusions{S} \subseteq \conclusions{\groundof{S}}$.\footnote{Equivalently, a \emph{differential PS} (resp. a \emph{PS}) is a differential in-PS (resp. an in-PS) $S$ such that $(\forall o \in \boxesatzero{S}) \temporaryConclusions{S}{o} = \conclusions{B_S(o)}$.}
\end{defi}

The set of \emph{cocontractions} of an in-PS $S$ is the set $\portsatzerooftype{\cod}{S} \setminus \boxesatzero{S}$. Notice that an in-PS is a differential in-PS with no co-contraction.

It is worth noticing that the binary relation $\leq_S$ on the set $\boxes{S} \cup \{ \emptysequence \}$ defined by $((o_1, \ldots, o_m) \leq_S (o'_1, \ldots, o'_n) \Leftrightarrow (m \leq n \wedge (o_1, \ldots, o_m) = (o'_1, \ldots, o'_m)))$ defines a tree with $\emptysequence$ as the root.

\begin{exa}
If $R$ is the PS of depth $2$ depicted in Figure~\ref{fig: new_example}, then we have $\boxesatzero{R} = \{ o_1, o_2, o_3, o_4 \}$, $\boxes{R} = \{ o_1, o_2, o_3, o_4, (o_2, o), (o_2, o'), (o_4, o), (o_4, o') \}$, $\exactboxes{R}{0} = \{ o_1, $ $(o_2, o),$ $ (o_2, o'), $ $o_3,$ $(o_4, o)$, $(o_4, o') \}$, $\exactboxes{R}{1} = \{ o_2, o_4 \}$, $\conclusions{R} = \{ p_1, $ $p_2, $ $p_3, $ $p_4, $ $p_5, $ $p_6, $ $p_7 \}$ and $\groundof{B_R(o_1)}$ is the ground-structure of Example~\ref{example: ground-structure}.
\begin{figure}
\centering
\resizebox{\textwidth}{!}{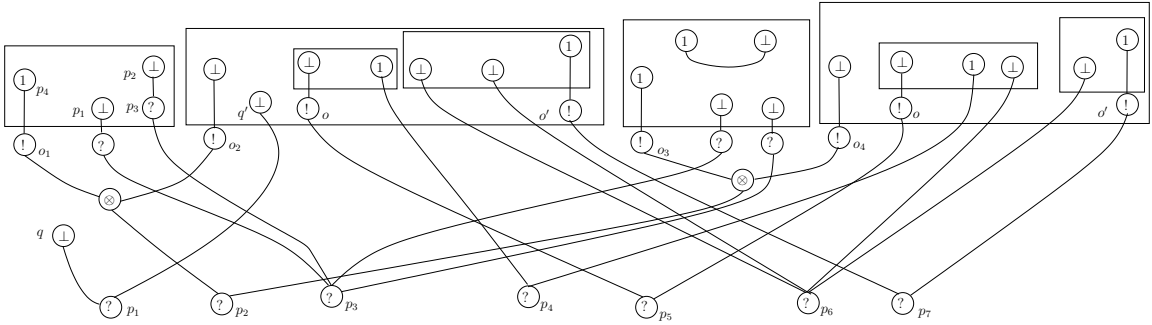}
\caption{The PS $R$}
\label{fig: new_example}
\end{figure}
\end{exa}

In the absence of axioms and cuts, our definition of PS through in-PS's is equivalent to our definition of PS in Definition~4 of \cite{Carvalho:CSL} through $\circ$-PS's. We removed $\circ$-ports because we simplified the proof of Proposition~\ref{prop: crucial} and after this simplification they would not play any role any more (actually we introduced a syntactic construction\footnote{See Definition~\ref{definition: adding contractions}} that, roughly speaking, can be seen as a partial recovery of these $\circ$-ports).

\begin{defi}
For any $l \in \{ \cod, \contr \}$, for any $p$, we denote by $l_p$ the PS $R$ of depth $0$ such that $\portsatzero{R} = \{ p \}$ and $\labelofport{\groundof{R}}(p) = l$.
\end{defi}

\begin{defi}\label{defin: typed differential PS}
For any $d \in \Nat$, we define, by induction on $d$, the set of \emph{typed differential in-PS's of depth $d$} (resp. the set of \emph{typed in-PS's of depth $d$}): it is the set of pairs $(S, \mathsf{T})$ such that $S$ is a differential in-PS (resp. an in-PS) and $\mathsf{T}$ is a function $\ports{S} \to \mathbb{T}$ such that:
\begin{itemize}
\item $(\groundof{S}, \restriction{\mathsf{T}}{\portsatzero{S}})$ is a typed simple differential net (resp. a typed ground-structure);
\item for any $o \in \boxesatzero{S}$, the pair $(B_S(o), \mathsf{T}_o)$ is a typed simple differential net, where $\mathsf{T}_o$ is the function $\ports{B_S(o)} \to \mathbb{T}$ defined by $\mathsf{T}_o(p) = \mathsf{T}(o, p)$ for any $p \in \ports{B_S(o)}$;
\item and, for any $o \in \boxesatzero{S}$, we have $(\forall q \in \temporaryConclusions{S}{o}) (\exists \zeta \in \{ \contr, \cod \}) \mathsf{T}(\target{S}(o, q)) = \zeta \mathsf{T}(o, q)$.
\end{itemize}

A \emph{typed differential PS} (resp. a \emph{typed PS}) is a typed differential in-PS (resp. a typed in-PS) $(S, \mathsf{T})$ such that $\conclusions{S} \subseteq \conclusions{\groundof{S}}$.
\end{defi}

\subsection{Isomorphisms}

We want to consider PS's up to the names of the ports, apart from the names of the shallow conclusions. We thus define the equivalence relation $\equiv$ on PS's; this relation is slightly finer than the equivalence relation $\simeq$, which ignores all the names of the ports.

\begin{defi}\label{definition: isomorphisms}
For any simple differential nets $\mathcal{G}$ and $\mathcal{G'}$, an isomorphism $\varphi$ from $\mathcal{G}$ to $\mathcal{G'}$ is a bijection $\ports{\mathcal{G}} \simeq \ports{\mathcal{G'}}$ such that:
\begin{itemize}
\item $\varphi[\wires{\mathcal{G}}] = \wires{\mathcal{G'}}$
\item $\varphi_\ast[\axioms{\mathcal{G}}] = \axioms{\mathcal{G'}}$
\item $\varphi_\ast[\cuts{\mathcal{G}}] = \cuts{\mathcal{G'}}$
\item $\varphi[\leftwires{\mathcal{G}}] = \leftwires{\mathcal{G'}}$
\item $\target{\mathcal{G'}} \circ \restriction{\varphi}{\wires{\mathcal{G}}} = \varphi \circ \target{\mathcal{G}}$
\item $\labelofport{\mathcal{G}} = \labelofport{\mathcal{G'}} \circ \varphi$
\end{itemize}
We write \emph{$\varphi: \mathcal{G} \simeq \mathcal{G'}$} to denote that $\varphi$ is an isomorphism from $\mathcal{G}$ to $\mathcal{G'}$; we write \emph{$\mathcal{G} \simeq \mathcal{G'}$} if there exists $\varphi$ such that $\varphi: \mathcal{G} \simeq \mathcal{G'}$.

Moreover, we write \emph{$\varphi : \mathcal{G} \equiv \mathcal{G'}$} to denote that $\varphi: \mathcal{G} \simeq \mathcal{G'}$ and $(\forall p \in \conclusions{\mathcal{G}}) \varphi(p) = p$; we write \emph{$\mathcal{G} \equiv \mathcal{G'}$} if there exists $\varphi$ such that $\varphi: \mathcal{G} \equiv \mathcal{G'}$.
\end{defi}

\begin{defi}\label{defin: isomorphism}
For any differential in-PS $S$ of depth $d$, for any differential in-PS $S'$, we define, by induction on $d$, the set of isomorphims from $S$ to $S'$: an isomorphism $\varphi$ from $S$ to $S'$ is a function $\ports{S} \to \ports{S'}$ such that:
\begin{itemize}
\item $(\forall p \in \portsatzero{S}) \varphi(p) \in \portsatzero{S'}$ and the function $\groundof{\varphi}: \begin{array}{rcl} \portsatzero{S} & \to & \portsatzero{S'}\\ p & \mapsto & \varphi(p) \end{array}$ is an isomorphism $\groundof{S} \simeq \groundof{S'}$;
\item $\varphi[\boxesatzero{S}] = \boxesatzero{S'}$;
\item $(\forall o \in \boxesatzero{S}) (\forall p \in \ports{B_S(o)}) (\exists p' \in \ports{B_{S'}(\varphi(o))}) \varphi(o, p) = (\varphi(o), p')$ and the function $$\varphi_o: \begin{array}{rcl} \ports{B_S(o)} & \to &\ports{B_{S'}(\varphi(o))}\\ p & \mapsto & \textit{$p'$ such that $\varphi(o, p) = (\varphi(o), p')$} \end{array}$$
is an isomorphism from $B_S(o)$ to $B_{S'}(\varphi(o))$;
\item $\dom{\target{S'}} = \varphi[\dom{\target{S}}]$ and, for any $p \in \dom{\target{S}}$, we have $(\varphi \circ\target{S})(p) = (\target{S'} \circ \varphi)(p)$.
\end{itemize}
We write \emph{$\varphi: S \simeq S'$} to denote that $\varphi$ is an isomorphism from $S$ to $S'$; we write \emph{$S \simeq S'$} if there exists $\varphi$ such that $\varphi: S \simeq S'$.

Moreover, we write \emph{$\varphi : S \equiv S'$} to denote that $\varphi: S \simeq S'$ and $(\forall p \in \conclusions{\groundof{S}}) \varphi(p) = p$; we write \emph{$S \equiv S'$} if there exists $\varphi$ such that $\varphi: S \equiv S'$.

Now, if $\mathcal{T}$ and $\mathcal{T'}$ are two sets of differential in-PS's, we write $\mathcal{T} \equiv \mathcal{T'}$ if there exists a bijection $\varphi : \mathcal{T} \simeq \mathcal{T'}$ such that, for any $T \in \mathcal{T}$, we have $T \equiv \varphi(T)$.

Finally, if $(S, \mathsf{T})$ and $(S', \mathsf{T'})$ are two typed differential in-PS's, then we write $(S, \mathsf{T}) \equiv (S', \mathsf{T'})$ if there exists $\varphi: S \equiv S'$ such that $\mathsf{T} = \mathsf{T'} \circ \varphi$.
\end{defi}

\begin{fact}\label{fact: iso of typed PS's}
Let $(S, \mathsf{T})$ and $(S', \mathsf{T'})$ be two cut-free typed differential in-PS's such that $S \equiv S'$. If $\restriction{\mathsf{T}}{\conclusions{S}} = \restriction{\mathsf{T'}}{\conclusions{S'}}$, then $(S, \mathsf{T}) \equiv (S', \mathsf{T'})$.
\end{fact}

Another variant of the notion of isomorphism will be defined in the next subsection (Definition~\ref{definition: equiv modulo}). A special case of isomorphism consists in renaming only ports at depth $0$:

\begin{defi}\label{defin: special iso}
Let $S$ and $S'$ be two differential in-PS's. Let $\varphi$ be a bijection $\mathcal{P} \simeq \mathcal{P'}$, where $\mathcal{P'} \cap (\portsatzero{S} \setminus \mathcal{P}) = \emptyset$. We say that $S'$ is obtained from $S$ by renaming the ports via $\varphi$ and we write $S' = S[\varphi]$ if the following properties hold:
\begin{itemize}
\item $\portsatzero{S'} = \overline{\varphi}[\portsatzero{S}]$
\item $\wiresatzero{S'} = \overline{\varphi}[\wiresatzero{S}]$
\item $\axiomsatzero{S'} = \{ \overline{\varphi}[a]; a \in \axiomsatzero{S} \}$ 
\item $\cutsatzero{S'} = \{ \overline{\varphi}[a]; a \in \cutsatzero{S} \}$ 
\item $\leftwiresatzero{S'} = \overline{\varphi}[\leftwiresatzero{S}]$
\item $\labelofport{\groundof{S'}} \circ \overline{\varphi} = \labelofport{\groundof{S}}$
\item $\target{\groundof{S'}} \circ \overline{\varphi} = \overline{\varphi} \circ \target{\groundof{S}}$
\item $\boxesatzero{S'} = \overline{\varphi}[\boxesatzero{S}]$
\item $\dom{\target{S'}} = \bigcup_{o \in \boxesatzero{S}} \{ (\overline{\varphi}(o), p) ; p \in \temporaryConclusions{S}{o} \}$ and, for any $(o, p) \in \dom{\target{S}}$, we have $\target{S'}(\overline{\varphi}(o), p) = \target{S}(o, p)$
\item and $B_{S'} = B_S \circ \overline{\varphi}$,
\end{itemize}
where $\overline{\varphi}$ is the function $\portsatzero{S} \to \portsatzero{S'}$ that associates with every $p \in \portsatzero{S}$ the following port of $\groundof{S'}$: $\left\lbrace \begin{array}{ll} p & \textit{if $p \in \portsatzero{S} \setminus \mathcal{P}$;} \\ \varphi(p) &  \textit{if $p \in \mathcal{P}$.} \end{array} \right.$

If $S' = S[\varphi]$ and $\varphi$ is, for some singleton $\mathcal{E} = \{ a \}$, the bijection $\portsatzero{S} \simeq \mathcal{E} \times \portsatzero{S}$ that associates with every $p \in \portsatzero{S}$ the pair $(a, p)$, then $S'$ is denoted by $\langle a, S \rangle$ too.
\end{defi}

\subsection{Some operations on differential proof-structures}

In this subsection, we describe some operations to obtain new PS's from old ones.

For any differential in-PS $S$, for any integer $i$, we define a differential in-PS $S^{\leq i}$ of depth $\leq i$, which is obtained from $S$ by removing some boxes:

\begin{defi}
Let $S$ be a differential in-PS and let $i \in \Nat$. We denote by $S^{\leq i}$ the differential in-PS such that $\groundof{S^{\leq i}} = \groundof{S}$, $\boxesatzero{S^{\leq i}} = \boxesatzerosmaller{S}{i}$, $B_{S^{\leq i}} = \restriction{B_S}{\boxesatzerosmaller{S}{i}}$ and $\target{S^{\leq i}} = \restriction{\target{S}}{\bigcup_{o \in \boxesatzerosmaller{S}{i}} (\{ o \} \times \temporaryConclusions{S}{o})}$.
\end{defi}

In particular $S^{\leq 0}$ is essentially the same object as $\groundof{S}$.

\begin{rem}\label{rem: substructure}
If $\depthof{T} < i$, then $T^{\leq i} = T$.
\end{rem}

\begin{rem}\label{rem: <= <=}
We have ${(S^{\leq i})}^{\leq i'} = S^{\leq \min \{ i, i' \}}$.
\end{rem}

\begin{exa}
The differential PS $R^{\leq 1}$, where $R$ is the PS depicted in Figure~\ref{fig: new_example}, is depicted in Figure~\ref{fig: Rleq1},~p.~\pageref{fig: Rleq1}.
\end{exa}

We can also erase some ports at depth $0$:

\begin{defi}\label{defin: substructure}
Let $S'$ and $S$ be two differential in-PS's. Let $\mathcal{Q} \subseteq \portsatzero{S}$. We write $S' \sqsubseteq_{\mathcal{Q}} S$ to denote that $\portsatzero{S'} \subseteq \portsatzero{S}$, $\wiresatzero{S'} = \{ w \in (\wiresatzero{S} \cap \portsatzero{S'}) \setminus (\mathcal{Q} \cap \exponentialportsatzero{S}); \target{\groundof{S}}(w) \in \portsatzero{S'} \}$, $\labelofport{\groundof{S'}} = \restriction{\labelofport{\groundof{S}}}{\portsatzero{S'}}$, $\target{\groundof{S'}} = \restriction{\target{\groundof{S}}}{\wiresatzero{S'}}$, $\leftwires{\groundof{S'}} = \leftwires{\groundof{S}} \cap \{ w \in \wiresatzero{S} ; \target{\groundof{S}}(w) \in \multiplicativeportsatzero{S'} \}$, $\axiomsatzero{S'} = \{ a \in \axiomsatzero{S} ; \bigcup a \subseteq \portsatzero{S'}\}$, $\cutsatzero{S'} = \{ a \in \cutsatzero{S} ; \bigcup a \subseteq \portsatzero{S'} \setminus (\mathcal{Q} \cap \exponentialportsatzero{S}) \}$, $\boxesatzero{S'} = \boxesatzero{S} \cap \portsatzero{S'}$, $B_{S'} = \restriction{B_S}{\boxesatzero{S'}}$ and $\target{S'} = \restriction{\target{S}}{\bigcup_{o \in \boxesatzero{S'}} (\{ o \} \times \temporaryConclusions{S}{o})}$.  

We write $S' \sqsubseteq S$ if there exists $\mathcal{Q}$ such that $S' \sqsubseteq_{\mathcal{Q}} S$.  
\end{defi}

\begin{rem}
We have $S' \sqsubseteq_{\mathcal{Q}} S$ if and only if $S' \sqsubseteq_{\mathcal{Q} \cap \exponentialportsatzero{S}} S$.
\end{rem}

\begin{rem}\label{rem: substructure and conclusions}
If $S' \sqsubseteq S$, then $\portsatzero{S'} \cap \conclusions{\groundof{S}} \subseteq \conclusions{\groundof{S'}}$.
\end{rem}

\begin{rem}\label{rem: unique substructure}
If $S', S'' \sqsubseteq_{\mathcal{Q}} S$ and $\portsatzero{S'} = \portsatzero{S''}$, then $S' = S''$. So, if, for some $\mathcal{P} \subseteq \portsatzero{S}$, there exists a differential in-PS $S'$ such that $\portsatzero{S'} = \mathcal{P}$ and $S' \sqsubseteq_{\emptyset} S$, then we can denote by $\restriction{S}{\mathcal{P}}$ the unique such differential in-PS $S'$.
\end{rem}

\begin{rem}
We have $S' \sqsubseteq_{\mathcal{Q}} S$ if, and only if, the following properties hold:
\begin{itemize}
\item ${S'}^{\leq 0} \sqsubseteq_{\mathcal{Q}} S^{\leq 0}$
\item $\boxesatzero{S'} = \boxesatzero{S} \cap \portsatzero{S'}$
\item $B_{S'} = \restriction{B_S}{\boxesatzero{S'}}$
\item $\target{S'} = \restriction{\target{S}}{\bigcup_{o \in \boxesatzero{S'}} (\{ o \} \times \temporaryConclusions{S}{o})}$  
\end{itemize}
\end{rem}

\begin{rem}
If $S'_1 \sqsubseteq_{\mathcal{Q}} S_1$ and $\varphi: S_1 \simeq S_2$, then there exists a unique $S'_2 \sqsubseteq_{\varphi[\mathcal{Q}]} S_2$ such that there exists an isomorphism $S'_1 \simeq S'_2$ associating with every port $p$ of $S'_1$ the port $\varphi(p)$ of $S_2$.
\end{rem}

\begin{fact}\label{fact: transitivity of substructures}
Let $S$, $S'$ and $S''$ be three differential in-PS's. Let $\mathcal{Q} \subseteq \portsatzero{S}$ and $\mathcal{Q'} \subseteq \portsatzero{S'}$. If $S'' \sqsubseteq_{\mathcal{Q'}} S'$ and $S' \sqsubseteq_{\mathcal{Q}} S$, then $S'' \sqsubseteq_{\mathcal{Q} \cup \mathcal{Q'}} S$.
\end{fact}

\begin{fact}
Let $S'$ and $S$ be two differential in-PS's and let $\mathcal{Q} \subseteq \exponentialportsatzero{S}$ such that $S' \sqsubseteq_{\mathcal{Q}} S$. Let $i \in \Nat$. Then $S'^{\leq i} \sqsubseteq_{\mathcal{Q}} S^{\leq i}$.
\end{fact}

The operator $\bigoplus$ glues together several differential in-PS's that share only shallow conclusions that are contractions:\footnote{This operation has nothing to do with the additive $\oplus$ of linear logic, it is rather essentially the mix of linear logic.}

\begin{defi}
Let $\mathcal{U}$ be a finite set of differential in-PS's. We say that $\mathcal{U}$ is \emph{gluable} if, for any $T, T' \in \mathcal{U}$ such that $T \not= T'$, we have $\portsatzero{T} \cap \portsatzero{T'} \subseteq \conclusions{\groundof{T}} \cap \portsatzerooftype{\contr}{T} \cap \conclusions{\groundof{T'}} \cap \portsatzerooftype{\contr}{T'}$ and, for any pair $(p_1, p_2) \in \portsatzero{T}$, we have $p_1 \notin \boxesatzero{T'}$ and, if $p_1$ is a pair $(p, p')$ too, then $p \notin \boxesatzero{T'}$.
 
If $\mathcal{U}$ is gluable, then $\bigoplus \mathcal{U}$ is the differential in-PS such that:
\begin{itemize}
\item $\portsatzero{\bigoplus \mathcal{U}} = \bigcup_{T \in \mathcal{U}} \portsatzero{T}$
\item $\wiresatzero{\bigoplus \mathcal{U}} = \bigcup_{T \in \mathcal{U}} \wiresatzero{T}$
\item $\labelofport{\groundof{\bigoplus \mathcal{U}}}(p) = \labelofport{\groundof{T}}(p)$ for any $p \in \portsatzero{\bigoplus \mathcal{U}}$ and any $T \in \mathcal{U}$ such that $p \in \portsatzero{T}$
\item $\axiomsatzero{\bigoplus \mathcal{U}} = \bigcup_{T \in \mathcal{U}} \axiomsatzero{T}$
\item $\cutsatzero{\bigoplus \mathcal{U}} = \bigcup_{T \in \mathcal{U}} \cutsatzero{T}$
\item $\leftwiresatzero{\bigoplus \mathcal{U}} = \bigcup_{T \in \mathcal{U}} \leftwiresatzero{T}$
\item $\target{\groundof{\bigoplus \mathcal{U}}}(p) = \target{\groundof{T}}(p)$ for any $p \in \wiresatzero{\bigoplus \mathcal{U}}$ and any $T \in \mathcal{U}$ such that $p \in \wiresatzero{T}$;
\item $\boxesatzero{\bigoplus \mathcal{U}}(p) = \bigcup_{T \in \mathcal{U}} \boxesatzero{T}$
\item $B_{\bigoplus \mathcal{U}}(o) = B_T(o)$ for any $o \in \boxesatzero{\bigoplus \mathcal{U}}$ and any $T \in \mathcal{U}$ such that $o \in \boxesatzero{T}$;
\item $\dom{\target{\bigoplus \mathcal{U}}} = \bigcup_{T \in \mathcal{U}} \dom{\target{T}}$ and $\target{\bigoplus \mathcal{U}}(p) = \target{T}(p)$ for any $p \in \dom{\target{\bigoplus \mathcal{U}}}$ and any $T \in \mathcal{U}$ such that $p \in \dom{\target{T}}$.
\end{itemize}
\end{defi}

\begin{rem}\label{rem: sum <=i}
If $\mathcal{U}$ is gluable, then $(\bigoplus \mathcal{U})^{\leq i} = \bigoplus \{ U^{\leq i} ; U \in \mathcal{U} \}$.
\end{rem}

We can add wires:

\begin{defi}
Let $S$ be a differential in-PS. Let $\mathcal{W} \subseteq \conclusions{S}$ and $\mathcal{W'} \subseteq \exponentialportsatzero{S} \setminus \boxesatzero{S}$ such that $(\forall p \in \mathcal{W} \cap \portsatzero{S}) (\forall p' \in \mathcal{W'}) \neg p \leq_{\groundof{S}} p'$. Let $t$ be a function $\mathcal{W} \to \mathcal{W'}$. Then we denote by $S@t$ the differential in-PS such that
\begin{itemize}
\item $\targetAnyPort{S@t}$ is the extension of $\targetAnyPort{S}$ such that $\dom{\targetAnyPort{S@t}} = \dom{\targetAnyPort{S}} \cup \mathcal{W}$ and $(\forall p \in \mathcal{W}) \targetAnyPort{S@t}(p) = t(p)$
\item and $\portsatzero{S@t} = \portsatzero{S}$, $\labelofport{\groundof{S@t}} = \labelofport{\groundof{S}}$, $\leftwiresatzero{S@t} = \leftwiresatzero{S}$, $\axiomsatzero{S@t} = \axiomsatzero{S}$, $\cutsatzero{S@t} = \cutsatzero{S}$, $\boxesatzero{S@t} = \boxesatzero{S}$, $B_{S@t} = B_S$.
\end{itemize}
\end{defi}

\begin{rem}\label{rem: @<=i}
We have $(S@t)^{\leq i} = (S^{\leq i})@\restriction{t}{\portsatdepthleq{S}{i}}$.
\end{rem}

We can remove shallow conclusions:

\begin{defi}
Let $T$ be a differential in-PS such that $\conclusions{\groundof{T}} \subseteq \exponentialportsatzero{T} \setminus \boxesatzero{T}$. Then $\overline{T}$ is the unique differential in-PS such that 
\begin{itemize}
\item $\portsatzero{\overline{T}} = \portsatzero{T} \setminus \conclusions{\groundof{T}}$
\item and $T = (\overline{T} \oplus \conclusions{\groundof{T}})@t$, where $t$ is the function that associates with every $p \in \dom{\targetAnyPort{T}}$ such that $\targetAnyPort{T}(p) \in \conclusions{\groundof{T}}$ the port $\targetAnyPort{T}(p)$.
\end{itemize}
If $\mathcal{T}$ is a set of differential in-PS's, then $\overline{\mathcal{T}} = \{ \overline{T} ; T \in \mathcal{T} \}$.
\end{defi}

This operation allows to define the following variant of the notion of isomorphism of differential in-PS's:

\begin{defi}\label{definition: equiv modulo}
Let $S$ and $U$ be two differential in-PS's. Let $o \in \boxesatzero{S}$ such that $\conclusions{U} \subseteq \temporaryConclusions{S}{o}$. Let $T$ be a differential in-PS such that $\conclusions{\groundof{T}} \subseteq \exponentialportsatzero{T} \setminus \boxesatzero{T}$. Then we write $\varphi: U \equiv_{(S, o)} T$ if $\varphi : U \simeq \overline{T}$ such that, for any $p \in \conclusions{U}$, we have $\target{S}(o, p) = \targetAnyPort{T}(\varphi(p))$. We write $U \equiv_{(S, o)} T$ if there exists $\varphi$ such that $\varphi: U \equiv_{(S, o)} T$.
\end{defi}

\begin{rem}\label{rem: iso modulo}
Notice that:
\begin{itemize}
\item While the relations $\simeq$ and $\equiv$ are symmetric, the relation $\equiv_{(S, o)}$ is \emph{not} symmetric.
\item If $\varphi: U \equiv_{(S, o)} T$ and $\psi: U' \equiv U$, then $\varphi \circ \psi: U' \equiv_{(S, o)} T$.
\item If $\varphi_1: U \equiv_{(S, o)} T_1$ and $\varphi_2: U \equiv_{(S, o)} T_2$, then $\psi: T_1 \equiv T_2$, where $\psi$ is the bijection $\ports{T_1} \simeq \ports{T_2}$ defined by $\psi(p) = \left\lbrace \begin{array}{ll} \varphi_2({\varphi_1}^{-1}(p)) & \textit{if $p \notin \conclusions{\groundof{T_1}}$;}\\ p & \textit{otherwise.} \end{array}\right.$
\end{itemize}
\end{rem}

The following operation consists in adding contractions as shallow conclusions to the content of some box $o$; the conclusions of the box $o$ that were contracted at depth $0$ are now contracted inside the box $o$. We can then define a complexity measure on in-PS's that decreases with this operation, which allows to prove Proposition~\ref{prop: crucial} by induction on this complexity measure.

\begin{defi}\label{definition: adding contractions}
Let $R$ and $R_o$  be two in-PS's. Let $o \in \boxesatzero{R} \cap \boxesatzero{R_o}$. Let $\mathcal{Q'} \subseteq \portsatzerooftype{\contr}{B_{R_o}(o)}$ and let $\varphi$ be a bijection $\contractionsUnder{R}{o} \simeq \mathcal{Q'}$. We say that \emph{$R_o$ is obtained from $R$ by adding, according to $\varphi$, contractions as shallow conclusions to the content of the box $o$} and we write $R_o = \varphi \cdot_o R$ if the following properties hold:
\begin{itemize}
\item ${R_o}^{\leq 0} = R^{\leq 0}$;
\item $\boxesatzero{R_o} = \boxesatzero{R}$;
\item $B_{R_o}(o') = \left\lbrace \begin{array}{ll} B_R(o') & \textit{if $o' \not= o$;}\\ R'_o & \textit{if $o' = o$;} \end{array} \right.$ with ${R'_o} = ({B_R(o)} \oplus \bigoplus_{q' \in \mathcal{Q'}} \contr_{q'})@t$, where $\mathcal{Q'}$ is a disjoint set from $\portsatzero{B_R(o)}$ and $t$ is the function $\temporaryConclusions{R}{o} \setminus \{ \aboveBang{R}{o} \} \to \mathcal{Q'}$ that associates with every $p \in \temporaryConclusions{R}{o} \setminus \{ \aboveBang{R}{o} \}$ the port $\varphi(\target{R}(o, p))$;
\item and $\target{R_o} = \restriction{\target{R}}{\dom{\target{R}} \setminus (\{ o \} \times (\temporaryConclusions{R}{o} \setminus \{ \aboveBang{R}{o} \}))}$.
\end{itemize}
\end{defi}

\begin{rem}\label{rem: action on PS}
For any $o' \in \boxesatzero{B_R(o)}$, we have 
\begin{itemize}
\item $\temporaryConclusions{{B_{R}(o)}}{o'} = $ $\{ p \in \temporaryConclusions{{B_{(\varphi \cdot_o R)}(o)}}{o'};$ $\target{B_{(\varphi \cdot_o R)}(o)}(o', p) \notin \mathcal{Q'} \}$
\item and $(\forall p \in \temporaryConclusions{{B_{R}(o)}}{o'}) \target{B_R(o)}(o', p) = \target{B_{(\varphi \cdot_o R)}(o)}(o', p)$.
\end{itemize}
Moreover, for any $q \in \temporaryConclusions{R}{o}$ such that $\target{B_{(\varphi \cdot_o R)}(o)}(q) \in \mathcal{Q'}$, we have $\target{B_{(\varphi \cdot_o R)}(o)}(q) = \varphi(\target{R}(o, q))$. 
\end{rem}

Notice that, if $R$ is a PS, then $B_{(\varphi \cdot_o R)}(o)$ is a PS too (while $B_R(o)$ is not necessarily a PS); we will implicitly use this property in the proofs of Lemma~\ref{lem: ground-structure} and Proposition~\ref{prop: boxes}.

\begin{exa}
The in-PS $\varphi \cdot_{o_3} R$, where $\varphi$ is some bijection $\contractionsUnder{R}{o_3} \simeq \mathcal{Q'}$, is depicted in Figure~\ref{figure: R_o3}.
\begin{figure}
\resizebox{\textwidth}{!}{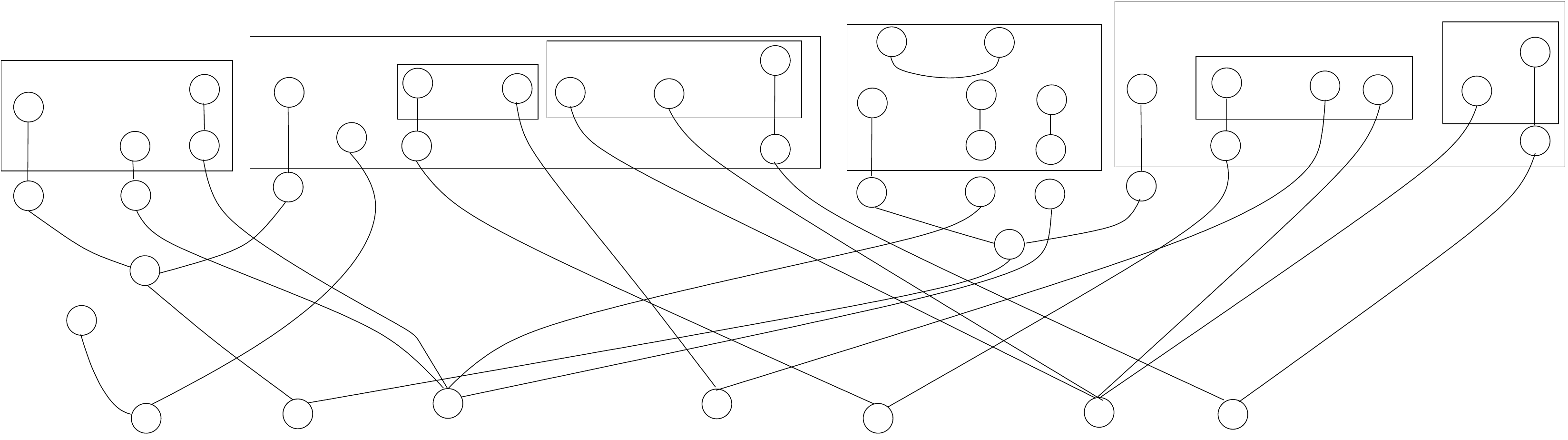}
\caption{The in-PS $\varphi \cdot_{o_3} R$}\label{figure: R_o3}
\end{figure}
\end{exa}

\begin{exa}
The in-PS $\varphi' \cdot_{o_4} R$, where $\varphi'$ is some bijection $\contractionsUnder{R}{o_4} \simeq \mathcal{Q'}$, is depicted in Figure~\ref{figure: R_o4}.
\begin{figure}
\resizebox{\textwidth}{!}{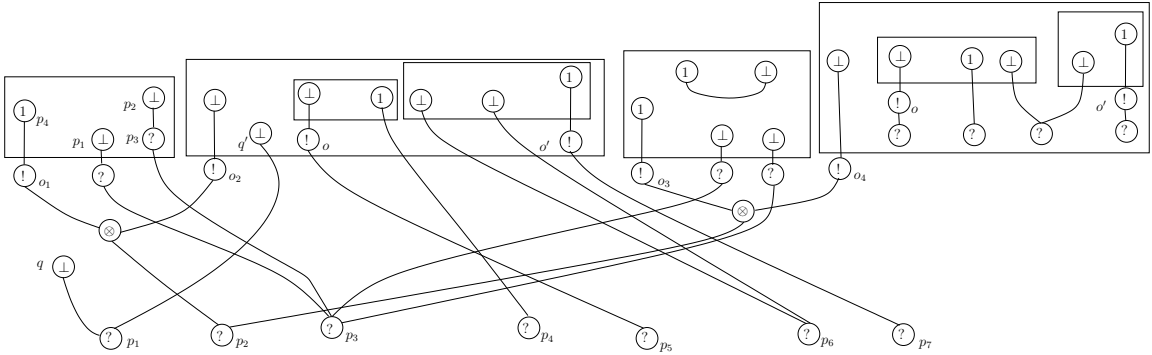}
\caption{The in-PS $\varphi' \cdot_{o_4} R$}\label{figure: R_o4}
\end{figure}
\end{exa}

\section{Taylor expansion}\label{section: Taylor expansion}

When Jean-Yves Girard introduced proof-nets in \cite{ll}, he also introduced \emph{experiments on proof-nets}. Experiments (see our Definition~\ref{defin: experiment} in Section~\ref{section: injectivity}) are a technology allowing to compute pointwisely the interpretation $\llbracket R \rrbracket$ of a proof-net $R$ in the model directly on the proof-net rather than through some sequent calculus proof obtained from one of its sequentializations: the set of \emph{results} of all the experiments on a given proof-net is its interpretation $\llbracket R \rrbracket$. In an untyped framework, experiments correspond to derivations of intersection types and results correspond to intersection types. 

Inspired by this notion, we introduce \emph{pseudo-experiments}. 

\begin{defi}\label{defin: pseudo-experiments}
For any differential in-PS $R$, we define, by induction on $\textit{depth}(R)$, the set $\pseudoExperiments{R}$ of \emph{pseudo-experiments on $R$}: it is the set of functions that associate with every $o \in \boxesatzero{R}$ a finite set of pseudo-experiments on $B_R(o)$ and with $\emptysequence$ some $m \in \Nat$.

Given a pseudo-experiment $e$ on a differential in-PS $R$, we define, by induction on $\textit{depth}(R)$, the function $e^\#: \boxes{R} \rightarrow \finitesubsets{\Nat}$ as follows: for any $o \in \boxesatzero{R}$, $e^\#(o) = \{ \Card{e(o)} \}$ and, for any $o' \in \boxes{B_R(o)}$, $e^\#(o, o') = \bigcup_{e_o \in e(o)} {e_o}^\#(o')$.
\end{defi}

Our definition is quite \emph{ad hoc}; actually, what we have in mind is the following notion of \emph{canonical pseudo-experiment}: A canonical pseudo-experiment on a differential in-PS $R$ is a function that associates with every $o \in \boxesatzero{R}$ a finite multiset of canonical pseudo-experiments on $B_R(o)$. Thus a canonical pseudo-experiment is an experiment without any labels on the axioms and any constraints on cuts. If we consider pseudo-experiments instead of canonical pseudo-experiments, it is only in order to be less verbose. For instance, if we used canonical pseudo-experiments instead of pseudo-experiments, the differential in-PS $S$ of Definition~\ref{defin: Taylor} could be defined by setting $$S = R^{\leq i} \oplus \bigoplus_{o \in \boxesatzerogeq{R}{i}} \bigoplus_{e_o \in \supp{e(o)}} \bigoplus_{z \in \{ 1, \ldots, e(e_o) \}} \langle o, \langle (e_o, z), \termofTaylor{B_R(o)}{e_o}{i} \rangle \rangle$$ instead of $S = R^{\leq i} \oplus \bigoplus_{o \in \boxesatzerogeq{R}{i}} \bigoplus_{e_o \in e(o)} \langle o, \langle e_o, \termofTaylor{B_R(o)}{e_o}{i} \rangle \rangle$. 

\begin{exa}\label{example: pseudo-experiment}
There exists a pseudo-experiment $e$ on the PS $R$ of Figure~\ref{fig: new_example} such that $e^\#(o_1) = \{ 10^{223} \}$, $e^\#(o_2) = \{ 10 \}$, $e^\#(o_3) = \{ 10^{224} \}$, $e^\#(o_4) = \{ 100 \}$, $e^\#((o_2, o))$ $=$ $\{ 10^3, \ldots, 10^{12} \}$, $e^\#((o_2, o'))$ $=$ $\{ 10^{13}, \ldots, 10^{22} \}$, $e^\#((o_4, o))$ $=$ $\{ 10^{23}, \ldots, 10^{122} \}$ and $e^\#((o_4, o'))$ $=$ $\{ 10^{123}, \ldots, 10^{222} \}$.
\end{exa}

For defining Taylor expansion, we only need to define $\termofTaylor{R}{e}{i}$ with $i = 0$, which is the differential in-PS obtained by (fully) expanding the boxes according to $e$. But a key tool of the proof is the introduction of the differential in-PS's $\termofTaylor{R}{e}{i}$ with $i > 0$, which are the differential in-PS's obtained by expanding only the boxes of depth at least $i$. We recall that the notation $\langle o, R \rangle$ for any differential in-PS $R$ was introduced in Definition~\ref{defin: special iso}.

\begin{defi}\label{defin: Taylor}
Let $R$ be an in-PS of depth $d$. Let $e$ be a pseudo-experiment on $R$. Let $i \in \Nat$. 
We define, by induction on $d$, a differential in-PS $\termofTaylor{R}{e}{i}$ of depth $\min \{ i, d \}$ and a function $\isaCopyfrom{R}{e}{i}: \ports{\termofTaylor{R}{e}{i}} \to \ports{R}$ as follows: We set $\termofTaylor{R}{e}{i} = S@t$, where 
$S = R^{\leq i} \oplus \bigoplus_{o \in \boxesatzerogeq{R}{i}} \bigoplus_{e_o \in e(o)} \langle o, \langle e_o, \termofTaylor{B_R(o)}{e_o}{i} \rangle \rangle$ 
and $t$ is the function $\mathcal{W}_0 \cup \mathcal{W}_{> 0} \to \exponentialportsatzero{R}$ with
\begin{itemize}
\item $\mathcal{W}_0 = \bigcup_{o \in \boxesatzerogeq{R}{i}} \bigcup_{e_o \in e(o)} \{ (o, (e_o, q)) ; (q \in \conclusions{\groundof{\termofTaylor{B_R(o)}{e_o}{i}}} \wedge \isaCopyfrom{B_R(o)}{e_o}{i}(q) \in \temporaryConclusions{R}{o}) \}$
\item $\mathcal{W}_{> 0} = \bigcup_{(o, (e_o, o')) \in \boxesatzero{S} \setminus \boxesatzero{R^{\leq i}}} \left\lbrace ((o, (e_o, o')), q) ; \begin{array}{ll} ((o', q) \in \nonshallowConclusions{\termofTaylor{B_R(o)}{e_o}{i}}\\ \wedge \isaCopyfrom{B_R(o)}{e_o}{i}(o', q) \in \temporaryConclusions{R}{o}) \end{array} \right\rbrace$
\item and $t(p) = \left\lbrace \begin{array}{ll} \target{R}(o, \isaCopyfrom{B_R(o)}{e_o}{i}(q)) & \textit{if $p = (o, (e_o, q)) \in \mathcal{W}_0$;}\\ \target{R}(o, \isaCopyfrom{B_R(o)}{e_o}{i}(o', q)) & \textit{if $p = ((o, (e_o, o')), q) \in \mathcal{W}_{> 0}$.} \end{array} \right.$
\end{itemize}
For any $p \in \ports{\termofTaylor{R}{e}{i}}$, the port $\isaCopyfrom{R}{e}{i}(p)$ of $R$ is the following one: 
$$\left\lbrace \begin{array}{ll} p & \textit{if $p \in \ports{R^{\leq i}}$;}\\ (o, \isaCopyfrom{B_R(o)}{e_o}{i}(p')) &\textit{if $p = (o, (e_o, p')) \in \portsatzero{\termofTaylor{R}{e}{i}} \setminus \portsatzero{R}$;}\\
(o, \isaCopyfrom{B_R(o)}{e_o}{i}(o', p')) & \textit{if $p = ((o, (e_o, o')), p')$ and $(o, (e_o, o')) \in \boxesatzero{\termofTaylor{R}{e}{i}} \setminus \boxesatzero{R^{\leq i}}$.} \end{array} \right.$$
If $o \in \boxesatzerogeq{R}{i}$ and $e_o$ is a pseudo-experiment on $B_R(o)$, then we set
\begin{align*}
  R\langle o, i, e_o \rangle = \langle o, \langle e_o, \termofTaylor{B_R(o)}{e_o}{i} \rangle \rangle.
\end{align*}
\end{defi}

\begin{rem}
We have $\mathcal{W}_0 \subseteq \wiresatzero{\termofTaylor{R}{e}{i}}$ and $\dom{\target{\termofTaylor{R}{e}{i}}}$ is the set
$$\dom{\target{R^{\leq i}}} \cup \left(\bigcup_{o \in \boxesatzerogeq{R}{i}} \bigcup_{e_o \in e(o)} \bigcup_{o' \in \boxesatzero{\termofTaylor{B_R(o)}{e_o}{i}}} \{ ((o, (e_o, o')), p) ; p \in \temporaryConclusions{\termofTaylor{B_R(o)}{e_o}{i}}{o'} \}\right) \cup \mathcal{W}_{> 0}$$
Moreover, for any $p \in \dom{\target{R^{\leq i}}}$, we have $\target{\termofTaylor{R}{e}{i}}(p) = \target{R}(p)$.
\end{rem}

Notice that if $R$ is cut-free, then $\termofTaylor{R}{e}{i}$ is cut-free too. The conclusions are the duplications of the conclusions; in particular, if $R$ is a PS, then $\termofTaylor{R}{e}{i}$ is a differential PS:

\begin{fact}\label{fact: conclusions}
Let $R$ be an in-PS. Let $e$ be a pseudo-experiment on $R$. Let $i \in \Nat$. Then we have $(\forall q \in \ports{\termofTaylor{R}{e}{i}}) (\isaCopyfrom{R}{e}{i}(q) \in \conclusions{R} \Leftrightarrow q \in \conclusions{\termofTaylor{R}{e}{i}})$.
\end{fact}

\begin{proof}
By induction on $\depthof{R}$. Let $q \in \ports{\termofTaylor{R}{e}{i}}$ such that $\isaCopyfrom{R}{e}{i}(q) \in \conclusions{R}$. We distinguish between two cases:
\begin{itemize}
\item $\isaCopyfrom{R}{e}{i}(q) \in \conclusions{R^{\leq i}}$: we have $q = \isaCopyfrom{R}{e}{i}(q) \in \conclusions{R^{\leq i}} \subseteq \conclusions{\termofTaylor{R}{e}{i}}$;
\item There exist $o \in \boxesatzerogeq{R}{i}$ and $p \in \conclusions{B_R(o)} \setminus \temporaryConclusions{R}{o}$ such that $\isaCopyfrom{R}{e}{i}(q) = (o, p)$: Then
\begin{itemize}
\item either there exist $e_o \in e(o)$ and $q' \in \portsatzero{\termofTaylor{B_R(o)}{e_o}{i}}$ such that $q = (o, (e_o, q'))$ and $\isaCopyfrom{B_R(o)}{e_o}{i}(q') = p$, and then, by induction hypothesis, $q' \in \conclusions{\termofTaylor{B_R(o)}{e_o}{i}}$, hence $q = (o, (e_o, q')) \in \conclusions{R \langle o, i, e_o \rangle}$; since $\isaCopyfrom{B_R(o)}{e_o}{i}(q') \notin \temporaryConclusions{R}{o}$, we have $q \notin \dom{\target{\groundof{\termofTaylor{R}{e}{i}}}}$ and we obtain $q \in \conclusions{\termofTaylor{R}{e}{i}}$;
\item or there exist $e_o \in e(o)$, $o' \in \boxesatzero{\termofTaylor{B_R(o)}{e_o}{i}}$ and $q' \in \ports{B_{\termofTaylor{B_R(o)}{e_o}{i}}(o')}$ such that $q = ((o, (e_o, o')), q')$ and $\isaCopyfrom{B_R(o)}{e_o}{i}(o', q') = p$, and then, by induction hypothesis, $(o', q') \in \conclusions{\termofTaylor{B_R(o)}{e_o}{i}}$, hence $q = ((o, (e_o, o')), q') \in \conclusions{R \langle o, i, e_o \rangle}$; since $\isaCopyfrom{B_R(o)}{e_o}{i}(o', q') \notin \temporaryConclusions{R}{o}$, we have $q \notin \dom{\target{\termofTaylor{R}{e}{i}}}$ and we obtain $q \in \conclusions{\termofTaylor{R}{e}{i}}$.
\end{itemize}
\end{itemize}
Conversely, let $q \in \conclusions{\termofTaylor{R}{e}{i}}$. We distinguish between three cases:
\begin{itemize}
\item $q \in \conclusions{R^{\leq i}}$: we have $\isaCopyfrom{R}{e}{i}(q) = q \in \conclusions{R^{\leq i}} \subseteq \conclusions{\termofTaylor{R}{e}{i}}$;
\item There exist $o \in \boxesatzerogeq{R}{i}$, $e_o \in e(o)$ and $q' \in \conclusions{\groundof{\termofTaylor{B_R(o)}{e_o}{i}}}$ such that $\isaCopyfrom{B_R(o)}{e_o}{i}(q') \notin \temporaryConclusions{R}{o}$ and $q = (o, (e_o, q'))$: by induction hypothesis, we have $\isaCopyfrom{B_R(o)}{e_o}{i}(q') \in \conclusions{B_R(o)}$, hence $\isaCopyfrom{R}{e}{i}(q) = (o, \isaCopyfrom{B_R(o)}{e_o}{i}(q'))$ with $\isaCopyfrom{B_R(o)}{e_o}{i}(q') \in \conclusions{B_R(o)} \setminus \temporaryConclusions{R}{o}$; we thus have $\isaCopyfrom{R}{e}{i}(q) \in \conclusions{R}$.
\item There exist $o \in \boxesatzerogeq{R}{i}$, $e_o \in e(o)$, $o' \in \boxesatzero{\termofTaylor{B_R(o)}{e_o}{i}}$ and $q' \in \conclusions{B_{\termofTaylor{R}{e}{i}}(o, (e_o, o'))} \setminus \temporaryConclusions{\termofTaylor{R}{e}{i}}{(o, (e_o, o'))}$ such that $\isaCopyfrom{B_R(o)}{e_o}{i}(o', q') \notin \temporaryConclusions{R}{o}$ and $q = ((o, (e_o, o')), q')$: we have $q' \in \conclusions{B_{\termofTaylor{B_R(o)}{e_o}{i}}(o')} \setminus \temporaryConclusions{\termofTaylor{B_R(o)}{e_o}{i}}{o'}$, hence $(o', q') \in \conclusions{\termofTaylor{B_R(o)}{e_o}{i}}$; by induction hypothesis, we have $\isaCopyfrom{B_R(o)}{e_o}{i}(o', q') \in \conclusions{B_R(o)}$; since $\isaCopyfrom{B_R(o)}{e_o}{i}(o', q') \notin \temporaryConclusions{R}{o}$, we have $\isaCopyfrom{R}{e}{i}(q) = (o, \isaCopyfrom{B_R(o)}{e_o}{i}(o', q')) \in \conclusions{R}$.
\qedhere
\end{itemize}
\end{proof}

\begin{figure}
\centering
\resizebox{\textwidth}{4cm}{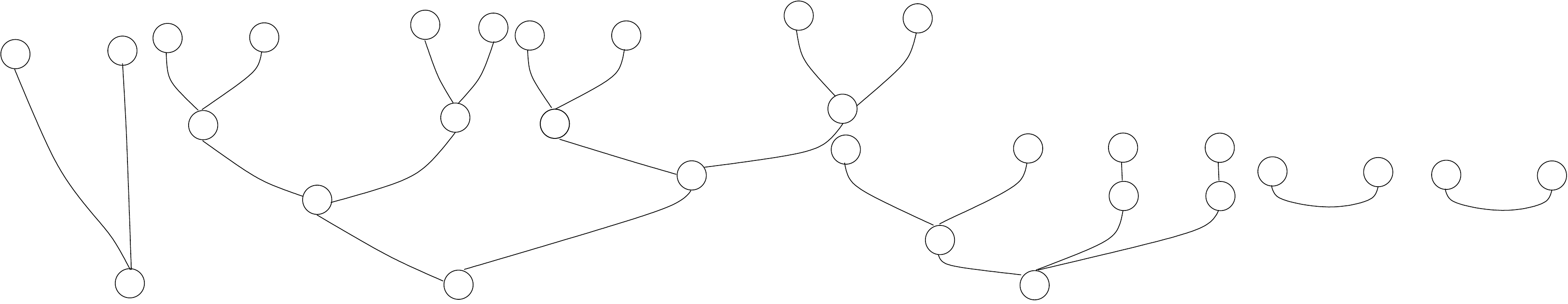}
\caption{The differential PS $S0_1$}
\label{fig: termofTaylor{R}{e}{0}: 1}
\centering
\resizebox{\textwidth}{!}{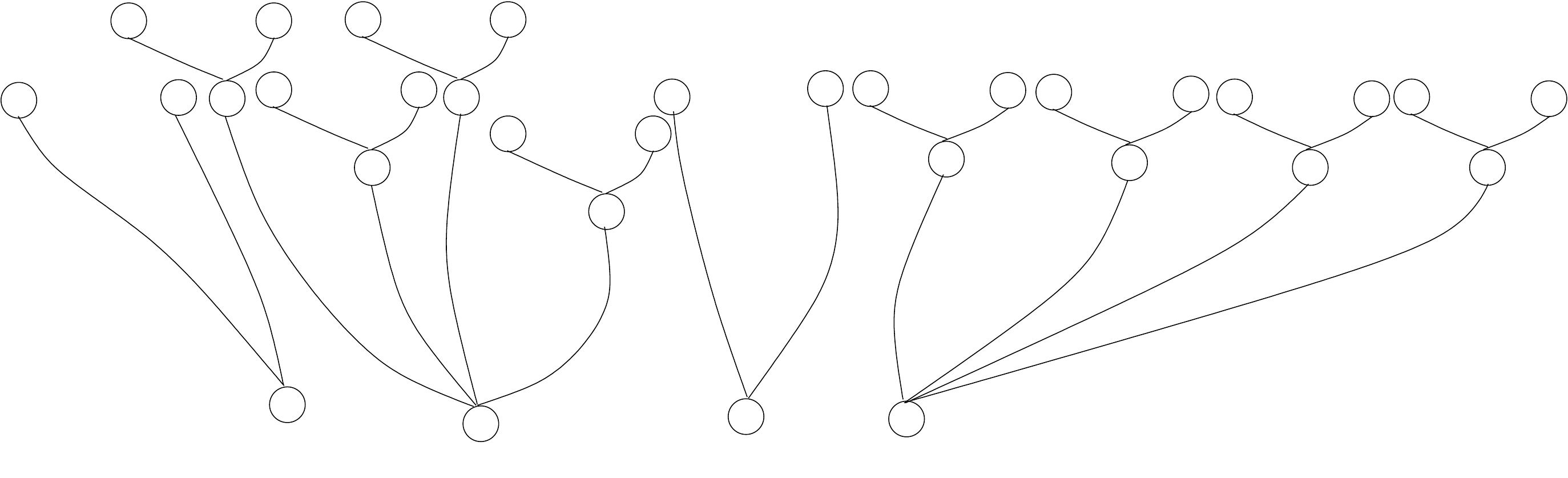}
\caption{The differential PS $S0_2$}
\label{fig: termofTaylor{R}{e}{0}: 2}
\end{figure}

\begin{exa}
Generally, given a PS $R$ and an integer $i$, the PS $\termofTaylor{R}{e}{i}$ does not depend only on $e^\#$ (we can have ${e_1}^\# = {e_2}^\#$ and not $\termofTaylor{R}{e_1}{i} \equiv \termofTaylor{R}{e_2}{i}$), but with the PS $R$ depicted in Figure~\ref{fig: new_example}, it is not the case: with \emph{this} PS $R$, we have $(\forall i \in \Nat) (\forall e_1, e_2 \in \pseudoExperiments{R}) ({e_1}^\# = {e_2}^\# \Rightarrow \termofTaylor{R}{e_1}{i} \equiv \termofTaylor{R}{e_2}{i})$. If $R$ is this PS and $e$ some pseudo-experiment such that $e^\#$ is as described in Example~\ref{example: pseudo-experiment}, then $\termofTaylor{R}{e}{0} = S0_1 \oplus S0_2$, where $S0_1$ and $S0_2$ are the differential PS's depicted in Figures~\ref{fig: termofTaylor{R}{e}{0}: 1} and~\ref{fig: termofTaylor{R}{e}{0}: 2} respectively, and $\termofTaylor{R}{e}{1} = S1_1 \oplus S1_2 \oplus S1_3$, where $S1_1$, $S1_2$ and $S1_3$ are the differential PS's depicted in Figures~\ref{fig: termofTaylor{R}{e}{1}: 1},~\ref{fig: termofTaylor{R}{e}{1}: 2} and~\ref{fig: termofTaylor{R}{e}{1}: 3} respectively.
\end{exa}

\begin{figure}
\centering
\begin{minipage}{0.3\textwidth}
\centering
\resizebox{.4 \textwidth}{!}{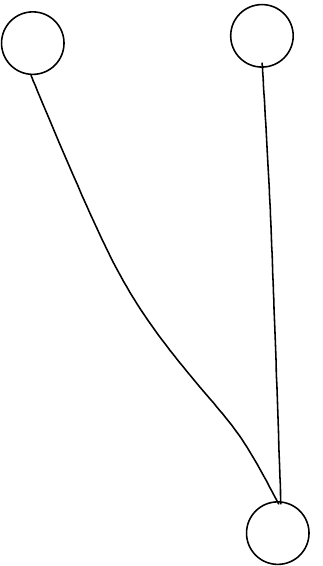}
\captionsetup{width=.9 \textwidth}
\caption{The differential PS $S1_1$}
\label{fig: termofTaylor{R}{e}{1}: 1}
\end{minipage}\hfill
\begin{minipage}{0.65 \textwidth}
\centering
\resizebox{\textwidth}{!}{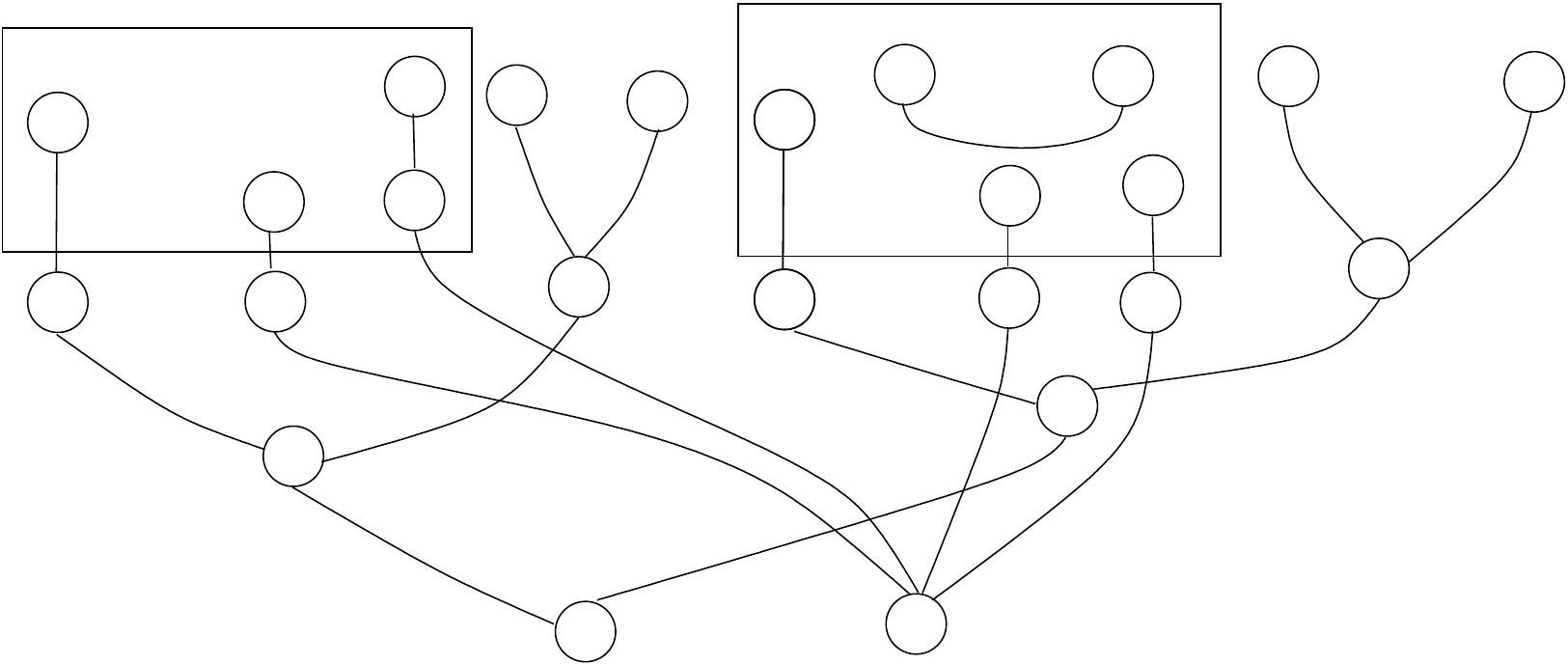}
\caption{The differential PS $S1_2$}
\label{fig: termofTaylor{R}{e}{1}: 2}
\end{minipage}\hfill
\centering
\resizebox{\textwidth}{!}{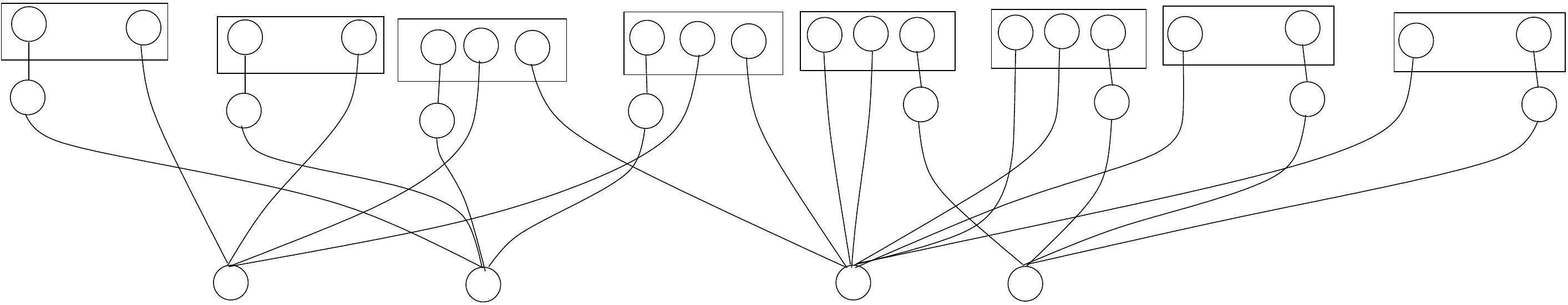}
\caption{The differential PS $S1_3$}
\label{fig: termofTaylor{R}{e}{1}: 3}
\end{figure}

The Taylor expansion we consider has no coefficients (i.e. has coefficients in the Boolean semiring $\mathbb{B} = \{ 0, 1 \}$, where $1+1 = 1$). In other words, we consider the support of Taylor expansion with coefficients:

\begin{defi}\label{defin: Taylor expansion}
Let $R$ be a PS. A \emph{Taylor expansion of $R$} is a set $\TaylorExpansion{R}$  of simple differential nets such that $\TaylorExpansion{R} \equiv \{ \termofTaylor{R}{e}{0} ; e \in \pseudoExperiments{R} \}$.
\end{defi}

It is clear that two Taylor expansions of $R$ are the same sets of simple differential nets up to the names of the ports that are not conclusions, that is why it makes sense to speak about ``the'' Taylor expansion of a PS.

An important case for our proof is the partial Taylor expansion of an in-PS of the form $\varphi \cdot_o R$ for some in-PS $R$, some box $o$ of $R$ at depth $0$ and some bijection $\varphi : \contractionsUnder{R}{o} \simeq \mathcal{Q'}$. Since we have $\pseudoExperiments{\varphi \cdot_o R} = \pseudoExperiments{R}$, we can compare $\termofTaylor{R}{e}{i}$ with $\termofTaylor{(\varphi \cdot_o R)}{e}{i}$. From $B_{B_{(\varphi \cdot_o R)}(o)} = B_{B_R(o)}$, we already deduce $B_{\termofTaylor{B_{(\varphi \cdot_o R)}(o)}{e}{i}} = B_{\termofTaylor{B_{R}(o)}{e}{i}}$; more information is given by the following lemma:

\begin{lem}\label{lem: termofTaylor of R_o}
Let $R$ be an in-PS. Let $o \in \boxesatzero{R}$. Let $e_o$ be a pseudo-experiment on $B_R(o)$. Let $\varphi$ be some bijection $\contractionsUnder{R}{o} \simeq \mathcal{Q'}$. Let $R_o$ be an in-PS such that $R_o = \varphi \cdot_o R$. Let $i \in \Nat$. Then we have:
\begin{itemize}
\item $\target{\termofTaylor{B_R(o)}{e_o}{i}} = \restriction{\target{\termofTaylor{B_{R_o}(o)}{e_o}{i}}}{\{ p \in \dom{\target{\termofTaylor{B_{R_o}(o)}{e_o}{i}}} ; \target{\termofTaylor{B_{R_o}(o)}{e_o}{i}}(p) \notin \mathcal{Q'} \}}$
\item and, for any $o_1 \in \boxesatzero{\termofTaylor{B_{R_o}(o)}{e_o}{i}}$, 
for any $p \in \temporaryConclusions{\termofTaylor{B_{R_o}(o)}{e_o}{i}}{o_1}$ such that
\begin{displaymath}
  \target{\termofTaylor{B_{R_o}(o)}{e_o}{i}}(o_1, p) \in \mathcal{Q'},
\end{displaymath}
we have $$\target{\termofTaylor{B_{R_o}(o)}{e_o}{i}}(o_1, p) = \varphi(\target{\termofTaylor{R}{e}{i}}((o, (e_o, o_1)), p))$$
\end{itemize}
\end{lem}

\begin{proof}
We set 
$$S = {B_{R}(o)}^{\leq i} \oplus \bigoplus_{o' \in \boxesatzerogeq{{B_{R}(o)}}{i}} \bigoplus_{e_{o'} \in e_o(o')} \langle o', \langle e_{o'}, \termofTaylor{B_{B_{R}(o)}(o')}{e_{o'}}{i} \rangle \rangle$$ 
and
$$S_o = {B_{R_o}(o)}^{\leq i} \oplus \bigoplus_{o' \in \boxesatzerogeq{{B_{R_o}(o)}}{i}} \bigoplus_{e_{o'} \in e_o(o')} \langle o', \langle e_{o'}, \termofTaylor{B_{B_{R_o}(o)}(o')}{e_{o'}}{i} \rangle \rangle$$ 
Notice that $$S_o = {B_{R_o}(o)}^{\leq i} \oplus \bigoplus_{o' \in \boxesatzerogeq{{B_{R}(o)}}{i}} \bigoplus_{e_{o'} \in e_o(o')} \langle o', \langle e_{o'}, \termofTaylor{B_{B_{R}(o)}(o')}{e_{o'}}{i} \rangle \rangle$$
and $\boxesatzero{S_o} = \boxesatzero{S}$. 
Let $p \in \dom{\target{\termofTaylor{B_R(o)}{e_o}{i}}}$ and let us show that $p \in \dom{\target{\termofTaylor{B_{R_o}(o)}{e_o}{i}}}$. We distinguish between two cases:
\begin{itemize}
\item $p \in \dom{\target{{B_{R}(o)}^{\leq i}}}$: Then $p \in \dom{\target{{B_{R_o}(o)}^{\leq i}}} \subseteq \dom{\target{\termofTaylor{B_{R_o}(o)}{e_o}{i}}}$.
\item $p \notin \dom{\target{{B_{R}(o)}^{\leq i}}}$: 
Then there exists $(o', (e_{o'}, o'')) \in \boxesatzero{S} \setminus \boxesatzero{{B_{R}(o)}^{\leq i}} = \boxesatzero{S_o} \setminus \boxesatzero{{B_{R_o}(o)}^{\leq i}}$ such that $p$ is an element of the set
\begin{eqnarray*}
& & \left\lbrace ((o', (e_{o'}, o'')), q) ; \begin{array}{ll} ((o'', q) \in \nonshallowConclusions{\termofTaylor{B_{B_R(o)}(o')}{e_{o'}}{i}}\\ \wedge \isaCopyfrom{B_{B_R(o)}(o')}{e_{o'}}{i}(o'', q) \in \temporaryConclusions{B_R(o)}{o'}) \end{array} \right\rbrace\\
& = & \left\lbrace ((o', (e_{o'}, o'')), q) ; \begin{array}{ll} ((o'', q) \in \nonshallowConclusions{\termofTaylor{B_{B_{R_o}(o)}(o')}{e_{o'}}{i}}\\ \wedge \isaCopyfrom{B_{B_{R_o}(o)}(o')}{e_{o'}}{i}(o'', q) \in \temporaryConclusions{B_R(o)}{o'}) \end{array} \right\rbrace\allowdisplaybreaks\\
& \subseteq & \left\lbrace ((o', (e_{o'}, o'')), q) ; \begin{array}{ll} ((o'', q) \in  \nonshallowConclusions{\termofTaylor{B_{B_{R_o}(o)}(o')}{e_{o'}}{i}}\\ \wedge \isaCopyfrom{B_{B_{R_o}(o)}(o')}{e_{o'}}{i}(o'', q) \in \temporaryConclusions{B_{R_o}(o)}{o'}) \end{array} \right\rbrace\\
& \subseteq & \dom{\target{\termofTaylor{B_{R_o}(o)}{e_o}{i}}}
\end{eqnarray*}
\end{itemize}
Conversely, let $p \in \dom{\target{\termofTaylor{B_{R_o}(o)}{e_o}{i}}}$ such that $\target{\termofTaylor{B_{R_o}(o)}{e_o}{i}}(p) \notin \mathcal{Q'}$; we will check that $p \in \dom{\target{\termofTaylor{B_{R}(o)}{e_o}{i}}}$ and $\target{\termofTaylor{B_{R_o}(o)}{e_o}{i}}(p) = \target{\termofTaylor{B_{R}(o)}{e_o}{i}}(p)$. We distinguish between two cases:
\begin{itemize}
\item $p \in \dom{\target{{B_{R_o}(o)}^{\leq i}}}$: Then $p \in \dom{\target{{({B_R(o)} \oplus \bigoplus_{q' \in \mathcal{Q'}} \contr_{q'})}^{\leq i}}} = \dom{\target{{{B_R(o)}}^{\leq i}}} \subseteq \dom{\target{\termofTaylor{B_R(o)}{e_o}{i}}}$. In this case, we have $\target{\termofTaylor{B_{R}(o)}{e_o}{i}}(p) = \target{B_{R}(o)}(p) = \target{B_{R_o}(o)}(p) = \target{\termofTaylor{B_{R_o}(o)}{e_o}{i}}(p)$.
\item $p \notin \dom{\target{{B_{R_o}(o)}^{\leq i}}}$: Then there exists $(o', (e_{o'}, o'')) \in \boxesatzero{S_o} \setminus \boxesatzero{{B_{R_o}(o)}^{\leq i}} = \boxesatzero{S} \setminus \boxesatzero{{B_{R}(o)}^{\leq i}}$ such that $p$ is an element of the set
\begin{eqnarray*}
& & \left\lbrace ((o', (e_{o'}, o'')), q) ; \begin{array}{ll} ((o'', q) \in \nonshallowConclusions{\termofTaylor{B_{B_{R_o}(o)}(o')}{e_{o'}}{i}}\\ \wedge \isaCopyfrom{B_{B_{R_o}(o)}(o')}{e_{o'}}{i}(o'', q) \in \temporaryConclusions{{B_{R_o}(o)}}{o'}) \end{array} \right\rbrace \\
& = & \left\lbrace ((o', (e_{o'}, o'')), q) ; \begin{array}{ll} ((o'', q) \in \nonshallowConclusions{\termofTaylor{B_{B_{R}(o)}(o')}{e_{o'}}{i}}\\ \wedge \isaCopyfrom{B_{B_{R}(o)}(o')}{e_{o'}}{i}(o'', q) \in \temporaryConclusions{{B_{R_o}(o)}}{o'}) \end{array} \right\rbrace\\
& & \textit{(because $B_{B_{R_o}(o)}(o') = B_{B_{R}(o)}(o')$)}\allowdisplaybreaks\\
& = & \left\lbrace ((o', (e_{o'}, o'')), q) ; \begin{array}{ll} ((o'', q) \in \nonshallowConclusions{\termofTaylor{B_{B_{R}(o)}(o')}{e_{o'}}{i}}\\ \wedge \isaCopyfrom{B_{B_{R}(o)}(o')}{e_{o'}}{i}(o'', q) \in \temporaryConclusions{{B_{R}(o)}}{o'}) \end{array} \right\rbrace\\
& \cup & \left\lbrace ((o', (e_{o'}, o'')), q) ; \begin{array}{ll} ((o'', q) \in \nonshallowConclusions{\termofTaylor{B_{B_{R}(o)}(o')}{e_{o'}}{i}}\\ \wedge \isaCopyfrom{B_{B_{R}(o)}(o')}{e_{o'}}{i}(o'', q) \in \temporaryConclusions{{B_{R_o}(o)}}{o'}\\ \wedge \target{B_{R_o}(o)}(o', \isaCopyfrom{B_{B_{R}(o)}(o')}{e_{o'}}{i}(o'', q)) \in \mathcal{Q'}
) \end{array}\right\rbrace \\
& & \textit{(by Remark~\ref{rem: action on PS}, $\temporaryConclusions{B_{R_o}(o)}{o'} = \temporaryConclusions{B_{R}(o)}{o'} \cup \{ q' \in \temporaryConclusions{B_{R_o}(o)}{o'}; \target{B_{R_o}(o)}(o', q') \in \mathcal{Q'} \}$)}\allowdisplaybreaks\\
& = & \left\lbrace ((o', (e_{o'}, o'')), q) ; \begin{array}{ll} ((o'', q) \in \nonshallowConclusions{\termofTaylor{B_{B_{R}(o)}(o')}{e_{o'}}{i}}\\ \wedge \isaCopyfrom{B_{B_{R}(o)}(o')}{e_{o'}}{i}(o'', q) \in \temporaryConclusions{{B_{R}(o)}}{o'}) \end{array} \right\rbrace\\
& \cup & \left\lbrace ((o', (e_{o'}, o'')), q) ; \begin{array}{ll} ((o'', q) \in \nonshallowConclusions{\termofTaylor{B_{B_{R}(o)}(o')}{e_{o'}}{i}}\\ \wedge 
\target{\termofTaylor{B_{R_o}(o)}{e_o}{i}}((o', (e_{o'}, o'')), q)
 \in \mathcal{Q'}
) \end{array} \right\rbrace\allowdisplaybreaks\\
& = & \left\lbrace ((o', (e_{o'}, o'')), q) ; \begin{array}{ll} ((o'', q) \in \nonshallowConclusions{\termofTaylor{B_{B_{R}(o)}(o')}{e_{o'}}{i}}\\ \wedge \isaCopyfrom{B_{B_{R}(o)}(o')}{e_{o'}}{i}(o'', q) \in \temporaryConclusions{{B_{R}(o)}}{o'}) \end{array} \right\rbrace\\
& & \textit{(by assumption, we have $\target{\termofTaylor{B_{R_o}(o)}{e_o}{i}}(p) \notin \mathcal{Q'}$)}\\
& \subseteq & \dom{\target{\termofTaylor{B_{R}(o)}{e_o}{i}}}
\end{eqnarray*}

In this case, if $q$ is such that $p = ((o', (e_{o'}, o'')), q)$, then we have 
\begin{eqnarray*}
\target{\termofTaylor{B_{R}(o)}{e_o}{i}}(p) & = & \target{B_R(o)}(o', \isaCopyfrom{B_{B_R(o)}(o')}{e_{o'}}{i}(o'', q))\\
& = & \target{B_{R_o}(o)}(o', \isaCopyfrom{B_{B_{R_o}(o)}(o')}{e_{o'}}{i}(o'', q))\\
& = & \target{\termofTaylor{B_{R_o}(o)}{e_o}{i}}(p)
\end{eqnarray*}
\end{itemize}
We thus have $\target{\termofTaylor{B_R(o)}{e_o}{i}} = \restriction{\target{\termofTaylor{B_{R_o}(o)}{e_o}{i}}}{\{ p \in \dom{\target{\termofTaylor{B_{R_o}(o)}{e_o}{i}}} ; \target{\termofTaylor{B_{R_o}(o)}{e_o}{i}}(p) \notin \mathcal{Q'} \}}$. 

Now, let $o_1 \in \boxesatzero{\termofTaylor{B_{R_o}(o)}{e_o}{i}}$ and $p \in \temporaryConclusions{\termofTaylor{B_{R_o}(o)}{e_o}{i}}{o_1}$ such that $\target{\termofTaylor{B_{R_o}(o)}{e_o}{i}}(o_1, p) \in \mathcal{Q'}$.
\begin{itemize}
\item If $o_1 \in \boxesatzerosmaller{B_{R_o}(o)}{i}$, then $(o_1, p) \in \temporaryConclusions{R}{o}$, hence $\target{\termofTaylor{B_{R_o}(o)}{e_o}{i}}(o_1, p) = \target{B_{R_o}(o)}(o_1, p) = \varphi(\target{R}(o, (o_1, p))) = \varphi(\target{R}(o, \isaCopyfrom{B_{R_o}(o)}{e_o}{i}(o_1, p)))$;
\item if $o_1 = (o', (e_{o'}, o''))$ with $o' \in \boxesatzerogeq{B_{R_o}(o)}{i}$, then $((o', (e_{o'}, o'')), p) \notin \conclusions{\termofTaylor{B_{R_o}(o)}{e_o}{i}}$, hence, by Fact~\ref{fact: conclusions}, $(o', \isaCopyfrom{B_{B_{R_o}(o)}(o')}{e_{o'}}{i}(o'', p)) = \isaCopyfrom{B_{R_o}(o)}{e_{o}}{i}((o', (e_{o'}, o'')), p) \notin \conclusions{B_{R_o}(o)}$. Moreover, since $\target{\termofTaylor{B_{R_o}(o)}{e_o}{i}}(o_1, p) \in \mathcal{Q'}$, we have $(o'', p) \in \conclusions{\termofTaylor{B_{B_{R_o}(o)}(o')}{e_{o'}}{i}}$ (otherwise, $\target{\termofTaylor{B_{B_{R_o}(o)}(o')}{e_{o'}}{i}}(o'', p) \in \portsatzero{\termofTaylor{B_{B_{R_o}(o)}(o')}{e_{o'}}{i}} = \portsatzero{\termofTaylor{B_{B_{R}(o)}(o')}{e_{o'}}{i}}$, which entails $\target{\termofTaylor{B_{R_o}(o)}{e_o}{i}}(o_1, p) 
\in \portsatzero{\termofTaylor{{B_{R}(o)}}{e_{o}}{i}}$, which contradicts $\target{\termofTaylor{B_{R_o}(o)}{e_o}{i}}(o_1, p) \in \mathcal{Q'}$). By Fact~\ref{fact: conclusions}, we obtain $\isaCopyfrom{B_{B_{R_o}(o)}(o')}{e_{o'}}{i}(o'', p) \in \conclusions{B_{B_{R_o}(o)}(o')}$. We showed $\isaCopyfrom{B_{B_{R_o}(o)}(o')}{e_{o'}}{i}(o'', p) \in \temporaryConclusions{B_{R_o}(o)}{o'}$. We thus have 
\begin{eqnarray*}
\target{\termofTaylor{B_{R_o}(o)}{e_o}{i}}(o_1, p) & = & \target{B_{R_o}(o)}(o', \isaCopyfrom{B_{B_{R_o}(o)}(o')}{e_{o'}}{i}(o'', p))\\
& = & \varphi(\target{R}(o, (o', \isaCopyfrom{B_{B_{R_o}(o)}(o')}{e_{o'}}{i}(o'', p))))\\
& = &\varphi(\target{R}(o, \isaCopyfrom{B_{R_o}(o)}{e_o}{i}(o_1, p)));
\end{eqnarray*}
\end{itemize}
now, we have $\target{R}(o, \isaCopyfrom{B_{R_o}(o)}{e_o}{i}(o_1, p)) = \target{R}(o, \isaCopyfrom{B_{R}(o)}{e_o}{i}(o_1, p)) = \target{\termofTaylor{R}{e}{i}}((o, (e_o, o_1)), p)$.
\end{proof}

The rest of this section is devoted to show Proposition~\ref{prop: arity}, which shows how one can compute, for any in-PS $R$, the arity in $\termofTaylor{R}{e}{i}$ of a port at depth $0$ of $R$. For that purpose, we introduce the function $b_S^{\geq i}$ that associates with every port $p$ of $S$ at depth greater than $i$ the deepest box of depth at least $i$ that contains $p$:

\begin{defi}
For any differential in-PS $S$, for any $i \in \Nat$, we define, by induction on $\depthof{S}$, the function $b_S^{\geq i} : \portsatdepthgreater{S}{i} \to \boxesgeq{S}{i}$ as follows: 
$$\begin{array}{ccl}
b_S^{\geq i}: \portsatdepthgreater{S}{i} & \to & \boxesgeq{S}{i}\\
(o, p) & \mapsto & \left\lbrace \begin{array}{ll} o & \textit{if $o \in \boxesatzerogeq{S}{i}$ and $p \in \portsatdepthleq{B_S(o)}{i}$;}\\ (o, b_{B_S(o)}^{\geq i}(p)) & \textit{if $o \in \boxesatzerogeq{S}{i}$ and $p \in \portsatdepthgreater{B_S(o)}{i}$;} \end{array} \right.
\end{array}$$
\end{defi}

\begin{fact}\label{fact: characterization of ports at depth 0}
Let $R$ be an in-PS. Let $e$ be a pseudo-experiment on $R$. Let $i \in \Nat$. Then we have 
$$(\forall p \in \ports{\termofTaylor{R}{e}{i}} \setminus \portsatzero{R}) (p \in \portsatzero{\termofTaylor{R}{e}{i}} \Leftrightarrow (b_R^{\geq 0} \circ \isaCopyfrom{R}{i}{e})(p) \in \boxesgeq{R}{i})$$
\end{fact}

\begin{proof}
First notice that, for any differential in-PS $S$, we have:
$$(\forall p \in \portsatdepthgreater{S}{0}) (b_S^{\geq 0}(p) \in \boxesgeq{S}{i} \Rightarrow p \in \portsatdepthgreater{S}{i}) \: \: \: (\ast)$$
Indeed: Let $p \in \portsatdepthgreater{S}{0}$ such that $b_S^{\geq 0}(p) \in \boxesgeq{S}{i}$ and let $o \in \boxesatzero{S}$ and $p' \in \ports{B_S(o)}$ such that $p = (o, p')$; we distinguish between two cases:
\begin{itemize}
\item $p' \in \portsatzero{B_S(o)}$: We have $b_S^{\geq 0}(p) = o \in \boxesatzerogeq{S}{i}$;
\item $p' \in \portsatdepthgreater{B_S(o)}{0}$: We have $b_S^{\geq 0}(p) = (o, b_{B_S(o)}^{\geq 0}(p')) \in \boxesatzerogeq{S}{i}$, hence $o \in \boxesgeq{S}{i}$;
\end{itemize}
in both cases we have $o \in \boxesatzerogeq{S}{i}$, hence $p = (o, p') \in \portsatdepthgreater{S}{i}$.

We prove now the fact by induction on $\depthof{R}$. If $\depthof{R} = 0$, then $\ports{\termofTaylor{R}{e}{i}} \setminus \portsatzero{R} = \emptyset$. Otherwise, let $p \in \ports{\termofTaylor{R}{e}{i}} \setminus \portsatzero{R} \not= \emptyset$:
\begin{itemize}
\item If $p \in \portsatzero{\termofTaylor{R}{e}{i}}$, then there exist $o \in \boxesatzerogeq{R}{i}$, $e_o \in e(o)$ and $p' \in \portsatzero{\termofTaylor{B_R(o)}{e_o}{i}}$ such that $p = (o, (e_o, p'))$; we distinguish between two cases:
\begin{itemize}
\item $p' \in \portsatzero{B_R(o)}$: We have $b_R^{\geq 0}(\isaCopyfrom{R}{e}{i}(p)) = b_R^{\geq 0}(o, p') = o \in \boxesgeq{R}{i}$.
\item $p' \in \portsatdepthgreater{B_R(o)}{0}$: By induction hypothesis, we have $b_{B_R(o)}^{\geq 0}(\isaCopyfrom{B_R(o)}{e_o}{i}(p')) \in \boxesgeq{B_R(o)}{i}$, hence, by $(\ast)$, $\isaCopyfrom{B_R(o)}{e_o}{i}(p') \in \portsatdepthgreater{B_R(o)}{i}$; we thus have $$b_R^{\geq 0}(\isaCopyfrom{R}{e}{i}(p)) = b_R^{\geq 0}(o, \isaCopyfrom{B_R(o)}{e_o}{i}(p')) = (o, b_{B_R(o)}^{\geq 0}(\isaCopyfrom{B_R(o)}{e_o}{ i}(p'))) \in \boxesgeq{R}{i}.$$
\end{itemize}
\item If $p \in \portsatdepthgreater{\termofTaylor{R}{e}{i}}{0}$, then there exist $o \in \boxesatzero{\termofTaylor{R}{e}{i}}$ and $p' \in \ports{B_{\termofTaylor{R}{e}{i}}(o)}$ such that $p = (o, p')$; we distinguish between two cases:
\begin{itemize}
\item $o \in \boxesatzerosmaller{R}{i}$: Then $p' \in \portsatdepthleq{B_R(o)}{i}$, hence $b_R^{\geq 0}(\isaCopyfrom{R}{e}{i}(p)) = b_R^{\geq 0}(o, p') = o \in \boxessmaller{R}{i}$;
\item $o = (o_1, (e_1, o'))$ with $o_1 \in \boxesatzerogeq{R}{i}$, $e_1 \in e(o_1)$ and $o' \in \boxesatzero{\termofTaylor{B_R(o_1)}{e_1}{i}}$: By induction hypothesis, we have $b_{B_R(o_1)}^{\geq 0}(\isaCopyfrom{B_R(o_1)}{e_1}{i}(o', p')) \in \boxessmaller{B_R(o_1)}{i}$, hence 
\begin{displaymath}
  b_R^{\geq 0}(\isaCopyfrom{R}{e}{i}(p)) = b_R^{\geq 0}(o_1, \isaCopyfrom{B_R(o_1)}{e_1}{i}(o', p')) = (o_1, b_{B_R(o_1)}^{\geq 0}(\isaCopyfrom{B_R(o_1)}{e_1}{i}(o', p'))) \in \boxessmaller{R}{i}.
\end{displaymath}

\end{itemize}  
\end{itemize}
\end{proof}

If a port $q$ of a PS $R$ is deep enough, then it is duplicated $\sum e^\#(b_R^{\geq i}(q))$ times in $\termofTaylor{R}{e}{i}$:

\begin{lem}\label{lem: preimage under kappa}
Let $R$ be an in-PS. Let $e$ be a pseudo-experiment on $R$. Let $i \in \Nat$. Let $q \in \portsatdepthgreater{R}{i}$. Then we have $$\Card{\{ p \in \ports{\termofTaylor{R}{e}{i}} ; \isaCopyfrom{R}{e}{i}(p) = q \}} = \sum e^\#(b_R^{\geq i}(q))$$
\end{lem}

\begin{proof}
By induction on $\depthof{R}$. If $\depthof{R} = 0$, then there is no such $q$. Otherwise: Let $o \in \boxesatzerogeq{R}{i}$ and $q' \in \ports{B_R(o)}$ such that $q = (o, q')$. We distinguish between three cases:
\begin{itemize}
\item $q' \in \portsatdepthleq{B_R(o)}{i}$: we have $b_R^{\geq 0}(q) = b_R^{\geq 0}(o, q') = o \in \boxesgeq{R}{i}$, hence, by Fact~\ref{fact: characterization of ports at depth 0}, we have $\{ p \in \ports{\termofTaylor{R}{e}{i}} ; \isaCopyfrom{R}{e}{i}(p) = q \} = \{ p \in \portsatzero{\termofTaylor{R}{e}{i}} \setminus \portsatzero{R} ; \isaCopyfrom{R}{e}{i}(p) = (o, q') \} = \{ (o, (e_o, q')) ; e_o \in e(o) \}$, hence 
\begin{eqnarray*}
\Card{\{ p \in \ports{\termofTaylor{R}{e}{i}} ; \isaCopyfrom{R}{e}{i}(p) = q \}} & = & \Card{e(o)}\\
& = & \sum e^\# (o)\\
& = & \sum e^\#(b_R^{\geq i}(q))
\end{eqnarray*}
\item $q' \in \portsatdepthgreater{B_R(o)}{i}$ and $b_{B_R(o)}^{\geq 0}(q') \in \boxesgeq{B_R(o)}{i}$: we have $b_R^{\geq 0}(q) = b_R^{\geq 0}(o, q') = (o, b_{B_R(o)}^{\geq 0}(q')) \in \boxesgeq{R}{i}$, hence, by Fact~\ref{fact: characterization of ports at depth 0}, $\{ p \in \ports{\termofTaylor{R}{e}{i}} ; \isaCopyfrom{R}{e}{i}(p) = q \} = \{ p \in \portsatzero{\termofTaylor{R}{e}{i}} \setminus \portsatzero{R}; \isaCopyfrom{R}{e}{i}(p) = (o, q') \} = \bigcup_{e_o \in e(o)} \{ (o, (e_o, p')) ; (p' \in \portsatzero{\termofTaylor{B_R(o)}{e_o}{i}} \wedge \isaCopyfrom{B_R(o)}{e_o}{i}(p') = q') \}$; for any $e_o \in e(o)$, by Fact~\ref{fact: characterization of ports at depth 0} again, we have 
$$\{ p' \in \ports{\termofTaylor{B_R(o)}{e_o}{i}} ; \isaCopyfrom{B_R(o)}{e_o}{i}(p') = q' \} = \{ p' \in \portsatzero{\termofTaylor{B_R(o)}{e_o}{i}} ; \isaCopyfrom{B_R(o)}{e_o}{i}(p') = q' \}$$ and, by induction hypothesis, we have 
$$\Card{\{ p' \in \ports{\termofTaylor{B_R(o)}{e_o}{i}} ; \kappa_{e_o, i}(p') = q' \}} =  \sum {e_o}^\#(b_{B_R(o)}^{\geq i}(q'))$$
We thus obtain 
\begin{eqnarray*}
& & \Card{\{ p \in \ports{\termofTaylor{R}{e}{i}} ; \isaCopyfrom{R}{e}{i}(p) = q \}} \\
& = & \Card{\bigcup_{e_o \in e(o)} \{ (o, (e_o, p')) ; (p' \in \ports{\termofTaylor{B_R(o)}{e_o}{i}} \wedge \isaCopyfrom{B_R(o)}{e_o}{i}(p') = q') \}} \allowdisplaybreaks\\
& = & \sum_{e_o \in e(o)} \Card{\{ p' \in \ports{\termofTaylor{B_R(o)}{e_o}{i}} ; \isaCopyfrom{B_R(o)}{e_o}{i}(p') = q' \}}\allowdisplaybreaks\\
& = & \sum_{e_o \in e(o)} \sum {e_o}^\#(b_{B_R(o)}^{\geq i}(q')) \allowdisplaybreaks\\
& = & \sum e^\#(o, b_{B_R(o)}^{\geq i}(q'))\\
& = & \sum e^\#(b_R^{\geq i}(q))
\end{eqnarray*}
\item $q' \in \portsatdepthgreater{B_R(o)}{i}$ and $b_{B_R(o)}^{\geq 0}(q') \in \boxessmaller{B_R(o)}{i}$: we have $b_R^{\geq 0}(q) = b_R^{\geq 0}(o, q') = (o, b_{B_R(o)}^{\geq 0}(q')) \in \boxessmaller{R}{i}$, hence, by Fact~\ref{fact: characterization of ports at depth 0}, $$\{ p \in \ports{\termofTaylor{R}{e}{i}} ; \isaCopyfrom{R}{e}{i}(p) = q \} = \{ p \in \portsatdepthgreater{\termofTaylor{R}{e}{i}}{0} ; \isaCopyfrom{R}{e}{i}(p) = (o, q') \};$$ for any $e_o \in e(o)$, by Fact~\ref{fact: characterization of ports at depth 0} again, we have $\{ p \in \ports{\termofTaylor{B_R(o)}{e_o}{i}} ; \isaCopyfrom{B_R(o)}{e_o}{i}(p) = q' \} = \{ p \in \portsatdepthgreater{\termofTaylor{B_R(o)}{e_o}{i}}{0} ; \isaCopyfrom{B_R(o)}{e_o}{i}(p) = q' \}$ and, by induction hypothesis, we have $$\Card{\{ p \in \ports{\termofTaylor{B_R(o)}{e_o}{i}} ; \isaCopyfrom{B_R(o)}{e_o}{i}(p) = q' \}} =  \sum {e_o}^\#(b_{B_R(o)}^{\geq i}(q'))$$ hence 
\begin{align*}
& \Card{\{ p \in \ports{\termofTaylor{R}{e}{i}} ; \isaCopyfrom{R}{e}{i}(p) = q \}}\\
={}& \Card{\bigcup_{e_o \in e(o)} \{ ((o, (e_o, o')), p') ; ((o', p') \in \portsatdepthgreater{\termofTaylor{B_R(o)}{e_o}{i}}{0} \wedge \isaCopyfrom{B_R(o)}{e_o}{i}(o', p') = q') \}}\allowdisplaybreaks\\
={}& \sum_{e_o \in e(o)} \Card{\{ p \in \portsatdepthgreater{\termofTaylor{B_R(o)}{e_o}{i}}{0} ; \isaCopyfrom{B_R(o)}{e_o}{i}(p) = q' \}}\allowdisplaybreaks\\
={}& \sum_{e_o \in e(o)} \Card{\{ p \in \ports{\termofTaylor{B_R(o)}{e_o}{i}} ; \isaCopyfrom{B_R(o)}{e_o}{i}(p) = q' \}}\\
& \textit{(because $(\forall e_o \in e(o)) (\forall p \in \portsatzero{\termofTaylor{B_R(o)}{e_o}{i}}) \isaCopyfrom{B_R(o)}{e_o}{i}(p) \notin \portsatdepthgreater{B_R(o)}{i}$)}\allowdisplaybreaks\\
={}& \sum_{e_o \in e(o)} \sum {e_o}^\#(b_{B_R(o)}^{\geq i}(q'))\allowdisplaybreaks \\
={}& \sum e^\#(o, b_{B_R(o)}^{\geq i}(q'))\\
={}& \sum e^\#(b_R^{\geq i}(q))
 \tag*{\qedhere}
\end{align*}
\end{itemize}
\end{proof}

We finally obtain Proposition~\ref{prop: arity}:

\begin{prop}\label{prop: arity}
Let $R$ be an in-PS. Let $e$ be a pseudo-experiment on $R$. Let $i \in \Nat$. Let $p \in \portsatzero{R}$. Then we have $$\arity{\termofTaylor{R}{e}{i}}(p) = \arity{R^{\leq i}}(p) + \sum_{\substack{p' \in \portsatdepthgreater{R}{i}\\ \target{R}(p') = p}} \sum e^\#(b_R^{\geq i}(p'))$$
\end{prop}

\begin{proof}
We have 
\begin{align*}
 & \sum_{o \in \boxesatzerogeq{R}{i}} \sum_{e_o \in e(o)} \Card{\left\lbrace q \in \conclusions{\groundof{\termofTaylor{B_R(o)}{e_o}{i}}} ; \begin{array}{l} (\isaCopyfrom{B_R(o)}{e_o}{i}(q) \in \temporaryConclusions{R}{o}\\ \wedge \target{R}(o, \isaCopyfrom{B_R(o)}{e_o}{i}(q)) = p) \end{array} \right\rbrace}\\
 = & \sum_{o \in \boxesatzerogeq{R}{i}} \sum_{e_o \in e(o)} \sum_{\substack{p' \in \temporaryConclusions{R}{o}\\ \target{R}(o, p') = p}} \Card{\left\lbrace q \in \conclusions{\groundof{\termofTaylor{B_R(o)}{e_o}{i}}} ;  \isaCopyfrom{B_R(o)}{e_o}{i}(q) = p'  \right\rbrace}\allowdisplaybreaks\\
 = & \sum_{o \in \boxesatzerogeq{R}{i}} \sum_{e_o \in e(o)} \sum_{\substack{p' \in \temporaryConclusions{R}{o}\\ \target{R}(o, p') = p}} \Card{\left\lbrace q \in \portsatzero{\termofTaylor{B_R(o)}{e_o}{i}} ;  \isaCopyfrom{B_R(o)}{e_o}{i}(q) = p'  \right\rbrace}\\
& \textit{   (by Fact~\ref{fact: conclusions})}\allowdisplaybreaks\\
 = & \sum_{o \in \boxesatzerogeq{R}{i}} \sum_{e_o \in e(o)} \sum_{\substack{p' \in \temporaryConclusions{R}{o}\\ \target{R}(o, p') = p}} \Card{\left\lbrace q \in \portsatzero{\termofTaylor{B_R(o)}{e_o}{i}} ;  \isaCopyfrom{R}{e}{i}(o, (e_o, q)) = (o, p')  \right\rbrace} \allowdisplaybreaks\\
 = & \sum_{o \in \boxesatzerogeq{R}{i}} \sum_{\substack{p' \in \temporaryConclusions{R}{o}\\ \target{R}(o, p') = p}} \sum_{e_o \in e(o)}  \Card{\left\lbrace q \in \portsatzero{\termofTaylor{B_R(o)}{e_o}{i}} ;  \isaCopyfrom{R}{e}{i}(o, (e_o, q)) = (o, p')  \right\rbrace}\\
= & \sum_{\substack{p' \in \portsatdepthgreater{R}{i}\\ \target{R}(p') = p}} \Card{\{ q \in \portsatzero{\termofTaylor{R}{e}{i}}; \isaCopyfrom{R}{e}{i}(q) = p' \}}
\end{align*}
and, for any $o \in \boxesatzerogeq{R}{i}$, for any $e_o \in e(o)$, we set 
\begin{align*}
& d_o(e_o)\\
= & \sum_{o' \in \boxesatzero{\termofTaylor{B_R(o)}{e_o}{i}}} \Card{\left\lbrace (o', q) \in \conclusions{\termofTaylor{B_R(o)}{e_o}{i}} ; \begin{array}{l}
(\isaCopyfrom{B_R(o)}{e_o}{i}(o', q) \in \temporaryConclusions{R}{o}\\
\wedge 
\target{R}(o, \isaCopyfrom{B_R(o)}{e_o}{i}(o', q)) = p) \end{array} \right\rbrace}
\end{align*}
we have
\begin{align*}
  & d_o(e_o)\allowdisplaybreaks\\
=  & \sum_{o' \in \boxesatzero{\termofTaylor{B_R(o)}{e_o}{i}}} \sum_{\substack{q' \in \temporaryConclusions{R}{o}\\ \target{R}(o, q') = p}} \Card{\left\lbrace (o', q) \in \conclusions{\termofTaylor{B_R(o)}{e_o}{i}} ; \isaCopyfrom{B_R(o)}{e_o}{i}(o', q) = q' \right\rbrace}\allowdisplaybreaks\\
=  & \sum_{o' \in \boxesatzero{\termofTaylor{B_R(o)}{e_o}{i}}} \sum_{\substack{q' \in \temporaryConclusions{R}{o}\\ \target{R}(o, q') = p}} \Card{\left\lbrace (o', q) \in \ports{\termofTaylor{B_R(o)}{e_o}{i}} ; \isaCopyfrom{B_R(o)}{e_o}{i}(o', q) = q' \right\rbrace}\\
& \textit{   (by Fact~\ref{fact: conclusions})}
\end{align*}
hence
\begin{align*}
  & \sum_{e_o \in e(o)} d_o(e_o)\allowdisplaybreaks\\
=  & \sum_{\substack{q' \in \temporaryConclusions{R}{o}\\ \target{R}(o, q') = p}} \sum_{e_o \in e(o)} \sum_{o' \in \boxesatzero{\termofTaylor{B_R(o)}{e_o}{i}}} \Card{\left\lbrace (o', q) \in \ports{\termofTaylor{B_R(o)}{e_o}{i}} ; \isaCopyfrom{B_R(o)}{e_o}{i}(o', q) = q' \right\rbrace}\allowdisplaybreaks\\
= &  \sum_{\substack{q' \in \temporaryConclusions{R}{o}\\ \target{R}(o, q') = p}} \sum_{e_o \in e(o)} \sum_{o' \in \boxesatzero{\termofTaylor{B_R(o)}{e_o}{i}}} \Card{\left\lbrace (o', q) \in \ports{\termofTaylor{B_R(o)}{e_o}{i}} ; \isaCopyfrom{R}{e}{i}((o, (e_o, o')), q) = (o, q') \right\rbrace}
\end{align*}
and $$\sum_{\substack{o \in \boxesatzerogeq{R}{i}}} \sum_{e_o \in e(o)} d_o(e_o) = \sum_{\substack{p' \in \portsatdepthgreater{R}{i}\\ \target{R}(p') = p}} \Card{\{ q \in \portsatdepthgreater{\termofTaylor{R}{e}{i}}{0} ; \isaCopyfrom{R}{e}{i}(q) = p' \}}$$
We thus have
\begin{align*}
& \arity{\termofTaylor{R}{e}{i}}(p)\\
= & \arity{R^{\leq i}}(p) + \sum_{\substack{o \in \boxesatzerogeq{R}{i}}} \sum_{e_o \in e(o)} d_o(e_o)\\
+ & \sum_{o \in \boxesatzerogeq{R}{i}} \sum_{e_o \in e(o)} \Card{\left\lbrace q \in \conclusions{\groundof{\termofTaylor{B_R(o)}{e_o}{i}}} ; \begin{array}{l} (\isaCopyfrom{B_R(o)}{e_o}{i}(q) \in \temporaryConclusions{R}{o}\\ \wedge \target{R}(o, \isaCopyfrom{B_R(o)}{e_o}{i}(q)) = p) \end{array} \right\rbrace}\allowdisplaybreaks\\
= & \arity{R^{\leq i}}(p) + \sum_{\substack{p' \in \portsatdepthgreater{R}{i}\\ \target{R}(p') = p}} \Card{\{ q \in \portsatdepthgreater{\termofTaylor{R}{e}{i}}{0} ; \isaCopyfrom{R}{e}{i}(q) = p' \}}\\
+ & \sum_{\substack{p' \in \portsatdepthgreater{R}{i}\\ \target{R}(p') = p}} \Card{\{ q \in \portsatzero{\termofTaylor{R}{e}{i}}; \isaCopyfrom{R}{e}{i}(q) = p' \}}\allowdisplaybreaks\\
= & \arity{R^{\leq i}}(p) + \sum_{\substack{p' \in \portsatdepthgreater{R}{i}\\ \target{R}(p') = p}} \Card{\{ q \in \ports{\termofTaylor{R}{e}{i}}; \isaCopyfrom{R}{e}{i}(q) = p' \}}\allowdisplaybreaks\\
= & \arity{R^{\leq i}}(p) + \sum_{\substack{p' \in \portsatdepthgreater{R}{i}\\ \target{R}(p') = p}} \sum e^\#(b_R^{\geq i}(p'))\\
& \textit{   (by Lemma~\ref{lem: preimage under kappa})}
 \tag*{\qedhere}
\end{align*}
\end{proof}

\begin{exa}
Consider the port $p_1$ of $R$ at depth $0$ (see Figure~\ref{fig: new_example}): We have $\{ p' \in \portsatdepthgreater{R}{1} ; \target{R}(p') = p_1 \} = \{ (o_2, q') \}$, $b_R^{\geq 1}((o_2, q')) = o_2$ and $\arity{R^{\leq 1}}(p_1) = 1$ (see Figure~\ref{fig: Rleq1},~p.~\pageref{fig: Rleq1}). Now, if $e$ is a pseudo-experiment as in Example~\ref{example: pseudo-experiment}, then $e^\#(o_2) = \{ 10 \}$, hence $\arity{R^{\leq 1}}(p_1) + \sum_{\substack{p' \in \portsatdepthgreater{R}{1}\\ \target{R}(p') = p_1}} \sum e^\#(b_R^{\geq 1}(p')) = 1 + \sum \{ 10 \} = 11$. And, indeed, we have $\arity{\termofTaylor{R}{e}{1}}(p_1) = 11$ (see Figure~\ref{fig: termofTaylor{R}{e}{1}: 1}).
\end{exa}

\begin{cor}\label{cor: arity for R_o}
Let $R$ be an in-PS. Let $o \in \boxesatzero{R}$. Let $\varphi$ be some bijection $\contractionsUnder{R}{o} \simeq \mathcal{Q'}$ and let $R_o$ be an in-PS such that $R_o = \varphi \cdot_o R$. Let $e_o \in e(o)$. Let $i \in \Nat$. Then, for any $p \in \portsatzero{\termofTaylor{B_R(o)}{e_o}{i}}$, we have $\arity{\termofTaylor{B_R(o)}{e_o}{i}}(p) = \arity{\termofTaylor{B_{R_o}(o)}{e_o}{i}}(p)$.
\end{cor}

\begin{proof}
Let $p \in \portsatzero{\termofTaylor{B_R(o)}{e_o}{i}}$. If $p \in \portsatzero{B_R(o)}$, then, by Proposition~\ref{prop: arity}, we have 
\begin{eqnarray*}
\arity{\termofTaylor{B_R(o)}{e_o}{i}}(p) & = & \arity{{B_R(o)}^{\leq i}}(p) + \sum_{\substack{p' \in \portsatdepthgreater{B_R(o)}{i}\\ \target{B_R(o)}(p') = p}} \sum {e_o}^\#(b_{B_R(o)}^{\geq i}(p'))\\
& = & \arity{{B_{R_o}(o)}^{\leq i}}(p) + \sum_{\substack{p' \in \portsatdepthgreater{B_{R_o}(o)}{i}\\ \target{B_{R_o}(o)}(p') = p}} \sum {e_o}^\#(b_{B_{R_o}(o)}^{\geq i}(p'))\\
&  = & \arity{\termofTaylor{B_{R_o}(o)}{e_o}{i}}(p)
\end{eqnarray*}
If $p \notin \portsatzero{B_R(o)}$, then there exist $o_1 \in \boxesatzerogeq{B_R(o)}{i}$, $e_1 \in e_o(o_1)$ and $p' \in \portsatzero{\termofTaylor{B_{B_R(o)}(o_1)}{e_1}{i}}$ such that $p = (o_1, (e_1, p'))$; moreover, 
\begin{align*}
\arity{\termofTaylor{B_R(o)}{e_o}{i}}(p) ={} & \arity{\termofTaylor{B_{B_R(o)}(o_1)}{e_1}{i}}(p')\\
={} & \arity{\termofTaylor{B_{B_{R_o}(o)}(o_1)}{e_1}{i}}(p')\\
={} & \arity{\termofTaylor{B_{R_o}(o)}{e_o}{i}}(p)
 \tag*{\qedhere}
\end{align*}
\end{proof}

\section{Rebuilding the proof-structure}\label{section: rebuilding}

The Taylor expansion of a PS is an infinite set of simple differential nets (for PS's of depth $> 0$). It was not known whether from this infinite set is was possible to rebuild the PS; indeed, \emph{a priori}, two different PS's could have the same Taylor expansion. We will not only show that it is possible to rebuild any PS $R$ from its Taylor expansion $\TaylorExpansion{R}$, we will also show something much stronger: We are already able to rebuild the PS with only one well-chosen simple differential net that appears in the Taylor expansion, chosen according to a specific information given by a second simple differential net of the Taylor expansion. 

The algorithm leading from the simple differential net $\termofTaylor{R}{e}{0}$ for some well-chosen pseudo-experiment $e$ on $R$ to the entire rebuilding of $R$ is done in several steps: In the intermediate steps, we obtain a partial rebuilding where some boxes have been recovered but not all of them; a convenient way to represent this information is to use differential PS's, which lie between the purely linear differential proof-nets and the non-linear proof-nets.

The rebuilding of the PS $R$ is done in $d$ steps, where $d$ is the depth of $R$. We first rebuild the occurrences of the boxes of depth $0$ (the deepest ones) and next we rebuild the occurrences of the boxes of depth $1$ and so on... This can be formalized using simple differential nets (possibly with boxes) as follows: starting from $\termofTaylor{R}{e}{0}$, the first step of the algorithm builds $\termofTaylor{R}{e}{1}$, the second step builds $\termofTaylor{R}{e}{2}$ from $\termofTaylor{R}{e}{1}$, and so on... until $\termofTaylor{R}{e}{\depthof{R}} = R$. We thus reduced the problem of rebuilding the PS to the problem of rebuilding $\termofTaylor{R}{e}{i+1}$ from $\termofTaylor{R}{e}{i}$ for some well-chosen pseudo-experiment $e$.

We can hardly obtain much more than the following property for a non-well-chosen pseudo-experiment:

\begin{lem}\label{lem: ground-structure: trivial}
Let $R$ be an in-PS. Let $e$ be a pseudo-experiment on $R$. Let $i \in \Nat$. Then we have:
\begin{itemize}
\item ${\termofTaylor{R}{e}{i+1}}^{\leq i} \sqsubseteq_{\emptyset} {\termofTaylor{R}{e}{i}}$ 
\item $(\forall o \in \boxesatzerosmaller{{\termofTaylor{R}{e}{i+1}}}{i}) \temporaryConclusions{{\termofTaylor{R}{e}{i+1}}}{o} = \temporaryConclusions{{\termofTaylor{R}{e}{i}}}{o}$;
\item and $\restriction{\isaCopyfrom{R}{e}{i+1}}{\portsatzero{\termofTaylor{R}{e}{i+1}}} = \restriction{\isaCopyfrom{R}{e}{i}}{\portsatzero{\termofTaylor{R}{e}{i+1}}}$.
\end{itemize}
\end{lem}

\begin{proof}
By induction on $\depthof{R}$, noticing, by applying Remark~\ref{rem: @<=i}, Remark~\ref{rem: sum <=i} and Remark~\ref{rem: <= <=}, that we have ${\termofTaylor{R}{e}{i+1}}^{\leq i}  = S@t$ with 
$$S  = R^{\leq i} \oplus \bigoplus_{o \in \boxesatzerogeq{R}{i+1}} \bigoplus_{e_o \in e(o)} \langle o, \langle e_o, \termofTaylor{B_R(o)}{e_o}{i+1}^{\leq i} \rangle \rangle$$ and $t$ is the function $\mathcal{W}_0 \cup \mathcal{W}_{> 0} \to \exponentialportsatzero{R}$, where
\begin{itemize}
\item $\mathcal{W}_0 = \bigcup_{o \in \boxesatzerogeq{R}{i+1}} \bigcup_{e_o \in e(o)} \{ (o, (e_o, q)) ; (q \in \conclusions{\groundof{\termofTaylor{B_R(o)}{e_o}{i+1}}} \wedge \isaCopyfrom{B_R(o)}{e_o}{i+1}(q) \in \temporaryConclusions{R}{o}) \}$
\item $\mathcal{W}_{> 0} = \bigcup_{(o, (e_o, o')) \in \boxesatzerosmaller{S}{i} \setminus \boxesatzero{R^{\leq i+1}}} \left\lbrace ((o, (e_o, o')), q) ; \begin{array}{ll} ((o', q) \in \nonshallowConclusions{\termofTaylor{B_R(o)}{e_o}{i+1}}\\ \wedge \isaCopyfrom{B_R(o)}{e_o}{i+1}(o', q) \in \temporaryConclusions{R}{o}) \end{array} \right\rbrace$
\item and $t(p) = \left\lbrace \begin{array}{ll} \target{R}(o, \isaCopyfrom{B_R(o)}{e_o}{i+1}(q)) & \textit{if $p = (o, (e_o, q)) \in \mathcal{W}_0$;}\\ \target{R}(o, \isaCopyfrom{B_R(o)}{e_o}{i+1}(o', q)) & \textit{if $p = ((o, (e_o, o')), q) \in \mathcal{W}_{> 0}$.} \end{array} \right.$
\qedhere
\end{itemize}
\end{proof}

We thus have to consider some special experiments. As a first requirement, the pseudo-experiments we will consider are \emph{exhaustive}:

\begin{defi}\label{definition: exhaustive}
A pseudo-experiment $e$ of an in-PS $R$ is said to be \emph{exhaustive} if, for any $o \in \boxes{R}$, we have $0 \notin e^\#(o)$.
\end{defi}

But this requirement is not strong enough. More specific pseudo-experiments have to be considered. 
\begin{itemize}
\item In \cite{injectcoh}, it was shown that given the result of an \emph{injective $k$-obsessional experiment} ($k$ big enough) of a cut-free proof-net in the fragment $A ::= X \vert \contr A \parr A \vert A \parr \contr A \vert A \otimes A \vert !A$, one can rebuild the entire experiment and, so, the entire proof-net. We recall that the LPS of a cut-free proof-net forgets the outline of the boxes but keeps the trace of the auxiliary doors (see Figure~\ref{fig: LPS},~p.~\pageref{fig: LPS} for an example). There, ``injective'' means that the experiment labels two different axioms with different atoms and ``obsessional'' means that different copies of the same axiom are labeled by the same atom. Obsessionality entails that the names of the atoms play some role and thus we cannot reduce such experiments to some pseudo-experiments.
\item In \cite{LPSinjectivity}, it was shown that, for any two cut-free MELL proof-nets $R$ and $R'$, we have $LPS(R) = LPS(R')$ if, and only if, for $k$ big enough\footnote{Interestingly, \cite{k=2}, following the approach of \cite{LPSinjectivity}, showed that, if these two proof-nets are assumed to be (recursively) connected, then we can take $k = 2$.}, there exist an \emph{injective $k$-experiment} on $R$ and an \emph{injective $k$-experiment} on $R'$ having the same result; as an immediate corollary we obtained the injectivity of the set of (recursively) connected proof-nets. There, ``injective'' means that not only the experiment labels two different axioms with different atoms, but it labels also different copies of the same axiom by different atoms. Given some proof-net $R$, there is exactly one injective $k$-experiment on $R$ up to the names of the atoms. Injectivity allows to reduce such experiments to some pseudo-experiments: it makes sense to define \emph{$k$-pseudo-experiments} and we could show in the same way that, for any two cut-free proof-nets $R$ and $R'$, the two following statements are equivalent:
\begin{itemize}
\item we have $LPS(R) = LPS(R')$;
\item for any $k \in \Nat$, for any $k$-pseudo-experiment $e$ on $R$, for any $k$-pseudo-experiment $e'$ on $R'$, we have $\termofTaylor{R}{e}{0} = \termofTaylor{R'}{e'}{0}$, where a \emph{$k$-pseudo-experiment} is a pseudo-experiment such that, for any $o \in \boxes{R}$, we have $e^\#(o) = \{ k \}$. 
\end{itemize}
Now, there are many different cut-free PS's with the same LPS (see Figures~\ref{fig: LPS}, \ref{fig: R_1}, \ref{fig: R_2}, \ref{fig: R_3} and \ref{fig: R_4}, p.~\pageref{fig: LPS} for an example).
\end{itemize}

In \cite{LPSinjectivity}, the interest for \emph{injective} experiments came from the remark that the result of an (atomic) \emph{injective} experiment on a \emph{cut-free} proof-net can be easily identified with a simple differential net of its Taylor expansion in a sum of simple differential nets \cite{EhrhardRegnier:DiffNets} (it is essentially the content of our Lemma~\ref{lem: Taylor expansion}). Thus any proof using injective experiments can be straightforwardly expressed in terms of simple differential nets and conversely. Since this identification is trivial, besides the idea of considering injective experiments instead of obsessional experiments, the use of the terminology of differential nets does not bring any new insight\footnote{For proof-nets with cuts, the situation is completely different: the great novelty of differential nets is that differential nets have a cut-elimination; the simple differential nets appearing in the Taylor expansion of a proof-net with cuts have cuts, while the semantics does not see these cuts. But the proofs of the injectivity only consider cut-free proof-nets.}, it just superficially changes the presentation.That is why we decided in \cite{LPSinjectivity} to avoid introducing explicitly differential nets. In \cite{Carvalho:CSL}, we made the opposite choice for the following reason: The simple differential net representing the result and the proof-net are both instances of the more general notion of ``simple differential nets possibly with boxes'', which are used to represent the partial information obtained during the algorithm execution. Moreover, this identification allows to see the injectivity of the relational semantics as a particular case of the invertibility of Taylor expansion.

In the present paper, we introduce the notion of \emph{$k$-heterogeneous \mbox{(pseudo-)}experiment} ($k$-heterogeneous pseudo-experiments on \emph{cut-free} PS's are an abstraction of injective $k$-heterogeneous experiments, where \emph{injective} has the same meaning as in \cite{LPSinjectivity}) and the simple differential net we will consider to rebuild entirely the PS is the simple differential net obtained by expanding the boxes according to any \emph{$k$-heterogeneous pseudo-experiment} (for $k$ big enough): We show that, for any cut-free PS $R$, given the result $\alpha$ of a \emph{$k$-heterogeneous experiment} on $R$ for $k$ big enough, if $\alpha \in \sm{R'}$, where $R'$ is any cut-free PS, then $R'$ is the same PS as $R$. The constraints on $k$ are given by the result of a $1$-experiment, so we show that two (well-chosen) points are enough to determine a PS. The expression ``$k$-heterogeneous'' means that, for any two different occurrences of boxes, the experiment never takes the same number of copies: it takes $k^{j_1}$ copies and $k^{j_2}$ copies with $j_1 \not= j_2$ (\emph{a contrario}, in \cite{injectcoh} and \cite{LPSinjectivity}, the experiments always take the same number of copies).
\begin{figure}
\centering
\resizebox{.5 \textwidth}{!}{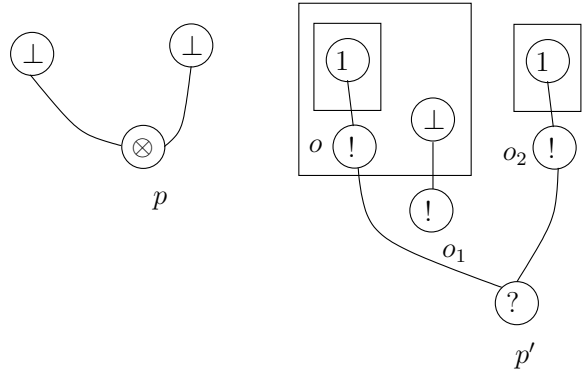}
\caption{Proof-net $S$}
\label{fig: no reconstruction of the experiment}
\end{figure}
As shown by the proof-net $S$ of Figure~\ref{fig: no reconstruction of the experiment}, it is impossible to rebuild the experiment from its result, since there exist five different $4$-heterogeneous experiments $e_1, e_2, e_3, e_4$ and $e_5$ on $S$ such that, for any $i \in \{ 1, 2, 3, 4, 5 \}$, we have $e_i(p) = (\ast, \ast)$, $e_i(o_1) = [\ast, \ast, \ast, \ast]$ and $e_i(p') = [[\underbrace{\ast, \ldots, \ast}_{4^2}], \ldots, [\underbrace{\ast, \ldots, \ast}_{4^6}]]$: The experiment $e_i$ takes $4$ copies of the box $o_1$ and $4^{i+1}$ copies of the box $o_2$.

Actually, more generally, we show that, for \emph{any} PS $R$, given the simple differential net $\termofTaylor{R}{e}{0}$ that belongs to the support of the Taylor expansion and that has been obtained by expanding the boxes according to any (atomic) injective $k$-heterogeneous pseudo-experiment $e$ on $R$ for $k$ big enough, if $\termofTaylor{R}{e}{0}$ belongs to the support of the Taylor expansion of any PS $R'$, then $R'$ is the same PS as $R$. Notice that in presence of cuts, the $k$-heterogeneous pseudo-experiment we consider is not necessarily induced  (see Definition~\ref{definition: experiment induces pseudo}) by an experiment.\footnote{In the case there is no such experiment, the simple differential net $\termofTaylor{R}{e}{0}$ reduces to $0$.}

\begin{defi}\label{defin: k-heterogeneous}
Let $k > 0$. A pseudo-experiment $e$ on an in-PS $R$ is said to be \emph{$k$-heterogeneous} if
\begin{itemize}
\item for any $o \in \boxes{R}$, for any $m \in e^\#(o)$, there exists $j > 0$ such that $m = k^j$;
\item for any $o \in \boxesatzero{R}$, for any $o' \in \boxes{B_R(o)}$, we have $(\forall e_1, e_2 \in e(o))$ $({e_1}^\#(o') \cap {e_2}^\#(o') \not= \emptyset \Rightarrow e_1 = e_2)$;
\item and, for any $o_1, o_2 \in \boxes{R}$, we have $({e}^\#(o_1) \cap {e}^\#(o_2) \not= \emptyset \Rightarrow o_1 = o_2)$.
\end{itemize}
\end{defi}

\begin{exa}\label{example: pseudo}
The pseudo-experiment $e$ of Example~\ref{example: pseudo-experiment} is a $10$-heterogeneous pseudo-experiment.
\end{exa}

$k$-heterogeneous experiments are characterized by (the arities of the co-contractions of) their corresponding terms of the Taylor expansion:

\begin{lem}\label{lem: M_0(e)}
For any in-PS $R$, we have $\bigcup_{o \in \boxes{R}} e^\#(o) = \arity{\termofTaylor{R}{e}{0}}[\portsatzerooftype{\cod}{\termofTaylor{R}{e}{0}}]$.
\end{lem}

\begin{proof}
By Proposition~\ref{prop: arity}, we have $\bigcup_{o \in \boxesatzero{R}} e^\#(o) = \arity{\termofTaylor{R}{e}{0}}[\boxesatzero{R}]$, hence
\begin{align*}
\bigcup_{o \in \boxes{R}} e^\#(o) ={} & \bigcup_{o \in \boxesatzero{R}} e^\#(o) \cup \bigcup_{o \in \boxesatzero{R}} \bigcup_{o' \in \boxes{B_R(o)}} e^\#(o, o')\\
={} & \arity{\termofTaylor{R}{e}{0}}[\boxesatzero{R}] \cup \bigcup_{o \in \boxesatzero{R}} \bigcup_{o' \in \boxes{B_R(o)}} \bigcup_{e_o \in e(o)} {e_o}^\#(o')\allowdisplaybreaks\\
={} & \arity{\termofTaylor{R}{e}{0}}[\boxesatzero{R}] \cup \bigcup_{o \in \boxesatzero{R}} \bigcup_{e_o \in e(o)} \arity{\termofTaylor{B_R(o)}{e_o}{0}}[\portsatzerooftype{\cod}{\termofTaylor{B_R(o)}{e_o}{0}}]\\
& \textit{    (by induction hypothesis)}\allowdisplaybreaks\\
={} & \arity{\termofTaylor{R}{e}{0}}[\boxesatzero{R}] \cup \bigcup_{o \in \boxesatzero{R}} \bigcup_{e_o \in e(o)} \arity{\termofTaylor{R}{e}{0}}[\portsatzerooftype{\cod}{R \langle o, 0, e_o \rangle}]\\
={} & \arity{\termofTaylor{R}{e}{0}}[\portsatzerooftype{\cod}{\termofTaylor{R}{e}{0}}]
 \tag*{\qedhere}
\end{align*}
\end{proof}

\begin{cor}\label{cor: characterization of k-heterogeneous experiments}
Let $R$ be an in-PS. Let $e$ be a pseudo-experiment on $R$. Let $k > 1$. Then $e$ is a $k$-heterogeneous pseudo-experiment on $R$ if, and only if, the two following properties hold together:
\begin{enumerate}
\item\label{property: arities are powers of k} $\arity{\termofTaylor{R}{e}{0}}[\portsatzerooftype{\cod}{\termofTaylor{R}{e}{0}}] \subseteq \{ k^j ; j > 0 \}$
\item\label{property: arities are distinct} $(\forall p_1, p_2 \in \portsatzerooftype{\cod}{\termofTaylor{R}{e}{0}}) (\arity{\termofTaylor{R}{e}{0}}(p_1) = \arity{\termofTaylor{R}{e}{0}}(p_2) \Rightarrow p_1 = p_2)$
\end{enumerate}
\end{cor}

\begin{proof}
For any pseudo-experiment $e$ on $R$, applying Lemma~\ref{lem: M_0(e)}, we obtain:
\begin{itemize}
\item For any $k > 1$, one has $(\forall o \in \boxes{R}) (\forall m \in e^\#(o)) (\exists j > 0) m = k^j$ if, and only if, one has $\arity{\termofTaylor{R}{e}{0}}[\portsatzerooftype{\cod}{\termofTaylor{R}{e}{0}}] \subseteq \{ k^j ; j > 0 \}$.
\item One has $(\forall o \in \boxesatzero{R}) (\forall o' \in \boxes{B_R(o)}) (\forall e_1, e_2 \in e(o)) ({e_1}^\#(o') \cap {e_2}^\#(o') \not= \emptyset \Rightarrow e_1 = e_2)$ if, and only if, one has $(\forall o \in \boxesatzero{R}) (\forall e_1, e_2 \in e(o))$ $(\arity{\termofTaylor{R}{e}{0}}[\portsatzerooftype{\cod}{R \langle o, 0, e_1 \rangle}] \cap \arity{\termofTaylor{R}{e}{0}}[\portsatzerooftype{\cod}{R \langle o, 0, e_2 \rangle}] \not= \emptyset \Rightarrow e_1 = e_2)$.
\item One has $(\forall o_1, o_2 \in \boxes{R}) ({e}^\#(o_1) \cap {e}^\#(o_2) \not= \emptyset \Rightarrow o_1 = o_2)$ if, and only if, one has $(\forall p_1, p_2 \in \portsatzerooftype{\cod}{\termofTaylor{R}{e}{0}}) (\arity{\termofTaylor{R}{e}{0}}(p_1) = \arity{\termofTaylor{R}{e}{0}}(p_2) \Rightarrow \isaCopyfrom{R}{e}{0}(p_1) = \isaCopyfrom{R}{e}{0}(p_2))$.
\end{itemize}
We prove, by induction on $\depthof{R}$, that, for any $k > 1$, the two following properties together:
\begin{itemize}
\item $(\forall o \in \boxesatzero{R}) (\forall e_1, e_2 \in e(o))$ $(\arity{\termofTaylor{R}{e}{0}}[\portsatzerooftype{\cod}{R \langle o, 0, e_1 \rangle}] \cap \arity{\termofTaylor{R}{e}{0}}[\portsatzerooftype{\cod}{R \langle o, 0, e_2 \rangle}] \not= \emptyset \Rightarrow e_1 = e_2)$
\item and $(\forall p_1, p_2 \in \portsatzerooftype{\cod}{\termofTaylor{R}{e}{0}}) (\arity{\termofTaylor{R}{e}{0}}(p_1) = \arity{\termofTaylor{R}{e}{0}}(p_2) \Rightarrow \isaCopyfrom{R}{e}{0}(p_1) = \isaCopyfrom{R}{e}{0}(p_2))$
\end{itemize}
imply $(\forall p_1, p_2 \in \portsatzerooftype{\cod}{\termofTaylor{R}{e}{0}}) (\arity{\termofTaylor{R}{e}{0}}(p_1) = \arity{\termofTaylor{R}{e}{0}}(p_2) \Rightarrow p_1 = p_2)$; the converse is trivial.
\end{proof}

We will rebuild $\termofTaylor{R}{e}{i+1}$ from $\termofTaylor{R}{e}{i}$ for any $k$-heterogeneous pseudo-experiment $e$ on $R$ (with $k \geq \basis{R}$, where $\basis{R}$ is an integer\footnote{The integer $\basis{R}$ is defined in Definition~\ref{defin: k big enough}.} provided by any $1$-pseudo-experiment on $R$). For this purpose we will introduce our notion of \emph{critical component} (elements of $\nontrivialconnected{\termofTaylor{R}{e}{i}}{\criticalports{\termofTaylor{R}{e}{i}}{k}{j}}{k}$ with $j \in \mathcal{N}_i(e)$)\footnote{The set $\mathcal{N}_i(e)$ is a set of integers that will be defined in Definition~\ref{defin: the algorithm (b)}, the set $\criticalports{\termofTaylor{R}{e}{i}}{k}{j}$ is a set of exponential ports of $\termofTaylor{R}{e}{i}$ at depth $0$ that will be defined in Definition~\ref{definition: critical} and the sets $\nontrivialconnected{S}{\mathcal{Q}}{k}$ of \emph{components $T$ of $S$ that are connected \emph{via} other ports than $\mathcal{Q}$, whose conclusions belong to $\mathcal{Q}$ and with $\cosize{T} < k$} will  be defined in Definition~\ref{defin: connected components}.}, which are special connected components of $\termofTaylor{R}{e}{i}$, \emph{connected} in a weak sense, and we will consider equivalence classes of critical components of $\termofTaylor{R}{e}{i}$ for the relation $\equiv$ that forgets names of ports apart from those of shallow conclusions.\footnote{This relation has been defined in Definition~\ref{definition: isomorphisms}.} We can summarize the three main ideas of our proof as follows:
\begin{enumerate}
\item The simple differential net that corresponds to a $k$-heterogeneous pseudo-experiments is informative enough to rebuild the entire PS if $k$ is big enough.
\item We introduce the notion of \emph{partial} Taylor expansion $\mathcal{T}_R[i]$ of a PS $R$ (with $i \in \Nat$) and reduce the problem of the rebuilding of a PS $R$ to the problem of the rebuilding of the simple differential net \emph{possibly with boxes} $\termofTaylor{R}{e}{i+1}$ that corresponds to a $k$-heterogeneous pseudo-experiment $e$ in the partial Taylor expansion $\mathcal{T}_R[i]$ from the simple differential net \emph{possibly with boxes} $\termofTaylor{R}{e}{i}$ that corresponds to the same $k$-heterogeneous pseudo-experiment $e$ in the partial Taylor expansion $\mathcal{T}_R[i]$.
\item We consider cardinalities of equivalence classes $\nontrivialconnected{\termofTaylor{R}{e}{i}}{\criticalports{\termofTaylor{R}{e}{i}}{k}{j}}{k}/\equiv$ of \emph{critical components} of $\termofTaylor{R}{e}{i}$ for the relation $\equiv$ in order to deduce the cardinalities of equivalence classes $\nontrivialconnected{\termofTaylor{R}{e}{i+1}}{\criticalports{\termofTaylor{R}{e}{i+1}}{k}{j}}{k}/\equiv$ of \emph{critical components} of $\termofTaylor{R}{e}{i+1}$ and the contents of the new boxes of depth $i$ of $\termofTaylor{R}{e}{i+1}$.
\end{enumerate}

\subsection{The borders of the boxes}\label{subsection: outline}

In this subsection we first show how to recover the set $\bigcup_{o \in \boxesgeq{R}{i}} \{ \log_k(m) ; m \in e^\#(o) \}$ and, therefore, the set $\portsatzerooftype{\cod}{\termofTaylor{R}{e}{i}} \setminus \boxesatzero{\termofTaylor{R}{e}{i}}$ of co-contractions of $\termofTaylor{R}{e}{i}$ (Lemma~\ref{lem: M_i}). Next, we show how to determine, from $\termofTaylor{R}{e}{i}$, the set $\exactboxesatzero{\termofTaylor{R}{e}{i+1}}{i}$ of ``new'' boxes and, for any such ``new'' box $o \in \exactboxesatzero{\termofTaylor{R}{e}{i+1}}{i}$, its \emph{border} i.e. the set $\target{\termofTaylor{R}{e}{i+1}}[\{ o \} \times \temporaryConclusions{\termofTaylor{R}{e}{i+1}}{o}]$ of exponential ports that are immediately below (Proposition~\ref{prop: critical ports below new boxes}), which, as shown by Example~\ref{example: border > LPS},~p.~\pageref{example: border > LPS}, is an information that is not provided by the LPS but is still too weak to rebuild the PS (the borders of the boxes do not allow to recover their outlines). In particular, we have $\exactboxesatzero{\termofTaylor{R}{e}{i+1}}{i} = !_{e, i}[\mathcal{N}_i(e)]$, where the set $\mathcal{N}_i(e) \subseteq \Nat$ is defined from the set $\mathcal{M}_0(e)$ of the numbers of copies of boxes taken by the pseudo--experiment $e$:

\begin{defi}\label{defin: the algorithm (b)}
Let $R$ be a differential in-PS. Let $k > 1$. Let $e$ be a $k$-heterogeneous pseudo-experiment on $R$. For any $i \in \Nat$, we define, by induction on $i$, $\mathcal{M}_i(e) \subseteq \Nat \setminus \{ 0 \}$ and $(m_{i, j}(e))_{j \in \Nat} \in \{ 0, \ldots, k-1 \}^\Nat$ as follows. 
We set $\mathcal{M}_0(e) = \bigcup_{o \in \boxes{R}{}} \{ j \in \Nat; k^j \in e^\#(o) \}$ and we write $\Card{\mathcal{M}_i(e)}$ in base $k$: $\Card{\mathcal{M}_i(e)} = \sum_{j \in \Nat} m_{i, j}(e) \cdot k^j$; we set $\mathcal{M}_{i+1}(e) =$ $\{ j > 0 ; m_{i, j}(e) \not= 0 \}$.

For any $i \in \Nat$, we set $\mathcal{N}_i(e) =  \mathcal{M}_i(e) \setminus \mathcal{M}_{i+1}(e)$.
\end{defi}

Notice that all the sets $\mathcal{M}_i(e)$ and $\mathcal{N}_i(e)$ can be computed from $\termofTaylor{R}{e}{0}$, since, by Lemma~\ref{lem: M_0(e)}, we have $\mathcal{M}_0(e) =$ $\log_k[\{ \arity{\termofTaylor{R}{e}{0}}(p) ;$ $p \in \portsatzerooftype{\cod}{\termofTaylor{R}{e}{0}} \}]$.

\begin{exa}\label{example: N_i(e)}
If $e$ is a $10$-heterogeneous pseudo-experiment as in Example~\ref{example: pseudo}, then $\mathcal{M}_0(e) = \{ 1, \ldots, 224 \}$. We have $\Card{\mathcal{M}_0(e)} = 4 + 2 \cdot 10^1 + 2 \cdot 10^2$, hence $\mathcal{M}_1(e) = \{ 1, 2 \}$ and $\mathcal{N}_0(e) = \{ 3, \ldots, 224 \}$. We have $\Card{\mathcal{M}_1(e)} = 2$, hence $\mathcal{M}_2(e) = \emptyset$ and $\mathcal{N}_1(e) = \{ 1, 2 \}$.  
\end{exa}

\begin{lem}\label{lem: bijection cod{e, i}}
Let $k > 1$. Let $R$ be an in-PS. If $e$ is a $k$-heterogeneous pseudo-experiment on $R$, then, for any $i \in \Nat$, there exists a unique bijection 
$$\cod_{e, i} : \log_k[\bigcup e^\#[\boxesgeq{R}{i}]] \simeq \portsatzerooftype{\cod}{\termofTaylor{R}{e}{i}} \setminus \boxesatzero{\termofTaylor{R}{e}{i}}$$ 
such that, for any $j \in \dom{\cod_{e, i}}$, the following properties hold:  
\begin{enumerate}
\item $k^j \in (e^\# \circ \isaCopyfrom{R}{e}{i} \circ \cod_{e, i})(j)$
\item\label{property: arity} and $(\arity{\termofTaylor{R}{e}{i}} \circ !_{e, i})(j) = k^j$. 
\end{enumerate}
\end{lem}

\begin{proof}
Notice first that, for any $o \in \boxesatzerogeq{R}{i}$, we have $\arity{\termofTaylor{R}{e}{i}}(o) = \Card{e(o)}$, hence the function that associates with every $j \in \bigcup_{o \in \boxesatzerogeq{R}{i}} \{ \log_k(m) ; m \in e^\#(o) \}$ the unique $o_j \in \boxesatzerogeq{R}{i}$ such that $e^\#(o_j) = \{ k^j \}$ is a bijection 
$$\bigcup_{o \in \boxesatzerogeq{R}{i}} \{ \log_k(m) ; m \in e^\#(o) \} \simeq \portsatzerooftype{\cod}{R} \setminus \boxesatzero{\termofTaylor{R}{e}{i}}$$
such that $\arity{\termofTaylor{R}{e}{i}}(o_j) = \sum e^\#(o_j) = k^j$.

Now, we prove the lemma by induction on $\depthof{R}$. We set $$\cod_{e, i}(j) = \left\lbrace \begin{array}{ll} o_j & \textit{if $j \in \bigcup_{o \in \boxesatzerogeq{R}{i}} \{ \log_k(m) ; m \in e^\#(o) \}$;}\\
(o, (e_o, \cod_{e_o, i}(j))) & \textit{if $o \in \boxesatzerogeq{R}{i}$ and $k^j \in ({e_o}^\# \circ \kappa_{e_o, i} \circ \cod_{e_o, i})(j)$.} \end{array} \right.$$
For checking Property~\ref{property: arity}, notice that, for any $o \in \boxesatzerogeq{R}{i}$, for any $e_o \in e(o)$, for any $p \in \portsatzero{\termofTaylor{B_R(o)}{e_o}{i}}$, we have $\arity{\termofTaylor{R}{e}{i}}(o, (e_o, p)) = \arity{\termofTaylor{B_R(o)}{e_o}{i}}(p)$.
\end{proof}

\subsubsection{Identifying the co-contractions that correspond to the new boxes}

This part is devoted to prove Lemma~\ref{lem: M_i}, which shows that, for any $k$-heterogeneous pseudo-experiment $e$ on $R$, for any $i \in \Nat$, the function $\cod_{e, i}$ is actually a bijection $\mathcal{M}_i(e) \to \portsatzerooftype{\cod}{\termofTaylor{R}{e}{i}} \setminus \boxesatzero{\termofTaylor{R}{e}{i}}$ such that, for any $j \in \mathcal{M}_i(e)$,  we have  $(\arity{\termofTaylor{R}{e}{i}} \circ \cod_{e, i})(j) = k^j$. 

\begin{lem}\label{lem: M_i: aux}
Let $k > 1$. For any $i \in \Nat$, for any in-PS $R$ such that $\Card{\boxes{R}} < k$, for any $k$-heterogeneous pseudo-experiment $e$ on $R$, the following properties hold:
\begin{enumerate}
\item\label{property: at depth 0 disjoint from at depth > 0} $P_1(R, e, i)$: $(\forall o, o' \in \boxesatzero{R}) (\forall e_{o} \in e(o)) (\forall j \in \mathcal{M}_i(e_o)) k^j \notin e^\#(o')$
\item\label{property: M_i are disjoint} $P_2(R, e, i)$: $(\forall o, o' \in \boxesatzero{R}) (\forall e_o \in e(o)) (\forall e_{o'} \in e(o')) (\mathcal{M}_i(e_o) \cap \mathcal{M}_i(e_{o'}) \not= \emptyset \Rightarrow (o, e_o) = (o', e_{o'}))$
\item\label{property: M_{i+1} empty} $P_3(R, e, i)$: $(\depthof{R} = i \Rightarrow \mathcal{M}_{i}(e) = \emptyset)$
\item\label{property: m_{i, 0}} $P_4(R, e, i)$: $m_{i, 0}(e) = \Card{\boxesatzerogeq{R}{i}}$
\item\label{property: inclusion} $P_5(R, e, i)$: $\mathcal{M}_{i}(e) \subseteq \mathcal{M}_{i-1}(e)$
\item\label{property: description of M_i(e)} $P_6(R, e, i)$: $\mathcal{M}_{i}(e) = \log_k[\bigcup e^\#[\boxesatzerogeq{R}{i}]] \cup \bigcup_{o \in \boxesatzerogeq{R}{i+1}} \bigcup_{e_o \in e(o)} \mathcal{M}_{i}(e_o)$
\end{enumerate}
\end{lem}

\begin{proof}
By complete induction on $\depthof{R} + i$. First, if $i = 0 = \depthof{R}$, then $\mathcal{M}_{-1} = \emptyset = \mathcal{M}_0$, hence $P_1(R, e, i)$, $P_2(R, e, i)$, $P_3(R, e, i)$, $P_4(R, e, i)$, $P_5(R, e, i)$ and $P_6(R, e, i)$ hold trivially. Now, let $i \in \Nat$. Notice that if $P_6(R, e, i)$ holds, then $P_3(R,e, i)$ holds; moreover, if, in addition, $P_1(R, e, i)$ and $P_2(R, e, i)$ hold, then 
\begin{eqnarray*}
& & \Card{\mathcal{M}_{i}(e)}\\
& = & \Card{\boxesatzerogeq{R}{i}} + \sum_{o \in \boxesatzerogeq{R}{i+1}} \sum_{e_o \in e(o)} \Card{\mathcal{M}_{i}(e_o)}\allowdisplaybreaks\\
& = & \Card{\boxesatzerogeq{R}{i}} + \sum_{o \in \boxesatzerogeq{R}{i+1}} \sum_{e_o \in e(o)} \sum_{j \in \mathcal{M}_{i+1}(e_o) \cup \{ 0 \}} m_{i, j}(e_o) \cdot k^j\\
&  = & \Card{\boxesatzerogeq{R}{i}}\\
& + & \sum_{o \in \boxesatzerogeq{R}{i+1}} \left(\Card{\boxesatzerogeq{B_R(o)}{i}} \cdot \Card{e(o)} + \sum_{e_o \in e(o)} \sum_{j \in \mathcal{M}_{i+1}(e_o)} m_{i, j}(e_o) \cdot k^j \right)\\
& & \textit{(by $P_4(B_R(o), e_o, i)$ for each $o \in \boxesatzerogeq{R}{i+1}$ and each $e_o \in e(o)$)}
\end{eqnarray*}
We thus have $P_4(R, e, i)$. Moreover, if $i > 0$, then by applying $P_6(R, e, i-1)$, $P_6(R, e, i)$ and $P_5(B_R(o), e_o, i)$ for each $o \in \boxesatzerogeq{R}{i+1}$ and each $e_o \in o$, and by noticing the inclusion $\log_k[\bigcup e^\#[\boxesatzerogeq{R}{i}]] \subseteq \log_k[\bigcup e^\#[\boxesatzerogeq{R}{i-1}]]$, we obtain $P_5(R, e, i)$.

\begin{itemize}
\item Let us assume that $i = 0$. For any in-PS $R$, for any $k$-heterogeneous pseudo-experiment $e$ on $R$, since $\bigcup e^\#[\boxes{R}] \subseteq \{ k^j ; j > 0 \}$, we have $\Card{\mathcal{M}_{-1}(e)} = \sum \bigcup e^\#[\boxes{R}] = \sum_{j \in \log_k[\bigcup e^\#[\boxes{R}]]} k^j$, hence $\mathcal{M}_0(e) = \log_k[\bigcup e^\#[\boxes{R}]]$. 

Let $R$ be an in-PS such that $\Card{\boxes{R}} < k$ and let $e$ be a $k$-heterogeneous experiment on $R$.

Let $o, o' \in \boxesatzero{R}$ and let $e_{o} \in e(o)$. Let $j \in \mathcal{M}_0(e_{o})$. There exists $o'' \in \boxes{B_R(o)}$ such that $k^j \in {e_{o}}^\#(o'')$. Assume that $k^j \in e^\#(o')$. We have $e^\#((o, o'')) \cap e^\#(o') \not= \emptyset$; but, by the definition of \emph{$k$-heterogeneous experiment} (Definition~\ref{defin: k-heterogeneous}), this entails that $o' = (o, o'')$; we thus obtain a contradiction with the following requirement of the definition of in-PS's (Definition~\ref{defin: differential PS}): for any $(p_1, p_2) \in \portsatzero{R}$, we have $p_1 \notin \boxesatzero{R}$. We showed that $P_1(R, e, 0)$ holds.

Let $o, o' \in \boxesatzero{R}$. Let $e_o \in e(o)$ and $e_{o'} \in e(o')$ such that $\mathcal{M}_0(e_o) \cap \mathcal{M}_0(e_{o'}) \not= \emptyset$. Let $j \in \mathcal{M}_0(e_o) \cap \mathcal{M}_0(e_{o'})$. There exist $o_1 \in \boxes{B_R(o)}$ such that $k^j \in {e_o}^\#(o_1) = e^\#(o, o_1)$ and $o_2 \in \boxes{B_R(o')}$ such that $k^j \in {e_{o'}}^\#(o_2) = e^\#(o', o_2)$. By the definition of \emph{$k$-heterogeneous experiment}, we have $(o, o_1) = (o', o_2)$. We have ${e_o}^\#(o_1) \cap {e_{o'}}^\#(o_1) \not= \emptyset$, hence, again by the definition of \emph{$k$-heterogeneous experiment}, we obtain $e_o = e_{o'}$. We showed that $P_2(R, e, 0)$ holds.

We have
\begin{eqnarray*}
\mathcal{M}_0(e) & = & \log_k[\bigcup e^\#[\boxes{R}]]\\
& = & \log_k[\bigcup e^\#[\boxesatzero{R}]] \cup \bigcup_{o \in \boxesatzerogeq{R}{1}} \log_k[\bigcup e^\#[\{ o \} \times \boxesatzero{B_R(o)}]]\allowdisplaybreaks\\
& = & \log_k[\bigcup e^\#[\boxesatzero{R}]] \cup \bigcup_{o \in \boxesatzerogeq{R}{1}} \bigcup_{e_o \in e(o)} \log_k[\bigcup {e_o}^\#[\boxesatzero{B_R(o)}]]\\
& = & \log_k[\bigcup e^\#[\boxesatzero{R}]] \cup \bigcup_{o \in \boxesatzerogeq{R}{1}} \bigcup_{e_o \in e(o)} \mathcal{M}_0(e_o)
\end{eqnarray*}
hence $P_6(R, e, 0)$ holds. We already know that it follows that $P_4(R, e, 0)$ holds.

Let $j \in \mathcal{M}_0(e)$: We have $k^j \in \bigcup e^\#[\boxes{R}]$, hence $k^j \leq \sum \bigcup e^\#[\boxes{R}]$. Since $j < k^j$, we have $j < \sum \bigcup e^\#[\boxes{R}]$; so $j \in \mathcal{M}_{-1}(e)$. We thus proved $P_5(R, e, 0)$.

\item Let us assume that $i > 0$. Let $R$ be an in-PS such that $\Card{\boxes{R}} < k$ and let $e$ be a $k$-heterogeneous experiment on $R$. By $P_1(R, e, i-1)$ and $P_5(R, e, i-1)$, we have $P_1(R, e, i)$. By $P_2(R, e, i-1)$ and $P_5(R, e, i-1)$, we have $P_2(R, e, i)$. 
We have
\begin{eqnarray*}
& & \Card{\mathcal{M}_{i-1}(e)} \\
& = & \Card{\boxesatzerogeq{R}{i-1}} + \sum_{o \in \boxesatzerogeq{R}{i}} \sum_{e_o \in e(o)} \Card{\mathcal{M}_{i-1}(e_o)} \\
& & \textit{(by $P_6(R, e, i-1)$, $P_1(R, e, i-1)$ and $P_2(R, e, i-1)$)}\allowdisplaybreaks\\
& = & \Card{\boxesatzerogeq{R}{i-1}} + \sum_{o \in \boxesatzerogeq{R}{i}} \sum_{e_o \in e(o)} \left(m_{i-1, 0}(e_o) + \sum_{j \in \mathcal{M}_{i}(e_o)} m_{i-1, j}(e_o) \cdot k^j \right)\allowdisplaybreaks\\
& = & \Card{\boxesatzerogeq{R}{i-1}} + \sum_{o \in \boxesatzerogeq{R}{i}} \sum_{e_o \in e(o)} \left(\Card{\boxesatzerogeq{B_R(o)}{i-1}} + \sum_{j \in \mathcal{M}_{i}(e_o)} m_{i-1, j}(e_o) \cdot k^j \right)\\
& & \textit{(by $P_4(B_R(o), e_o, i-1)$ for each $o \in \boxesatzerogeq{R}{i}$ and each $e_o \in e(o)$)} \allowdisplaybreaks \\
& = & \Card{\boxesatzerogeq{R}{i-1}}\\
& + & \sum_{o \in \boxesatzerogeq{R}{i}} \left( \Card{\boxesatzerogeq{B_R(o)}{i-1}} \cdot \Card{e(o)} + \sum_{e_o \in e(o)} \sum_{j \in \mathcal{M}_{i}(e_o)} m_{i-1, j}(e_o) \cdot k^j\right) \allowdisplaybreaks \\
& = & \Card{\boxesatzerogeq{R}{i-1}} +\sum_{o \in \boxesatzerogeq{R}{i}} \Card{\boxesatzerogeq{B_R(o)}{i-1}} \cdot \sum e^\#(o) \\
& & + \sum_{o \in \boxesatzerogeq{R}{i+1}} \sum_{e_o \in e(o)} \sum_{j \in \mathcal{M}_{i}(e_o)} m_{i-1, j}(e_o) \cdot k^j \\
& & \textit{(by $P_3(B_R(o), e_o, i)$ for each $o \in \exactboxesatzero{R}{i}$ and each $e_o \in e(o)$)}
\end{eqnarray*}
By $P_1(R, e, i)$ and $P_2(R, e, i)$, we have $\{ j > 0 ; m_{i-1, j}(e) \not= 0 \} = \bigcup_{o \in \boxesatzerogeq{R}{i}} \{ j \in \Nat ; k^j = \sum e^\#(o) \} \cup \bigcup_{o \in \boxesatzerogeq{R}{i+1}} \bigcup_{e_o \in e(o)} \mathcal{M}_{i}(e_o)$, hence $P_6(R, e, i)$ holds.
 \qedhere
\end{itemize}
\end{proof}

\begin{lem}\label{lem: M_i}
Let $R$ be an in-PS. 
Let $k > \Card{\boxes{R}{}}$. 
For any $k$-heterogeneous pseudo-experiment $e$ on $R$, for any $i \in \Nat$, we have $\mathcal{M}_i(e) = \log_k[\bigcup e^\#[\boxesgeq{R}{i}]]$, hence $\mathcal{N}_i(e) = \log_k[\bigcup e^\#[\exactboxes{R}{i}]]$ and $(\isaCopyfrom{R}{e}{i} \circ \cod_{e, i})[\mathcal{N}_i(e)] = \exactboxes{R}{i}$.
\end{lem}

\begin{proof}
By induction on $\depthof{R}$. We have
\begin{eqnarray*}
\mathcal{M}_i(e) & = & \log_k[\bigcup e^\#[\boxesatzerogeq{R}{i}]] \cup \bigcup_{o \in \boxesatzerogeq{R}{i+1}} \bigcup_{e_o \in e(o)} \mathcal{M}_{i}(e_o)\\
& & \textit{(by Lemma~\ref{lem: M_i: aux})} \allowdisplaybreaks\\
& = & \log_k[\bigcup e^\#[\boxesatzerogeq{R}{i}]] \cup \bigcup_{o \in \boxesatzerogeq{R}{i+1}} \bigcup_{e_o \in e(o)} \log_k[\bigcup {e_o}^\#[\boxesgeq{B_R(o)}{i}]]\\
& & \textit{(by the induction hypothesis)} \allowdisplaybreaks\\
& = & \log_k[\bigcup e^\#[\boxesatzerogeq{R}{i}] \cup \bigcup_{o \in \boxesatzerogeq{R}{i+1}} \bigcup_{e_o \in e(o)} \bigcup {e_o}^\#[\boxesgeq{B_R(o)}{i}]]\allowdisplaybreaks\\
& = & \log_k[\bigcup e^\#[\boxesatzerogeq{R}{i}] \cup \bigcup_{o \in \boxesatzerogeq{R}{i+1}} \bigcup_{e_o \in e(o)} \bigcup {e}^\#[\{ o \} \times  \boxesgeq{B_R(o)}{i}]]\allowdisplaybreaks\\
& = & \log_k[\bigcup e^\#[\boxesatzerogeq{R}{i} \cup \bigcup_{o \in \boxesatzerogeq{R}{i+1}} \bigcup_{e_o \in e(o)} (\{ o \} \times  \boxesgeq{B_R(o)}{i})]]\allowdisplaybreaks\\
& = & \log_k[\bigcup e^\#[\boxesgeq{R}{i}]]
\end{eqnarray*}
Hence 
\begin{eqnarray*}
\mathcal{N}_i(e) & = & \mathcal{M}_{i}(e) \setminus \mathcal{M}_{i+1}(e)\\
& = & \log_k[\bigcup e^\#[\boxesgeq{R}{i}]] \setminus \log_k[\bigcup e^\#[\boxesgeq{R}{i+1}]]\\
& = & \log_k[\bigcup e^\#[\exactboxes{R}{i}]]
\end{eqnarray*}
Moreover, since $(\forall j \in \mathcal{N}_i(e)) k^j \in (e^\# \circ \isaCopyfrom{R}{e}{i} \circ \cod_{e, i})(j)$, we obtain $(\isaCopyfrom{R}{e}{i} \circ \cod_{e, i})[\mathcal{N}_i(e)] = \exactboxes{R}{i}$.
\end{proof}

\begin{exa}~\label{example: N_i(e) - 2}
(Continuation of Example~\ref{example: N_i(e)}) We thus have $\mathcal{M}_1(e) = \{ 1, 2 \}$ and \linebreak[4] $\portsatzerooftype{\cod}{\termofTaylor{R}{e}{i}} \setminus \boxesatzero{\termofTaylor{R}{e}{1}} = \linebreak[2] \{ o_2, o_4 \}$ with $\arity{\termofTaylor{R}{e}{1}}(o_2)$ $=$ $10^1$ and $\arity{\termofTaylor{R}{e}{1}}(o_4)$ $=$ $10^2$ (see Figures~\ref{fig: termofTaylor{R}{e}{1}: 1},~\ref{fig: termofTaylor{R}{e}{1}: 2} and~\ref{fig: termofTaylor{R}{e}{1}: 3} - we recall that $\termofTaylor{R}{e}{1} = S1_1 \oplus S1_2 \oplus S1_3$).
\end{exa}

\subsubsection{Determining the contractions immediately below the new boxes}

The set $\criticalports{S}{k}{\mathcal{N}_i(e)}$ of ``critical ports'' is a set of exponential ports that will play a crucial role in our algorithm.

\begin{defi}\label{definition: critical}
Let $S$ be a differential in-PS. Let $ k > 1$. For any $p \in \portsatzero{S}$, we define the sequence $(m_{k, j}(S)(p))_{j \in \Nat} \in \{ 0, \ldots, k-1 \}^{\Nat}$ as follows: $\arity{S}(p) = \sum_{j \in \Nat} m_{k, j}(S)(p) \cdot k^j$. For any $j \in \Nat$, we set $\criticalports{S}{k}{j} = $ $\{ p \in \portsatzero{S} ; m_{k, j}(S)(p) \not= 0 \} \cap \exponentialports{\groundof{S}}$ and, for any $J \subseteq \Nat$, we set $\criticalports{S}{k}{J} =$  $\bigcup_{j \in J} \criticalports{S}{k}{j}$. 
\end{defi}

In particular, for any $j \in \mathcal{M}_i(e)$, we have $\cod_{e, i}(j) \in \criticalports{\termofTaylor{R}{e}{i}}{k}{j}$.

\begin{exa}
We have $\criticalports{\termofTaylor{R}{e}{1}}{10}{1} = \{ p_1, p_4, p_5, p_6, p_7, o_2 \}$ and $\criticalports{\termofTaylor{R}{e}{1}}{10}{2} = \{ p_4, $ $p_5, $ $p_6, $ $p_7, $ $o_4 \}$, where $\termofTaylor{R}{e}{1} = S1_1 \oplus S1_2 \oplus S1_3$ with $S1_1$, $S1_2$ and $S1_3$ depicted in Figures~\ref{fig: termofTaylor{R}{e}{1}: 1},~\ref{fig: termofTaylor{R}{e}{1}: 2} and~\ref{fig: termofTaylor{R}{e}{1}: 3} respectively. So we have $\criticalports{\termofTaylor{R}{e}{1}}{10}{\{ 1, 2\}} = \{ p_1, p_4, p_5, p_6, p_7, o_2, o_4 \}$.
\end{exa}

Critical ports are defined by their arities. Proposition~\ref{prop: critical ports below new boxes} shows that they are exponential ports that are immediately below the ``new'' boxes. 

In particular, this proposition highlights one more essential difference between the $k$-experiments of \cite{phdtortora, injectcoh, LPSinjectivity, k=2} and our $k$-heterogeneous experiments. There, such a $k$-experiment labelling some contraction $p$ with a multiset of cardinality $\sum_j m_j \cdot k^j$ (where $0 \leq m_j < k$ for any $j$) gives the information that immediately above the contraction $p$ there are exactly $m_{j_0}$ series of exactly $j_0$ auxiliary doors. Here, whenever a $k$-heterogeneous experiment labels some contraction $p$ with a multiset of cardinality $\sum_j m_j \cdot k^j$ (where $0 \leq m_j < k$ for any $j$), the integer $j_0$ is not related to the number of auxiliary doors in series anymore; it corresponds, in the case $m_{j_0} > 0$, with the existence of a box that has an occurrence taking $k^{j_0}$ copies of its content, the box having, among all its auxiliary doors, exactly $m_{j_0}$ auxiliary doors that are, each of them, \emph{the first one} (i.e. the deepest one) of a series of auxiliary doors immediately above the contraction $p$. 

By the way, the constraints on the experiments are completely different: The constraints in \cite{phdtortora, injectcoh, LPSinjectivity, k=2} give a lower bound on the arities of the co-contractions, while the constraints here are on the \emph{basis} $k$. For instance, in the case of Figure~\ref{fig: basis}, the co-size is $< 100$ and there are only $2$ boxes, but still it is not enough to consider a $10$-heterogeneous experiments with powers $\geq 2$: by taking an experiment $e$ with $e^\#(o_1) = 10^3$ and $e^\#(o_2) = 10^2$, we get a contraction of arity $9800 = 9 \cdot 10^3 + 8 \cdot 10^2$ in the corresponding term of the Taylor expansion, which, following our decomposition, would correspond to the PS of Figure~\ref{fig: basis2}.

\begin{figure}
\begin{minipage}{0.45\textwidth}
\centering
\resizebox{\textwidth}{!}{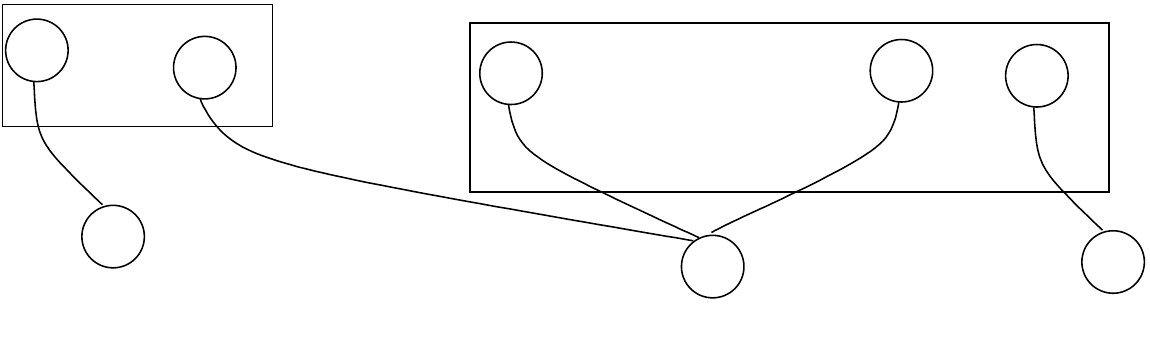}
\captionsetup{width=.9 \textwidth}
\caption{Lower bounding the arities of the co-contractions is not enough}
\label{fig: basis}
\end{minipage}\hfill
\begin{minipage}{0.45\textwidth}
\centering
\resizebox{\textwidth}{!}{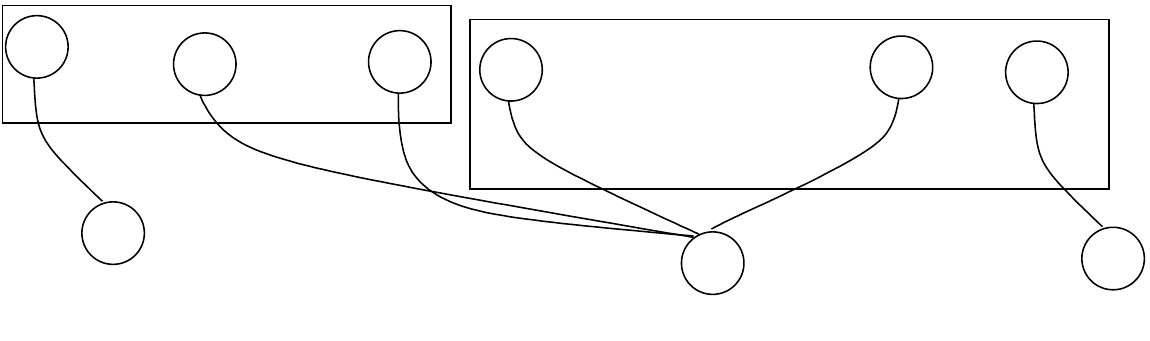}
\captionsetup{width=.9 \textwidth}
\caption{What our decomposition would provide}
\label{fig: basis2}
\end{minipage}\hfill
\end{figure}

\begin{exa}\label{example: outline of boxes}
(Continuation of Example~\ref{example: N_i(e) - 2}) We thus have $!_{f, 1}[\mathcal{N}_1(f)]$ $= \{ o_2, o_4 \}$; and indeed $o_2$ and $o_4$ are the boxes of depth $1$ at depth $0$ of $\termofTaylor{R}{e}{2} \equiv R$ (see Figure~\ref{fig: new_example}). Moreover we have $\criticalports{\termofTaylor{R}{e}{1}}{10}{1} = \{ p_1, p_4, p_5, p_6, p_7, o_2 \}$ and $\criticalports{\termofTaylor{R}{e}{1}}{10}{2} = \{ p_4, p_5, p_6, p_7, o_4 \}$; and indeed, in Figure~\ref{fig: new_example}, we have $\target{\termofTaylor{R}{e}{2}}[\{ o_2 \} \times \temporaryConclusions{\termofTaylor{R}{e}{2}}{o_2}] =$ $\{ p_1, p_4, p_5, p_6, p_7, o_2 \}$ and $\target{\termofTaylor{R}{e}{2}}[\{ o_4 \} \times \temporaryConclusions{\termofTaylor{R}{e}{2}}{o_4}] =$ $\{ p_4, p_5, p_6, p_7, o_4 \}$.
\end{exa}

\begin{defi}
Let $R$ be an in-PS. For any $p \in \portsatzero{R}$, for any $i \in \Nat$, we define a subset $\mathcal{B}_R^{\geq i}(p)$ of $\boxesgeq{R}{i}$: we set $\mathcal{B}_R^{\geq i}(p) = \{ b_R^{\geq i}(q) ; (q \in \portsatdepthgreater{R}{i} \wedge \target{R}(q) = p) \}$.
\end{defi}

A crucial lemma is the following one:

\begin{lem}\label{lem: arity{taylor{e}{i}}(c) for c at depth 0}
Let $R$ be an in-PS. Let $p \in \portsatzero{R}$. Let $k > \cosize{R}$. Let $e$ be a $k$-heterogeneous pseudo-experiment on $R$.  Let $i \in \Nat$. Let $p \in \portsatzero{R}$. Then, for any $j > 0$, we have $p \in \criticalports{\termofTaylor{R}{e}{i}}{k}{j}$ if, and only if, there exists $o \in \mathcal{B}_R^{\geq i}(p)$ such that $k^j \in e^\#(o)$. Moreover we have $\arity{\termofTaylor{R}{e}{i}}(p) \textit{ mod } k = \arity{R^{\leq  i}}(p)$.
\end{lem}

\begin{proof}
By Proposition~\ref{prop: arity}, we have $$\arity{\termofTaylor{R}{e}{i}}(p) = \arity{R^{\leq i}}(p) + \sum_{\substack{p' \in \portsatdepthgreater{R}{i}\\ \target{R}(p') = p}} \sum e^\#(b_R^{\geq i}(p'))$$ Moreover  we have
\begin{itemize}
\item $\arity{R^{\leq i}}(p) < k$;
\item $\Card{\{ p' \in \portsatdepthgreater{R}{i} ; \target{R}(p') = p \}} < k$;
\item and $(\forall p' \in \portsatdepthgreater{R}{i}) (\forall m \in e^\#(b_R^{\geq i}(p'))) (\exists j > 0) k^j \in e^\#(b_R^{\leq i}(p'))$.
\end{itemize}
Hence
\begin{itemize}
\item $\arity{\termofTaylor{R}{e}{i}}(p) \textit{ mod } k = \arity{R^{\leq  i}}(p)$
\item and $p \in \criticalports{\termofTaylor{R}{e}{i}}{k}{j}$ if, and only if, $(\exists p' \in \portsatdepthgreater{R}{i}) (\target{R}(p') = p \wedge k^j \in e^\#(b_R^{\geq i}(p')))$, i.e. $(\exists o \in \mathcal{B}_R^{\geq i}(p)) k^j \in e^\#(o)$.
 \qedhere
\end{itemize}
\end{proof}

\begin{fact}\label{fact: critical ports}
Let $R$ be an in-PS. Let $e$ be a pseudo-experiment on $R$. Let $i \in \Nat$. Let $o \in \exactboxesatzero{\termofTaylor{R}{e}{i+1}}{i}$. Let $q \in \portsatdepthgreater{R}{i}$. Then we have $b_R^{\geq i}(q) = \isaCopyfrom{R}{e}{i+1}(o)$ if, and only if, there exists $q' \in \ports{B_{\termofTaylor{R}{e}{i+1}}(o)}$ such that $q = \isaCopyfrom{R}{e}{i+1}(o, q')$.
\end{fact}

\begin{proof}
By induction on $\depthof{R}$. We assume that $b_R^{\geq i}(q) = \isaCopyfrom{R}{e}{i+1}(o)$ and we distinguish between two cases:
\begin{itemize}
\item $o \in \exactboxesatzero{R}{i}$: we have $\isaCopyfrom{R}{e}{i+1}(o) = o$, hence there exists $q' \in \portsatdepthleq{B_R(o)}{i}$ such that $q = (o, q')$; but $B_R(o) = B_{\termofTaylor{R}{e}{i+1}}(o)$;
\item $o = (o_1, (e_1, o'))$ for some $o_1 \in \boxesatzerogeq{R}{i}$, $e_1 \in e(o_1)$ and $o' \in \exactboxesatzero{\termofTaylor{B_R(o_1)}{i}{e_1}}{i}$: we have $\isaCopyfrom{R}{e}{i+1}(o) = (o_1, \isaCopyfrom{B_R(o_1)}{e_1}{i+1}(o'))$, hence there exists $q_0 \in \portsatdepthgreater{B_R(o_1)}{i}$ such that $q = (o_1, q_0)$ and $b_{B_R(o_1)}^{\geq i}(q_0) = \isaCopyfrom{B_R(o_1)}{e_1}{i+1}(o')$; by induction hypothesis, there exists $q' \in \ports{B_{\termofTaylor{B_R(o_1)}{e_1}{i+1}}(o')}$ such that $q_0 = \isaCopyfrom{B_R(o_1)}{e_1}{i+1}(o_1, q')$; but $B_{\termofTaylor{B_R(o_1)}{e_1}{i+1}}(o') = B_{\termofTaylor{R}{e}{i+1}}(o)$.
\end{itemize}
Conversely, we assume that there exists $q' \in \ports{B_{\termofTaylor{R}{e}{i+1}}(o)}$ such that $q = \isaCopyfrom{R}{e}{i+1}(o, q')$ and we distinguish between two cases:
\begin{itemize}
\item $o \in \exactboxesatzero{R}{i}$: we have $b_R^{\geq i}(\isaCopyfrom{R}{e}{i+1}(o, q')) = b_R^{\geq i}(o, q') = o = \isaCopyfrom{R}{e}{i+1}(o)$;
\item $o = (o_1, (e_1, o'))$ for some $o_1 \in \boxesatzerogeq{R}{i}$, $e_1 \in e(o_1)$ and $o' \in \exactboxesatzero{\termofTaylor{B_R(o_1)}{i}{e_1}}{i}$: we have $b_R^{\geq i}(\isaCopyfrom{R}{e}{i+1}(o, q')) = b_R^{\geq i}(o_1, (\isaCopyfrom{B_R(o_1)}{e_1}{i+1}(o', q'))) = (o_1, b_{B_R(o_1)}^{\geq i}(\isaCopyfrom{B_R(o_1)}{e_1}{i+1}(o', q')))$; since $B_{\termofTaylor{R}{e}{i+1}}(o) = B_{\termofTaylor{B_R(o_1)}{e_1}{i+1}}(o')$, we can apply the induction hypothesis and we thus obtain $b_R^{\geq i}(\isaCopyfrom{R}{e}{i+1}(o, q')) = (o_1, \isaCopyfrom{B_R(o_1)}{e_1}{i+1}(o')) = \isaCopyfrom{R}{e}{i+1}(o)$.
 \qedhere
\end{itemize}
\end{proof}

\begin{prop}\label{prop: critical ports below new boxes}
Let $R$ be an in-PS. 
Let $k > \Card{\boxes{R}{}}, \cosize{R}$. Let $e$ be a $k$-heterogeneous pseudo-experiment on $R$ and let $i \in \Nat$. Then we have $\exactboxesatzero{\termofTaylor{R}{e}{i+1}}{i} = !_{e, i}[\mathcal{N}_{i, e}]$. Moreover, the set $\criticalports{\termofTaylor{R}{e}{i}}{k}{\Nat \setminus (\mathcal{M}_i(e) \cup \{ 0 \})}$ is empty. Furthermore, for any $j \in \mathcal{N}_i(e)$, we have $\criticalports{\termofTaylor{R}{e}{i}}{k}{j} = \target{\termofTaylor{R}{e}{i+1}}[\{ \cod_{e, i}(j) \} \times \temporaryConclusions{\termofTaylor{R}{e}{i+1}}{\cod_{e, i}(j)}]$ and, if $\cod_{e, i}(j) \notin \exactboxesatzero{R}{i}$, then there exist $o \in \boxesatzerogeq{R}{i+1}$ and $e_o \in e(o)$ such that $j \in \mathcal{N}_i(e_o)$ and $\criticalports{\termofTaylor{R}{e}{i}}{k}{j} \setminus \portsatzero{R} = \{ o \} \times ( \{ e_o \} \times \criticalports{\termofTaylor{B_R(o)}{e_o}{i}}{k}{j})$. 
In particular, we have $\criticalports{\termofTaylor{R}{e}{i}}{k}{\mathcal{N}_i(e)} \subseteq \exponentialportsatzero{\termofTaylor{R}{e}{i+1}}$.
\end{prop}

\begin{proof}
It is trivial to check, by induction on $\depthof{R}$, that we have 
$$(\forall o \in \boxesatzero{\termofTaylor{R}{e}{i}}) (\forall p \in \ports{B_{\termofTaylor{R}{e}{i}}(o)}) \isaCopyfrom{R}{e}{i}(o, p) \notin \exactboxes{R}{i} \: \: \: (\ast)$$

Now, we prove, by induction on $\depthof{R}$, that $\exactboxesatzero{\termofTaylor{R}{e}{i+1}}{i} = !_{e, i}[\mathcal{N}_i(e)]$:
\begin{itemize}
\item Let $j \in \mathcal{N}_i(e)$. By Lemma~\ref{lem: M_i}, there exists $o \in \exactboxes{R}{i}$ such that $j \in \log_k[e^\#(o)]$, hence $k^j \in e^\#(o)$. Since $k^j \in e^\#((\isaCopyfrom{R}{e}{i} \circ \cod_{e, i})(j))$, we have $(\isaCopyfrom{R}{e}{i} \circ \cod_{e, i})(j) = o \in \exactboxes{R}{i}$. If $o \in \exactboxesatzero{R}{i}$, then $\cod_{e, i}(j) =o$. Otherwise, there exist $o_1 \in \boxesatzero{R}$ and $o' \in \exactboxes{B_R(o_1)}{i}$ such that $o = (o_1, o')$: By $(\ast)$, there exist $e_1 \in e(o_1)$ and $o'' \in \portsatzero{\termofTaylor{B_R(o_1)}{e_1}{i}}$ such that $\isaCopyfrom{B_R(o_1)}{e_1}{i}(o'') = o'$ and $\cod_{e, i}(j) = (o_1, (e_1, o''))$, hence $o'' = \cod_{e_1, i}(j)$. By induction hypothesis, we have $\cod_{e_1, i}(j) \in \exactboxesatzero{\termofTaylor{B_R(o_1)}{i+1}{e_1}}{i}$, hence $(o_1, (e_1, o'')) \in \exactboxesatzero{\termofTaylor{R}{i+1}{e}}{i}$.
\item Conversely, let $o \in \exactboxesatzero{\termofTaylor{R}{e}{i+1}}{i}$. Let $j \in \log_k[e^\#(o)]$. By Lemma~\ref{lem: M_i}, we have $j \in \mathcal{N}_i(e)$. We have $k^j \in e^\#(o)$ and $k^j \in e^\#((\isaCopyfrom{R}{e}{i} \circ \cod_{e, i})(j))$, hence $(\isaCopyfrom{R}{e}{i} \circ \cod_{e, i})(j) = o$. If $o \in \exactboxesatzero{R}{i}$, then $\cod_{e, i}(j) = o$. Otherwise, there exist $o_1 \in \boxesatzerogeq{R}{i+1}$, $e_1 \in e(o_1)$ and $o' \in \exactboxesatzero{B_{\termofTaylor{B_R(o_1)}{e_1}{i+1}}}{i}$ such that $o = (o_1, (e_1, o'))$. By induction hypothesis, there exists $j \in \mathcal{N}_i(e_1)$ such that $\cod_{e_1, i}(j) = o'$: we have $\cod_{e, i}(j) = (o_1, (e_1, o'))$. 
\end{itemize}

By Lemma~\ref{lem: M_i} and Lemma~\ref{lem: arity{taylor{e}{i}}(c) for c at depth 0}, we have $\criticalports{\termofTaylor{R}{e}{i}}{k}{\Nat \setminus (\mathcal{M}_i(e) \cup \{ 0 \})} \cap \portsatzero{R} = \emptyset$. Moreover, for any $o \in\boxesatzerogeq{R}{i}$, for any $e_o \in e(o)$, again by Lemma~\ref{lem: M_i}, we have $\mathcal{M}_i(e_o) \subseteq \mathcal{M}_i(e)$, hence 
\begin{eqnarray*}
& & \criticalports{\termofTaylor{R}{e}{i}}{k}{\Nat \setminus (\mathcal{M}_i(e) \cup \{ 0 \})} \cap (\{ o \} \times (\{ e_o \} \times \portsatzero{\termofTaylor{B_R(o)}{e_o}{i}}))\\ 
& = & \{ o \} \times (\{ e_o \} \times \criticalports{\termofTaylor{B_R(o)}{e_o}{i}}{k}{\Nat \setminus (\mathcal{M}_i(e) \cup \{ 0 \})})\\
& \subseteq &  \{ o \} \times (\{ e_o \} \times \criticalports{\termofTaylor{B_R(o)}{e_o}{i}}{k}{\Nat \setminus (\mathcal{M}_i(e_o) \cup \{ 0 \})})\allowdisplaybreaks\\
& = & \{ o \} \times (\{ e_o \} \times \emptyset)\\
& = & \emptyset \textit{   (by induction hypothesis)}
\end{eqnarray*}
We showed $\criticalports{\termofTaylor{R}{e}{i}}{k}{\Nat \setminus (\mathcal{M}_i(e) \cup \{ 0 \})} = \emptyset$ $(\ast \ast)$.

Now, let $j \in \mathcal{N}_i(e)$. We distinguish between two cases:
\begin{itemize}
\item $\cod_{e, i}(j) \in \boxesatzero{R}$: By Lemma~\ref{lem: arity{taylor{e}{i}}(c) for c at depth 0}, we have $\criticalports{\termofTaylor{R}{e}{i}}{k}{j} \cap \portsatzero{R} = \{ p \in \portsatzero{R} ; \cod_{e, i}(j) \in \mathcal{B}_R^{\geq i}(p) \} = \target{R}[\{ \cod_{e, i}(j) \} \times \temporaryConclusions{R}{\cod_{e, i}(j)}] = \target{\termofTaylor{R}{e}{i+1}}[\{ \cod_{e, i}(j) \} \times \temporaryConclusions{\termofTaylor{R}{e}{i+1}}{\cod_{e, i}(j)}]$. Moreover, by $(\ast \ast)$, we have $\criticalports{\termofTaylor{R}{e}{i}}{k}{j} \subseteq \portsatzero{R}$. We thus have $\criticalports{\termofTaylor{R}{e}{i}}{k}{j} = \target{\termofTaylor{R}{e}{i+1}}[\{ \cod_{e, i}(j) \} \times \temporaryConclusions{\termofTaylor{R}{e}{i+1}}{\cod_{e, i}(j)}]$.
\item $\cod_{e, i}(j) = (o, (e_o, \cod_{e_o, i}(j)))$ for some $o \in \boxesatzerogeq{R}{i}$ and $e_o \in e(o)$: By induction hypothesis, we have $\criticalports{\termofTaylor{B_R(o)}{e_o}{i}}{k}{j} = \target{\termofTaylor{B_R(o)}{e_o}{i+1}}[\{ \cod_{e_o, i}(j) \} \times \temporaryConclusions{\termofTaylor{B_R(o)}{e_o}{i+1}}{\cod_{e_o, i}(j)}]$. By $(\ast \ast)$, we have $\criticalports{\termofTaylor{R}{e}{i}}{k}{j} \subseteq \portsatzero{R} \cup \portsatzero{R \langle o, i, e_o \rangle}$, hence 
\begin{eqnarray*}
\criticalports{\termofTaylor{R}{e}{i}}{k}{j} \setminus \portsatzero{R} & = & \criticalports{R \langle o, i, e_o \rangle}{k}{j}\\
& = & (\{ o \} \times ( \{ e_o \} \times \criticalports{\termofTaylor{B_R(o)}{e_o}{i}}{k}{j})) \\
& = & \{ o \} \times ( \{ e_o \} \times \target{\termofTaylor{B_R(o)}{e_o}{i+1}}[\{ \cod_{e_o, i}(j) \} \times \temporaryConclusions{\termofTaylor{B_R(o)}{e_o}{i+1}}{\cod_{e_o, i}(j)}])\\
& = & \target{\termofTaylor{R}{e}{i+1}}[\{ \cod_{e, i}(j) \} \times \temporaryConclusions{\termofTaylor{R}{e}{i+1}}{\cod_{e, i}(j)}] \setminus \portsatzero{R}
\end{eqnarray*}
 We have 
\begin{eqnarray*}
\criticalports{\termofTaylor{R}{e}{i}}{k}{j} \cap \portsatzero{R} & = & \{ p \in \portsatzero{R} ; (\exists o \in \mathcal{B}_R^{\geq i}(p)) k^j \in e^\#(o) \}\\
& & \textit{   (by Lemma~\ref{lem: arity{taylor{e}{i}}(c) for c at depth 0})}\\
& = & \{ p \in \portsatzero{R} ; \isaCopyfrom{R}{e}{i}(\cod_{e, i}(j)) \in \mathcal{B}_R^{\geq i}(p) \} \\
& & \textit{   (by Lemma~\ref{lem: bijection cod{e, i}})}\allowdisplaybreaks \\
& = & \{ p \in \portsatzero{R} ;  (\exists q \in \portsatdepthgreater{R}{i}) (b_R^{\geq i}(q) = \isaCopyfrom{R}{e}{i}(\cod_{e, i}(j)) \wedge \target{R}(q) =p) \} \allowdisplaybreaks\\
& = & \{ p \in \portsatzero{R} ;  (\exists q \in \portsatdepthgreater{R}{i}) (b_R^{\geq i}(q) = \isaCopyfrom{R}{e}{i+1}(\cod_{e, i}(j)) \wedge \target{R}(q) =p) \}\\
& & \textit{   (by Lemma~\ref{lem: ground-structure: trivial})}\allowdisplaybreaks \\
& = & \{ \target{R}(\isaCopyfrom{R}{e}{i+1}(\cod_{e, i}(j), q')) ; q' \in \ports{B_{\termofTaylor{R}{e}{i+1}}(\cod_{e, i}(j))} \}\\
& & \textit{   (by Fact~\ref{fact: critical ports})}\allowdisplaybreaks\\
& = & \{ \target{R}(\isaCopyfrom{R}{e}{i+1}(\cod_{e, i}(j), q')) ; q' \in \conclusions{B_{\termofTaylor{R}{e}{i+1}}(\cod_{e, i}(j))} \}\\
& & \textit{   (by Fact~\ref{fact: conclusions})}\\
& = & \target{\termofTaylor{R}{e}{i+1}}[\{ \cod_{e, i}(j) \} \times \temporaryConclusions{\termofTaylor{R}{e}{i+1}}{\cod_{e, i}(j)}] \cap \portsatzero{R}
\end{eqnarray*}
\end{itemize}
\end{proof}

As the following example shows, the information we obtain is already non-trivial, but far away to be strong enough.

\begin{exa}\label{example: border > LPS}
The PS's $R_1$, $R_2$, $R_3$ and $R_4$ of Figure~\ref{fig: R_1}, Figure~\ref{fig: R_2}, Figure~\ref{fig: R_3} and Figure~\ref{fig: R_4} respectively have the same LPS, which is depicted in Figure~\ref{fig: LPS}. But if we know that $p \in \target{R}[\{ o_1 \} \times \temporaryConclusions{R}{o_1}] \cap \target{R}[\{ o_2 \} \times \temporaryConclusions{R}{o_2}]$, then we know that $R \not= R_3$ and $R \not= R_4$: The information we already obtain is an information that is not obtained from $LPS(R)$. Still we are not yet able to distinguish between $R_1$ and $R_2$. Since $LPS(R_1) = LPS(R_2)$ and $\depthof{R_1} = 1 = \depthof{R_2}$, the sets of pseudo-experiments on $R_1$ and $R_2$ coincide. Now, for any pseudo-experiment on these two proof-nets $R_1$ and $R_2$, we have $\arity{\termofTaylor{R_1}{e}{0}}(o_1) = \arity{\termofTaylor{R_2}{e}{0}}(o_1)$, $\arity{\termofTaylor{R_1}{e}{0}}(o_2) = \arity{\termofTaylor{R_2}{e}{0}}(o_2)$ and $\arity{\termofTaylor{R_1}{e}{0}}(p) = \arity{\termofTaylor{R_2}{e}{0}}(p)$, which shows that the arity of exponential ports in the Taylor expansion is not sufficient to recover the PS.

\begin{figure}
\centering
\begin{minipage}[b]{0.45\linewidth}
\centering
\resizebox{\textwidth}{!}{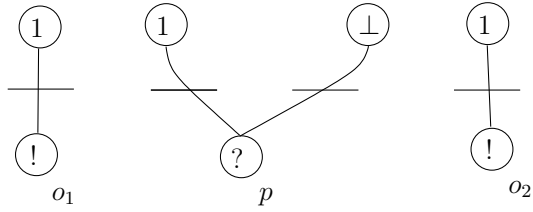}
\captionsetup{width= \textwidth}
\caption{LPS of $R_1$, $R_2$, $R_3$ and $R_4$}\label{fig: LPS}
\end{minipage}
\end{figure}
\begin{figure}
\centering
\begin{minipage}[t]{0.45 \linewidth}
\raggedleft
\resizebox{.9 \textwidth}{!}{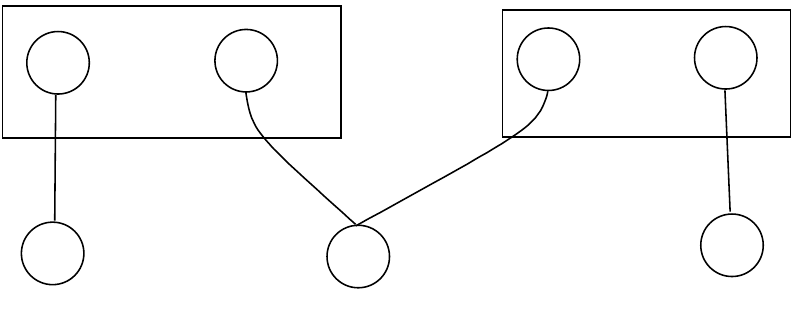}
\caption{$R_1$}\label{fig: R_1}
\end{minipage}\hfill
\begin{minipage}[t]{0.45\linewidth}
\raggedright
\resizebox{.9 \textwidth}{!}{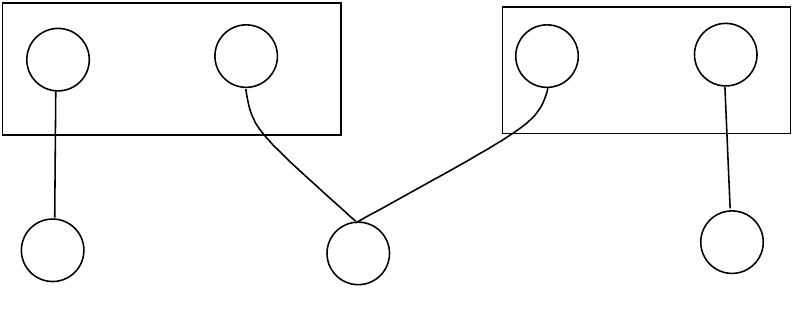}
\caption{$R_2$}\label{fig: R_2}
\end{minipage}
\centering
\begin{minipage}[t]{0.45\linewidth}
\raggedleft
\resizebox{\textwidth}{!}{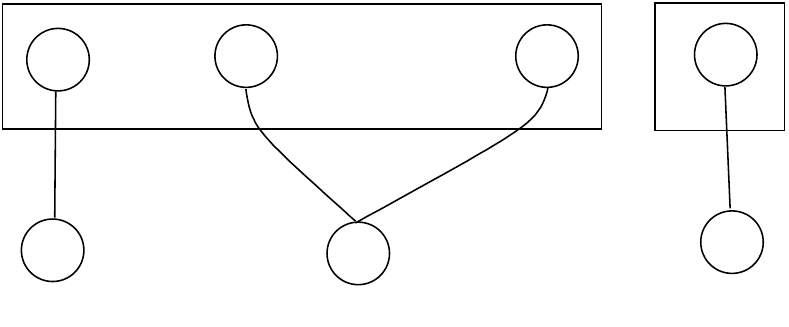}
\caption{$R_3$}
\label{fig: R_3}
\end{minipage}\hfill
\begin{minipage}[t]{0.45 \linewidth}
\raggedright
\resizebox{\textwidth}{!}{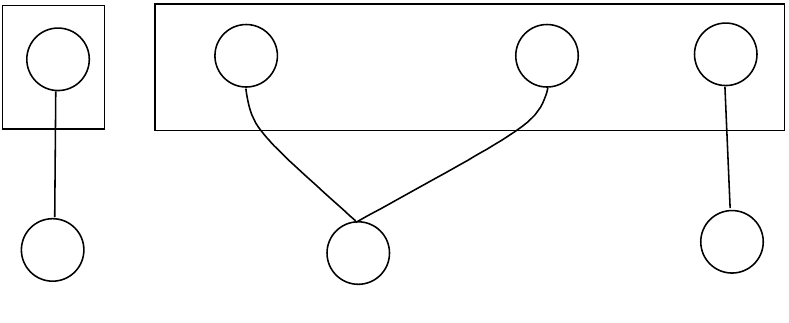}
\caption{$R_4$}
\label{fig: R_4}
\end{minipage}
\end{figure}
\end{exa}

\begin{cor}\label{cor: new critical contractions of R_o}
Let $R$ be an in-PS. Let $o \in \boxesatzero{R}$. Let $\varphi$ be some bijection $\contractionsUnder{R}{o} \simeq \mathcal{Q'}$. Let $R_o$ be an in-PS such that $R_o = \varphi \cdot_o R$. Let $k > \Card{\boxes{R}{}}, \cosize{R}$. Let $e$ be a $k$-heterogeneous pseudo-experiment on $R$, let $e_o \in e(o)$, let $i \in \Nat$ and let $j \in \mathcal{N}_i(e_o)$. Then $\criticalports{\termofTaylor{B_{R_o}(o)}{e_o}{i}}{k}{j} \cap \mathcal{Q'} \subseteq \varphi[\contractionsUnder{R}{o} \cap \criticalports{\termofTaylor{R}{e}{i}}{k}{j}]$.
\end{cor}

\begin{proof}
Notice first that, since the set $\mathcal{N}_i(e_o)$ is non-empty, by Proposition~\ref{prop: critical ports below new boxes}, we have $o \in \boxesatzerogeq{R}{i+1}$. By Proposition~\ref{prop: critical ports below new boxes} again, it is enough to show that 
\begin{eqnarray*}
& & \target{\termofTaylor{B_{R_o}(o)}{e_o}{i+1}}[\{ \cod_{e_o, i}(j) \} \times \temporaryConclusions{\termofTaylor{B_{R_o}(o)}{e_o}{i+1}}{\cod_{e_o, i}(j)}] \cap \mathcal{Q'} \\
& \subseteq & \varphi[\contractionsUnder{R}{o} \cap \target{\termofTaylor{R}{e}{i+1}}[\{ \cod_{e, i}(j) \} \times \temporaryConclusions{\termofTaylor{R}{e}{i+1}}{\cod_{e, i}(j)}]]
\end{eqnarray*}
Now, let $p \in \temporaryConclusions{\termofTaylor{B_{R_o}(o)}{e_o}{i+1}}{\cod_{e_o, i}(j)}$ such that $\target{\termofTaylor{B_{R_o}(o)}{e_o}{i+1}}(\cod_{e_o, i}(j), p) \in \mathcal{Q'}$. By Lemma~\ref{lem: termofTaylor of R_o}, we have $p \notin \temporaryConclusions{\termofTaylor{B_{R_o}(o)}{e_o}{i+1}}{\cod_{e_o, i}(j)}$ and $\target{\termofTaylor{B_{R_o}(o)}{e_o}{i+1}}(\cod_{e_o, i}(j), p) = \varphi(\target{R}(o, \isaCopyfrom{B_{R_o}(o)}{e_o}{i+1}(\cod_{e_o, i}(j), p)))$. Moreover, since $p \in \conclusions{B_{\termofTaylor{B_{R}(o)}{e_o}{i+1}}(\cod_{e_o, i}(j))}$, we obtain $(\cod_{e_o, i}(j), p) \in \conclusions{\termofTaylor{B_R(o)}{e_o}{i+1}}$, hence, by Fact~\ref{fact: conclusions}, $\isaCopyfrom{B_R(o)}{i+1}{e_o}(\cod_{e_o, i}(j), p) \in \conclusions{B_R(o)}$: We thus showed $\isaCopyfrom{B_R(o)}{e_o}{i+1}(\cod_{e_o, i}(j), p) \in \temporaryConclusions{R}{o}$, which entails $p \in \temporaryConclusions{\termofTaylor{R}{e}{i+1}}{\cod_{e, i}(j)}$ and $\target{R}(o, \isaCopyfrom{B_{R_o}(o)}{e_o}{i+1}(\cod_{e_o, i}(j), p)) = \target{\termofTaylor{R}{e}{i+1}}(\cod_{e, i}(j), p)$. 
\end{proof}

\subsection{Connected components}

In order to achieve the rebuilding of the ground-structure of $\termofTaylor{R}{e}{i+1}$ and to recover the content of its boxes, we introduce our notion of \emph{connected component} (Definition~\ref{defin: connected components}). This subsection is devoted to relate connected components of $\termofTaylor{R}{e}{i}$ with connected components of $\termofTaylor{R}{e}{i+1}$ (Proposition~\ref{prop: ground-structure} and Proposition~\ref{prop: crucial}).

The relation $\coh_S$ formalizes the notion of ``connectedness" between two ports of $S$ at depth $0$. But be aware that, here, ``connected'' has nothing to do with ``connected'' in the sense of \cite{LPSinjectivity}: here, any two doors of the same box are always ``connected''.

\begin{defi}
Let $T$ be a differential in-PS. We define the binary relation $\coh_T$ on $\portsatzero{T}$ as follows: for any $p, p' \in \portsatzero{T}$, we have $p \coh_T p'$ iff $\{ p, p' \} \in \axiomsatzero{T} \cup \cutsatzero{T}$ or $(p \in \wiresatzero{T}$ and $p' = \target{\groundof{T}}(p))$ or $(p' \in \wiresatzero{T}$ and $p = \target{\groundof{T}}(p'))$ or $(\exists o \in \boxesatzero{T}) (\exists q, q' \in \temporaryConclusions{T}{o}) \{ p, p' \} = \{ \target{T}(o, q), \target{T}(o, q') \}$.

Let $\mathcal{Q} \subseteq \exponentialportsatzero{T}$ and let $p, p' \in \portsatzero{T}$. A \emph{path in $T$ from $p$ to $p'$ without crossing $\mathcal{Q}$} is a finite sequence $(p_0, \ldots, p_n)$ of elements of $\portsatzero{T}$ such that $p_0 = p$, $p_n = p'$ and, for any $j \in \{ 0, \ldots, n-1 \}$, we have $p_j \coh_T p_{j+1}$ and $(p_j \in \mathcal{Q} \Rightarrow j = 0)$.

We say that $T$ is \emph{connected through ports not in $\mathcal{Q}$} if, for any $p, p' \in \portsatzero{T}$, there exists a path in $T$ from $p$ to $p'$ without crossing $\mathcal{Q}$.
\end{defi}

\begin{defi}
Let $T$ and $S$ two differential in-PS's and let $\mathcal{Q}$ such that $T \sqsubseteq_{\mathcal{Q}} S$. We write $T \trianglelefteq_{\mathcal{Q}} S$ if the following property holds: 
$$(\forall p \in \portsatzero{T} \setminus \mathcal{Q}) (\forall q \in \portsatzero{S}) (p \coh_S q \Rightarrow q \in \portsatzero{T})$$ 
\end{defi}

As we can expect, the connected components of $\termofTaylor{R}{e}{i+1}$ that do not cross any expanded box are exactly the connected components of $\termofTaylor{R}{e}{i}$ that do not cross any expanded box:

\begin{fact}\label{fact: components of R^{<=i}}
Let $R$ be an in-PS. Let $e$ be an exhaustive pseudo-experiment on $R$. Let $i \in \Nat$. Let $\mathcal{Q} \subseteq \portsatzero{R}$. Let $T$ be a differential in-PS such that $T \sqsubseteq_{\mathcal{Q}} R^{\leq i}$. Then we have $T \trianglelefteq_{\mathcal{Q}} \termofTaylor{R}{e}{i}$ if and only if $T \trianglelefteq_{\mathcal{Q}} \termofTaylor{R}{e}{i+1}$.
\end{fact}

\begin{proof}
Assume $T \trianglelefteq_{\mathcal{Q}} \termofTaylor{R}{e}{i}$. Let $p \in \portsatzero{T} \setminus \mathcal{Q}$ and let $q \in \portsatzero{\termofTaylor{R}{e}{i+1}}$ such that $p \coh_{\termofTaylor{R}{e}{i+1}} q$. Let us show that there is no $o \in \exactboxesatzero{R}{i}$ such that $p \in \target{R}[\{ o \} \times \temporaryConclusions{R}{o}]$: Let us assume that there is such an $o$, let $e_o \in e(o)$ (such an $e_o$ exists since $e$ is exhaustive) and let $p' \in \temporaryConclusions{R}{o}$ such that $\target{R}(o, p') = p$; we have $p \coh_{\termofTaylor{R}{e}{i}} (o, (e_o, p'))$, hence $(o, (e_o, p')) \in \portsatzero{T}$, which contradicts $T \sqsubseteq R^{\leq i}$. This entails $p \coh_{\termofTaylor{R}{e}{i}} q$.

Conversely, assume that $T \trianglelefteq_{\mathcal{Q}} \termofTaylor{R}{e}{i+1}$. Let $p \in \portsatzero{T} \setminus \mathcal{Q}$ and let $q \in \portsatzero{\termofTaylor{R}{e}{i}}$ such that $p \coh_{\termofTaylor{R}{e}{i}} q$. If $q = (o, (e_o, q'))$ for some $o \in \exactboxesatzero{R}{i}$, $e_o \in e(o)$ and $q' \in \temporaryConclusions{R}{i}$ such that $\target{R}(o, q') = q$, then $p \in \target{R}[\{ o \} \times \temporaryConclusions{R}{o}]$, hence $p \coh_{\termofTaylor{R}{e}{i+1}} o$; then $o \in \portsatzero{T}$, which contradicts $T \sqsubseteq R^{\leq i}$. We thus have $q \in \portsatzero{\termofTaylor{R}{e}{i+1}}$, hence $p \coh_{\termofTaylor{R}{e}{i+1}} q$.
\end{proof}

The sets $\nontrivialconnected{S}{\mathcal{Q}}{k}$ of \emph{components $T$ of $S$ that are connected \emph{via} other ports than $\mathcal{Q}$, whose conclusions belong to $\mathcal{Q}$ and with $\cosize{T} < k$} will play a crucial role in the algorithm of the rebuilding of $\termofTaylor{R}{e}{i+1}$ from $\termofTaylor{R}{e}{i}$, where we will take for $\mathcal{Q}$ a subset of the critical ports $\criticalports{\termofTaylor{R}{e}{i}}{k}{\mathcal{N}_i(e)}$ that were considered in the previous subsection.\footnote{In presence of cuts, we cannot say any more that these components are ``above'' $\mathcal{Q}$.} The reader already knows that, here, ``connected'' has nothing to do with the ``connected proof-nets" of \cite{LPSinjectivity}: there, the crucial tool used was rather the ``bridges'', which put together two doors of the same copy of some box only if they are connected in the LPS of the proof-net.

\begin{defi}\label{defin: connected components}
Let $k \in \Nat$. 
Let $S$ be a differential in-PS. Let $\mathcal{Q} \subseteq \exponentialportsatzero{S}$. We set 
$$\nontrivialconnected{S}{\mathcal{Q}}{k} = \left\lbrace T \trianglelefteq_\mathcal{Q} S ; \begin{array}{l} 
 (\cosize{T} < k \textit{ and } \conclusions{\groundof{T}} \subseteq \mathcal{Q} \textit{ and} \\ \portsatzero{T} \setminus \mathcal{Q} \not= \emptyset \textit{ and } T \textit{ is connected through ports not in $\mathcal{Q}$})  \end{array} \right\rbrace$$
\end{defi}

\begin{rem}\label{rem: conclusions of connected components}
If $T \in \nontrivialconnected{S}{\mathcal{Q}}{k}$, then $\conclusions{\groundof{T}} = \mathcal{Q} \cap \portsatzero{T}$.
\end{rem}

A port at depth $0$ of $S$ that is not in $\mathcal{Q}$ cannot belong to two different components:

\begin{fact}\label{fact: two components with a common port}
Let $k \in \Nat$. Let $S$ be a differential in-PS. Let $\mathcal{Q} \subseteq \portsatzero{S}$. Let $T, T' \in \nontrivialconnected{S}{\mathcal{Q}}{k}$ such that $(\portsatzero{T} \cap \portsatzero{T'}) \setminus \mathcal{Q} \not= \emptyset$. Then $T = T'$.
\end{fact}

\begin{proof}
By Remark~\ref{rem: unique substructure}, it is enough to check that $\portsatzero{T} = \portsatzero{T'}$.
\end{proof}

\begin{exa}
We have 
\begin{itemize}
\item $\Card{\nontrivialconnected{S}{\{ p_1, p_4, p_5, p_6, p_7, o_2 \}}{10}} = 241$
\item and $\Card{\nontrivialconnected{S}{\{ p_4, p_5, p_6, p_7, o_4 \}}{10}} = 320$
\end{itemize}
with $S = S1_1 \oplus S1_2 \oplus S1_3$, where $S1_1$, $S1_2$ and $S1_3$ are the differential PS's of Figures~\ref{fig: termofTaylor{R}{e}{1}: 1}, \ref{fig: termofTaylor{R}{e}{1}: 2} and \ref{fig: termofTaylor{R}{e}{1}: 3} respectively.
\end{exa}

The connected components we consider do not mix several copies of boxes. More precisely:

\begin{lem}\label{lemma: ports of connected components}
Let $R$ be an in-PS. Let $k > \cosize{R}$. Let $e$ be a $k$-heterogeneous pseudo-experiment on $R$. Let $i \in \Nat$. Let $\mathcal{P} \subseteq \portsatzero{\termofTaylor{R}{e}{i}}$. Let $T \in \nontrivialconnected{\termofTaylor{R}{e}{i}}{\mathcal{P}}{k}$. Let $o \in \boxesatzerogeq{R}{i}$ and $e_o \in e(o)$ such that $\portsatzero{T} \cap \portsatzero{R \langle o, i, e_o \rangle} \not= \emptyset$. Then the following properties hold:
\begin{enumerate}
\item\label{fact: ports} $\portsatzero{T} \subseteq (\target{R}[\{ o \} \times \temporaryConclusions{R}{o}] \cap \mathcal{P}) \cup \portsatzero{R \langle o, i, e_o \rangle}$
\item\label{fact: shallow conclusions} $\portsatzero{T} \cap \target{R}[\{ o \} \times \temporaryConclusions{R}{o}] \subseteq \conclusions{\groundof{T}}$
\item\label{fact: wires} $\wiresatzero{T} \subseteq \portsatzero{R \langle o, i, e_o \rangle}$
\end{enumerate}
\end{lem}

\begin{proof}
Notice first that the two following properties hold:
\begin{enumerate}[label=(\roman*)]
\item\label{property: 1} For any $p \in \temporaryConclusions{R}{o}$, we have $\arity{\termofTaylor{R}{e}{i}}(\target{R}(o, p))  \geq k$. Indeed, for any $p \in \temporaryConclusions{R}{o}$, the set $\{ p' \in \portsatdepthgreater{R}{i} ; \target{R}(p') = p \}$ is non-empty, hence this property is obtained by applying Proposition~\ref{prop: arity}.
\item\label{property: 2} For any $p \in \portsatzero{T} \setminus \mathcal{P}$, we have $\arity{T}(p) = \arity{\termofTaylor{R}{e}{i}}(p)$.
\end{enumerate}
By $\ref{property: 1}$ and $\ref{property: 2}$, we obtain $$\target{R}[\{ o \} \times \temporaryConclusions{R}{o}] \cap \portsatzero{T} \subseteq \mathcal{P} \: \: \: (\ast)$$

Let $p \in \portsatzero{T} \setminus \portsatzero{R \langle o, i, e_o \rangle}$. There exist $q_0 \in \portsatzero{T} \setminus \mathcal{P}$ and a path $(q_0, \ldots, q_n)$ in $T$ from $q_0$ to $p = q_n$ without crossing $\mathcal{P}$. We set $\iota_0 = \min \{ \iota \in \{ 1, \ldots, n \} ; q_{\iota} \notin \portsatzero{R \langle o, i, e_o \rangle} \}$. Since $q_{\iota_0} \coh_ {\termofTaylor{R}{e}{i}} q_{\iota_0+1}$, there exists $p' \in \temporaryConclusions{R}{o}$ such that $q_{\iota_0+1} = \target {R}(o, p')$, hence, by $(\ast)$, $q_{\iota_0+1} \in \mathcal{P}$; we thus have $\iota_0+1 = n$ and then $p = q_n \in \target{R}[\{ o \} \times \temporaryConclusions{R}{o}] \cap \mathcal{P}$. We showed Property~\ref{fact: ports}.

Property~\ref{fact: shallow conclusions} is obtained by applying Property~\ref{fact: ports} and Remark~\ref{rem: conclusions of connected components}.

Finally, Property~\ref{fact: wires} is an immediate consequence of Properties~\ref{fact: ports} and \ref{fact: shallow conclusions}.
\end{proof}

As a consequence, if $T \in \nontrivialconnected{\termofTaylor{R}{e}{i}}{\mathcal{P}}{k}$ for some $k$-heterogeneous pseudo-experiment $e$ on $R$, then $T$ has at most one co-contraction, which is necessarily a conclusion of $T$ (hence an element of $\mathcal{P}$) and of arity $1$.

A PS $R$ with cuts might have some connected components without any conclusion; copies of such components will occur in the PS $\termofTaylor{R}{e}{i}$ and we want to recover from which boxes they come from. That is why we will consider $k$-heterogeneous experiments with $k > \numberInvisibleComponents{R}$, where $\numberInvisibleComponents{R}$ is defined as follows:

\begin{defi}\label{defin: components}
Let $R$ be an in-PS. We set 
$$\connectedcomponents{R} = \left\lbrace T \trianglelefteq_{\emptyset} R ; \begin{array}{l} 
 (\portsatzero{T} \not= \emptyset \wedge T \textit{ is connected through ports not in $\emptyset$})   \end{array} \right\rbrace$$ 

We denote by $\connectedcomponentsContaining{R}$ the function $\ports{R} \to \connectedcomponents{R}$ that associates with every $p \in \ports{R}$ the unique $T \in \connectedcomponents{R}$ such that $p \in \ports{T}$.

We define, by induction on $\depthof{R}$, the integer $\numberInvisibleComponents{R}$ as follows: 
$$\numberInvisibleComponents{R} = \Card{\{ U \in \connectedcomponents{R} ; \conclusions{U} = \emptyset \}} + \sum_{o \in \boxesatzero{R}} \numberInvisibleComponents{B_R(o)}$$
\end{defi}

Notice that, if $R$ is a cut-free in-PS, then $\numberInvisibleComponents{R} = 0$.

\begin{exa}
If $R$ is the PS of Figure~\ref{fig: new_example}, then $\Card{\connectedcomponents{B_R(o_3)}} = 4$ and $\numberInvisibleComponents{R} = 1$.
\end{exa}

The set $\connectedcomponents{R}$ is an alternative way to describe an in-PS $R$:

\begin{fact}\label{fact: connected components}
Let $R$ be an in-PS. We have $R = \bigoplus \connectedcomponents{R}$.
\end{fact}

\begin{defi}\label{defin: k big enough}
Let $R$ be an in-PS. We set 
$$\basis{R} = \max \{ \Card{\boxes{R}}, \cosize{R}, \numberInvisibleComponents{R}, 1 \} + 1$$
\end{defi}

\begin{lem}\label{lemm: components adding contractions}
Let $R$ and $R_o$ be two in-PS's. Let $o \in \boxesatzero{R}$. Let $\varphi$ be a bijection $\contractionsUnder{R}{o} \simeq \mathcal{Q'}$ such that $R_o = \varphi \cdot_o R$. Let $k > \cosize{R}$. There exists a bijection $\theta: \nontrivialconnected{B_{R_o}(o)}{\mathcal{Q'}}{k} \simeq \connectedcomponents{B_R(o)} \setminus \{ \connectedcomponentsContaining{B_R(o)}(\aboveBang{R}{o}) \}$ such that, for any $T \in \nontrivialconnected{B_{R_o}(o)}{\mathcal{Q'}}{k}$, we have $\theta(T) = \overline{T}$.
\end{lem}

\begin{proof}
Notice first $(\forall p, q \in \portsatzero{B_R(o)}) (p \coh_{B_R(o)} q \Rightarrow p \coh_{B_{R_o}(o)} q) (\ast)$.

We check that, for any $T \in \nontrivialconnected{B_{R_o}(o)}{\mathcal{Q'}}{k}$, we have $\overline{T} \in \connectedcomponents{B_R(o)} \setminus \{ \connectedcomponentsContaining{B_R(o)}(\aboveBang{R}{o}) \}$; let $T \in \nontrivialconnected{B_{R_o}(o)}{\mathcal{Q'}}{k}$:
\begin{itemize}
\item By Remark~\ref{rem: substructure and conclusions}, $\portsatzero{\overline{T}} = \portsatzero{T} \setminus \mathcal{Q'} \subseteq (\portsatzero{B_{R_o}(o)} \setminus \{ \aboveBang{R}{o} \}) \setminus \mathcal{Q'} = \portsatzero{B_{R}(o)} \setminus \{ \aboveBang{R}{o} \}$;
\item 
\begin{eqnarray*}
\wiresatzero{\overline{T}} & = & \{ w \in \wiresatzero{T} ; \target{\groundof{T}}(w) \notin \mathcal{Q'} \}\\
&  = & \{ w \in (\wiresatzero{B_{R_o}(o)} \cap \portsatzero{T}) \setminus \mathcal{Q'} ; \target{\groundof{B_{R_o}(o)}}(w) \in \portsatzero{T} \setminus \mathcal{Q'} \}\\
&  = & \{ w \in (\wiresatzero{B_R(o)} \cap \portsatzero{\overline{T}} ; \target{\groundof{B_R(o)}}(w) \in \portsatzero{\overline{T}} \}
\end{eqnarray*}
\item by $(\ast)$, we have $\overline{T} \trianglelefteq_{\emptyset} B_R(o)$
\item and $(\forall p \in \portsatzero{T} \setminus \mathcal{Q'}) (\forall q \in \portsatzero{T} \setminus \mathcal{Q'}) (p \coh_{T} q \Rightarrow p \coh_{\overline{T}} q)$, hence $\overline{T}$ is connected through ports not in $\emptyset$.
\end{itemize}

Now, for any $U \in \connectedcomponents{B_R(o)} \setminus \{ \connectedcomponentsContaining{B_R(o)}(\aboveBang{R}{o}) \}$, we set $T_U = (U \oplus \bigoplus_{p \in \temporaryConclusions{R}{o} \cap \ports{U}} \contr_{\varphi(\target{R}(o, p))})@t$, where $t$ is the function $\temporaryConclusions{R}{o} \cap \ports{U} \to \varphi[\target{R}[\{ o \} \times (\temporaryConclusions{R}{o} \cap \ports{U})]]$ that associates with every $p \in \temporaryConclusions{R}{o} \cap \ports{U}$ the port $\varphi(\target{R}(o, p))$ of $B_{R_o}(o)$, and we check that $T_U  \in \nontrivialconnected{B_{R_o}(o)}{\mathcal{Q'}}{k}$:
\begin{itemize}
\item $\portsatzero{T_U} \subseteq \portsatzero{U} \cup \mathcal{Q'} \subseteq \portsatzero{B_R(o)} \cup \mathcal{Q'} = \portsatzero{B_{R_o}(o)}$
\item 
\begin{eqnarray*}
\wiresatzero{T_U} & = & \wiresatzero{U} \cup (\temporaryConclusions{R}{o} \cap \portsatzero{U})\\
& = & \{ w \in \wiresatzero{B_R(o)} \cap \portsatzero{U} ; \target{\groundof{B_R(o)}}(w) \in \portsatzero{U} \} \cup (\temporaryConclusions{R}{o} \cap \portsatzero{U})\\
&  = & \{ w \in \wiresatzero{B_R(o)} \cap \portsatzero{U} ; \target{\groundof{B_{R_o}(o)}}(w) \in \portsatzero{U} \}\\
&  & \cup \{ w \in \temporaryConclusions{R}{o} \cap \portsatzero{U} ; \target{\groundof{B_{R_o}(o)}}(w) \in \portsatzero{T_U} \}\\
&  = & \{ w \in \wiresatzero{B_{R_o}(o)} \cap \portsatzero{U} ; \target{\groundof{B_{R_o}(o)}}(w) \in \portsatzero{T_U} \}\\
&  = & \{ w \in (\wiresatzero{B_{R_o}(o)} \cap \portsatzero{U}) \setminus (\mathcal{Q'} \cap \exponentialportsatzero{B_{R_o}(o)} ; \target{\groundof{B_{R_o}(o)}}(w) \in \portsatzero{T_U} \}\end{eqnarray*}
\item $\conclusions{\groundof{T_U}} = \{ \varphi(\target{R}(o, p)) ; p \in \temporaryConclusions{R}{o} \cap \ports{U} \} \subseteq \mathcal{Q'}$
\item 
\begin{eqnarray*}
\target{T_U} & = & \restriction{\target{B_{R_o}(o)}}{\bigcup_{o' \in \boxesatzero{U}} (\{ o' \} \times (\temporaryConclusions{B_R(o)}{o'} \cup \{ p \in \ports{B_{B_R(o)}(o')} ; (o', p) \in \temporaryConclusions{R}{o} \}))}\\
& = & \restriction{\target{B_{R_o}(o)}}{\bigcup_{o' \in \boxesatzero{T_U}} (\{ o' \} \times (\temporaryConclusions{B_R(o)}{o'} \cup \{ p \in \ports{B_{B_R(o)}(o')} ; (o', p) \in \temporaryConclusions{R}{o} \}))}\\
& = &\restriction{\target{B_{R_o}(o)}}{\bigcup_{o' \in \boxesatzero{T_U}} (\{ o' \} \times \temporaryConclusions{B_{R_o}(o)}{o'})} 
\end{eqnarray*}
\item by $(\ast)$, we have $T_U \trianglelefteq_{\mathcal{Q'}} B_{R_o}(o)$ and $T_U$ is connected through ports not in $\mathcal{Q'}$.
\end{itemize}
Moreover we have $\overline{T_U} = U$, which shows that $\theta$ is a surjection and, for any $T \in \nontrivialconnected{B_{R_o}(o)}{\mathcal{Q'}}{k}$, we have $T = (\overline{T_U} \oplus \bigoplus_{p \in \temporaryConclusions{R}{o} \cap \ports{\overline{T_U}}} \contr_{\varphi(\target{R}(o, p))})@t$ where $t$ is the function $\temporaryConclusions{R}{o} \cap \ports{\overline{T_U}} \to \varphi[\target{R}[\{ o \} \times (\temporaryConclusions{R}{o} \cap \ports{\overline{T_U}})]]$ that associates with every $p \in \temporaryConclusions{R}{o} \cap \ports{\overline{T_U}}$ the port $\varphi(\target{R}(o, p))$ of $B_{R_o}(o)$, which shows that $\theta$ is an injection.
\end{proof}

\begin{exa}
Let us consider the PS $R$ of Figure~\ref{fig: new_example}. We set $\mathcal{Q'} = \{ p'_1, p'_4, p'_5, p'_6, p'_7 \}$ and we define the bijection $\varphi : \contractionsUnder{R}{o_2} \simeq \mathcal{Q'}$ as follows: $\varphi(p_1) = p'_1$; $\varphi(p_4) = p'_4$; $\varphi(p_5) = p'_5$; $\varphi(p_6) = p'_6$; $\varphi(p_7) = p'_7$. If $k$ is big enough, then $\nontrivialconnected{(\varphi \cdot_{o_2} R)}{\mathcal{Q'}}{k} = \{ H'_1, H'_2, H'_3 \}$, where $H'_1$, $H'_2$ and $H'_3$ and the differential in-PS's of Figure~\ref{fig: H'_1}, Figure~\ref{fig: H'_2} and Figure~\ref{fig: H'_3} respectively.
\begin{figure}
\centering
\begin{minipage}{0.25 \textwidth}
\centering
\resizebox{.35 \textwidth}{!}{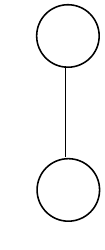}
\captionsetup{width=.9 \textwidth}
\caption{$H'_1$}
\label{fig: H'_1}
\end{minipage}\hfill
\begin{minipage}{0.3\textwidth}
\centering
\resizebox{.7 \textwidth}{!}{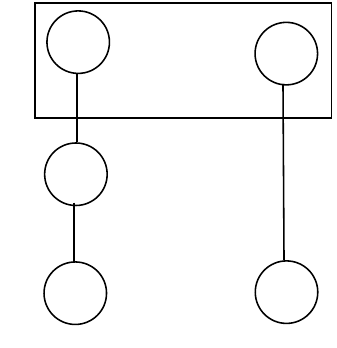}
\caption{$H'_2$}
\label{fig: H'_2}
\end{minipage}\hfill
\begin{minipage}{0.3\textwidth}
\centering
\resizebox{\textwidth}{!}{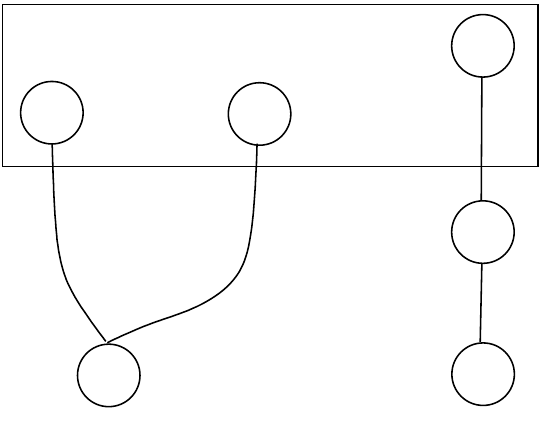}
\caption{$H'_3$}
\label{fig: H'_3}
\end{minipage}\hfill
\end{figure}
\end{exa}

\begin{lem}\label{lem: basis of R_o}
Let $R$ be an in-PS. Let $o \in \boxesatzero{R}$. Let $\mathcal{Q'}$ be some set, let $\varphi$ be some bijection $\contractionsUnder{R}{o} \simeq \mathcal{Q'}$, let $R_o$ be an in-PS obtained from $R$ by adding, according to $\varphi$, contractions as shallow conclusions to the content of the box $o$. Then $\basis{R_o} \leq \basis{R}$.
\end{lem}

\begin{proof}
Let us show that $\numberInvisibleComponents{R_o} \leq \numberInvisibleComponents{R}$:
\begin{itemize}
\item If $\conclusions{\connectedcomponentsContaining{R_o}(o)} = \emptyset$, then $(\{ o \} \times \mathcal{Q'}) \cap \ports{\connectedcomponentsContaining{R_o}(o)} = \emptyset$, hence $\mathcal{Q'} = \emptyset$, which entails that $R_o = R$ and then $\connectedcomponentsContaining{R_o}(o) = \connectedcomponentsContaining{R}(o)$.
\item Moreover we have $\{ T \in \connectedcomponents{B_{R_o}(o)} ; \conclusions{T} = \emptyset \} \subseteq \nontrivialconnected{B_{R_o}(o)}{\mathcal{Q'}}{k}$, hence, by Lemma~\ref{lemm: components adding contractions}, we have $\{ T \in \connectedcomponents{B_{R_o}(o)} ; \conclusions{T} = \emptyset \} \subseteq \connectedcomponents{B_R(o)}$.
 \qedhere
\end{itemize}
\end{proof}

\begin{lem}\label{lem: connected components of R_o are connected components of R}
Let $R$ be an in-PS. Let $k > \cosize{R}$. Let $e$ be a $k$-heterogeneous pseudo-experiment on $R$. Let $i \in \Nat$. Let $\mathcal{P} \subseteq \portsatzero{\termofTaylor{R}{e}{i}}$. Let $o \in \boxesatzerogeq{R}{i}$. We set $\mathcal{Q} = \contractionsUnder{R}{o}$. Let $\varphi_o$ be some bijection $\contractionsUnder{R}{o} \simeq \mathcal{Q'}$ and let $R_o$ be an in-PS such that $R_o = \varphi_o \cdot_o R$. Then, for any $T \in \nontrivialconnected{\termofTaylor{R}{e}{i}}{\mathcal{P}}{k}$,  we have $T \sqsubseteq_{\mathcal{P}} {R_o}^{\leq i}$ if, and only if, $T \sqsubseteq_{\mathcal{P}} R^{\leq i}$. 
\end{lem}

\begin{proof}
Assume $o \in \boxesatzerogeq{R}{i}$. Let $T \in \nontrivialconnected{\termofTaylor{R}{e}{i}}{\mathcal{P}}{k}$. Notice that, since $T \sqsubseteq \termofTaylor{R}{e}{i}$, we have $o \notin \boxesatzero{T}$.

Assume $T \sqsubseteq_{\mathcal{P}} {R_o}^{\leq i}$. We have $T^{\leq 0} \sqsubseteq_{\mathcal{P}} {({R_o}^{\leq i})}^{\leq 0} = {R_o}^{\leq 0} = R^{\leq 0} = {({R}^{\leq i})}^{\leq 0}$. Moreover:
\begin{itemize}
\item $\boxesatzero{T} = \boxesatzero{{R_o}^{\leq i}} \cap \portsatzero{T} = \boxesatzerosmaller{R_o}{i} \cap \portsatzero{T} = \boxesatzerosmaller{R}{i} \cap \portsatzero{T} = \boxesatzero{R^{\leq i}} \cap \portsatzero{T}$;
\item and, since $o \notin \boxesatzero{T}$, we have $B_T = \restriction{B_R}{\boxesatzero{T}}$ and $\target{T} = \restriction{\target{R^{\leq i}}}{\bigcup_{o' \in \boxesatzero{T}} (\{ o' \} \times \temporaryConclusions{R^{\leq i}}{o'})}$.
\end{itemize}
We obtained $T \sqsubseteq_{\mathcal{P}} R^{\leq i}$.

Conversely, assume $T \sqsubseteq_{\mathcal{P}} R^{\leq i}$. We have $T^{\leq 0} \sqsubseteq_{\mathcal{P}} {(R^{\leq i})}^{\leq 0} = R^{\leq 0} = {R_o}^{\leq 0} = {({R_o}^{\leq i})}^{\leq 0}$. Moreover:
\begin{itemize}
\item $\boxesatzero{T} = \boxesatzero{R^{\leq i}} \cap \portsatzero{T} = \boxesatzerosmaller{R}{i} \cap \portsatzero{T} = \boxesatzerosmaller{R_o}{i} \cap \portsatzero{T} = \boxesatzero{{R_o}^{\leq i}} \cap \portsatzero{T}$;
\item and, since $o \notin \boxesatzero{T}$, we have $B_T = \restriction{B_{R_o}}{\boxesatzero{T}}$ and $\target{T} = \restriction{\target{{R_o}^{\leq i}}}{\bigcup_{o' \in \boxesatzero{T}} (\{ o' \} \times \temporaryConclusions{{R_o}^{\leq i}}{o'})}$.
\end{itemize}
We obtained $T \sqsubseteq_{\mathcal{P}} {R_o}^{\leq i}$.
\end{proof}

The proof of the following lemma is postponed in the appendix.

\begin{lem}\label{lem: connected components of R_o}
Let $R$ be an in-PS. Let $k > \cosize{R}$. Let $e$ be a $k$-heterogeneous pseudo-experiment on $R$. Let $i \in \Nat$. Let $\mathcal{P} \subseteq \portsatzero{\termofTaylor{R}{e}{i}}$. Let $o \in \boxesatzerogeq{R}{i}$. We set $\mathcal{Q} = \contractionsUnder{R}{o}$. Let $\varphi_o$ be some bijection $\mathcal{Q} \simeq \mathcal{Q}'$ and let $R_o$ be an in-PS such that $R_o = \varphi_o \cdot_o R$. Let $\varphi_{e_o}$ be the bijection $\mathcal{Q} \simeq \{ o \} \times (\{ e_o \} \times \mathcal{Q'})$ defined by $\varphi_{e_o}(p) = (o, (e_o, \varphi_o(p)))$ for any $p \in \mathcal{Q}$. We set $\mathcal{P}_{e_o} = (\mathcal{P} \setminus \mathcal{Q}) \cup \varphi_{e_o}[\mathcal{P} \cap \mathcal{Q}]$. Then:
\begin{enumerate}
\item\label{lem: connected components of R_o item: [varphi]} for any $T \in \nontrivialconnected{\termofTaylor{R}{e}{i}}{\mathcal{P}}{k}$ such that $\portsatzero{T} \cap \portsatzero{R \langle o, i, e_o \rangle} \not= \emptyset$, we have $T[\varphi_{e_o}] \in \nontrivialconnected{\termofTaylor{R_o}{e}{i}}{\mathcal{P}_{e_o}}{k}$;
\item\label{lem: connected components of R_o item: [inverse of varphi]} for any $T \in \nontrivialconnected{\termofTaylor{R_o}{e}{i}}{\mathcal{P}_{e_o}}{k}$ such that $\portsatzero{T} \subseteq \portsatzero{R_o \langle o, i, e_o \rangle}$, we have $T[{\varphi_{e_o}}^{-1}] \in \nontrivialconnected{\termofTaylor{R}{e}{i}}{\mathcal{P}}{k}$;
\item\label{lem: connected components of R_o item: R_o -> R} for any $T \in \nontrivialconnected{\termofTaylor{R_o}{e}{i}}{\mathcal{P}}{k}$ such that $(\forall e_1 \in e(o)) \portsatzero{T} \cap \portsatzero{R_o \langle o, i, e_1 \rangle} = \emptyset$, we have $T \in \nontrivialconnected{\termofTaylor{R}{e}{i}}{\mathcal{P}}{k}$;
\item\label{lem: connected components of R_o item: R -> R_o} for any $T \in \nontrivialconnected{\termofTaylor{R}{e}{i}}{\mathcal{P}}{k}$ such that $(\forall e_1 \in e(o)) \portsatzero{T} \cap \portsatzero{R_o \langle o, i, e_1 \rangle} = \emptyset$, we have $T \in \nontrivialconnected{\termofTaylor{R_o}{e}{i}}{\mathcal{P}}{k}$;
\end{enumerate}
\end{lem}

The set of ``critical components'' we will consider in our algorithm is the set $$\bigcup_{j \in \mathcal{N}_i(e)} \nontrivialconnected{\termofTaylor{R}{e}{i}}{\criticalports{\termofTaylor{R}{e}{i}}{k}{j}}{k}$$ 
From this set we will build the contents of the new boxes. In particular, each port of $\termofTaylor{R}{e}{i}$ at depth $0$ that goes inside a new box of $\termofTaylor{R}{e}{i+1}$ (i.e. each element of $\portsatzero{\termofTaylor{R}{e}{i}} \setminus \portsatzero{\termofTaylor{R}{e}{i+1}}$) belongs to a ``critical component'':

\begin{lem}\label{lem: ground-structure}
Let $R$ be a PS. Let $k > \Card{\boxes{R}}, \cosize{R}$. Let $e$ be a $k$-heterogeneous experiment on $R$. Let $i \in \Nat$. For any $p \in \portsatzero{\termofTaylor{R}{e}{i}} \setminus \portsatzero{\termofTaylor{R}{e}{i+1}}$, there exists $T \in \bigcup_{j \in \mathcal{N}_i(e)} \nontrivialconnected{\termofTaylor{R}{e}{i}}{\criticalports{\termofTaylor{R}{e}{i}}{k}{j}}{k}$ such that $$p \in \portsatzero{T} \setminus \criticalports{\termofTaylor{R}{e}{i}}{k}{j} \subseteq \portsatzero{\termofTaylor{R}{e}{i}} \setminus \portsatzero{\termofTaylor{R}{e}{i+1}}$$
\end{lem}

\begin{proof}
By induction on $\depthof{R}$. If $\depthof{R} = 0$, then $\portsatzero{\termofTaylor{R}{e}{i}} \setminus \portsatzero{\termofTaylor{R}{e}{i+1}} = \emptyset$. Assume that $\depthof{R} > 0$ and let $p \in \portsatzero{\termofTaylor{R}{e}{i}} \setminus \portsatzero{\termofTaylor{R}{e}{i+1}}$. There exist $o \in \boxesatzerogeq{R}{i}$ and $e_o \in e(o)$ such that $p \in \portsatzero{R \langle o, i, e_o \rangle}$. We distinguish between two cases:
\begin{itemize}
\item $o \in \exactboxesatzero{R}{i}$: By Proposition~\ref{prop: critical ports below new boxes}, there exists $j_0 \in \mathcal{N}_i(e)$ such that $\cod_{e, i}(j_0) = o$. We have $R \langle o, i, e_o \rangle = \langle o, \langle e_o, B_R(o) \rangle \rangle$, hence there exists $T' \in \connectedcomponents{B_R(o)}$ such that $p \in \portsatzero{\langle o, \langle e_o, T' \rangle \rangle}$. We set $T = (\langle o, \langle e_o, T' \rangle \rangle \oplus \sum_{p \in \conclusions{T'}} {\labelofport{\groundof{R}}(\target{R}(o, p))}_{\target{R}(o, p)})@\restriction{\target{R}}{\{ o \} \times \conclusions{T'}}$. We have $\conclusions{\groundof{T}} = \target{R}[\{ o \} \times \conclusions{T'}] = \target{\termofTaylor{R}{e}{i+1}}[\{ o \} \times \conclusions{T'}] \subseteq \criticalports{\termofTaylor{R}{e}{i}}{k}{j_0}$ (by Proposition~\ref{prop: critical ports below new boxes}), hence $T \in \nontrivialconnected{\termofTaylor{R}{e}{i}}{\criticalports{\termofTaylor{R}{e}{i}}{k}{j_0}}{k}$.  
\item $o \in \boxesatzerogeq{R}{i+1}$: There exists $p' \in \portsatzero{\termofTaylor{B_R(o)}{e_o}{i}} \setminus \portsatzero{\termofTaylor{B_R(o)}{e_o}{i+1}}$ such that $p = (o, (e_o, p'))$. Let $R_o$ be an in-PS such that $R_o = \varphi \cdot_o R$, where $\varphi$ is some bijection $\contractionsUnder{R}{o} \simeq \mathcal{Q}'$. We have $p' \in \portsatzero{\termofTaylor{B_{R_o}(o)}{e_o}{i}} \setminus \portsatzero{\termofTaylor{B_{R_o}(o)}{e_o}{i+1}}$, hence, by induction hypothesis, there exists $T' \in \bigcup_{j \in \mathcal{N}_i(e_o)} \nontrivialconnected{\termofTaylor{B_{R_o}(o)}{e_o}{i}}{\criticalports{\termofTaylor{B_{R_o}(o)}{e_o}{i}}{k}{j}}{k}$ such that $p' \in \portsatzero{T'}$; let $j_0 \in \mathcal{N}_i(e_o)$ such that $T' \in \nontrivialconnected{\termofTaylor{B_{R_o}(o)}{e_o}{i}}{\criticalports{\termofTaylor{B_{R_o}(o)}{e_o}{i}}{k}{j_0}}{k}$. By Corollary~\ref{cor: arity for R_o}, we have $T' \in \nontrivialconnected{\termofTaylor{B_{R_o}(o)}{e_o}{i}}{(\mathcal{Q}' \cap \criticalports{\termofTaylor{B_{R_o}(o)}{e_o}{i}}{k}{j_0}) \cup \criticalports{\termofTaylor{B_{R}(o)}{e_o}{i}}{k}{j_0}}{k}$. By Corollary~\ref{cor: new critical contractions of R_o}, we have $T' \in \nontrivialconnected{\termofTaylor{B_{R_o}(o)}{e_o}{i}}{\varphi[\contractionsUnder{R}{o} \cap \criticalports{\termofTaylor{R}{e}{i}}{k}{j_0}] \cup \criticalports{\termofTaylor{B_{R}(o)}{e_o}{i}}{k}{j_0}}{k}$. We denote by $\varphi_{e_o}$ the bijection $\contractionsUnder{R}{o} \simeq \{ o \} \times (\{ e_o \} \times \mathcal{Q}')$ that associates $(o, (e_o, \varphi(q)))$ with every $q \in \contractionsUnder{R}{o}$. By Proposition~\ref{prop: critical ports below new boxes}, we have $\langle o, \langle e_o, T' \rangle \rangle \in \nontrivialconnected{\termofTaylor{R_o}{e}{i}}{\varphi_{e_o}[\contractionsUnder{R}{o} \cap \criticalports{\termofTaylor{R}{e}{i}}{k}{j_0}] \cup \criticalports{\termofTaylor{R}{e}{i}}{k}{j_0}}{k}$. We set $T = \langle o, \langle e_o, T' \rangle \rangle [{\varphi_{e_o}}^{-1}]$. By Lemma~\ref{lem: connected components of R_o}~(\ref{lem: connected components of R_o item: [inverse of varphi]}), we have $T \in \nontrivialconnected{\termofTaylor{R}{e}{i}}{(\contractionsUnder{R}{o} \cap \criticalports{\termofTaylor{R}{e}{i}}{k}{j_0}) \cup \criticalports{\termofTaylor{R}{e}{i}}{k}{j_0}}{k}$; we thus have $T \in \nontrivialconnected{\termofTaylor{R}{e}{i}}{\criticalports{\termofTaylor{R}{e}{i}}{k}{j_0}}{k}$.
 \qedhere
\end{itemize}
\end{proof}

Notice that not each port belonging to a ``critical component'' goes inside a new box. That is why, for now, we are not yet able to describe exactly the ground-structure of ${\termofTaylor{R}{e}{i+1}}$ but we can only obtain an approximant of ${\termofTaylor{R}{e}{i+1}}^{\leq 0}$:

\begin{prop}\label{prop: ground-structure}
Let $R$ be a PS. Let $k > \Card{\boxes{R}}, \cosize{R}$. Let $e$ be a $k$-heterogeneous pseudo-experiment on $R$. Let $i \in \Nat$. We set 
$$\mathcal{P} = \left(\portsatzero{\termofTaylor{R}{e}{i}} \setminus \bigcup_{j \in \mathcal{N}_i(e)} \bigcup_{\begin{array}{c} T \in \nontrivialconnected{\termofTaylor{R}{e}{i}}{\criticalports{\termofTaylor{R}{e}{i}}{k}{j}}{k} \end{array}} \portsatzero{T}\right) \cup \criticalports{\termofTaylor{R}{e}{i}}{k}{\mathcal{N}_i(e)}$$ Then we have $\restriction{{\termofTaylor{R}{e}{i}}^{\leq 0}}{\mathcal{P}} \sqsubseteq_{\emptyset} {\termofTaylor{R}{e}{i+1}}^{\leq 0}$.
\end{prop}

\begin{proof}
By Lemma~\ref{lem: ground-structure: trivial}, we have:
\begin{enumerate}
\item\label{prop1}  $\wiresatzero{\termofTaylor{R}{e}{i+1}} \subseteq \wiresatzero{\termofTaylor{R}{e}{i}}$ and $\target{\groundof{\termofTaylor{R}{e}{i+1}}} = \restriction{\target{\groundof{\termofTaylor{R}{e}{i}}}}{\wiresatzero{\termofTaylor{R}{e}{i+1}}}$ 
\item\label{prop2} $\wiresatzero{\termofTaylor{R}{e}{i}} \cap \portsatzero{\termofTaylor{R}{e}{i+1}} \subseteq \wiresatzero{\termofTaylor{R}{e}{i+1}}$ 
\end{enumerate}

By~\ref{prop1}., we have 
\begin{eqnarray*}
& & \{ w \in \wires{\groundof{\termofTaylor{R}{e}{i+1}}} \cap \mathcal{P} ; \target{\groundof{\termofTaylor{R}{e}{i+1}}}(w) \in \mathcal{P} \}\\
& \subseteq & \{ w \in \wires{\groundof{\termofTaylor{R}{e}{i}}} \cap \mathcal{P} ; \target{\groundof{\termofTaylor{R}{e}{i}}}(w) \in \mathcal{P} \}
\end{eqnarray*}
By Lemma~\ref{lem: ground-structure} and Proposition~\ref{prop: critical ports below new boxes}, we have $\mathcal{P} \subseteq \portsatzero{\termofTaylor{R}{e}{i+1}}$ $(\ast)$. By $(\ast)$ and \ref{prop2}., we have 
\begin{eqnarray*}
& & \{ w \in \wires{\groundof{\termofTaylor{R}{e}{i}}} \cap \mathcal{P} ; \target{\groundof{\termofTaylor{R}{e}{i}}}(w) \in \mathcal{P} \}\\
& \subseteq & \{ w \in \wires{\groundof{\termofTaylor{R}{e}{i+1}}} \cap \mathcal{P} ; \target{\groundof{\termofTaylor{R}{e}{i+1}}}(w) \in \mathcal{P} \}
\end{eqnarray*}
We thus showed
\begin{align*} 
& \wires{\restriction{\groundof{\termofTaylor{R}{e}{i}}}{\mathcal{P}}}\\
={} & \{ w \in \wires{\groundof{\termofTaylor{R}{e}{i}}} \cap \mathcal{P} ; \target{\groundof{\termofTaylor{R}{e}{i}}}(w) \in \mathcal{P} \}\\
={} & \{ w \in \wires{\groundof{\termofTaylor{R}{e}{i+1}}} \cap \mathcal{P} ; \target{\groundof{\termofTaylor{R}{e}{i+1}}}(w) \in \mathcal{P} \}
 \tag*{\qedhere}
\end{align*}
\end{proof}

\begin{exa}\label{example: approximant}
(Continuation of Example~\ref{example: outline of boxes}) Figure~\ref{fig: approximant_of_ground_of_R} depicts the differential PS $\restriction{{\termofTaylor{R}{e}{1}}^{\leq 0}}{\mathcal{P}}$ obtained by applying Proposition~\ref{prop: ground-structure}: We have 
$$\mathcal{P} = \left(\portsatzero{\termofTaylor{R}{e}{1}} \setminus \bigcup_{j \in \mathcal{N}_1(e)} \bigcup_{\begin{array}{c} T \in \nontrivialconnected{\termofTaylor{R}{e}{1}}{\criticalports{\termofTaylor{R}{e}{1}}{10}{j}}{k} \end{array}} \portsatzero{T}\right) \cup \criticalports{\termofTaylor{R}{e}{1}}{10}{\mathcal{N}_1(e)}$$
with $\criticalports{\termofTaylor{R}{e}{i}}{10}{\mathcal{N}_1(e)} = \{ p_1, p_4, p_5, p_6, p_7, o_2, o_4 \}$.
\begin{figure}
\centering
\resizebox{\textwidth}{!}{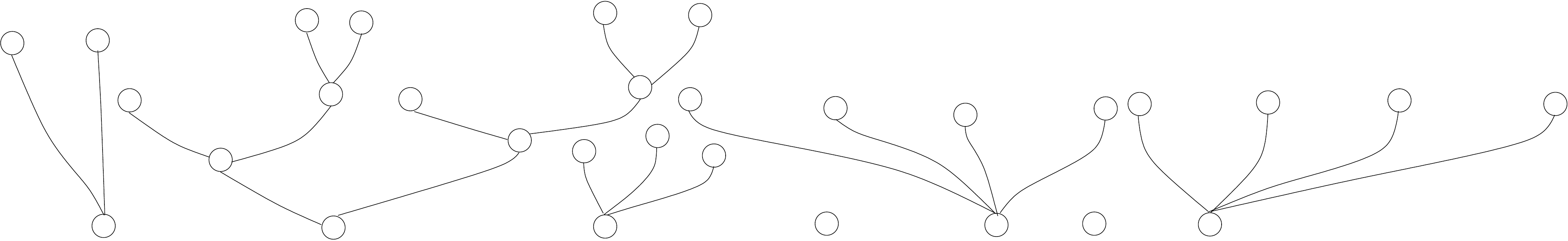}
\caption{The differential PS ${\termofTaylor{R}{e}{1}}^{\leq 0}$}
\label{fig: Ground_of_T_e_1}
\centering
\resizebox{\textwidth}{!}{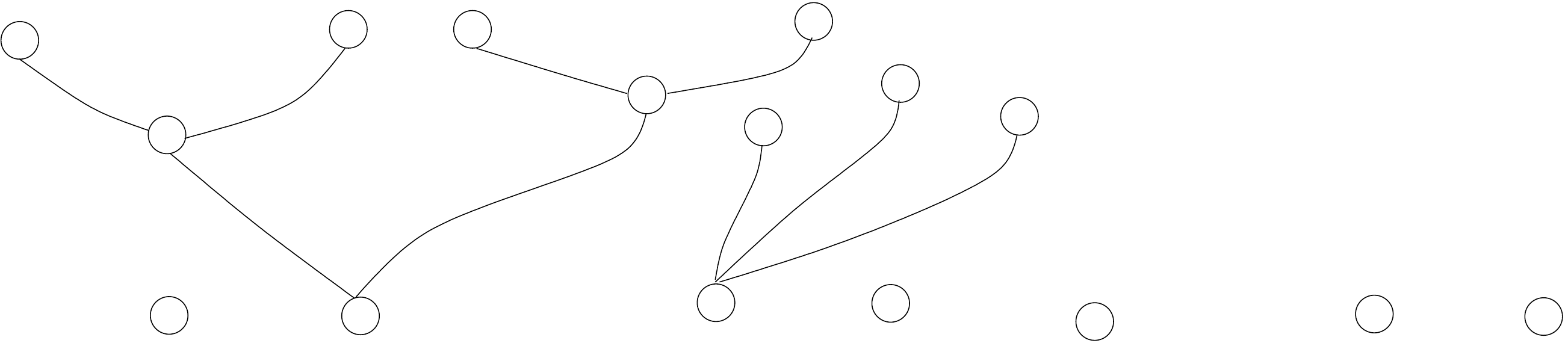}
\caption{The approximant $\restriction{{\termofTaylor{R}{e}{1}}^{\leq 0}}{\mathcal{P}}$ of the differential PS ${\termofTaylor{R}{e}{2}}^{\leq 0}$}
\label{fig: approximant_of_ground_of_R}
\centering
\resizebox{\textwidth}{!}{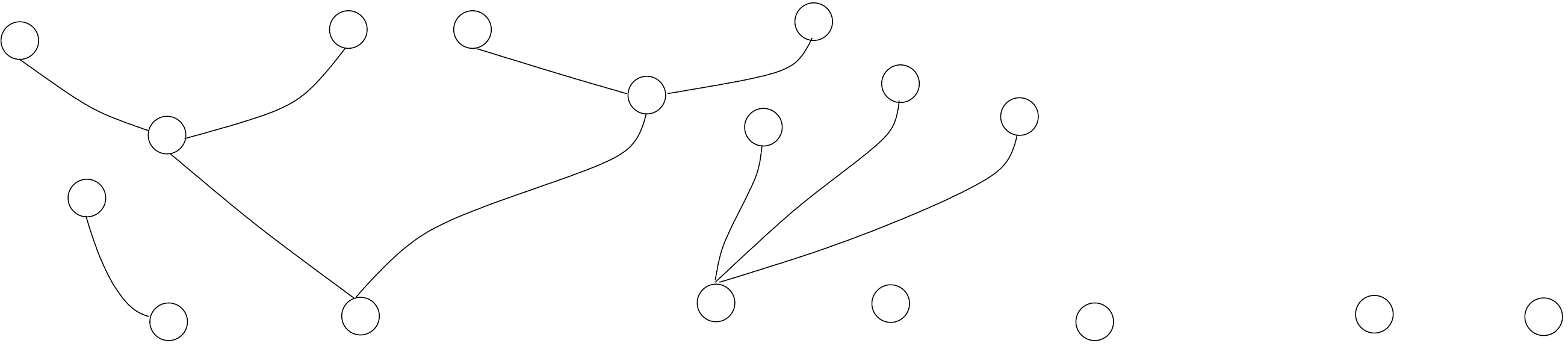}
\caption{The differential PS ${\termofTaylor{R}{e}{2}}^{\leq 0}$}
\label{fig: Ground_of_R}
\end{figure}
\end{exa}

There is no other connected component of $\termofTaylor{R}{e}{i+1}$ whose conclusions belong to critical ports of $\termofTaylor{R}{e}{i}$ than critical components of $\termofTaylor{R}{e}{i}$:

\begin{prop}\label{prop: connected components of termofTaylor{R}{e}{i+1}}
Let $R$ be an in-PS. Let $k > \cosize{R}$. Let $e$ be a $k$-heterogeneous pseudo-experiment on $R$. Let $i \in \Nat$. Then $\nontrivialconnected{\termofTaylor{R}{e}{i+1}}{\criticalports{\termofTaylor{R}{e}{i}}{k}{\mathcal{N}_i(e)}}{k} \subseteq \nontrivialconnected{\termofTaylor{R}{e}{i}}{\criticalports{\termofTaylor{R}{e}{i}}{k}{\mathcal{N}_i(e)}}{k}$.
\end{prop}

\begin{proof}
Let $T \in \nontrivialconnected{\termofTaylor{R}{e}{i+1}}{\criticalports{\termofTaylor{R}{e}{i}}{k}{\mathcal{N}_i(e)}}{k}$. Notice that $\depthof{T} < i$. We have 
$$T \sqsubseteq_{\criticalports{\termofTaylor{R}{e}{i}}{k}{\mathcal{N}_i(e)}} \termofTaylor{R}{e}{i+1},$$ 
hence, by Remark~\ref{rem: substructure}, Lemma~\ref{lem: ground-structure: trivial} and Fact~\ref{fact: transitivity of substructures}, $T = T^{\leq i} \sqsubseteq_{\criticalports{\termofTaylor{R}{e}{i}}{k}{\mathcal{N}_i(e)}} \termofTaylor{R}{e}{i}$.

If $T \sqsubseteq_{\criticalports{\termofTaylor{R}{e}{i}}{k}{\mathcal{N}_i(e)}} R^{\leq i}$, then one can apply Fact~\ref{fact: components of R^{<=i}}. 

Otherwise: Let $p \in \portsatzero{T} \setminus \criticalports{\termofTaylor{R}{e}{i}}{k}{\mathcal{N}_i(e)}$ and $q \in \portsatzero{\termofTaylor{R}{e}{i}} \setminus \portsatzero{\termofTaylor{R}{e}{i+1}}$ such that $p \coh_{\termofTaylor{R}{e}{i}} q$. By Lemma~\ref{lem: ground-structure}, there exists $T' \in \nontrivialconnected{\termofTaylor{R}{e}{i}}{\criticalports{\termofTaylor{R}{e}{i}}{k}{\mathcal{N}_i(e)}}{k}$ such that $p, q \in \portsatzero{T'} \setminus \criticalports{\termofTaylor{R}{e}{i}}{k}{\mathcal{N}_i(e)} \subseteq \portsatzero{\termofTaylor{R}{e}{i}} \setminus \portsatzero{\termofTaylor{R}{e}{i+1}}$, which contradicts $p \in \portsatzero{T} \subseteq \portsatzero{\termofTaylor{R}{e}{i+1}}$.
\end{proof}

We will apply Proposition~\ref{prop: crucial} with $\mathcal{P} = \criticalports{\termofTaylor{R}{e}{i}}{k}{j_0}$, where $j_0 \in \mathcal{N}_i(e)$, but, since the proof is by induction, we need to slightly generalize it.

\begin{prop}\label{prop: crucial}
Let $R$ be an in-PS. Let $k \geq \basis{R}$. 
Let $e$ be a $k$-heterogeneous pseudo-experiment on $R$. Let $i \in \Nat$. Let $\mathcal{P} \subseteq \exponentialportsatzero{\termofTaylor{R}{e}{i+1}} \setminus \boxesatzero{\termofTaylor{R}{e}{i}}$. Let $T \in \nontrivialconnected{\termofTaylor{R}{e}{i}}{\mathcal{P}}{k}$ such that $\conclusions{T} \subseteq \conclusions{\groundof{T}}$. We set 
\begin{itemize}
\item $\mathcal{T} = \{ T' \in \nontrivialconnected{\termofTaylor{R}{e}{i}}{\mathcal{P}}{k}; T \equiv T' \}$
\item $\mathcal{T'} = \{ T' \in \nontrivialconnected{\termofTaylor{R}{e}{i+1}}{\mathcal{P}}{k} ; T \equiv T' \}$
\item $\mathcal{B} = \{ o \in \boxesatzerogeq{R}{i} ; \conclusions{\groundof{T}} \cap \target{R}[\{ o \} \times \temporaryConclusions{R}{o}] \not= \emptyset \}$
\item $\mathcal{B'} = \{ o \in \boxesatzerogeq{R}{i+1} ; \conclusions{\groundof{T}} \cap \target{R}[\{ o \} \times \temporaryConclusions{R}{o}] \not= \emptyset \}$
\end{itemize}
Let $(m_j)_{j \in \Nat}, (m'_j)_{j \in \Nat} \in \{ 0, \ldots, k-1 \}^\Nat$ such that $\Card{\mathcal{T}}  = \sum_{j \in \Nat} m_j \cdot k^j$ 
and $ \Card{\mathcal{T'}} = \sum_{j \in \Nat} m'_j \cdot k^j$. Then the following properties hold:
\begin{itemize}
\item $\{ j \in \Nat \setminus \{ 0 \} ; m_j \not= 0 \} \subseteq \mathcal{M}_i(e)$
\item $\{ j \in \Nat \setminus \{ 0 \} ; m'_j \not= 0 \} \subseteq \mathcal{M}_{i+1}(e)$
\item $(\forall j \in \mathcal{M}_{i+1}(e)) m'_j = m_j$
\item $(\forall j \in \mathcal{N}_i(e)) m_j = \Card{\{ U \in \connectedcomponents{B_{\termofTaylor{R}{e}{i+1}}(\cod_{e, i}(j))} ; U \equiv_{(\termofTaylor{R}{e}{i+1}, \cod_{e, i}(j))} T \}}$
\end{itemize}
Moreover, if $\mathcal{B} \not= \emptyset$, then the following properties hold:
\begin{itemize}
\item $\cod_{e, i}[\{ j \in \Nat \setminus \{ 0 \} ; m_j \not= 0 \}] \subseteq \mathcal{B} \cup \bigcup_{o \in \mathcal{B}} \bigcup_{e_o \in e(o)} \cod_{e, i}[\mathcal{M}_i(e_o)]$
\item $\cod_{e, i}[\{ j \in \Nat \setminus \{ 0 \} ; m'_j \not= 0 \}] \subseteq \mathcal{B'} \cup \bigcup_{o \in \mathcal{B'}} \bigcup_{e_o \in e(o)} \cod_{e, i}[\mathcal{M}_{i+1}(e_o)]$
\end{itemize}
Finally, if $\mathcal{P} \subseteq \portsatzero{R}$, then $m_0 = \Card{\{ T' \in \mathcal{T} ; T' \sqsubseteq_{\mathcal{P}} R^{\leq i} \}} = m'_0$.
\end{prop}

\begin{proof} 
We prove the proposition by induction on $\left(\depthof{R}, \Card{\mathcal{B}}\right)$ lexicographically ordered. Part~I) is devoted to the case where $\depthof{R} = 0$, Part~II) is devoted to the case where $\depthof{R} > 0$ and $\mathcal{B} = \emptyset$, Part~III) is devoted to the case where $\depthof{R} > 0$, $\mathcal{B} \not= \emptyset$ and $\conclusions{\groundof{T}} \cap \boxesatzerogeq{R}{i} = \emptyset$, and Part IV) is devoted to the case where $\depthof{R} > 0$, $\mathcal{B} \not= \emptyset$ and $\conclusions{\groundof{T}} \cap \boxesatzerogeq{R}{i} \not= \emptyset$. 

\paragraph*{Part I }$\depthof{R} = 0$: Then $\termofTaylor{R}{e}{i} = R = \termofTaylor{R}{e}{i+1}$ and $R^{\leq i} = R$, hence $\mathcal{T} = \{ T' \in \mathcal{T} ; T' \sqsubseteq_{\mathcal{P}} R^{\leq i} \} = \mathcal{T'}$;
\begin{itemize}
\item $\conclusions{\groundof{T}} \not= \emptyset$: $\Card{\{ T' \in \mathcal{T} ; T' \sqsubseteq_{\mathcal{P}} R^{\leq i} \}} \leq \cosize{R} < k$;
\item $\conclusions{\groundof{T}} = \emptyset$: $\Card{\{ T' \in \mathcal{T} ; T' \sqsubseteq_{\mathcal{P}} R^{\leq i} \}} \leq \numberInvisibleComponents{R} < k$;
\end{itemize}
so, in both cases, $\Card{\mathcal{T}} = \Card{\mathcal{T'}} < k$, which entails $m_0 = \Card{\mathcal{T}} = \Card{\mathcal{T'}} = m'_0$ and $\{ j \in \Nat \setminus \{ 0 \} ; m_j \not= 0 \} = \emptyset = \{ j \in \Nat \setminus \{ 0 \} ; m'_j \not= 0 \}$; moreover, $\mathcal{M}_{i+1}(e) = \emptyset = \mathcal{N}_i(e)$.

\paragraph*{Part II }$\depthof{R} > 0$ and $\mathcal{B} = \emptyset$:\footnote{As an instance of this case, take for $R$ the PS that is depicted in Figure~\ref{fig: new_example}, take $i = 0$, take for $e$ a $10$-heterogeneous experiment with $e^\#$ like in Example~\ref{example: pseudo-experiment}, take for $\mathcal{P}$ any subset of $\exponentialportsatzero{\termofTaylor{R}{e}{1}} \setminus \boxesatzero{\termofTaylor{R}{e}{0}}$ and for $T$ one of the $10^{224}$ copies of \scalebox{\scalefactter}{\begin{picture}(0,0)%
\includegraphics{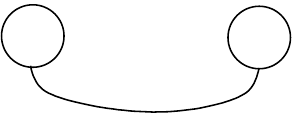}%
\end{picture}%
%
%
\setlength{\unitlength}{3947sp}%
\begingroup\makeatletter\ifx\SetFigFont\undefined%
\gdef\SetFigFont#1#2#3#4#5{%
  \reset@font\fontsize{#1}{#2pt}%
  \fontfamily{#3}\fontseries{#4}\fontshape{#5}%
  \selectfont}%
\fi\endgroup%
\begin{picture}(1403,534)(1272,-461)
\put(2448,-147){\makebox(0,0)[lb]{\smash{{\SetFigFont{12}{14.4}{\rmdefault}{\mddefault}{\updefault}{\color[rgb]{0,0,0}$\bottom$}%
}}}}
\put(1358,-156){\makebox(0,0)[lb]{\smash{{\SetFigFont{12}{14.4}{\rmdefault}{\mddefault}{\updefault}{\color[rgb]{0,0,0}$\one$}%
}}}}
\end{picture}%
}.} Then we distinguish between two cases (Case 1) and Case 2)):
\begin{itemize}
\item Case 1) There exist $o \in \boxesatzerogeq{R}{i}$ and $e_o \in e(o)$ such that the set $\portsatzero{T} \cap \portsatzero{R \langle o, i, e_o \rangle}$ is non-empty: By Lemma~\ref{lemma: ports of connected components}, we have $\portsatzero{T} \subseteq \portsatzero{R \langle o, i, e_o \rangle}$; we set $\mathcal{P}_o = \{ p \in \portsatzero{\termofTaylor{B_R(o)}{e_o}{i}} ; (o, (e_o, p)) \in \mathcal{P} \}$ and $T_0 \in \nontrivialconnected{\termofTaylor{B_R(o)}{e_o}{i}}{\mathcal{P}_o}{k}$ such that $T = \langle o, \langle e_o, T_0 \rangle \rangle$; we have $\mathcal{T} = \{ \langle o, \langle e_o, T' \rangle \rangle ; (T' \in \nontrivialconnected{\termofTaylor{B_R(o)}{e_o}{i}}{\mathcal{P}_o}{k} \wedge T_0 \equiv T') \}$ and $\mathcal{T'} = \{ \langle o, \langle e_o, T' \rangle \rangle ; (T' \in \nontrivialconnected{\termofTaylor{B_R(o)}{e_o}{i+1}}{\mathcal{P}_o}{k} \wedge T_0 \equiv T') \}$. We have $\basis{B_R(o)} \leq \basis{R}$ and $\conclusions{T_0} \subseteq \conclusions{\groundof{T_0}}$, hence we can apply the induction hypothesis: we obtain
\begin{itemize}
\item $\{ j \in \Nat \setminus \{ 0 \} ; m_j \not= 0 \} \subseteq \mathcal{M}_i(e_o) \subseteq \mathcal{M}_i(e)$ (by Lemma~\ref{lem: M_i})
\item $\{ j \in \Nat \setminus \{ 0 \} ; m'_j \not= 0 \} \subseteq \mathcal{M}_{i+1}(e_o) \subseteq \mathcal{M}_{i+1}(e)$ (by Lemma~\ref{lem: M_i})
\item and $(\forall j \in \mathcal{M}_{i+1}(e_o)) m'_j = m_j$, which entails $(\forall j \in \mathcal{M}_{i+1}(e)) m'_j = m_j$, since, by Lemma~\ref{lem: M_i}, we have $\mathcal{M}_{i+1}(e) \cap \mathcal{M}_i(e_o) = \mathcal{M}_{i+1}(e_o)$.
\end{itemize}

Now, let $j \in \mathcal{N}_i(e)$, let $U \in \connectedcomponents{B_{\termofTaylor{R}{e}{i+1}}(\cod_{e, i}(j))}$ such that $U \equiv_{(\termofTaylor{R}{e}{i+1}, \cod_{e, i}(j))} T$ and let $p \in \conclusions{\groundof{U}}$; there exists $p' \in \conclusions{\groundof{T}}$ such that $p' = \target{\termofTaylor{R}{e}{i+1}}(\cod_{e, i}(j), p)$; by Proposition~\ref{prop: critical ports below new boxes}, we have $\cod_{e, i}(j) \in \exactboxesatzero{R}{i}$ or there exist $o' \in \boxesatzerogeq{R}{i+1}$ and $e_{o'} \in e(o')$ such that $j \in \mathcal{N}_i(e_{o'})$:
\begin{itemize}
\item $\cod_{e, i}(j) \in \exactboxesatzero{R}{i}$: then we have $p' \in \portsatzero{R}$, hence $p' \notin \mathcal{P}$, which is in contradiction with $\conclusions{\groundof{T}} \subseteq \mathcal{P}$;
\item or ($j \in \mathcal{N}_i(e_{o'})$ for some $o' \in \boxesatzerogeq{R}{i+1}$ and $e_{o'} \in e(o')$): then we have $$p' = (o', (e_{o'}, \target{\termofTaylor{B_R(o')}{e_{o'}}{i+1}}(\cod_{e_{o'}, i}(j), p))) \in \mathcal{P}$$ hence $(o', e_{o'}) = (o, e_o)$.
\end{itemize}
This shows that, for any $j \in \mathcal{N}_i(e) \setminus \mathcal{N}_i(e_o)$, we have 
$$\Card{\{ U \in \connectedcomponents{B_{\termofTaylor{R}{e}{i+1}}(\cod_{e, i}(j))} ; U \equiv_{(\termofTaylor{R}{e}{i+1}, \cod_{e, i}(j))} T \}} = 0$$
but, since ${j} \notin \mathcal{M}_i(e_o)$, we have $m_{j} = 0$.

For any $j \in \mathcal{N}_i(e_o)$, we have $\cod_{e, i}(j) = (o, (e_o, \cod_{e_o, i}(j)))$; applying the induction hypothesis, we obtain
\begin{eqnarray*}
& & \Card{\{ U \in \connectedcomponents{B_{\termofTaylor{R}{e}{i+1}}(\cod_{e, i}(j))} ; U \equiv_{(\termofTaylor{R}{e}{i+1}, \cod_{e, i}(j))} T \}} \\
& = & \Card{\{ U \in \connectedcomponents{B_{\termofTaylor{R}{e}{i+1}}(o, (e_o, \cod_{e, i}(j)))} ; U \equiv_{(\termofTaylor{R}{e}{i+1}, (o, (e_o, \cod_{e_o, i}(j))))} T \}} \\
& = & \Card{\{ U \in \connectedcomponents{B_{\termofTaylor{B_R(o)}{e_o}{i+1}}(\cod_{e_o, i}(j))} ; U \equiv_{(\termofTaylor{B_R(o)}{e_o}{i+1}, \cod_{e_o, i}(j))} T_0 \}} \\
& = & m_{j}
\end{eqnarray*}
\item Case 2) For any $o \in \boxesatzerogeq{R}{i}$, for any $e_o \in e(o)$, we have $\portsatzero{T} \cap \portsatzero{R \langle o, i, e_o \rangle} = \emptyset$: Notice that then $T \sqsubseteq R$. We set $\mathcal{P}_0 = \mathcal{P} \cap \portsatzero{R}$. We have $\mathcal{T} = \{ T' \in \mathcal{T} ; T' \sqsubseteq_{\mathcal{P}_0} R^{\leq i} \}$, $\mathcal{T'} = \{ T' \in \mathcal{T'} ; T' \sqsubseteq_{\mathcal{P}_0} R^{\leq i} \}$ and $\Card{\{ T' \in \mathcal{T} ; T' \sqsubseteq_{\mathcal{P}_0} R^{\leq i} \}} < k$, hence, by Fact~\ref{fact: components of R^{<=i}}, we have $m_0 = \Card{\mathcal{T}} = \Card{\mathcal{T'}} = m'_0$ and $\{ j \in \Nat \setminus \{ 0 \} ; m_j \not= 0 \} = \emptyset = \{ j \in \Nat \setminus \{ 0 \} ; m'_j \not= 0 \}$.

Since $T \sqsubseteq R$ and $\mathcal{B} = \emptyset$, for any $o \in \exactboxesatzero{\termofTaylor{R}{e}{i+1}}{i}$, we have $$\{ U \in \connectedcomponents{B_{\termofTaylor{R}{e}{i+1}}(o)} ; U \equiv_{(\termofTaylor{R}{e}{i+1}, o)} T \} = \emptyset$$
\end{itemize}

\paragraph*{Part III }$\depthof{R} > 0$, $\mathcal{B} \not= \emptyset$ and $\conclusions{\groundof{T}} \cap \boxesatzerogeq{R}{i} = \emptyset$:\footnote{As an instance of this case, take for $R$ the PS that is depicted in Figure~\ref{fig: new_example}, take $i = 1$, take for $e$ a $10$-heterogeneous experiment with $e^\#$ like in Example~\ref{example: pseudo-experiment}, take $\mathcal{P} = \{ p_1 \}$ and for $T$ the differential PS that is depicted in Figure~\ref{fig: main_prop} (see Example~\ref{example: crucial prop} on page~\pageref{example: crucial prop}).} Then let $o \in \mathcal{B}$: Let $R_o$ be an in-PS and $\varphi$ be some bijection $\contractionsUnder{R}{o} \simeq \mathcal{Q}'$ such that $R_o = \varphi \cdot_o R$. 

Roughly speaking, the number $\sum_{j \in \Nat} m_j \cdot k^j$ (resp. $\sum_{j \in \Nat} m'_j \cdot k^j$) of components of $\termofTaylor{R}{e}{i}$ (resp. $\termofTaylor{R}{e}{i+1}$) that are equivalent to the connected component $T$ is the sum of the number $\sum_{j \in \Nat} p_j \cdot k^j$ (resp. $\sum_{j \in \Nat} p'_j \cdot k^j$) of such components that come from the expansion of the box $o$ and the number $\sum_{j \in \Nat} n_j \cdot k^j$ (resp. $\sum_{j \in \Nat} n'_j \cdot k^j$) of such components that do not come from the expansion of the box $o$, but we define the sequence $(n_j)_{j \in \Nat}$ (resp. the sequence $(n'_j)_{j \in \Nat}$) through $\termofTaylor{\varphi \cdot_o R}{e}{i}$ (resp. $\termofTaylor{\varphi \cdot_o R}{e}{i+1}$), and not through $\termofTaylor{R}{e}{i}$ (resp. $\termofTaylor{R}{e}{i+1}$), in order to be able to apply the induction hypothesis. 

Let $j_o \in \mathcal{M}_i(e)$ such that $\cod_{e, i}(j_o) = o$. We define a subset $\mathcal{N}$ of $\mathcal{N}_i(e)$ as follows: We set $\mathcal{N} = \left\lbrace \begin{array}{ll} \{ j_o \} & \textit{if $o \in \exactboxesatzero{R}{i}$;}\\ \bigcup_{e_o \in e(o)} \mathcal{N}_i(e_o) & \textit{otherwise.} \end{array} \right.$ For every $e_o \in e(o)$, let $\varphi_{e_o}$ be the bijection $\contractionsUnder{R}{o} \simeq \{ o \} \times (\{ e_o \} \times \mathcal{Q}')$ defined by $\varphi_{e_o}(q) = (o, (e_o, \varphi(q)))$ for any $q \in \contractionsUnder{R}{o}$; we set $\mathcal{P}_{e_o} = (\mathcal{P} \setminus \contractionsUnder{R}{o}) \cup \varphi_{e_o}[\mathcal{P} \cap \contractionsUnder{R}{o}]$.

Let $(p_j)_{j \in \Nat} \in \{ 0, \ldots, k-1 \}^{\Nat}$ such that 
$$\Card{\{ T' \in \mathcal{T} ; \portsatzero{T'} \cap \bigcup_{e_o \in e(o)} \portsatzero{R \langle o, i, e_o \rangle} \not= \emptyset \}} = \sum_{j \in \Nat} p_j \cdot k^j$$
For every $e_o \in e(o)$, let $(p_{e_o, j})_{j \in \Nat} \in \{ 0, \ldots, k-1 \}^{\Nat}$ such that $$\Card{\{ T' \in \nontrivialconnected{\termofTaylor{R_o}{e}{i}}{\mathcal{P}_{e_o}}{k} ; T' \equiv T[\varphi_{e_o}] \}} = \sum_{j \in \Nat} p_{e_o, j} \cdot k^j$$

Let $(p'_j)_{j \in \Nat} \in \{ 0, \ldots, k-1 \}^{\Nat}$ such that 
\begin{itemize}
\item if $o \in \boxesatzerogeq{R}{i+1}$, then
$$\Card{\{ T' \in \mathcal{T'} ; \portsatzero{T'} \cap \bigcup_{e_o \in e(o)} \portsatzero{R \langle o, i+1, e_o \rangle} \not= \emptyset \}} = \sum_{j \in \Nat} p'_j \cdot k^j$$
\item if $o \in \exactboxesatzero{R}{i}$, then $0 = \sum_{j \in \Nat} p'_j \cdot k^j$.
\end{itemize}
For any $e_o \in e(o)$, let $(p'_{e_o, j})_{j \in \Nat} \in \{ 0, \ldots, k-1 \}^{\Nat}$ such that 
\begin{itemize}
\item if $o \in \boxesatzerogeq{R}{i+1}$, then
$$\Card{\{ T' \in \nontrivialconnected{\termofTaylor{R_o}{e}{i+1}}{\mathcal{P}_{e_o}}{k} ; T' \equiv T[\varphi_{e_o}] \}} = \sum_{j \in \Nat} p'_{e_o, j} \cdot k^j$$
\item if $o \in \exactboxesatzero{R}{i}$, then $0 = \sum_{j \in \Nat} p'_{e_o, j} \cdot k^j$.
\end{itemize}

By Lemma~\ref{lemma: ports of connected components}, we have 
\begin{eqnarray*}
& & \Card{\{ T' \in \mathcal{T} ; \portsatzero{T'} \cap \bigcup_{e_o \in e(o)} \portsatzero{R \langle o, i, e_o \rangle} \not= \emptyset \}}\\
& = & \sum_{e_o \in e_o} \Card{\{ T' \in \mathcal{T} ; \portsatzero{T'} \cap \portsatzero{R \langle o, i, e_o \rangle} \not= \emptyset \}}
\end{eqnarray*}
hence, by Lemma~\ref{lem: connected components of R_o}~(\ref{lem: connected components of R_o item: [varphi]})~and~(\ref{lem: connected components of R_o item: [inverse of varphi]}), we have:
\begin{itemize}
\item $\sum_{j \in \Nat} p_j \cdot k^j = \sum_{e_o \in e(o)} \sum_{j \in \Nat} p_{e_o, j} \cdot k^j$ 
\item $\sum_{j \in \Nat} p'_j \cdot k^j = \sum_{e_o \in e(o)} \sum_{j \in \Nat} p'_{e_o, j} \cdot k^j$
\end{itemize}

Let $(n_j)_{j \in \Nat} \in \{ 0, \ldots, k-1 \}^{\Nat}$ such that
\begin{displaymath}
  \Card{\{ T' \in \nontrivialconnected{\termofTaylor{R_o}{e}{i}}{\conclusions{\groundof{T}}}{k} ; T' \equiv T \}} = \sum_{j \in \Nat} n_j \cdot k^j
\end{displaymath}
and let $(n'_j)_{j \in \Nat} \in \{ 0, \ldots, k-1 \}^{\Nat}$ such that
\begin{displaymath}
  \Card{\{ T' \in \nontrivialconnected{\termofTaylor{R_o}{e}{i+1}}{\conclusions{\groundof{T}}}{k} ; T' \equiv T \}} = \sum_{j \in \Nat} n'_j \cdot k^j.
\end{displaymath}
By Lemma~\ref{lem: connected components of R_o}~(\ref{lem: connected components of R_o item: R_o -> R})~and~(\ref{lem: connected components of R_o item: R -> R_o}), we have:
\begin{itemize}
\item $\sum_{j \in \Nat} m_j \cdot k^j = \sum_{j \in \Nat} (p_j + n_j) \cdot k^j$
\item $\sum_{j \in \Nat} m'_j \cdot k^j = \sum_{j \in \Nat} (p'_j + n'_j) \cdot k^j$
\end{itemize}
Now, the proof will be in three steps (Step 1), Step 2) and Step 3)):

\subparagraph*{Step 1) } It consists in proving properties about the sequences $(n_j)_{j \in \Nat}$ and $(n'_j)_{j \in \Nat}$. We distinguish between two cases:
\begin{itemize}
\item $\mathcal{B} = \{ o \}$: Then $\{ T' \in \nontrivialconnected{\termofTaylor{R_o}{e}{i}}{\mathcal{P}}{k} ; T' \equiv T \} = \{ T' \in \nontrivialconnected{\termofTaylor{R_o}{e}{i}}{\mathcal{P}}{k} ; (T' \equiv T \wedge T' \sqsubseteq_{\mathcal{P}} {R_o}^{\leq i}) \}$ and $\{ T' \in \nontrivialconnected{\termofTaylor{R_o}{e}{i+1}}{\mathcal{P}}{k} ; T' \equiv T \} = \{ T' \in \nontrivialconnected{\termofTaylor{R_o}{e}{i+1}}{\mathcal{P}}{k} ; (T' \equiv T \wedge T' \sqsubseteq_{\mathcal{P}} {R_o}^{\leq i}) \}$, hence $\Card{\{ T' \in \nontrivialconnected{\termofTaylor{R_o}{e}{i}}{\mathcal{P}}{k} ; T' \equiv T \}} < k$ and, by Fact~\ref{fact: components of R^{<=i}}, we have
$$\Card{\{ T' \in \nontrivialconnected{\termofTaylor{R_o}{e}{i}}{\mathcal{P}}{k} ; T' \equiv T \}} = \Card{\{ T' \in \nontrivialconnected{\termofTaylor{R_o}{e}{i+1}}{\mathcal{P}}{k} ; T' \equiv T \}}$$ 
so $\{ j \in \Nat \setminus \{ 0 \} ; n_j \not= 0 \} = \emptyset = \{ j \in \Nat \setminus \{ 0 \} ; n'_j \not= 0 \}$ and $$n_0 = \Card{\{ T' \in \nontrivialconnected{\termofTaylor{R_o}{e}{i}}{\mathcal{P}}{k} ; (T' \equiv T \wedge T' \sqsubseteq_{\mathcal{P}} {R_o}^{\leq i}) \}} = n'_0$$ (in particular, $(\forall j \in \Nat) n'_j = n_j$). If $\mathcal{P} \subseteq \portsatzero{R}$, then, by Lemma~\ref{lem: connected components of R_o are connected components of R}, we have $n_0 = \Card{\{ T' \in \mathcal{T} ; T' \sqsubseteq_{\mathcal{P}} R^{\leq i} \}} = n'_0$.
\item $\mathcal{B} \setminus \{ o \} \not= \emptyset$: By Lemma~\ref{lem: basis of R_o}, we can apply the induction hypothesis: We have 
\begin{itemize}
\item $\cod_{e, i}[\{ j \in \Nat \setminus \{ 0 \} ; n_j \not= 0 \}] \subseteq (\mathcal{B} \setminus \{ o \}) \cup \bigcup_{o' \in \mathcal{B} \setminus \{ o \}} \bigcup_{e_{o'} \in e(o')} \cod_{e, i}[\mathcal{M}_i(e_{o'})]$
\item $\cod_{e, i}[\{ j \in \Nat \setminus \{ 0 \} ; n'_j \not= 0 \}] \subseteq (\mathcal{B'} \setminus \{ o \}) \cup \bigcup_{o' \in \mathcal{B'} \setminus \{ o \}} \bigcup_{e_{o'} \in e(o')} \cod_{e, i}[\mathcal{M}_{i+1}(e_{o'})]$
\item $(\forall j \in \mathcal{M}_{i+1}(e)) n'_j = n_j$
\item if $\mathcal{P} \subseteq \portsatzero{R}$, then $n_0 = \Card{\{ T' \in \nontrivialconnected{\termofTaylor{R_o}{e}{i}}{\mathcal{P}}{k} ; (T' \equiv T \wedge T' \sqsubseteq_{\mathcal{P}} {R_o}^{\leq i}) \}} = n'_0$, hence, by Lemma~\ref{lem: connected components of R_o are connected components of R}, we have $n_0 = \Card{\{ T' \in \mathcal{T} ; T' \sqsubseteq_{\mathcal{P}} R^{\leq i} \}} = n'_0$.
\end{itemize}
\end{itemize}

Now, for any $j \in \mathcal{N}_i(e) \setminus \mathcal{N}$, we have $B_{\termofTaylor{R_o}{e}{i+1}}(\cod_{e, i}(j)) = B_{\termofTaylor{R}{e}{i+1}}(\cod_{e, i}(j))$, hence $$\Card{\{ U \in \connectedcomponents{B_{\termofTaylor{R}{e}{i+1}}(\cod_{e, i}(j))} ; U \equiv_{(\termofTaylor{R}{e}{i+1}, \cod_{e, i}(j))} T \}} = n_{j}$$

\subparagraph*{Step 2) }
In order to prove 
\begin{itemize}
\item $\{ j \in \Nat ; p_j \not= 0 \} \subseteq \{ 0, j_o \} \cup \bigcup_{e_o \in e(o)} \mathcal{M}_i(e_o)$ 
\item $(p_0 \not= 0 \Rightarrow n_0 = 0)$
\item $\{ j \in \Nat ; p'_j \not= 0 \} \subseteq \{ 0, j_o \} \cup \bigcup_{e_o \in e(o)} \mathcal{M}_{i+1}(e_o)$
\item $(o \in \boxesatzerogeq{R}{i+1} \Rightarrow (\forall j \in \{ j_o \} \cup \bigcup_{e_o \in e(o)} \mathcal{M}_{i+1}(e_o)) p'_j = p_j)$
\item $(p'_0 \not= 0 \Rightarrow n'_0 = 0)$
\item $(\forall j \in \mathcal{N}) p_j = \Card{\{ U \in \connectedcomponents{B_{\termofTaylor{R}{e}{i+1}}(\cod_{e, i}(j))} ; U \equiv_{(\termofTaylor{R}{e}{i+1}, \cod_{e, i}(j))} T \}}$
\end{itemize}
we distinguish between three cases (Case a), Case b) and Case c)):
\begin{itemize}
\item Case a) $\conclusions{\groundof{T}} \subseteq \contractionsUnder{R}{o}$: It is worth noticing that the differential in-PS $T$ has no co-contraction. If there is no $e_1 \in e(o)$ such that $\{ T' \in \nontrivialconnected{\termofTaylor{R_o}{e}{i}}{\mathcal{P}_{e_1}}{k} ; T' \equiv T[\varphi_{e_1}] \} \not= \emptyset$, then, by Lemma~\ref{lem: connected components of R_o}~(\ref{lem: connected components of R_o item: [varphi]}), for any $j \in \Nat$, we have $p_j = 0 = p'_j$, so there is nothing to prove. From now on, let us assume that there exist $e_1 \in e(o)$ and $T'_1 \in \{ T' \in \nontrivialconnected{\termofTaylor{R_o}{e}{i}}{\mathcal{P}_{e_1}}{k} ; T' \equiv T[\varphi_{e_1}] \}$; let $\zeta : T'_1 \equiv T[\varphi_{e_1}]$ and let $T_1 \in \nontrivialconnected{\termofTaylor{B_{R_o}(o)}{e_1}{i}}{\mathcal{Q}'}{k}$ such that $T'_1 = \langle o, \langle e_1, T_1 \rangle \rangle$; notice that, for any $e_o \in e(o)$, we have $T[\varphi_{e_o}] \equiv \langle o, \langle e_o, T_1 \rangle \rangle$, hence
\begin{eqnarray*}
& & \{ T' \in \nontrivialconnected{\termofTaylor{R_o}{e}{i}}{\mathcal{P}_{e_o}}{k} ; T' \equiv T[\varphi_{e_o}] \}\\
& = & \{ \langle o, \langle e_o, T' \rangle \rangle ; (T' \in \nontrivialconnected{\termofTaylor{B_{R_o}(o)}{e_o}{i}}{\mathcal{Q}'}{k} \wedge T' \equiv T_1) \} \: \: \: (\ast)
\end{eqnarray*}
we have $\conclusions{T_1} \subseteq \conclusions{\groundof{T_1}} \subseteq \mathcal{Q}' \subseteq \portsatzero{B_{R_o}(o)}$, hence, since by Lemma~\ref{lem: basis of R_o} we can apply the induction hypothesis, we have 
$$p_{e_1, 0} = \Card{\{ T' \in \nontrivialconnected{\termofTaylor{B_{R_o}(o)}{e_1}{i}}{\mathcal{Q}'}{k} ; (T' \equiv T_1 \wedge T' \sqsubseteq_{\mathcal{Q}'} {B_{R_o}(o)}^{\leq i}) \}}$$ 
and 
\begin{eqnarray*}
& & \sum_{e_o \in e(o)} \Card{\{ T' \in \nontrivialconnected{\termofTaylor{R_o}{e}{i}}{\mathcal{P}_{e_o}}{k} ; T' \equiv T[\varphi_{e_o}] \}}\\
& = & \sum_{e_o \in e(o)} \Card{\{ T' \in \nontrivialconnected{\termofTaylor{B_{R_o}(o)}{e_o}{i}}{\mathcal{Q}'}{k} ; T' \equiv T_1 \}} \: \: \: \textit{(by $(\ast)$)} \allowdisplaybreaks\\
& = & \Card{\{ T' \in \nontrivialconnected{\termofTaylor{B_{R_o}(o)}{e_1}{i}}{\mathcal{Q}'}{k} ; (T' \equiv T_1 \wedge T' \sqsubseteq_{\mathcal{Q}'} {B_{R_o}(o)}^{\leq i}) \}} \cdot \Card{e(o)}\\
& + & \sum_{e_o \in e(o)} \sum_{j \in \mathcal{M}_i(e_o)} p_{e_o, j} \cdot k^j
\end{eqnarray*}
hence $\cod_{e, i}[\{ j \in \Nat ; p_j \not= 0 \}] \subseteq \{ o \} \cup \bigcup_{e_o \in e(o)} \cod_{e, i}[\mathcal{M}_i(e_o)]$.

If $o \in \exactboxesatzero{R}{i}$, then there exists a bijection $\{ T' \in \nontrivialconnected{B_{R_o}(o)}{\mathcal{Q}'}{k} ; T' \equiv T_1 \} \simeq \{ U \in \connectedcomponents{B_R(o)} ; U \equiv_{(\termofTaylor{R}{e}{i+1}, o)} T \}$ that associates with every $T' \in \nontrivialconnected{B_{R_o}(o)}{\mathcal{Q}'}{k}$ such that $T' \equiv T_1$ the differential in-PS $\overline{T'}$. Indeed:
\begin{itemize}
\item By Lemma~\ref{lemm: components adding contractions}, there exists a bijection $\nontrivialconnected{B_{R_o}(o)}{\mathcal{Q}'}{k}  \simeq \connectedcomponents{B_R(o)}$ that associates with every $T' \in \nontrivialconnected{B_{R_o}(o)}{\mathcal{Q}'}{k}$ the differential in-PS $\overline{T'}$;
\item If $T' \in \nontrivialconnected{B_{R_o}(o)}{\mathcal{Q}'}{k}$ and $\psi: T' \equiv T_1$, then the function $\delta$ that associates with every $p \in \ports{\overline{T'}}$ the port $\zeta(o_1, (e_1, \psi(p)))$ of $\overline{T}$ is an isomorphism $\overline{T'} \simeq \overline{T}$ such that, 
for any $p \in \conclusions{\overline{T'}}$, we have $\targetAnyPort{T'}(p) = \targetAnyPort{B_{R_o}(o)}(p) = \varphi(\target{R}(o, p))$, hence 
\begin{eqnarray*}
\varphi_{e_1}(\targetAnyPort{T}(\delta(p))) & = & \varphi_{e_1}(\targetAnyPort{T}(\zeta(o_1, (e_1, \psi(p))))\\
& = & \targetAnyPort{T[\varphi_{e_1}]}(\zeta(o_1, (e_1, \psi(p))))\allowdisplaybreaks\\
& = & \zeta(\targetAnyPort{T'_1}(o_1, (e_1, \psi(p))))\allowdisplaybreaks\\
& = & \targetAnyPort{T'_1}(o_1, (e_1, \psi(p)))\allowdisplaybreaks\\
& = & (o_1, (e_1, \targetAnyPort{T_1}(\psi(p))))\allowdisplaybreaks\\
& = & (o_1, (e_1, \psi(\targetAnyPort{T'}(p))))\allowdisplaybreaks\\
& = & (o_1, (e_1, \targetAnyPort{T'}(p)))\allowdisplaybreaks\\
& = & (o_1, (e_1, \varphi(\target{R}(o, p))))\\
& = & \varphi_{e_1}(\target{R}(o, p))
\end{eqnarray*}
So, 
$\target{\termofTaylor{R}{e}{i+1}}(o, p) = \target{R}(o, p) = \targetAnyPort{T}(\delta(p))$.
\end{itemize}

If $o \in \boxesatzerogeq{R}{i+1}$, then, by induction hypothesis, $p'_{e_1, 0} = \Card{\{ T' \sqsubseteq_{\mathcal{Q}'} B_{R_o}(o) ; T' \equiv T_1 \}}$ and 
\begin{eqnarray*}
& & \sum_{e_o \in e(o)} \Card{\{ T' \in \nontrivialconnected{\termofTaylor{R_o}{e}{i+1}}{\mathcal{P}_{e_o}}{k} ; T' \equiv T[\varphi_{e_o}] \}}\\
& = & \sum_{e_o \in e(o)} \Card{\{ T' \in \nontrivialconnected{\termofTaylor{B_{R_o}(o)}{e_o}{i+1}}{\mathcal{Q}'}{k} ; T' \equiv T_1 \}}\\
& = & \Card{\{ T' \in \nontrivialconnected{\termofTaylor{B_{R_o}(o)}{e_1}{i}}{\mathcal{Q}'}{k} ; (T' \equiv T_1 \wedge T' \sqsubseteq_{\mathcal{Q}'} B_{R_o}(o)) \}} \cdot \Card{e(o)}\\
& + & \sum_{e_o \in e(o)} \sum_{j \in \mathcal{M}_{i+1}(e_o)} p'_{e_o, j} \cdot k^j
\end{eqnarray*}
hence $\cod_{e, i}[\{ j \in \Nat ; p'_j \not= 0 \} \subseteq \{ o \} \cup \bigcup_{e_o \in e(o)} \cod_{e, i}[\mathcal{M}_{i+1}(e_o)]$; by induction hypothesis, for any $e_o \in e(o)$, we have $(\forall j \in \mathcal{M}_{i+1}(e_o) \cup \{ 0 \}) p'_{e_o, j} = p_{e_o, j}$; moreover, we have $$p_{j_o} = \Card{\{ T' \in \nontrivialconnected{\termofTaylor{B_{R_o}(o)}{e_1}{i}}{\mathcal{Q}'}{k} ; (T' \equiv T_1 \wedge T' \sqsubseteq_{\mathcal{Q}'} B_{R_o}(o)) \}} = p'_{j_o}$$ now, let $j \in \mathcal{N}$ and let $e_o \in e(o)$ such that $\cod_{e, i}(j) = (o, (e_o, \cod_{e_o, i}(j)))$: We have 
\begin{eqnarray*}
& & \Card{\{ U \in \connectedcomponents{B_{\termofTaylor{R}{e}{i+1}}(\cod_{e, i}(j))} ; U \equiv_{(\termofTaylor{R}{e}{i+1}, \cod_{e, i}(j))} T \}}\\
& = & \Card{\{ U \in \connectedcomponents{B_{\termofTaylor{R_o}{e}{i+1}}(\cod_{e, i}(j))} ; U \equiv_{(\termofTaylor{R_o}{e}{i+1}, \cod_{e, i}(j))} T[\varphi_{e_o}] \}}\\
& = & p_j \textit{   (by induction hypothesis)}
\end{eqnarray*}
\item Case b) There exists $e_1 \in e(o)$ such that $\conclusions{\groundof{T}} \cap \portsatzero{R \langle o, i, e_1 \rangle} \not= \emptyset$: By Lemma~\ref{lemma: ports of connected components}, we have $\portsatzero{T} \subseteq \portsatzero{R \langle o, i, e_1 \rangle} \cup \target{R}[\{ o \} \times \temporaryConclusions{R}{o}]$, hence $(\forall j \in \Nat) n_j = 0 = n'_j$; we set $\mathcal{P'} = \mathcal{Q}' \cup \{ p ; (o, (e_1, p)) \in \mathcal{P} \}$; let $T_0 \in \nontrivialconnected{\termofTaylor{B_{R_o}(o)}{e_1}{i}}{\mathcal{P'}}{k}$ such that $T = \langle o, \langle e_1, T_0 \rangle \rangle$; we have $\conclusions{T_0} \subseteq \conclusions{\groundof{T_0}}$ and, by Lemma~\ref{lem: basis of R_o}, $\basis{R_o} \leq \basis{R} \leq k$, hence we can apply the induction hypothesis: We obtain 
\begin{eqnarray*}
&  & \sum_{e_o \in e(o)} \Card{\{ T' \in \mathcal{T} ; \portsatzero{T'} \cap \portsatzero{R \langle o, i, e_o \rangle} \not= \emptyset \}} \\
& = & \Card{\{ T' \in \mathcal{T} ; \portsatzero{T'} \cap \portsatzero{R \langle o, i, e_1 \rangle} \not= \emptyset \}}\\
& = & \Card{\{ T' \in \nontrivialconnected{\termofTaylor{R_o}{e}{i}}{\mathcal{P}_{e_1}}{k} ; T' \equiv T[\varphi_{e_1}] \}}\\
& = & \Card{\{ T' \in \nontrivialconnected{\termofTaylor{B_{R_o}(o)}{e_1}{i}}{\mathcal{P'}}{k} ; T' \equiv T_0 \}}\\
& = & \sum_{j \in \mathcal{M}_i(e_1) \cup \{ 0 \}} p_{e_1, j} \cdot k^j
\end{eqnarray*}
hence we have $\{ j \in \Nat ; p_j \not= 0 \} \subseteq \mathcal{M}_i(e_1) \cup \{ 0 \}$. 

If $o \in \exactboxesatzero{R}{i}$ then $\mathcal{M}_i(e_1) = \emptyset$, hence $p_{j_o} = 0$; from the other hand, since $\conclusions{\groundof{T}} \cap \portsatzero{R \langle o, i, e_1 \rangle} \not= \emptyset$, we have $\{ U \in \connectedcomponents{B_{\termofTaylor{R}{e}{i+1}}(\cod_{e, i}(j_o))} ; U \equiv_{(\termofTaylor{R}{e}{i+1}, \cod_{e, i}(j))} T \} = \{ U \in \connectedcomponents{B_{R}(o)} ; U \equiv_{(\termofTaylor{R}{e}{i+1}, \cod_{e, i}(j))} T \} = \emptyset$.

If $o \in \boxesatzerogeq{R}{i+1}$, then, by induction hypothesis, we have
\begin{eqnarray*}
&  & \sum_{e_o \in e(o)} \Card{\{ T' \in \mathcal{T'} ; \portsatzero{T'} \cap \portsatzero{R \langle o, i+1, e_o \rangle} \not= \emptyset \}} \\
& = & \Card{\{ T' \in \mathcal{T'} ; \portsatzero{T'} \cap \portsatzero{R \langle o, i+1, e_1 \rangle} \not= \emptyset \}}\\
& = & \Card{\{ T' \in \nontrivialconnected{\termofTaylor{R_o}{e}{i+1}}{\mathcal{P}_{e_1}}{k} ; T' \equiv T[\varphi_{e_1}] \}}\\
& = & \Card{\{ T' \in \nontrivialconnected{\termofTaylor{B_{R_o}(o)}{e_1}{i+1}}{\mathcal{P'}}{k} ; T' \equiv T_0 \}}\\
& = & \sum_{j \in \mathcal{M}_{i+1}(e_1) \cup \{ 0 \}} p'_{e_1, j} \cdot k^j
\end{eqnarray*}
and $\{ j \in \Nat ; p'_{e_1, j} \not= 0 \} \subseteq \mathcal{M}_{i+1}(e_1) \cup \{ 0 \}$ and $(\forall j \in \mathcal{M}_{i+1}(e_1) \cup \{ 0 \}) p'_{e_1, j} = p_{e_1, j}$; let $j \in \mathcal{N}$ and let $e_o \in e(o)$ such that $\cod_{e, i}(j) = (o, (e_o, \cod_{e_o, i}(j)))$: 
\begin{itemize}
\item if $e_o \not= e_1$, then $j \notin \mathcal{M}_i(e_1)$, hence $p_j = 0$; from the other hand, we have 
$$\{ U \in \connectedcomponents{B_{\termofTaylor{R}{e}{i+1}}(\cod_{e, i}(j))} ; U \equiv_{(\termofTaylor{R}{e}{i+1}, \cod_{e, i}(j))} T \} = \emptyset$$
\item if $e_o = e_1$, then we have 
\begin{eqnarray*}
& & \Card{\{ U \in \connectedcomponents{B_{\termofTaylor{R}{e}{i+1}}(\cod_{e, i}(j))} ; U \equiv_{(\termofTaylor{R}{e}{i+1}, \cod_{e, i}(j))} T \}}\\
& = & \Card{\{ U \in \connectedcomponents{B_{\termofTaylor{R_o}{e}{i+1}}(\cod_{e, i}(j))} ; U \equiv_{(\termofTaylor{R_o}{e}{i+1}, \cod_{e, i}(j))} T[\varphi_{e_1}] \}}\\
& = & p_j \textit{   (by induction hypothesis)}
\end{eqnarray*}
\end{itemize}
\item Case c) $\conclusions{\groundof{T}} \cap (\portsatzero{R} \setminus \{ \target{R}(o, p) ; p \in \conclusions{B_R(o)} \}) \not= \emptyset$: By Lemma~\ref{lemma: ports of connected components}, for any $e_o \in e(o)$, we have $\{ T' \in \mathcal{T} ; \portsatzero{T'} \cap \portsatzero{R \langle o, i, e_o \rangle} \not= \emptyset \} = \emptyset = \{ T' \in \mathcal{T'} ; \portsatzero{T'} \cap \portsatzero{R \langle o, i, e_o \rangle} \not= \emptyset \}$, hence $(\forall j \in \Nat) p_j = 0 = p'_j$. From the other hand, for any $j \in \mathcal{N}_i(e)$, the set $\{ U \in \connectedcomponents{B_{\termofTaylor{R}{e}{i+1}}(\cod_{e, i}(j))} ; U \equiv_{(\termofTaylor{R}{e}{i+1}, \cod_{e, i}(j))} T \}$ is empty.
\end{itemize}

\subparagraph*{Step 3) }
We distinguish between two cases:
\begin{itemize}
\item Case a) $o \in \exactboxesatzero{R}{i}$: We have $\cod_{e, i}[\{ j \in \Nat ; p_j \not= 0 \} \cap \mathcal{M}_{i+	1}(e)] \subseteq \cod_{e, i}[\{ j \in \Nat ; p_j \not= 0 \}] \cap \cod_{e, i}[\mathcal{M}_{i+1}(e)] \subseteq (\{ o \} \cup \bigcup_{e_o \in e(o)} \cod_{e, i}[\mathcal{M}_i(e_o)]) \cap \cod_{e, i}[\mathcal{M}_{i+1}(e)] = (\{ o \} \cap \cod_{e, i}[\mathcal{M}_{i+1}(e)]) \cup \bigcup_{e_o \in e(o)} (\cod_{e, i}[\mathcal{M}_i(e_o)] \cap \cod_{e, i}[\mathcal{M}_{i+1}(e)]) \subseteq \emptyset \cup \bigcup_{e_o \in e(o)} \cod_{e, i}[\mathcal{M}_i(e_o) \cap \mathcal{M}_{i+1}(e)] =\emptyset$, hence, for any $j \in \Nat$, $m_j = n_j$; by Lemma~\ref{lem: connected components of R_o}~(\ref{lem: connected components of R_o item: R_o -> R})~and(\ref{lem: connected components of R_o item: R -> R_o}), we have 
\begin{eqnarray*}
\Card{\mathcal{T'}} & = & \Card{\{ T' \in \nontrivialconnected{\termofTaylor{R_o}{e}{i+1}}{\conclusions{\groundof{T}}}{k} ; T' \equiv T \}}
\end{eqnarray*}
hence, for any $j \in \Nat$, we have $m'_j = n'_j$; since, for any $j \in \mathcal{M}_{i+1}(e)$, we have $n'_j = n_j$, we obtain $(\forall j \in \mathcal{M}_{i+1}(e)) m'_j = m_j$.
\item Case b) $o \in \boxesatzerogeq{R}{i+1}$: We have $\{ j \in \Nat ; m'_j \not= 0 \} = \{ j \in \Nat ; n'_j \not= 0 \} \cup \{ j \in \Nat ; p'_j \not= 0 \}$, hence $\cod_{e, i}[\{ j \in \Nat \setminus \{ 0 \} ; m'_j \not= 0 \}] \subseteq \mathcal{B'} \cup \bigcup_{o \in \mathcal{B'}} \bigcup_{e_o \in e(o)} \cod_{e, i}[\mathcal{M}_{i+1}(e_o)]$. For any $j \in \{ j_o \} \cup \bigcup_{e_o \in e(o)} \mathcal{M}_{i+1}(e_o)$, we have $n'_j = n_j$ and $p'_j = p_j$, hence $m'_j = m_j$; for any $j \in \mathcal{M}_{i+1}(e) \setminus (\{ j_o \} \cup \bigcup_{e_o \in e(o)} \mathcal{M}_{i+1}(e_o))$, we have $n'_j = n_j$ and $p'_j = 0 = p_0$, hence $m'_j = m_j$.
\end{itemize}

\paragraph*{Part IV }$\depthof{R} > 0$, $\mathcal{B} \not= \emptyset$ and $\conclusions{\groundof{T}} \cap \boxesatzerogeq{R}{i} \not= \emptyset$: Let $o \in \conclusions{\groundof{T}} \cap \boxesatzerogeq{R}{i}$. There exists $e_o \in e(o)$ such that $(o, (e_o, \aboveBang{R}{o})) \in \portsatzero{T}$ and $\target{T}(o, (e_o, \aboveBang{R}{o})) = o$. Notice that $o$ is a co-contraction of $T$ such that $\arity{\groundof{T}}(o) > 0$, which entails that the set $\{ T' \in \mathcal{T} ; T' \sqsubseteq_{\mathcal{P}} R^{\leq i} \}$ is empty. We distinguish between two cases: Case a), where $\conclusions{\groundof{T}} \subseteq \target{R}[\{o \} \times \temporaryConclusions{R}{o}]$ and Case b), where $\conclusions{\groundof{T}} \setminus \target{R}[\{o \} \times \temporaryConclusions{R}{o}] \not= \emptyset$.
\begin{itemize}
\item Case a) $\conclusions{\groundof{T}} \subseteq \target{R}[\{o \} \times \temporaryConclusions{R}{o}]$: Let $j_o \in \mathcal{M}_i(e)$ such that $\cod_{e, i}(j_o) = o$. \begin{itemize}
\item We have $\Card{\mathcal{T}} = k^{j_o}$.
\item We have $\Card{\mathcal{T'}} = \left\lbrace \begin{array}{ll} k^{j_o} & \textit{if $o \in \boxesatzerogeq{R}{i+1}$;}\\ 0 & \textit{otherwise.} \end{array} \right.$
\item We have $\{ U \in \connectedcomponents{B_{\termofTaylor{R}{e}{i+1}}(\cod_{e, i}(j_o))} ; U \equiv_{(\termofTaylor{R}{e}{i+1}, \cod_{e, i}(j_o))} T \} = \connectedcomponentsContaining{B_R(o)}(\aboveBang{R}{o})$, hence, for any $j \in \mathcal{N}_i(e)$,
\begin{eqnarray*}
& & \Card{\{ U \in \connectedcomponents{B_{\termofTaylor{R}{e}{i+1}}(\cod_{e, i}(j))} ; U \equiv_{(\termofTaylor{R}{e}{i+1}, \cod_{e, i}(j))} T \}}\\
& = & \left\lbrace \begin{array}{ll} 1 & \textit{if $j = j_o \in \mathcal{N}_i(e)$;}\\ 0 & \textit{otherwise.} \end{array} \right.
\end{eqnarray*}
\end{itemize}
\item Case b) $\conclusions{\groundof{T}} \setminus \target{R}[\{o \} \times \temporaryConclusions{R}{o}] \not= \emptyset$: By Lemma~\ref{lemma: ports of connected components}, $\portsatzero{T} \setminus \target{R}[\{ o \} \times \temporaryConclusions{R}{o}] \subseteq \portsatzero{R \langle o, i, e_o \rangle}$, hence $\mathcal{P} \subseteq \portsatzero{R}$ does not hold. 
\begin{itemize}
\item For any $T' \in \mathcal{T} \cup \mathcal{T'}$, we have $(o, (e_o, \aboveBang{R}{o})) \in \portsatzero{T'}$, hence $\Card{\mathcal{T}} \leq 1$ and $\Card{\mathcal{T'}} \leq 1$.
\item For any $j \in \mathcal{N}_i(e)$, we have
  \begin{align*}
    \{ U \in \connectedcomponents{B_{\termofTaylor{R}{e}{i+1}}(\cod_{e, i}(j))} ; U \equiv_{(\termofTaylor{R}{e}{i+1}, \cod_{e, i}(j))} T \} = \emptyset \tag*{\qedhere}
  \end{align*}

\end{itemize}
\end{itemize}
\end{proof}

\begin{figure}
\centering
\begin{minipage}{0.45\textwidth}
\centering
\resizebox{.4 \textwidth}{!}{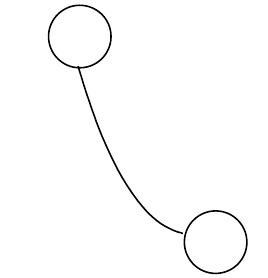}
\captionsetup{width= .9 \textwidth}
\caption{The differential in-PS \mbox{$T \in \nontrivialconnected{\termofTaylor{R}{e}{1}}{\{ p_1 \}}{10}$}}
\label{fig: main_prop}
\end{minipage}\hfill
\begin{minipage}{0.45\textwidth}
\centering
\resizebox{.15 \textwidth}{!}{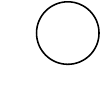}
\captionsetup{width= .6 \textwidth}
\caption{The in-PS \mbox{$U_0 \in \connectedcomponents{B_{\termofTaylor{R}{e}{2}}(\cod_{e, 1}(1))}$}}
\label{fig: U_0}
\end{minipage}
\end{figure}

\begin{exa}\label{example: crucial prop}
(Continuation of Example~\ref{example: approximant}) Let $T$ be the differential in-PS depicted in Figure~\ref{fig: main_prop}.
We have $T \in \nontrivialconnected{\termofTaylor{R}{e}{1}}{\{ p_1 \}}{10}$, $\Card{\mathcal{T}} = 11$ and $\Card{\mathcal{T'}} = 1$, $\mathcal{B} = \{ o_2 \}$ and $\mathcal{B'} = \emptyset$. We thus have:
\begin{itemize}
\item $m_0 = 1 = m_1$ and $m_j = 0$ for any $j \geq  2$
\item and $m'_0 = 1$ and $m'_j = 0$ for any $j \geq 1$.
\end{itemize}
We recall (see Example~\ref{example: N_i(e)}) that $\mathcal{M}_1(e) = \{ 1, 2 \} = \mathcal{N}_1(e)$ and $\mathcal{M}_2(e) = \emptyset$. Notice that we have $\{ U \in \connectedcomponents{B_{\termofTaylor{R}{e}{2}}(\cod_{e, 1}(1))} ; U \equiv_{(\termofTaylor{R}{e}{i+1}, \cod_{e, i}(j))} T \} = \{ U_0 \}$, where $U_0$ is the in-PS depicted in Figure~\ref{fig: U_0}.

\begin{figure}
\centering
\resizebox{\textwidth}{!}{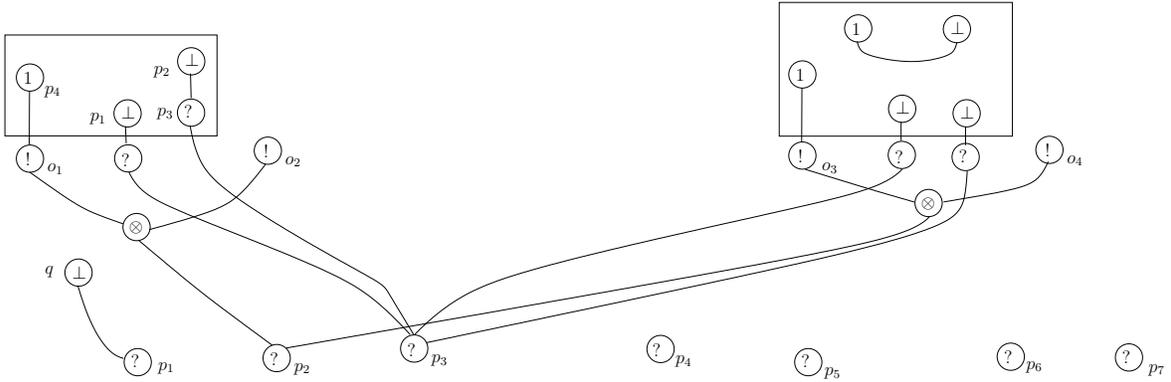}
\caption{The differential PS $R^{\leq 1}$}
\label{fig: Rleq1}
\end{figure}
The differential PS $R^{\leq 1}$ is depicted in Figure~\ref{fig: Rleq1}: We have $\{ T' \in \mathcal{T} ; T' \sqsubseteq_{\{ p_1 \}} R^{\leq 1} \} = \{ T \}$.
\end{exa}

\subsection{The content of the boxes}

\begin{lem}\label{lem: bijection => equiv}
Let $S$ be a differential in-PS. Let $o \in \boxesatzero{S}$. Let $\mathcal{T}$ be a set of differential in-PS's that is gluable. Assume that there exists a bijection $\gamma : \connectedcomponents{B_S(o)} \simeq \mathcal{T}$ such that, for any $V \in \connectedcomponents{B_S(o)}$, we have $V \equiv_{(S, o)} \gamma(V)$. Then $B_S(o) \equiv_{(S, o)} \bigoplus \mathcal{T}$.
\end{lem}

\begin{proof}
For each $V \in \connectedcomponents{B_S(o)}$, we are given $\varphi_V : V \equiv_{(S, o)} \gamma(V)$. Now, we define $\varphi: \bigoplus \connectedcomponents{B_S(o)} \equiv_{(S, o)} \bigoplus \mathcal{T}$ as follows: for any $p \in \ports{\bigoplus \connectedcomponents{B_S(o)}}$, we set $\varphi(p) = \varphi_{\connectedcomponentsContaining{B_S(o)}(p)}(p)$. But, by Fact~\ref{fact: connected components}, we have $\bigoplus \connectedcomponents{B_S(o)} = B_S(o)$.
\end{proof}

\begin{fact}\label{fact: equiv modulo}
Let $R$ be an in-PS. Let $o \in \boxesatzerogeq{R}{i}$. Let $e_o \in e(o)$. Let $i \in \Nat$. Let $o_1 \in \boxesatzero{\termofTaylor{B_R(o)}{e_o}{i}}$. Let $V \in \connectedcomponents{B_{\termofTaylor{B_R(o)}{e_o}{i}}(o_1)}$ and let $T$ be an in-PS such that $V \equiv_{(\termofTaylor{B_R(o)}{e_o}{i}, o_1)} T$. Then $V \equiv_{(\termofTaylor{R}{e}{i}, (o, (e_o, o_1)))} \langle o, \langle e_o, T \rangle \rangle$.
\end{fact}

\begin{proof}
Let $\varphi: V \equiv_{(\termofTaylor{B_R(o)}{e_o}{i}, o_1)} T$. Let $\varphi'$ be the function $\ports{V} \to \ports{\langle o, \langle e_o, T \rangle \rangle}$ defined by $\varphi'(p) = (o, (e_o, \varphi(p)))$ for any $p$. 
For any $p \in \conclusions{V}$, we have 
\begin{eqnarray*}
\target{\termofTaylor{R}{e}{i}}((o, (e_o, o_1)), p) & = & (o, (e_o, \target{\termofTaylor{B_R(o)}{e_o}{i}}(o_1, p)))\\
& = & (o, (e_o, \target{T}(\varphi(p))))\\
& = & \target{\langle o, \langle e_o, T \rangle \rangle}(\varphi'(p))
\end{eqnarray*}
We thus have $\varphi': V \equiv_{(\termofTaylor{R}{e}{i}, (o, (e_o, o_1)))} \langle o, \langle e_o, T \rangle \rangle$.
\end{proof}

The rebuilding of the content of the boxes of $\termofTaylor{R}{e}{i+1}$ is achieved by the following proposition:

\begin{prop}\label{prop: boxes}
Let $R$ be a PS. Let $k \geq \basis{R}$. 
Let $e$ be a $k$-heterogeneous pseudo-experiment on $R$. Let $i \in \Nat$. Let $j_0 \in \mathcal{N}_i(e)$. 
We set $\mathfrak{T} = \nontrivialconnected{\termofTaylor{R}{e}{i}}{\criticalports{\termofTaylor{R}{e}{i}}{k}{j_0}}{k} / \equiv$. For any $\mathcal{T} \in \mathfrak{T}$, let $(m_j^{\mathcal{T}})_{j \in \Nat} \in \{ 0, \ldots, k-1 \}^\Nat$ such that $\Card{\mathcal{T}}  = \sum_{j \in \Nat} m_j^{\mathcal{T}} \cdot k^j$. Let $\mathcal{U} \subseteq \nontrivialconnected{\termofTaylor{R}{e}{i}}{\criticalports{\termofTaylor{R}{e}{i}}{k}{j_0}}{k}$ such that, for any $\mathcal{T} \in \mathfrak{T}$, we have $\Card{\mathcal{U} \cap \mathcal{T}} = m_{j_0}^{\mathcal{T}}$. Then we have $$B_{\termofTaylor{R}{e}{i+1}}(\cod_{e, i}(j_0)) \equiv_{(\termofTaylor{R}{e}{i+1}, \cod_{e, i}(j_0))} \bigoplus \mathcal{U}$$
\end{prop}

\begin{proof}
We first prove, by induction on $\depthof{R}$, that, for any $j \in \mathcal{N}_i(e)$, there exists an injection $\xi: \connectedcomponents{B_{\termofTaylor{R}{e}{i+1}}(\cod_{e, i}(j))} \to \nontrivialconnected{\termofTaylor{R}{e}{i}}{\criticalports{\termofTaylor{R}{e}{i}}{k}{j}}{k}$ such that, for any $V \in \connectedcomponents{B_{\termofTaylor{R}{e}{i+1}}(\cod_{e, i}(j))}$, we have $V \equiv_{(\termofTaylor{R}{e}{i+1}, \cod_{e, i}(j))} \xi(V)$. $(\ast)$ 

Let $j \in \mathcal{N}_i(e)$ and let $V \in \connectedcomponents{B_{\termofTaylor{R}{e}{i+1}}(\cod_{e, i}(j))}$.

If $\cod_{e, i}(j) \in \exactboxesatzero{R}{i}$, then $V \in \connectedcomponents{B_{R}(\cod_{e, i}(j))}$: Let $e_1 \in e(\cod_{e, i}(j))$. Let $\xi(V)$ be the following differential PS: 
$$\xi(V) = (\langle \cod_{e, i}(j), \langle e_1, V \rangle \rangle \oplus \bigoplus_{p \in \conclusions{B_R(\cod_{e, i}(j))} \cap \ports{V}} \labelofport{\groundof{R}}(\target{R}(\cod_{e, i}(j), p))_{\target{R}(\cod_{e, i}(j), p)})@t$$ where $t$ is the function $\{ \cod_{e, i}(j) \} \times (\{ e_1 \} \times (\conclusions{B_R(\cod_{e, i}(j))} \cap \ports{V})) \to \target{R}[\{ \cod_{e, i}(j) \} \times (\conclusions{B_R(\cod_{e, i}(j))} \cap \ports{V})]$ that associates with every $(\cod_{e, i}(j), (e_1, p))$ for some $p \in \conclusions{B_R(\cod_{e, i}(j))} \cap \ports{V}$ the port $\target{R}(\cod_{e, i}(j), p)$ of $\groundof{R}$.
By Proposition~\ref{prop: critical ports below new boxes}, we have 
\begin{eqnarray*}
\conclusions{\groundof{\xi(V)}} & = & \target{R}[\{ \cod_{e, i}(j) \} \times (\conclusions{B_R(\cod_{e, i}(j))} \cap \ports{V})]\\
& \subseteq & \target{R}[\{ \cod_{e, i}(j) \} \times \conclusions{B_R(\cod_{e, i}(j))}]\\
& = & \target{R}[\{ \cod_{e, i}(j) \} \times \temporaryConclusions{R}{\cod_{e, i}(j)}]\\
& = & \target{\termofTaylor{R}{e}{i+1}}[\{ \cod_{e, i}(j) \} \times \temporaryConclusions{\termofTaylor{R}{e}{i+1}}{\cod_{e, i}(j)}]\\
& = & \criticalports{\termofTaylor{R}{e}{i}}{k}{j}
\end{eqnarray*}
hence $\xi(V) \in \nontrivialconnected{\termofTaylor{R}{e}{i}}{\criticalports{\termofTaylor{R}{e}{i}}{k}{j}}{k}$. Moreover, we have $\overline{\xi(V)} = \langle \cod_{e, i}(j), \langle e_1, V \rangle \rangle$ 
and $V \equiv_{(\termofTaylor{R}{e}{i+1}, \cod_{e, i}(j))} \xi(V)$.

Now, if $\cod_{e, i}(j) = (o, (e_o, \cod_{e_o, i}(j)))$ for some $o \in \boxesatzerogeq{R}{i+1}$ and $e_o \in e(o)$, then we have $V \in \connectedcomponents{B_{\termofTaylor{B_R(o)}{e_o}{i+1}}(\cod_{e_o, i}(j))}$: Let $\varphi$ be some bijection $\contractionsUnder{R}{o} \simeq \mathcal{Q}'$ and let $R_o$ be an in-PS such that $R_o = \varphi \cdot_o R$; we have $V \in \connectedcomponents{B_{\termofTaylor{B_{R_o}(o)}{e_o}{i+1}}(\cod_{e_o, i}(j))}$, hence, by induction hypothesis, there exists $T'_V \in \nontrivialconnected{\termofTaylor{B_{R_o}(o)}{e_o}{i}}{\criticalports{\termofTaylor{B_{R_o}(o)}{e_o}{i}}{k}{j}}{k}$ such that $V \equiv_{(\termofTaylor{B_{R_o}(o)}{e_o}{i+1}, \cod_{e_o, i}(j))} T'_V$. Let $\varphi_{e_o}$ be the bijection $\contractionsUnder{R}{o} \simeq \{ o \} \times (\{ e_o \} \times \mathcal{Q}')$ defined by $\varphi_{e_o}(p) = (o, (e_o, \varphi(p)))$ for any $p \in \contractionsUnder{R}{o}$. We have $\langle o, \langle e_o, T'_V\rangle \rangle \in \nontrivialconnected{\termofTaylor{R_o}{e}{i}}{(\criticalports{\termofTaylor{R}{e}{i}}{k}{j} \setminus \contractionsUnder{R}{o}) \cup \varphi_{e_o}[\criticalports{\termofTaylor{R}{e}{i}}{k}{j} \cap \contractionsUnder{R}{o}]}{k}$, hence, by Lemma~\ref{lem: connected components of R_o}~(\ref{lem: connected components of R_o item: [inverse of varphi]}), $\langle o, \langle e_o, T'_V \rangle \rangle[{\varphi_{e_o}}^{-1}] \in \nontrivialconnected{\termofTaylor{R}{e}{i}}{\criticalports{\termofTaylor{R}{e}{i}}{k}{j}}{k}$. Moreover, by Fact~\ref{fact: equiv modulo}, we have $V \equiv_{(\termofTaylor{R_o}{e}{i+1}, \cod_{e, i}(j))} \langle o, \langle e_o, T'_V\rangle \rangle$, hence  $V \equiv_{(\termofTaylor{R}{e}{i+1}, \cod_{e, i}(j))} \langle o, \langle e_o, T'_V\rangle \rangle[{\varphi_{e_o}}^{-1}]$. We can thus set $\xi(V) = \langle o, \langle e_o, T'_V\rangle \rangle[{\varphi_{e_o}}^{-1}]$.

We proved $(\ast)$. Now, let us show that there exists a bijection $$\gamma: \connectedcomponents{B_{\termofTaylor{R}{e}{i+1}}(\cod_{e, i}(j_0))} \simeq \mathcal{U}$$ such that, for any $V \in \connectedcomponents{B_{\termofTaylor{R}{e}{i+1}}(\cod_{e, i}(j_0))}$, we have $V \equiv_{(\termofTaylor{R}{e}{i+1}, \cod_{e, i}(j_0))} \gamma(V)$: For any $T \in \im{\xi}$, by Proposition~\ref{prop: crucial}, there exists a bijection $$\gamma_T : \{ U \in \connectedcomponents{B_{\termofTaylor{R}{e}{i+1}}(\cod_{e, i}(j_0))} ; U \equiv_{\termofTaylor{R}{e}{i+1}, \cod_{e, i}(j_0)} T \} \simeq \{ T' \in \mathcal{U} ; T' \equiv T \}$$
We define the function $\gamma$ by setting $\gamma(V) = \gamma_{\xi(V)}(V)$ for any $V$. Since $\xi$ is an injection, $\gamma$ is an injection. Let us check that $\gamma$ is a surjection too: Let $T \in \mathcal{U}$; by Proposition~\ref{prop: crucial}, there exists $V \in \connectedcomponents{B_{\termofTaylor{R}{e}{i+1}}(\cod_{e, i}(j_0))}$ such that $V \equiv_{(\termofTaylor{R}{e}{i+1}, \cod_{e, i}(j))} T$; by Remark~\ref{rem: iso modulo}, we have $T \in \im{\gamma_{\xi(V)}}$.

Finally, we can apply Lemma~\ref{lem: bijection => equiv} to obtain $B_{\termofTaylor{R}{e}{i+1}}(\cod_{e, i}(j_0)) \equiv_{(\termofTaylor{R}{e}{i+1}, \cod_{e, i}(j_0))} \bigoplus \mathcal{U}$.
\end{proof}

\begin{exa}
Consider the $3$-heterogeneous experiment $e$ on the PS's $R_1$ and $R_2$ (which are depicted in Figure~\ref{fig: R_1} and in Figure~\ref{fig: R_2} respectively) such that $e^\#(o_1) = \{ 3 \}$ and $e^\#(o_2) = \{ 9 \}$. We set $\mathfrak{T}_1 = \nontrivialconnected{\termofTaylor{R_1}{e}{0}}{\criticalports{\termofTaylor{R_1}{e}{0}}{3}{1}}{3} / \equiv$ and $\mathfrak{T}_2 = \nontrivialconnected{\termofTaylor{R_2}{e}{0}}{\criticalports{\termofTaylor{R_2}{e}{0}}{3}{1}}{3} / \equiv$. We have $\Card{\mathfrak{T}_1} = 3 = \Card{\mathfrak{T}_2}$.
\begin{figure}
\begin{minipage}{0.33\textwidth}
\centering
\resizebox{.15 \textwidth}{!}{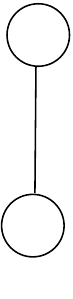}
\captionsetup{width= .9 \textwidth}
\caption{The differential PS $V_1$}
\label{fig: V_1}
\end{minipage}\hfill
\begin{minipage}{0.33\textwidth}
\centering
\resizebox{.35 \textwidth}{!}{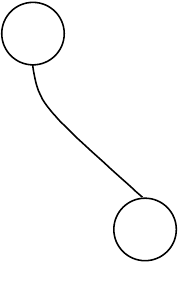}
\captionsetup{width= .9 \textwidth}
\caption{The differential PS $V_2$}
\label{fig: V_2}
\end{minipage}\hfill
\begin{minipage}{0.33\textwidth}
\centering
\resizebox{.5 \textwidth}{!}{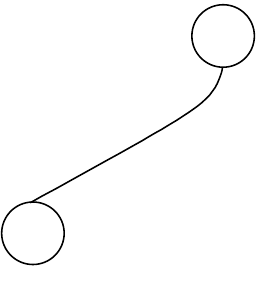}
\captionsetup{width= .9 \textwidth}
\caption{The differential PS $V_3$}
\label{fig: V_3}
\end{minipage}\hfill
\end{figure}
We have 
\begin{itemize}
\item $\Card{\{ V \in \bigcup \mathfrak{T}_1 ; V \equiv V_1 \}} = 3 = \Card{\{ V \in \bigcup \mathfrak{T}_2 ; V \equiv V_1 \}}$,
\item $\Card{\{ V \in \bigcup \mathfrak{T}_1 ; V \equiv V_2 \}} = 3 = \Card{\{ V \in \bigcup \mathfrak{T}_2 ; V \equiv V_3 \}}$
\item and $\Card{\{ V \in \bigcup \mathfrak{T}_1 ; V \equiv V_3 \}} = 9 = \Card{\{ V \in \bigcup \mathfrak{T}_2 ; V \equiv V_2 \}}$, 
\end{itemize}
where $V_1$, $V_2$ and $V_3$ are the differential PS's that are depicted in Figure~\ref{fig: V_1}, Figure~\ref{fig: V_2} and Figure~\ref{fig: V_3} respectively.
\end{exa}

\subsection{Characterizing a PS by some finite subset of its Taylor expansion}

The proof of the following proposition contains a full description of the algorithm leading from $\termofTaylor{R}{e}{i}$ to $\termofTaylor{R}{e}{i+1}$:

\begin{prop}\label{prop: from i to i+1}
Let $R$ and $R'$ be two PS's. Let $k \geq \max \{ \basis{R}, \basis{R'} \}$. Let $e$ be a $k$-heterogeneous pseudo-experiment on $R$ and let $e'$ be a $k$-heterogeneous pseudo-experiment on $R'$ such that $\termofTaylor{R}{e}{0} \equiv \termofTaylor{R'}{e'}{0}$. Then, for any $i \in \Nat$, we have $\termofTaylor{R}{e}{i} \equiv \termofTaylor{R'}{e'}{i}$.
\end{prop}

\begin{proof}
By induction on $i$. We assume that we are given $\varphi: \termofTaylor{R}{e}{i} \equiv \termofTaylor{R'}{e'}{i}$. 

We set $\mathcal{M} = \mathcal{M}_i(e)$ (resp. $\mathcal{M'} = \mathcal{M}_i(e')$) and $\mathcal{N} = \mathcal{N}_i(e)$ (resp. $\mathcal{N'} = \mathcal{N}_i(e')$). There is a bijection $\cod : \mathcal{M} \simeq \portsatzerooftype{\cod}{\termofTaylor{R}{e}{i}} \setminus \boxesatzero{{\termofTaylor{R}{e}{i}}}$ (resp. a bijection $\cod' : \mathcal{M}' \simeq \portsatzerooftype{\cod}{\termofTaylor{R'}{e'}{i}} \setminus \boxesatzero{{\termofTaylor{R'}{e'}{i}}}$) such that, for any $j \in \mathcal{M}$,  we have $(\arity{{\termofTaylor{R}{e}{i}}} \circ \cod)(j) = k^j$ (resp. $(\arity{{\termofTaylor{R'}{e'}{i}}} \circ \cod')(j) = k^j$). For any $j \in \mathcal{N}$, we set $\mathcal{T}_j = \nontrivialconnected{\termofTaylor{R}{e}{i}}{\criticalports{\termofTaylor{R}{e}{i}}{k}{j}}{k}$ (resp. $\mathcal{T}'_j = \nontrivialconnected{\termofTaylor{R'}{e'}{i}}{\criticalports{\termofTaylor{R'}{e'}{i}}{k}{j}}{k}$). We set $\mathfrak{T} = \bigcup_{j \in \mathcal{N}} (\mathcal{T}_j / \equiv)$ (resp. $\mathfrak{T}' = \bigcup_{j \in \mathcal{N}'} (\mathcal{T}'_j / \equiv)$). We set $\mathcal{P} =  (\portsatzero{{\termofTaylor{R}{e}{i}}} \setminus \bigcup_{j \in \mathcal{N}} \bigcup_{\begin{array}{c} T \in \mathcal{T}_j \end{array}} \portsatzero{T}) \cup \criticalports{{\termofTaylor{R}{e}{i}}}{k}{\mathcal{N}}$ (resp. $\mathcal{P}' =  (\portsatzero{{\termofTaylor{R'}{e'}{i}}} \setminus \bigcup_{j \in \mathcal{N'}} \bigcup_{\begin{array}{c} T' \in \mathcal{T}'_j \end{array}} \portsatzero{T'}) \cup \criticalports{{\termofTaylor{R'}{e'}{i}}}{k}{\mathcal{N'}}$). For any $\mathcal{T} \in \mathfrak{T}$ (resp. $\mathcal{T}' \in \mathfrak{T}'$), we define $(m_j^{\mathcal{T}}) \in {\{ 0, \ldots, k-1 \}}^\Nat$ (resp. $(n_j^{\mathcal{T}'}) \in {\{ 0, \ldots, k-1 \}}^\Nat$) as follows: $\Card{\mathcal{T}} = \sum_{j \in \Nat} m_j^{\mathcal{T}} \cdot k^j$ (resp. $\Card{\mathcal{T}'} = \sum_{j \in \Nat} n_j^{\mathcal{T}'} \cdot k^j$). 

Notice that $\mathcal{M} = \mathcal{M}'$, $\mathcal{N} = \mathcal{N}'$ and, for any $j \in \mathcal{N}$, $\varphi[\criticalports{\termofTaylor{R}{e}{i}}{k}{j}] = \criticalports{\termofTaylor{R'}{e'}{i}}{k}{j}$. Moreover, there exists a bijection $\sigma: \bigcup_{j \in \mathcal{N}} \mathcal{T}_j \simeq \bigcup_{j \in \mathcal{N}} \mathcal{T}'_j$ such that, for any $T \in \bigcup_{j \in \mathcal{N}} \mathcal{T}_j$, there is an isomorphism $T \simeq \sigma(T)$ associating with every port $p$ of ${T}$ the port $\varphi(p)$ of $\termofTaylor{R'}{e'}{i}$, hence $\varphi[\mathcal{P}] = \mathcal{P}'$. Also, notice that, for any $T_1, T_2 \in \bigcup_{j \in \mathcal{N}} \mathcal{T}_j$, we have $(T_1 \equiv T_2 \Leftrightarrow \sigma(T_1) \equiv \sigma(T_2))$, which entails that there exists a bijection $\overline{\sigma} : \mathfrak{T} \simeq \mathfrak{T}'$ such that, for any $\mathcal{T} \in \mathfrak{T}$, for any $\mathcal{T}' \in \mathfrak{T}'$, for any $T \in \mathcal{T}$, for any $T' \in \mathcal{T}'$, we have $(\sigma(T) = T' \Rightarrow \overline{\sigma}(\mathcal{T}) = \mathcal{T}')$; moreover, for any $\mathcal{T} \in \mathfrak{T}$, we have $\Card{\mathcal{T}} = \Card{\overline{\sigma}(\mathcal{T})}$. For any $\mathcal{T} \in \mathfrak{T}$, we thus have $(n_j^{\overline{\sigma}(\mathcal{T})})_{j \in \Nat} = (m_j^{\mathcal{T}})_{j \in \Nat}$. 

For any $j \in \mathcal{N}$, we set $\mathcal{V}_j = \nontrivialconnected{\termofTaylor{R}{e}{i+1}}{\criticalports{\termofTaylor{R}{e}{i}}{k}{j}}{k}$. We set $\mathfrak{V} = \bigcup_{j \in \mathcal{N}} (\mathcal{V}_j / \equiv)$. For any $j \in \mathcal{N}$, we are given $\mathcal{U}_j \subseteq \mathcal{T}_j$ such that, for any $\mathcal{T} \in \mathcal{T}_j/\equiv$, we have $\Card{\mathcal{U}_j \cap \mathcal{T}} = m_j^{\mathcal{T}}$.  Then one can describe the differential PS $\termofTaylor{R}{e}{i+1}$ as follows:
\begin{itemize}
\item $\restriction{{\termofTaylor{R}{e}{i}}^{\leq 0}}{\mathcal{P}} \sqsubseteq_\emptyset {\termofTaylor{R}{e}{i+1}}^{\leq 0}$ (by Proposition~\ref{prop: ground-structure});
\item ${\termofTaylor{R}{e}{i+1}}^{\leq i} \sqsubseteq_{\emptyset} {\termofTaylor{R}{e}{i}}$  (by Lemma~\ref{lem: ground-structure: trivial});
\item $(\forall j \in \mathcal{N}) \mathcal{V}_j \subseteq \mathcal{T}_{j}$ (by Proposition~\ref{prop: connected components of termofTaylor{R}{e}{i+1}});
\item for any $\mathcal{V} \in \mathfrak{V}$, we have $\Card{\mathcal{V}} = \sum_{j \notin \mathcal{N}} m_j^{\mathcal{T}} \cdot k^j$, where $\mathcal{T} \in \mathfrak{T}$ such that $\mathcal{V} \subseteq \mathcal{T}$ (by Proposition~\ref{prop: crucial});
\item $\boxesatzero{{\termofTaylor{R}{e}{i+1}}} = (\boxesatzero{{\termofTaylor{R}{e}{i}}} \cap \portsatzero{{\termofTaylor{R}{e}{i+1}}}) \cup \cod[\mathcal{N}]$ (by Lemma~\ref{lem: ground-structure: trivial} and Proposition~\ref{prop: critical ports below new boxes});
\item and, for any $j \in \mathcal{N}$, there exists $\varphi_j:  B_{{\termofTaylor{R}{e}{i+1}}}(\cod(j)) \equiv_{({\termofTaylor{R}{e}{i+1}}, \cod(j))} \bigoplus \mathcal{U}_j$ (by Proposition~\ref{prop: boxes}).
\end{itemize}

In the same way: For any $j \in \mathcal{N}$, we set $\mathcal{V}'_j = \nontrivialconnected{\termofTaylor{R'}{e'}{i+1}}{\criticalports{\termofTaylor{R'}{e'}{i}}{k}{j}}{k}$. We set $\mathfrak{V}' = \bigcup_{j \in \mathcal{N}} (\mathcal{V}'_j / \equiv)$. For any $j \in \mathcal{N}$, we set $\mathcal{U}'_j = \sigma[\mathcal{U}_j]$. Then one can describe the differential PS $\termofTaylor{R'}{e'}{i+1}$ as follows:
\begin{itemize}
\item $\restriction{{\termofTaylor{R'}{e'}{i}}^{\leq 0}}{\mathcal{P'}} \sqsubseteq_\emptyset {\termofTaylor{R'}{e'}{i+1}}^{\leq 0}$ (by Proposition~\ref{prop: ground-structure});
\item ${\termofTaylor{R'}{e'}{i+1}}^{\leq i} \sqsubseteq_{\emptyset} {\termofTaylor{R'}{e'}{i}}$  (by Lemma~\ref{lem: ground-structure: trivial});
\item $(\forall j \in \mathcal{N'}) \mathcal{V}'_j \subseteq \mathcal{T}'_{j}$ (by Proposition~\ref{prop: connected components of termofTaylor{R}{e}{i+1}});
\item for any $\mathcal{V}' \in \mathfrak{V}'$, we have $\Card{\mathcal{V}'} = \sum_{j \notin \mathcal{N}'} n_j^{\mathcal{T}'} \cdot k^j$, where $\mathcal{T}' \in \mathfrak{T}'$ such that $\mathcal{V}' \subseteq \mathcal{T}'$ (by Proposition~\ref{prop: crucial});
\item $\boxesatzero{{\termofTaylor{R'}{e'}{i+1}}} = (\boxesatzero{{\termofTaylor{R'}{e'}{i}}} \cap \portsatzero{{\termofTaylor{R'}{e'}{i+1}}}) \cup \cod'[\mathcal{N}']$ (by Lemma~\ref{lem: ground-structure: trivial} and Proposition~\ref{prop: critical ports below new boxes});
\item and, for any $j \in \mathcal{N}'$, there exists $\varphi'_j:  B_{{\termofTaylor{R'}{e'}{i+1}}}(\cod'(j)) \equiv_{({\termofTaylor{R'}{e'}{i+1}}, \cod'(j))} \bigoplus \mathcal{U}'_j$ (by Proposition~\ref{prop: boxes}).
\end{itemize}

So, for any $\mathcal{T} \in \mathfrak{T}$, there exists a bijection $\tau_{\mathcal{T}}: \mathcal{T} \simeq \overline{\sigma}(\mathcal{T})$ such that 
$$\tau_{\mathcal{T}}[\mathcal{T} \cap (\bigcup_{j \in \mathcal{N}} \nontrivialconnected{{\termofTaylor{R}{e}{i+1}}}{\criticalports{{\termofTaylor{R}{e}{i}}}{k}{j}}{k})] = \overline{\sigma}(\mathcal{T}) \cap (\bigcup_{j \in \mathcal{N}'} \nontrivialconnected{{\termofTaylor{R'}{e'}{i+1}}}{\criticalports{{\termofTaylor{R'}{e'}{i}}}{k}{j}}{k})$$ 
hence there exist a bijection $\tau: \bigcup_{j \in \mathcal{N}} \mathcal{T}_j \simeq \bigcup_{j \in \mathcal{N}} \mathcal{T}'_j$ and a sequence $(\psi_T)_{T \in \bigcup_{j \in \mathcal{N}} \mathcal{T}_j}$ such that $(\forall T \in \bigcup_{j \in \mathcal{N}} \mathcal{T}_j) \psi_T: \sigma(T) \equiv \tau(T)$ and 
$$\tau[\bigcup_{j \in \mathcal{N}} \nontrivialconnected{{\termofTaylor{R}{e}{i+1}}}{\criticalports{{\termofTaylor{R}{e}{i}}}{k}{j}}{k}] = \bigcup_{j \in \mathcal{N'}} \nontrivialconnected{{\termofTaylor{R'}{e'}{i+1}}}{\criticalports{{\termofTaylor{R'}{e'}{i}}}{k}{j}}{k}$$ 
For any $p \in \portsatzero{\termofTaylor{R'}{e'}{i}} \setminus \mathcal{P'}$, let $H(p)$ be the unique $T' \in \bigcup_{j \in \mathcal{N}} \mathcal{T}'_j$ such that $p \in \portsatzero{T'}$.

We can thus define a bijection $\psi: \ports{\termofTaylor{R}{e}{i+1}} \simeq \ports{\termofTaylor{R'}{e'}{i+1}}$ such that:
\begin{itemize}
\item for any $p \in \mathcal{P}$, we have $\psi(p) = \varphi(p) \in \mathcal{P}'$;
\item for any $o \in \boxesatzero{\termofTaylor{R}{e}{i+1}} \cap \mathcal{P}$, for any $p \in \ports{B_{\termofTaylor{R}{e}{i+1}}(o)}$, we have $\psi(o, p) = \varphi(o, p)$;
\item for any $p \in \portsatzero{\termofTaylor{R}{e}{i+1}} \setminus \mathcal{P}$, we have $\psi(p) = \psi_{H(\varphi(p))}(\varphi(p))$;
\item for any $o \in \boxesatzerosmaller{\termofTaylor{R}{e}{i+1}}{i} \setminus \mathcal{P}$, for any $p \in \ports{B_{\termofTaylor{R}{e}{i+1}}(o)}$, we have $\psi(o, p) = \psi_{H(\varphi(o))}(\varphi(o, p))$;
\item and, for any $j \in \mathcal{N}$, for any $p \in \ports{B_{{\termofTaylor{R}{e}{i+1}}}(\cod(j))}$, we have $\psi(\cod(j), p) = (\varphi(\cod(j)), ({\varphi'_j}^{-1} \circ \varphi \circ \varphi_j)(p))$.
\end{itemize}
It is straightforward to check that $\psi: \termofTaylor{R}{e}{i+1} \equiv \termofTaylor{R'}{e'}{i+1}$. In particular:
\begin{itemize}
\item for any $p \in \wiresatzero{\termofTaylor{R}{e}{i+1}} \setminus \mathcal{P}$ such that $\target{\groundof{\termofTaylor{R}{e}{i+1}}}(p) \in \mathcal{P}$, we have $\varphi(\target{\groundof{\termofTaylor{R}{e}{i}}}(p)) \in \mathcal{P}' \cap \portsatzero{H(\varphi(p))} = \conclusions{H(\varphi(p))}$, hence
\begin{eqnarray*}
\psi(\target{\groundof{\termofTaylor{R}{e}{i+1}}}(p)) & = & \varphi(\target{\groundof{\termofTaylor{R}{e}{i+1}}}(p))\\
& = & \varphi(\target{\groundof{\termofTaylor{R}{e}{i}}}(p))\\
& = & \psi_{H(\varphi(p))}(\varphi(\target{\groundof{\termofTaylor{R}{e}{i}}}(p)))\allowdisplaybreaks\\
& = & \psi_{H(\varphi(p))}(\target{\groundof{\termofTaylor{R'}{e'}{i}}}(\varphi(p)))\allowdisplaybreaks\\
& = & \target{\groundof{\termofTaylor{R'}{e'}{i}}}(\psi_{H(\varphi(p))}(\varphi(p)))\\
& = & \target{\groundof{\termofTaylor{R'}{e'}{i+1}}}(\psi_{H(\varphi(p))}(\varphi(p)))\\
& = & \target{\groundof{\termofTaylor{R'}{e'}{i+1}}}(\psi(p))
\end{eqnarray*}
\item we have 
\begin{eqnarray*}
\psi[\boxesatzero{\termofTaylor{R}{e}{i+1}}] & = & \psi[\boxesatzero{\termofTaylor{R}{e}{i+1}} \cap \mathcal{P}] \cup \psi[\boxesatzero{\termofTaylor{R}{e}{i+1}} \setminus \mathcal{P}]\\
& = & \varphi[\boxesatzero{\termofTaylor{R}{e}{i+1}} \cap \mathcal{P}] \cup \{ \psi_{H(\varphi(o))}(\varphi(o)) ; o \in \boxesatzero{\termofTaylor{R}{e}{i+1}} \setminus \mathcal{P} \}\allowdisplaybreaks\\
& = & (\boxesatzero{\termofTaylor{R'}{e'}{i+1}} \cap \mathcal{P}') \cup (\boxesatzero{\termofTaylor{R'}{e'}{i+1}} \setminus \mathcal{P}') \\
& = & \boxesatzero{\termofTaylor{R'}{e'}{i+1}}
\end{eqnarray*}
\item For any $j \in \mathcal{N}$, we have an isomorphism $\varphi_{\cod(j)}:\bigoplus \mathcal{U}_j \simeq \bigoplus \mathcal{U}'_j$ that associates with every $p \in \ports{\bigoplus \mathcal{U}_j}$ the port $p \in \ports{\bigoplus \mathcal{U}'_j}$, hence $(\varphi'_j)^{-1} \circ \varphi_{\cod(j)} \circ \varphi_j$ is an isomorphism $B_{\termofTaylor{R}{e}{i+1}}(\cod(j)) \simeq B_{\termofTaylor{R'}{e'}{i+1}}(\cod'(j))$.
  \qedhere
\end{itemize}
\end{proof}

A set of two well-chosen terms of the Taylor expansion are already enough to characterize a PS:

\begin{thm}\label{theorem: finite characterization}
For any PS $R$ having $\mathcal{T}$ as Taylor expansion, there exists a finite subset $\mathcal{T}_0$ of $\mathcal{T}$ with $\Card{\mathcal{T}_0} = 2$ such that, for any PS $R'$ having $\mathcal{T'}$ as Taylor expansion, for any $\mathcal{T}_0' \subseteq \mathcal{T'}$, we have $(\mathcal{T}_0 \equiv \mathcal{T}_0' \Rightarrow R \equiv R')$.
\end{thm}

\begin{proof}
From $\termofTaylor{R}{e_1}{0}$ for $e_1$ a $1$-experiment, one can compute $\basis{R}$. Indeed, one can easily check by induction on $\depthof{R}$ that we have:
\begin{itemize}
\item $\cosize{R} = \cosize{\termofTaylor{R}{e_1}{0}}$
\item $\Card{\boxes{R}} = \Card{\boxes{\termofTaylor{R}{e_1}{0}}}$
\item $\numberInvisibleComponents{R} = \numberInvisibleComponents{\termofTaylor{R}{e_1}{0}}$
\end{itemize}
One can then take for $\mathcal{T}_0$ the set $\{ \termofTaylor{R}{e_1}{0}, \termofTaylor{R}{e}{0} \}$, where $e$ is any $k$-heterogeneous pseudo-experiment on $R$ with $k \geq \basis{R}$. Indeed:

Let $R'$ be a PS having $\mathcal{T'}$ as Taylor expansion and let $\{ T'_1, T'_2 \} \subseteq \mathcal{T'}$ such that $T'_1 \equiv \termofTaylor{R}{e_1}{0}$ and $T'_2 \equiv \termofTaylor{R}{e}{0}$. There exist a pseudo-experiment $e'_1$ on $R'$ such that $T'_1 = \termofTaylor{R'}{e'_1}{0}$ and a pseudo-experiment $e'$ on $R'$ such that $T'_2 = \termofTaylor{R'}{e'}{0}$. By Corollary~\ref{cor: characterization of k-heterogeneous experiments}, the pseudo-experiment $e'$ is a $k$-heterogeneous experiment on $R'$. So we can apply Proposition~\ref{prop: from i to i+1}. 
\end{proof}

In particular, we obtain the invertibility of Taylor expansion:

\begin{cor}\label{cor: invertibility of Taylor expansion}
Let $R_1$ and $R_2$ be two PS's having respectively $\mathcal{T}_1$ and $\mathcal{T}_2$ as Taylor expansions. If $\mathcal{T}_1 \equiv \mathcal{T}_2$, then $R_1 \equiv R_2$.
\end{cor}

The finite subset $\mathcal{T}_0$ of Theorem~\ref{theorem: finite characterization} has cardinality $2$. A natural question is to ask if the theorem could be refined in such a way that one could have a singleton $\mathcal{T}_0$ (for any PS). The answer is: no, it is not possible.

\begin{prop}\label{prop: non-principal typing}
There exists a PS $R$ having $\mathcal{T}$ as Taylor expansion such that, for any $T \in \mathcal{T}$, there exist a PS $R'$ having $\mathcal{T'}$ as Taylor expansion and $T' \in \mathcal{T'}$ such that $T \equiv T'$ holds but $R \equiv R'$ does not hold.
\end{prop}

\begin{proof}
For $R$, we take the PS of depth $1$ depicted in Figure~\ref{fig: non-principal typing-1}.
\begin{figure}
\centering
\resizebox{.3 \textwidth}{!}{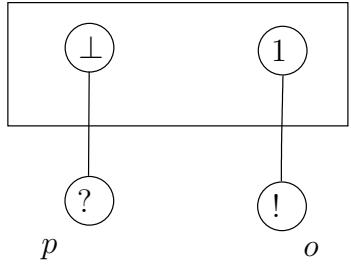}
\captionsetup{width=.9 \textwidth}
\caption{The PS $R$ of the proof of Proposition~\ref{prop: non-principal typing}}
\label{fig: non-principal typing-1}
\end{figure}
The PS $R$ has as a Taylor expansion the set $\{ T_n ; n \in \Nat \}$, where, for any $n \in \Nat$, the differential PS $T_n$ is as depicted in Figure~\ref{fig: non-principal typing-differential}.
\begin{figure}
\centering
\begin{minipage}{0.45\textwidth}
\raggedright
\resizebox{\textwidth}{!}{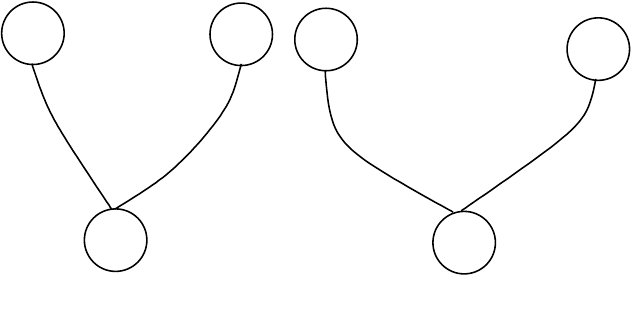}
\captionsetup{width=.9 \textwidth}
\caption{The differential PS's $T_n$ of Proposition~\ref{prop: non-principal typing}}
\label{fig: non-principal typing-differential}
\end{minipage}\hfill
\begin{minipage}{0.45\textwidth}
\resizebox{.7 \textwidth}{!}{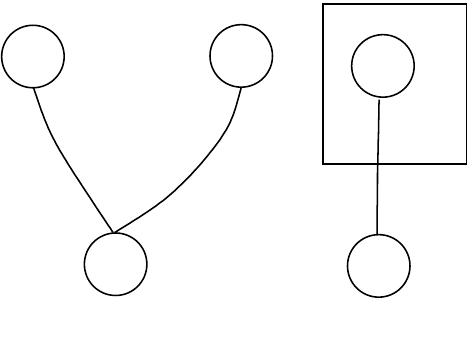}
\captionsetup{width=.9 \textwidth}
\caption{The PS's $R'_n$ of Proposition~\ref{prop: non-principal typing}}
\label{fig: non-principal typing-2}
\end{minipage}\hfill
\end{figure}
For any $n \in \Nat$, let $R'_n$ be the PS as depicted in Figure~\ref{fig: non-principal typing-2} having $\mathcal{T}_n'$ as Taylor expansion; there exists $T' \in \mathcal{T}_n'$ such that $T' \equiv T_n$.
\end{proof}

\section{Relational semantics}\label{section: injectivity}

\subsection{Untyped framework}

For the semantics of PS's in the multiset based relational semantics, we are given a set $\mathcal{A}$ that does not contain any couple nor any $3$-tuple and such that $\ast \notin \mathcal{A}$. We define, by induction on $n$, the sets $D_{\mathcal{A}, n}$ for any $n \in \Nat$:
\begin{itemize}
\item $D_{\mathcal{A}, 0} = \{ +, - \} \times (\mathcal{A} \cup \{ \ast \})$
\item $D_{\mathcal{A}, n+1} = D_{\mathcal{A}, 0} \cup (\{ +, - \} \times D_{\mathcal{A}, n} \times D_{\mathcal{A}, n}) \cup (\{ +, - \} \times \finitemultisets{D_{\mathcal{A}, n}})$
\end{itemize}
We set $D_\mathcal{A} = \bigcup_{n \in \Nat} D_{\mathcal{A}, n}$. 
\begin{itemize}
\item For any $\alpha \in D_{\mathcal{A}}$, we denote by $\height{\alpha}$ the integer $\min \{ n \in \Nat ; \alpha \in D_{\mathcal{A}, n} \}$. For any finite set $\mathcal{P}$, for any function $x: \mathcal{P} \to D_{\mathcal{A}}$, we set $\height{x} = \max \{ \height{x(p)} ; p \in \mathcal{P} \}$.
\item For any $\alpha \in D_{\mathcal{A}}$, we define the integer $\size{\alpha}$ by induction on $\height{\alpha}$: for any $\delta \in \{ +, - \}$, we set $\size{(\delta, \gamma)} = 1$ if $\gamma \in \mathcal{A} \cup \{ \ast \}$, $\size{(\delta, \alpha_1, \alpha_2)} = 1 + \size{\alpha_1} + \size{\alpha_2}$ if $\alpha_1, \alpha_2 \in D_{\mathcal{A}}$, and $\size{(\delta, \multi{\alpha_1, \ldots, \alpha_m})} = 1 + \sum_{j=1}^m \size{\alpha_j}$. For any finite set $\mathcal{P}$, for any function $x: \mathcal{P} \to D_{\mathcal{A}}$, we set $\size{x} = \sum_{p \in \mathcal{P}} \size{x(p)}$.
\item For any $\alpha \in D_{\mathcal{A}}$, we denote by $\atoms{\alpha}$ the set of $\gamma \in \mathcal{A}$ that occur in $\alpha$. For any finite set $\mathcal{P}$, for any function $x: \mathcal{P} \to D_{\mathcal{A}}$, we set $\atoms{x} = \bigcup_{\alpha \in \im{x}} \atoms{\alpha}$.
\end{itemize}

\begin{defi}
For any $\alpha \in D_\mathcal{A}$, we define $\alpha^\perp \in D_\mathcal{A}$ as follows:
\begin{itemize}
\item if $\alpha \in \mathcal{A} \cup \{ \ast \}$ and $\delta \in \{ +, - \}$, then $(\delta, \alpha)^\perp = (\delta^\perp, \alpha)$;
\item if $\alpha = (\delta, \alpha_1, \alpha_2)$ with $\delta \in \{ +, - \}$ and $\alpha_1, \alpha_2 \in D_{\mathcal{A}}$, then $\alpha^\perp = (\delta^\perp, {\alpha_1}^\perp, {\alpha_2}^\perp)$;
\item if $\alpha = (\delta, [\alpha_1, \ldots, \alpha_m])$ with $\delta \in \{ +, - \}$ and $\alpha_1, \ldots, \alpha_m \in D_{\mathcal{A}}$, then $\alpha^\perp = (\delta^\perp, [{\alpha_1}^\perp, \ldots, {\alpha_m}^\perp])$;
\end{itemize}
where $+^\perp = -$ and $-^\perp = +$.
\end{defi}

The following definition is an adaptation of the original definition in \cite{ll} to our framework. If $e$ is an experiment on some PS $S$, then $\restriction{\ports{e}}{\conclusions{S}}$ is its \emph{result}. Then the interpretation of a PS is the set of results of its experiments:

\begin{defi}\label{defin: experiment in an untyped framework}
Let $S$ be a differential in-PS. We define, by induction on $\textit{depth}(S)$, the set of \emph{experiments on $S$}: it is the set of pairs $e = (\ports{e}, \boxes{e})$, where 
\begin{itemize}
\item $\ports{e}$ is a function that associates with every $p \in \portsatzero{S}$ an element of $D_{\mathcal{A}}$ and with every $p \in \portsatdepthgreater{S}{0}$ an element of $\finitemultisets{D_{\mathcal{A}}}$,
\item and $\boxes{e}$ is a function which associates with every $o \in \boxesatzero{S}$ a finite multiset of experiments on $B_S(o)$ 
\end{itemize}
such that
\begin{itemize}
\item for any $p \in \multiplicativeportsatzero{S}$, for any $w_1, w_2 \in \wiresatzero{S}$ such that $\target{\groundof{S}}(w_1) = p = \target{\groundof{S}}(w_2)$, $w_1 \in \leftwires{\groundof{S}}$ and $w_2 \notin \leftwires{\groundof{S}}$, we have $\ports{e}(p) = (\delta, \ports{e}(w_1), \ports{e}(w_2))$ with $\delta \in \{ +, - \}$ and $(\labelofport{\groundof{S}}(p) = \tens \Leftrightarrow \delta = +)$;
\item for any $\{ p, p' \} \in \axiomsatzero{S} \cup \cutsatzero{S}$, we have $\ports{e}(p) = {\ports{e}(p')}^\perp$;
\item for any $p \in \exponentialportsatzero{S}$,
we have $\ports{e}(p) = (\delta, \sum_{\substack{w \in \wiresatzero{S}\\ \target{\groundof{S}}(w) = p}} \multi{\ports{e}(w)} +  \sum_{\substack{q \in \portsatdepthgreater{S}{0}\\ \target{S}(q) = p}} \ports{e}(q))$ and $\delta \in \{ +, - \}$ and $(\labelofport{\groundof{S}}(p) = \cod \Leftrightarrow \delta = +)$;
\item for any $p \in \portsatzerooftype{\one}{S}$, we have $\ports{e}(p) = (+, \ast)$ and, for any $p \in \portsatzerooftype{\bot}{S}$, we have $\ports{e}(p) = (-, \ast)$;
\item for any $o \in \boxesatzero{S}$, for any $p \in \portsatzero{B_S(o)}$, we have $\ports{e}(p) = \sum_{e_o \in \supp{\boxes{e}(o)}}$ $\boxes{e}(o)(e_o) \cdot \multi{\ports{e_o}(p)}$;
\item for any $o \in \boxesatzero{S}$, for any $p \in \portsatdepthgreater{B_S(o)}{0}$, we have $\ports{e}(p) = \sum_{e_o \in \supp{\boxes{e}(o)}}$ $\boxes{e}(o)(e_o) \cdot \ports{e_o}(p)$.
\end{itemize}
If $S$ is a PS, then we set $\sm{S} =$ $\{ \restriction{\ports{e}}{\conclusions{S}} ;$ $e \textit{ is an experiment on } S \}$.
\end{defi}

Every experiment induces a pseudo-experiment:

\begin{defi}\label{definition: experiment induces pseudo}
Let $S$ be an in-PS. Let $e$ be an experiment on $S$. Then we define, by induction on $\depthof{S}$, a pseudo-experiment $\overline{e}$ on $S$ as follows: 
\begin{itemize}
\item $\overline{e}(\emptysequence) = 0$
\item and, for any $o \in \boxesatzero{S}$, $\overline{e}(o) = \bigcup_{e_o \in \supp{e(o)}} \{ \overline{e_o}[\emptysequence \mapsto i] ; 1 \leq i \leq \boxes{e}(e_o) \}$.
\end{itemize}
\end{defi}

\begin{defi}\label{defin: injective}
Let $r \in \finitemultisets{D_{\mathcal{A}}}$. We say that $r$ is \emph{$\mathcal{A}$-injective} if, for any $(\delta, \gamma) \in \{ +, - \} \times {\mathcal{A}}$, there is at most one occurrence of $(\delta, \gamma)$ in $r$. 

For any set $\mathcal{P}$, for any function $x : \mathcal{P} \to D_{\mathcal{A}}$, we say that $x$ is \emph{$\mathcal{A}$-injective} if $\sum_{p \in \mathcal{P}} [x(p)]$ is $\mathcal{A}$-injective; moreover a subset $D_0$ of $(D_{\mathcal{A}})^{\mathcal{P}}$ is said to be $\mathcal{A}$-injective if, for any $x \in D_0$, the function $x$ is $\mathcal{A}$-injective.

An experiment $e$ on a differential PS $S$ is said to be \emph{injective} if $\restriction{\ports{e}}{\conclusions{R}}$ is $\mathcal{A}$-injective.
\end{defi}

\begin{defi}
Let $\sigma : {\mathcal{A}} \to D_{\mathcal{A}}$. For any $\alpha \in D_{\mathcal{A}}$, we define $\sigma \cdot \alpha \in D_{\mathcal{A}}$ as follows:
\begin{itemize}
\item if $\alpha \in {\mathcal{A}}$, then $\sigma \cdot (+, \alpha) = \sigma(\alpha)$ and $\sigma \cdot (-, \alpha) = \sigma(\alpha)^\perp$;
\item if $\alpha = (\delta, \ast)$ for some $\delta \in \{ +, - \}$, then $\sigma \cdot \alpha = \alpha$;
\item if $\delta \in \{ +, - \}$ and $\alpha_1, \alpha_2 \in D_{\mathcal{A}}$, then $\sigma \cdot (\delta, \alpha_1, \alpha_2) = (\delta, \sigma \cdot \alpha_1, \sigma \cdot \alpha_2)$;
\item if $\delta \in \{ +, - \}$ and $\alpha_1, \ldots, \alpha_m \in D_{\mathcal{A}}$, then $\sigma \cdot (\delta, [\alpha_1, \ldots, \alpha_m]) = (\delta, [\sigma \cdot \alpha_1, \ldots, \sigma \cdot \alpha_m])$.
\end{itemize}
For any set $\mathcal{P}$, for any function $x : \mathcal{P} \to D_{\mathcal{A}}$, we define a function $\sigma \cdot x : \mathcal{P} \to D_{\mathcal{A}}$ by setting: $(\sigma \cdot x)(p) = \sigma \cdot x(p)$ for any $p \in \mathcal{P}$.
\end{defi}

\begin{fact}\label{fact: substitutions compatible with duality}
For any $\sigma: A \to D_{\mathcal{A}}$, for any $\alpha \in D_{\mathcal{A}}$, we have $\sigma \cdot {\alpha}^{\perp} = {(\sigma \cdot \alpha)}^{\perp}$.
\end{fact}

\begin{proof}
By induction on $\alpha$.
\end{proof}

\begin{lem}\label{lem: substitutions are actions}
For any functions $\sigma, \sigma' : {\mathcal{A}} \to D_{\mathcal{A}}$, for any $\alpha \in D_{\mathcal{A}}$, we have $\sigma \cdot (\sigma' \cdot \alpha) = (\sigma \cdot \sigma') \cdot \alpha$.
\end{lem}

\begin{proof}
By induction on $\alpha$, applying Fact~\ref{fact: substitutions compatible with duality}.
\end{proof}

\begin{defi}
Let $S$ be a differential in-PS. Let $e$ be an experiment on $S$. Let $\sigma : {\mathcal{A}} \to D_{\mathcal{A}}$. We define, by induction of $\depthof{S}$, an experiment $\sigma \cdot e$ on $S$ by setting 
\begin{itemize}
\item $\ports{\sigma \cdot e} = \sigma \cdot \ports{e}$
\item $\boxes{\sigma \cdot e}(o) = \sum_{e_o \in \supp{\boxes{e}(o)}} \boxes{e}(o)(e_o) \cdot [\sigma \cdot e_o]$ for any $o \in \boxesatzero{S}$.
\end{itemize}
\end{defi}

Since we deal with untyped proof-nets, we cannot assume that the proof-nets are $\eta$-expanded and that experiments label the axioms only by atoms. That is why we introduce the notion of \emph{atomic experiment}:

\begin{defi}\label{definition: atomic experiment}
For any differential in-PS $S$, we define, by induction on $\depthof{S}$, the set of \emph{atomic experiments on $S$}: it is the set of experiments $e$ on $S$ such that
\begin{itemize}
\item for any $\{ p, q \} \in \axioms{\groundof{S}}$, we have $\ports{e}(p), \ports{e}(q) \in \{ +, - \} \times {\mathcal{A}}$;
\item and, for any $o \in \boxesatzero{S}$, the multiset $\boxes{e}(o)$ is a multiset of atomic experiments on $B_S(o)$.
\end{itemize}
\end{defi}

\begin{defi}\label{definition: atomic point}
Let $\mathcal{P}$ be a set.  

Let $x \in {(D_{\mathcal{A}})}^\mathcal{P}$. A \emph{renaming of $x$} is a function $\sigma: \mathcal{A} \to D_{\mathcal{A}}$ such that $(\forall \gamma \in \atoms{x}) \sigma(\gamma) \in \{ +, - \} \times \mathcal{A}$.

Let $\mathcal{D} \subseteq {(D_{\mathcal{A}})}^\mathcal{P}$. 
Let $x \in \mathcal{D}$, we say that $x$ is \emph{$\mathcal{D}$-atomic} if we have 
$$(\forall \sigma \in {(D_{\mathcal{A}})}^{\mathcal{A}}) (\forall y \in \mathcal{D}) (\sigma \cdot y = x \Rightarrow \sigma \in \mathfrak{R}(y))$$

We denote by $\mathcal{D}_{\textit{At}}$ the subset of $\mathcal{D}$ consisting of the $\mathcal{D}$-atomic elements of $\mathcal{D}$.
\end{defi}

\begin{fact}\label{fact: st renaming => t renaming}
Let $x \in D_{\mathcal{A}}$. Let $\sigma$ and $\tau$ be two applications $\mathcal{A} \to D_{\mathcal{A}}$. Then we have $(\sigma \cdot \tau \in \mathfrak{R}(x) \Rightarrow (\tau \in \mathfrak{R}(x) \wedge \sigma \in \mathfrak{R}(\tau \cdot x)))$.
\end{fact}

\begin{proof}
Let us assume that $\sigma \cdot \tau \in \mathfrak{R}(x)$. By Lemma~\ref{lem: substitutions are actions}, we have $(\sigma \cdot \tau) \cdot x = \sigma \cdot (\tau \cdot x)$. So, we have $\size{\tau \cdot x} \leq \size{\sigma \cdot (\tau \cdot x)} = \size{(\sigma \cdot \tau) \cdot x} = \size{x} \leq \size{\tau \cdot x} \leq \size{\sigma \cdot (\tau \cdot x)}$, hence $\size{\tau \cdot x} = \size{x}$ and $\size{\sigma \cdot (\tau \cdot x)} = \size{\tau \cdot x}$, which entails that $\tau \in \mathfrak{R}(x)$ and $\sigma \in \mathfrak{R}(\tau \cdot x)$.
\end{proof}

\begin{fact}\label{fact: injectivity for depth 0}
Let $S$ and $S'$ be two cut-free differential PS's of depth $0$. Let $e$ and $e'$ be two atomic experiments on $S$ and $S'$ respectively such that $\restriction{\ports{e}}{\conclusions{S}} = \restriction{\ports{e}}{\conclusions{S'}}$ is $\mathcal{A}$-injective. Then $S \equiv S'$.
\end{fact}

\begin{proof}
By induction on $\Card{\portsatzero{S}}$.
\end{proof}

For any PS $R$, any $\sm{R}$-atomic injective point is the result of some atomic experiment on $R$:

\begin{fact}\label{fact: from point to experiment}
Let $R$ be a cut-free in-PS. Let $x \in {\sm{R}}_{\textit{At}}$. We assume that the set $\mathcal{A}$ is infinite or $x$ is $\mathcal{A}$-injective. Then there exists an atomic experiment $e$ on $R$ such that $\restriction{e}{\conclusions{R}} = x$.
\end{fact}

\begin{proof}
We prove, by induction on $(\depthof{R}, \Card{\portsatzero{R}})$ lexicographically ordered, that, for any non-atomic experiment $e'$ on $R$, there exist an experiment $e$ on $R$, a function $\sigma : {\mathcal{A}} \to D_{\mathcal{A}}$ such that $\sigma \cdot e = e'$ and $\gamma \in \atoms{\restriction{e}{\conclusions{R}}}$ such that $\sigma(\gamma) \notin \{ +, - \} \times {\mathcal{A}}$.
\end{proof}

The converse does not necessarily hold (but see Lemma~\ref{lem: results of typed experiments are atomic} about \emph{typable} cut-free PS's): for some cut-free PS's $R$, there are results of atomic injective experiments on $R$ that are not $\sm{R}$-atomic. Indeed, consider Figure~\ref{fig: atomic experiments}. 
\begin{figure}[t]
\centering
\resizebox{\textwidth}{!}{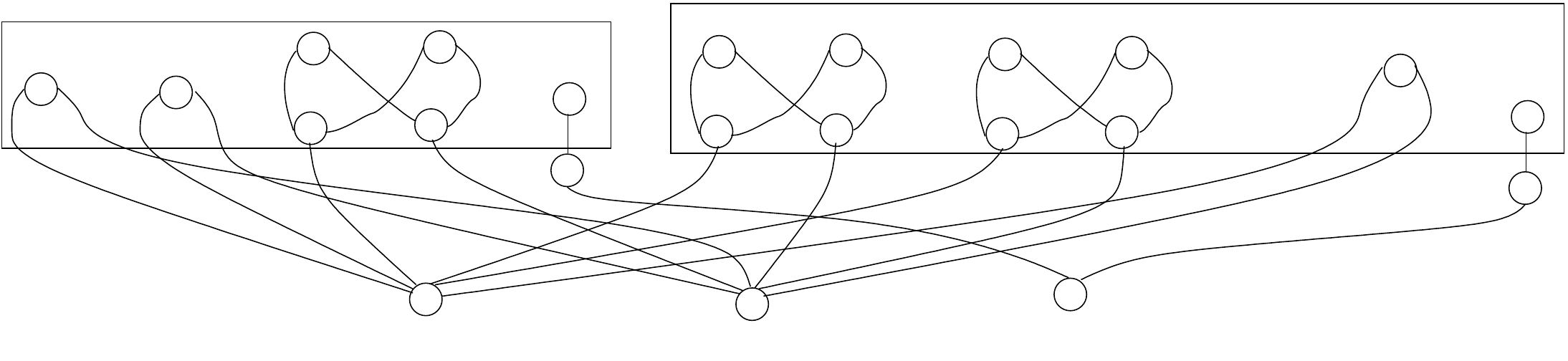}
\caption{PS $R''$}
\label{fig: atomic experiments}
\end{figure}
There exists an atomic injective experiment $e$ on $R''$ such that 
\begin{itemize}
\item $\ports{e}(p_1) = (-, [(+, \gamma_1), \ldots, (+, \gamma_7), (+, (+, \gamma_{8}), (+, \gamma_{9})), \ldots, (+, (+, \gamma_{22}), (+, \gamma_{23}))])$, 
\item $\ports{e}(p_2) = (-, [(-, \gamma_1), \ldots, (-, \gamma_7), (-, (-, \gamma_{8}), (-, \gamma_{9})), \ldots, (-, (-, \gamma_{22}), (-, \gamma_{23}))])$
\item and $\ports{e}(p_3) = (-, [(+, [(+, \ast), (+, \ast)]), (+, [(+, \ast), (+, \ast), (+, \ast)]) ])$,
\end{itemize}
where $\{ \gamma_1, \ldots, \gamma_{23} \} \subseteq {\mathcal{A}}$. But $\restriction{e}{\{ p_1, p_2, p_3 \}}$ is not in ${(\sm{R''})}_{\textit{At}}$: there exists an atomic injective experiment $e'$ on $R'$ such that 
\begin{itemize}
\item $\ports{e'}(p_1) = (-, [(+, \gamma_1), \ldots, (+, \gamma_8), (+, (+, \gamma_{10}), (+, \gamma_{11})), \ldots, (+, (+, \gamma_{22}), (+, \gamma_{23}))])$, 
\item $\ports{e'}(p_2) = (-, [(-, \gamma_1), \ldots, (-, \gamma_8), (-, (-, \gamma_{10}), (-, \gamma_{11})), \ldots, (-, (-, \gamma_{22}), (-, \gamma_{23}))])$ 
\item and $\ports{e'}(p_3) = (-, [(+, [(+, \ast), (+, \ast)]), (+, [(+, \ast), (+, \ast), (+, \ast)]) ])$;
\end{itemize}
we set $\sigma(\gamma) = \left\lbrace \begin{array}{ll} (+, \gamma) & \textit{if $\gamma \in {\mathcal{A}} \setminus \{ \gamma_8 \}$;}\\ (+, (+, \gamma_8), (+, \gamma_9)) & \textit{if $\gamma = \gamma_8$;} \end{array} \right.$ - we have $\sigma \cdot \restriction{e'}{\{ p_1, p_2, p_3 \}} = \restriction{e}{\{ p_1, p_2, p_3 \}}$.

But it does not matter, because yet atomic points are enough to generate all the points:

\begin{fact}\label{fact: atomic are enough}
Let $R$ be an in-PS. For any $y \in \sm{R}$, there exist $x \in {\sm{R}}_\textit{At}$ and $\sigma : {\mathcal{A}} \to D_{\mathcal{A}}$ such that $\sigma \cdot x = y$.
\end{fact}

\begin{proof}
By induction on $\size{\sum_{p \in \conclusions{R}} [y(p)]}$: if $y \in {\sm{R}}_\textit{At}$, then we can set $x = y$ and $\sigma = id_{\mathcal{A}}$; if $y \notin {\sm{R}}_{\textit{At}}$, then there exist a function $\sigma' : {\mathcal{A}} \to D_{\mathcal{A}}$, $y' \in \sm{R}$ such that $\sigma' \cdot y' = y$ and $\gamma \in \atoms{y'}$ such that $\sigma'(\gamma) \notin {\mathcal{A}}$, hence $\size{\sum_{p \in \conclusions{R}} [y'(p)]} < \size{\sum_{p \in \conclusions{R}} [y(p)]}$. By induction hypothesis, there exist $x \in {\sm{R}}_\textit{At}$ and $\sigma'' : {\mathcal{A}} \to D_{\mathcal{A}}$ such that $\sigma'' \cdot x = y'$. We set $\sigma = \sigma' \cdot \sigma''$: we have $\sigma \cdot x = (\sigma' \cdot \sigma'') \cdot x = \sigma' \cdot (\sigma'' \cdot x) = \sigma' \cdot y' = y$.
\end{proof}

\begin{lem}\label{lem: atomic experiment induces atomic experiment on Taylor expansion}
Let $R$ be an in-PS. Let $e$ be an (resp. atomic) experiment on $R$. Then there exists an (resp. atomic) experiment $\mathcal{T}_R(e)$ on $\termofTaylor{R}{\overline{e}}{0}$ such that $\restriction{\ports{\mathcal{T}_R(e)}}{\conclusions{\termofTaylor{R}{\overline{e}}{0}}} = \restriction{\ports{e}}{\conclusions{R}}$.
\end{lem}

\begin{proof}
We prove, by induction on $\depthof{R}$, that, for any in-PS $R$, for any experiment $e$ on $R$, there exists an experiment $\mathcal{T}_R(e)$ on $\termofTaylor{R}{\overline{e}}{0}$ such that  
\begin{itemize}
\item for any $p \in \portsatzero{R}$, we have $\ports{e}(p) = \ports{\mathcal{T}_R(e)}(p)$;
\item and, for any $p \in \portsatdepthgreater{R}{0}$, we have $\ports{e}(p) = \sum_{\substack{q \in \ports{\termofTaylor{R}{\overline{e}}{0}}\\ \isaCopyfrom{R}{e}{0}(q) = p}} \multi{\ports{\mathcal{T}_R(e)}(q)}$.
\end{itemize}
One can take for $\mathcal{T}_R(e)$ the following experiment on $\termofTaylor{R}{\overline{e}}{0}$:
\begin{itemize}
\item for any $p \in \portsatzero{R}$, $\ports{\mathcal{T}_R(e)}(p) = \ports{e}(p)$;
\item for any $o \in \boxesatzero{R}$, for any $e_o \in \supp{\boxes{e}(o)}$, for any $ p \in \portsatzero{\termofTaylor{B_R(o)}{\overline{e_o}}{0}}$,
  \begin{align*}
   \ports{\mathcal{T}_R(e)}((o, (e_o, p)) = \ports{\mathcal{T}_{B_R(o)}(e_o)}(p).
\tag*{\qedhere}
  \end{align*}
\end{itemize}
\end{proof}

\begin{defi}
Let $k > 1$. An experiment $e$ on some in-PS $R$ is said to be \emph{$k$-heterogeneous} if the pseudo-experiment $\overline{e}$ is $k$-heterogeneous.
\end{defi}

\begin{defi}\label{definition: k-heterogeneous point}
Let $k > 1$. 
For any function $x: \mathcal{P} \to D_{\mathcal{A}}$, where $\mathcal{P}$ is any set, we say that $x$ is \emph{$k$-heterogeneous} if the following properties hold:
\begin{itemize}
\item for any multiset $a$, the pair $(+, a)$ occurs at most once in $\im{x}$ and, if it occurs, then there exists $j > 0$ such that the cardinality of $a$ is $k^j$;
\item for any multisets $a_1$ and $a_2$ having the same cardinality such that the pairs $(+, a_1)$ and $(+, a_2)$ occur in $\im{x}$, we have $a_1 = a_2$.
\end{itemize}
\end{defi}

\begin{lem}\label{lem: characterization of k-heterogeneous experiments}
Let $R$ be a cut-free in-PS. Let $k > 1$. For any experiment $e$ on $R$, if $\restriction{\ports{e}}{\conclusions{R}}$ is $k$-heterogeneous, then $e$ is $k$-heterogeneous. Conversely, if $e$ is an atomic $k$-heterogeneous experiment on $R$, then $\restriction{\ports{e}}{\conclusions{R}}$ is $k$-heterogeneous.
\end{lem}

\begin{proof}
Let $e$ be an experiment on $R$ such that $\restriction{\ports{e}}{\conclusions{R}}$ is $k$-heterogeneous. By Lemma~\ref{lem: atomic experiment induces atomic experiment on Taylor expansion}, there exists an experiment $\mathcal{T}_R(e)$ on $\termofTaylor{R}{\overline{e}}{0}$ such that $\restriction{\ports{\mathcal{T}_R(e)}}{\conclusions{\termofTaylor{R}{\overline{e}}{0}}} = $ $\restriction{\ports{e}}{\conclusions{R}}$. Since $\termofTaylor{R}{\overline{e}}{0}$ is cut-free, we have:
\begin{enumerate}
\item $\arity{\termofTaylor{R}{\overline{e}}{0}}[\portsatzerooftype{\cod}{\termofTaylor{R}{\overline{e}}{0}}] \subseteq \{ k^j ; j > 0 \}$
\item $(\forall p_1, p_2 \in \portsatzerooftype{\cod}{\termofTaylor{R}{\overline{e}}{0}}) (\arity{\termofTaylor{R}{\overline{e}}{0}}(p_1) = \arity{\termofTaylor{R}{\overline{e}}{0}}(p_2) \Rightarrow p_1 = p_2)$
\end{enumerate}
By Corollary~\ref{cor: characterization of k-heterogeneous experiments}, the pseudo-experiment $\overline{e}$ is $k$-heterogeneous.

Conversely, let $e$ be an atomic $k$-heterogeneous experiment on $R$. Then, by Lemma~\ref{lem: atomic experiment induces atomic experiment on Taylor expansion}, there exists an atomic experiment $\mathcal{T}_R(e)$ of $\termofTaylor{R}{\overline{e}}{0}$ such that $\restriction{\ports{\mathcal{T}_R(e)}}{\conclusions{\termofTaylor{R}{\overline{e}}{0}}} = \restriction{\ports{e}}{\conclusions{R}}$. By Corollary~\ref{cor: characterization of k-heterogeneous experiments}, we have:
\begin{enumerate}
\item $\arity{\termofTaylor{R}{\overline{e}}{0}}[\portsatzerooftype{\cod}{\termofTaylor{R}{\overline{e}}{0}}] \subseteq \{ k^j ; j > 0 \}$
\item $(\forall p_1, p_2 \in \portsatzerooftype{\cod}{\termofTaylor{R}{\overline{e}}{0}}) (\arity{\termofTaylor{R}{\overline{e}}{0}}(p_1) = \arity{\termofTaylor{R}{\overline{e}}{0}}(p_2) \Rightarrow p_1 = p_2)$
\end{enumerate}
Since $\mathcal{T}_R(e)$ is atomic, $\restriction{\ports{\mathcal{T}_R(e)}}{\conclusions{\termofTaylor{R}{\overline{e}}{0}}}$ is $k$-heterogeneous.
\end{proof}

\begin{lem}\label{lem: Taylor expansion}
Let $R$ and $R'$ be two cut-free PS's such that $\conclusions{R} = \conclusions{R'}$. Let $e$ be an injective atomic experiment on $R$ and let $e'$ be an injective atomic experiment on $R'$ such that $\restriction{\ports{e}}{\conclusions{R}} = \restriction{\ports{e'}}{\conclusions{R'}}$. Then $\termofTaylor{R}{\overline{e}}{0} \equiv \termofTaylor{R'}{\overline{e'}}{0}$.
\end{lem}

\begin{proof}
By Lemma~\ref{lem: atomic experiment induces atomic experiment on Taylor expansion}, there exists an atomic experiment $\mathcal{T}_R(e)$ on $\termofTaylor{R}{\overline{e}}{0}$ such that $\restriction{\ports{\mathcal{T}_R(e)}}{\conclusions{\termofTaylor{R}{\overline{e}}{0}}} = \restriction{\ports{e}}{\conclusions{R}}$ and an atomic experiment $\mathcal{T}_{R'}(e')$ on $\termofTaylor{R'}{\overline{e'}}{0}$ such that $\restriction{\ports{\mathcal{T}_{R'}(e')}}{\conclusions{\termofTaylor{R'}{\overline{e'}}{0}}} = \restriction{\ports{e'}}{\conclusions{R'}}$. We apply Fact~\ref{fact: injectivity for depth 0}.
\end{proof}

\begin{thm}\label{thm: injectivity in the untyped framework}
Let $R$ be a cut-free PS. If the set $\mathcal{A}$ is infinite, then there exists an $\mathcal{A}$-injective subset $D_0$ of ${\sm{R}}_{\textit{At}}$ with $\Card{D_0} = 2$ such that, for any cut-free PS $R'$ with $\conclusions{R} = \conclusions{R'}$, one has $(D_0 \subseteq {\sm{R'}}_{\textit{At}} \Rightarrow R \equiv R')$.
\end{thm}

\begin{proof}
Let $f$ be an injective $1$-experiment on $R$ (its existence is ensured by the assumption that the set $\mathcal{A}$ is infinite). By Fact~\ref{fact: atomic are enough}, there exists $x_1 \in {\sm{R}}_{\textit{At}}$ and $\sigma_1: {\mathcal{A}} \to D_{\mathcal{A}}$ such that $\sigma_1 \cdot x_1 = \restriction{f}{\conclusions{R}}$; the point $x_1$ is $\mathcal{A}$-injective too. Let $k_1 \in \Nat$ be the greatest cardinality of the negative multisets that occur in $\im{x_1}$. Let $k_2 \in \Nat$ be the number of occurrences of positive multisets in $\im{x_1}$. Let $k > \max{\{ k_1, k_2 \}}$ and let $e$ be an injective $k$-heterogeneous experiment on $R$ (its existence is ensured by the assumption that the set $\mathcal{A}$ is infinite). By Fact~\ref{fact: atomic are enough}, there exist $x \in {\sm{R}}_{\textit{At}}$ and $\sigma: {\mathcal{A}} \to D_{\mathcal{A}}$ such that $\sigma \cdot x = \restriction{e}{\conclusions{R}}$. We can take $D_0 = \{ x_1, x \}$. Indeed: 

Since the point $\restriction{e}{\conclusions{R}}$ is $\mathcal{A}$-injective, the point $x$ is $\mathcal{A}$-injective too. By Lemma~\ref{lem: characterization of k-heterogeneous experiments}, $\restriction{e}{\conclusions{R}}$ is $k$-heterogeneous, hence $x$ too. Let $R'$ be a cut-free PS such that $\conclusions{R} = \conclusions{R'}$ and $D_0 \subseteq {\sm{R'}}_{\textit{At}}$. By Fact~\ref{fact: from point to experiment}, there exist an atomic experiment $f_0$ (resp. $f_0'$) on $R$ (resp. of $R'$) such that $\restriction{\ports{f_0}}{\conclusions{R}} = x_1$ (resp. $\restriction{\ports{f_0'}}{\conclusions{R'}} = x_1$) and an atomic experiment $e_0$ (resp. $e_0'$) on $R$ (resp. of $R'$) such that $\restriction{\ports{e_0}}{\conclusions{R}} = x$ (resp. $\restriction{\ports{e_0'}}{\conclusions{R'}} = x$). By Lemma~\ref{lem: Taylor expansion}, we have $\termofTaylor{R}{\overline{e_0}}{0} \equiv \termofTaylor{R'}{\overline{e_0'}}{0}$. 

Moreover, the experiments $f_0$ and $f_0'$ are $1$-experiments on $R$ and $R'$ respectively (all the positive multisets of $x_1$ have cardinality $1$). We have:
\begin{itemize}
\item $\cosize{R} = k_1 = \cosize{R'}$;
\item $\Card{\boxes{R}} = k_2 = \Card{\boxes{R'}}$;
\item and $\numberInvisibleComponents{R} = 0 = \numberInvisibleComponents{R}$ (because $R$ and $R'$ are cut-free).
\end{itemize}
We thus have $k \geq \max{\{ \basis{R}, \basis{R'} \}}$. Finally, by Lemma~\ref{lem: characterization of k-heterogeneous experiments}, the experiments $e_0$ and $e_0'$ are $k$-heterogeneous. We can then apply Proposition~\ref{prop: from i to i+1} to obtain $R \equiv R'$. 
\end{proof}

\begin{cor}\label{cor: injectivity in the untyped framework}
Let $R$ and $R'$ be two cut-free PS's such that $\conclusions{R} = \conclusions{R'}$. If the set $\mathcal{A}$ is infinite, then one has $(R \equiv R' \Leftrightarrow {\sm{R}}_{\textit{At}} = {\sm{R'}}_{\textit{At}})$.
\end{cor}

\subsection{Typed framework}\label{subsection: Typed}

We want to apply Theorem~\ref{thm: injectivity in the untyped framework} to obtain a similar result for typed PS's (Theorem~\ref{theorem: typed framework}). For that, we need to relate $\sm{(R, \mathsf{T})}$ to ${\sm{R}}_{\textit{At}}$; it is the role of Lemma~\ref{lem: relating typed semantics and untyped semantics}.

\begin{defi}\label{defin: experiment}
We assume that we are given a set $\sm{X}$ for each $X \in \propvar$. Then, for any $C \in \mathbb{T}$, we define, by induction on $C$, the set $\sm{C}$ as follows: $\sm{1} = \{ \ast \} = \sm{\bot}$; $\sm{(C_1 \tens C_2)} = \sm{C_1} \times \sm{C_2} = \sm{(C_1 \parr C_2)}$; $\sm{\cod C} = \finitemultisets{\sm{C}} = \sm{\contr C}$.

Let $(R, \mathsf{T})$ be a typed differential in-PS. We define, by induction on $\textit{depth}(R)$, the set of \emph{experiments on $(R, \mathsf{T})$}: it is the set of pairs $e = (\ports{e}, \boxes{e})$, where 
\begin{itemize}
\item $\ports{e}$ is a function that associates with every $p \in \portsatzero{R}$ an element of $\sm{\mathsf{T}(p)}$ and with every $p \in \portsatdepthgreater{R}{0}$ an element of $\finitemultisets{\sm{\mathsf{T}(p)}}$,
\item and $\boxes{e}$ is a function which associates with every $o \in \boxesatzero{R}$ a finite multiset of experiments on $(B_R(o), \restriction{\mathsf{T}}{\ports{B_R(o)}})$ 
\end{itemize}
such that
\begin{itemize}
\item for any $p \in \multiplicativeportsatzero{R}$, for any $w_1, w_2 \in \wiresatzero{R}$ such that $\target{\groundof{R}}(w_1) = p = \target{\groundof{R}}(w_2)$, $w_1 \in \leftwires{\groundof{R}}$ and $w_2 \notin \leftwires{\groundof{R}}$, we have $\ports{e}(p) = (\ports{e}(w_1), \ports{e}(w_2))$;
\item for any $\{ p, p' \} \in \axiomsatzero{R} \cup \cutsatzero{R}$, we have $\ports{e}(p) = \ports{e}(p')$;
\item for any $p \in \exponentialportsatzero{R}$,
we have $\ports{e}(p) = \sum_{\substack{w \in \wiresatzero{R}\\ \target{\groundof{R}}(w) = p}} [\ports{e}(w)] +  \sum_{o \in \boxesatzero{R}} \sum_{e_o \in \supp{\boxes{e}(o)}}$ $\sum_{\substack{q \in \conclusions{B_R(o)}\\ \target{R}(o, q) = p}} \boxes{e}(o)(e_o) \cdot \ports{e_o}(q)$;
\item for any $o \in \boxesatzero{R}$, for any $p \in \portsatzero{B_R(o)}$, we have $\ports{e}(p) = \sum_{e_o \in \supp{\boxes{e}(o)}} \boxes{e}(o)(e_o) \cdot \multi{\ports{e_o}(p)}$;
\item for any $o \in \boxesatzero{R}$, for any $p \in \portsatdepthgreater{B_R(o)}{0}$, we have $\ports{e}(p) = \sum_{e_o \in \supp{\boxes{e}(o)}} \boxes{e}(o)(e_o) \cdot \ports{e_o}(p)$.
\end{itemize}
For any experiment $e = ((R, \mathsf{T}), \ports{e}, \boxes{e})$, we set $\ports{e} = \ports{e}$ and $\boxes{e} = \boxes{e}$. We set $\sm{(R, \mathsf{T})} =$ $\{ \restriction{\ports{e}}{\conclusions{R}} ;$ $e \textit{ is an experiment on } (R, \mathsf{T}) \}$.
\end{defi}

From now on, we assume that, for any $X \in \propvar$, the set $\sm{X}$ does not contain any couple nor any $3$-tuple and $\ast \notin \mathcal{A}$ and we assume that $\mathcal{A} = \bigcup_{X \in \propvar} \sm{X}$. We define, by induction on $n$, the sets ${\overline{D}}_{\mathcal{A}, n}$ for any $n \in \Nat$:
\begin{itemize}
\item $\overline{D}_{\mathcal{A}, 0} = \mathcal{A} \cup \{ \ast \}$
\item $\overline{D}_{\mathcal{A}, n+1} = \overline{D}_{\mathcal{A}, 0} \cup (\overline{D}_{\mathcal{A}, n} \times \overline{D}_{\mathcal{A}, n}) \cup \finitemultisets{\overline{D}_{\mathcal{A}, n}}$
\end{itemize}
We set $\overline{D}_\mathcal{A} = \bigcup_{n \in \Nat} \overline{D}_{\mathcal{A}, n}$. We define the function $U: D_{\mathcal{A}} \to \overline{D}_{\mathcal{A}}$ as follows:
\begin{itemize}
\item if $\alpha = (\delta, \gamma)$ with $\delta \in \{ +, - \}$ and $\gamma \in \mathcal{A} \cup \{ \ast \}$, then $U(\alpha) = \gamma$;
\item if $\alpha = (\delta, \alpha_1, \alpha_2)$ with $\delta \in \{ +, - \}$ and $\alpha_1, \alpha_2 \in D$, then $U(\delta, \alpha_1, \alpha_2) = (U(\alpha_1), U(\alpha_2))$;
\item if $\alpha = (\delta, \alpha_0)$ with $\delta \in \{ +, - \}$ and $\alpha_0 \in D$, then $U(\delta, \alpha_0) = U(\alpha_0)$.
\end{itemize}

\begin{defi}
Let $\sigma : {\mathcal{A}} \to \mathcal{A}$. For any $\alpha \in \overline{D}_{\mathcal{A}}$, we define $\sigma \cdot \alpha \in \overline{D}_{\mathcal{A}}$ as follows:
\begin{itemize}
\item if $\alpha \in {\mathcal{A}}$, then $\sigma \cdot \alpha = \sigma(\alpha)$;
\item if $\alpha = \ast$, then $\sigma \cdot \alpha = \alpha$;
\item if $\alpha_1, \alpha_2 \in \overline{D}_{\mathcal{A}}$, then $\sigma \cdot (\alpha_1, \alpha_2) = (\sigma \cdot \alpha_1, \sigma \cdot \alpha_2)$;
\item if $\alpha_1, \ldots, \alpha_m \in \overline{D}_{\mathcal{A}}$, then $\sigma \cdot [\alpha_1, \ldots, \alpha_m] = [\sigma \cdot \alpha_1, \ldots, \sigma \cdot \alpha_m]$.
\end{itemize}
For any set $\mathcal{P}$, for any function $x : \mathcal{P} \to \overline{D}_{\mathcal{A}}$, we define a function $\sigma \cdot x : \mathcal{P} \to \overline{D}_{\mathcal{A}}$ by setting: $(\sigma \cdot x)(p) = \sigma \cdot x(p)$ for any $p \in \mathcal{P}$.
\end{defi}

For any function $\sigma: \mathcal{A} \to D_{\mathcal{A}}$, we define the function $U(\sigma): \mathcal{A} \to \mathcal{A}$ as follows: $$U(\sigma)(\gamma) = \left\lbrace \begin{array}{ll} U(\sigma(\gamma)) & \textit{if $\sigma(\gamma) \in \{ +, - \} \times \mathcal{A}$;} \\ \gamma & \textit{otherwise.} \end{array} \right.$$

\begin{fact}\label{fact: substitutions compatible with U}
For any $\alpha \in D_{\mathcal{A}}$, for any $\sigma \in \mathfrak{R}(\alpha)$, we have $U(\sigma \cdot \alpha) = U(\sigma) \cdot U(\alpha)$.
\end{fact}

\begin{fact}\label{fact: experiments on typable PS's}
Let $(R, \mathsf{T})$ be a typed in-PS. Then, for any experiment $e$ on $R$, for any $p \in \portsatzero{R}$, we have $\height{e(p)} \geq \height{\mathsf{T}(p)}$.
\end{fact}

\begin{defi}
Any $\alpha \in D_{\mathcal{A}}$ is said  to be \emph{uniform} if, for any occurrence of any multiset $a$ that occurs in $\alpha$, for any $\beta, \beta' \in \supp{a}$, we have $\height{\beta} = \height{\beta'}$.

For any finite set $\mathcal{P}$, any function $x: \mathcal{P} \to D_{\mathcal{A}}$ is said to be \emph{uniform} if, for any $p \in \mathcal{P}$, $x(p)$ is uniform.
\end{defi}

\begin{fact}\label{fact: experiments on typed PS's provide uniform points}
Let $(R, \mathsf{T})$ be a typed in-PS. Then, for any atomic experiment $e$ on $R$, for any $p \in \portsatzero{R}$, $e(p)$ is uniform and we have $\height{e(p)} = \height{\mathsf{T}(p)}$.
\end{fact}

\begin{fact}\label{fact: substitutions providing uniform points}
Let $\alpha, \alpha' \in D_{\mathcal{A}}$ and $\sigma : \mathcal{A} \to D_{\mathcal{A}}$ such that $\alpha'$ is uniform and $\sigma \cdot \alpha' = \alpha$. We have $\height{\alpha} > \height{\alpha'}$ if, and only if, $\sigma \notin \mathfrak{R}(\alpha')$.
\end{fact}

\begin{proof}
By induction on $\height{\alpha}$.
\end{proof}

\begin{lem}\label{lem: results of typed experiments are atomic}
We assume that, for any $X \in \propvar$, the set $\sm{X}$ is infinite. Let $(R, \mathsf{T})$ be a cut-free typed PS. Then, for any atomic experiment $e$ on $R$, we have $\restriction{e}{\conclusions{R}} \in {\sm{R}}_{\textit{At}}$.
\end{lem}

\begin{proof}
Let $e$ be an atomic experiment on $R$. Let $x' \in \sm{R}$ and let $\sigma : \mathcal{A} \to D_{\mathcal{A}}$ such that $\sigma \cdot x' = \restriction{e}{\conclusions{R}}$. By Fact~\ref{fact: atomic are enough}, there exist $x \in \sm{R}_{\textit{At}}$ and $\tau: \mathcal{A} \to D_{\mathcal{A}}$ such that $\tau \cdot x = x'$. By Fact~\ref{fact: from point to experiment} and Fact~\ref{fact: experiments on typed PS's provide uniform points}, the point $x$ is uniform and $(\forall p \in \conclusions{R}) \height{x(p)} = \height{\mathsf{T}(p)} = e(p)$. 
By Lemma~\ref{lem: substitutions are actions}, we have $(\sigma \cdot \tau) \cdot x = \sigma \cdot (\tau \cdot x) = \sigma \cdot x' = \restriction{e}{\conclusions{R}}$. By Fact~\ref{fact: substitutions providing uniform points}, we have $\sigma \cdot \tau \in \mathfrak{R}(x)$, hence, by Fact~\ref{fact: st renaming => t renaming}, $\sigma \in \mathfrak{R}(\tau \cdot x) = \mathfrak{R}(x')$, which entails $\restriction{e}{\conclusions{R}} \in \sm{R}_{\textit{At}}$.
\end{proof}

\begin{lem}\label{lem: relating typed semantics and untyped semantics}
Let us assume that, for any $X \in \propvar$, the set $\sm{X}$ is infinite. Let $(R, \mathsf{T})$ be a cut-free typed PS. Then $\{ U \circ x ; x \in {\sm{R}}_{\textit{At}} \} = \bigcup_{x \in \sm{(R, \mathsf{T})}} \{ \sigma \cdot x ; \sigma \in \mathcal{A}^{\mathcal{A}} \}$.
\end{lem}

\begin{proof}
By Lemma~\ref{lem: results of typed experiments are atomic}, we have $\sm{(R, \mathsf{T})} \subseteq \{ U \circ x ; x \in {\sm{R}}_{\textit{At}} \}$. Let $x \in \sm{(R, \mathsf{T})}$. 
Now, for any $x' \in \sm{R}_{\textit{At}}$, for any $\sigma \in \mathfrak{R}(x')$, we have $\sigma \cdot x' \in \sm{R}_{\textit{At}}$; indeed, let $x \in \sm{R}$ and $\tau : \mathcal{A} \to D_{\mathcal{A}}$ such that $\tau \cdot x = \sigma \cdot x'$; by Fact~\ref{fact: from point to experiment} and Fact~\ref{fact: experiments on typed PS's provide uniform points}, we have $\height{x} = \height{x'}$, hence, by Fact~\ref{fact: substitutions providing uniform points}, $\tau \in \mathfrak{R}(x)$. So, by Fact~\ref{fact: substitutions compatible with U}, we have $\bigcup_{x \in \sm{(R, \mathsf{T})}} \{ \sigma \cdot x ; \sigma \in \mathcal{A}^{\mathcal{A}} \} \subseteq \{ U(\sigma) \cdot (U \circ x) ; (x \in {\sm{R}}_{\textit{At}} \wedge \sigma \in \mathcal{A}^{\mathcal{A}}) \} = \{ U \circ (\sigma \cdot x) ; (x \in {\sm{R}}_{\textit{At}} \wedge \sigma \in \mathfrak{R}(x)) \} = \{ U \circ x ; x \in {\sm{R}}_{\textit{At}} \}$.

Conversely, let $x \in {\sm{R}}_{\textit{At}}$. If $\propvar = \emptyset$, then $U \circ x \in \sm{(R, \mathsf{T})}$. Otherwise, the set $\mathcal{A}$ is infinite, hence there exist an injective atomic experiment $e$ on $R$ and $\sigma \in \mathfrak{R}(x)$ such that $x = \sigma \cdot \restriction{e}{\conclusions{R}}$. Let $\tau: \mathcal{A} \to D_{\mathcal{A}}$ such that:
\begin{itemize}
\item for any $p \in \portsatzerooftype{\textit{ax}}{R}$, for any $\gamma \in \mathcal{A}$, for any $\delta \in \{ +, - \}$ such that $e(p) = (\delta, \gamma)$, we have $U(\tau)(\gamma) \in \sm{\mathsf{T}(p)}$;
\item and, for any $p \in \portsoftype{\textit{ax}}{R} \cap \portsatdepthgreater{R}{0}$, for any $\gamma \in \mathcal{A}$, for any $\delta \in \{ +, - \}$ such that $(\delta, \gamma) \in \supp{e(p)}$, we have $U(\tau)(\gamma) \in \sm{\mathsf{T}(p)}$.
\end{itemize}
We have $U(\tau) \cdot (U \circ \restriction{e}{\conclusions{R}}) \in \sm{(R, \mathsf{T})}$. It is clear that there exists $\sigma' \in \mathfrak{R}(\tau \cdot \restriction{e}{\conclusions{R}})$ such that $\sigma' \cdot (\tau \cdot \restriction{e}{\conclusions{R}}) = \sigma \cdot \restriction{e}{\conclusions{R}}$. We have $U \circ x = U \circ (\sigma \cdot \restriction{e}{\conclusions{R}}) = U \circ (\sigma' \cdot (\tau \cdot \restriction{e}{\conclusions{R}})) = U(\sigma') \cdot(U \circ (\tau \cdot \restriction{e}{\conclusions{R}})) = U(\sigma') \cdot (U(\tau) \cdot (U \circ \restriction{e}{\conclusions{R}}))$ (by applying Fact~\ref{fact: substitutions compatible with U} twice).
\end{proof}

\begin{defi}
Let $\mathcal{D} \subseteq {(D_{\mathcal{A}})}^{\mathcal{P}}$ for some set $\mathcal{P}$. We say that $\mathcal{D}$ \emph{reflects renamings} if the following property holds: $(\forall x \in \mathcal{D}) (\forall y \in D_{\mathcal{A}}) (\forall \sigma \in \mathfrak{R}(y)) (\sigma \cdot y = x \Rightarrow y \in \mathcal{D})$
\end{defi}

\begin{defi}
For any finite set $\mathcal{P}$, any function $x: \mathcal{P} \to D_{\mathcal{A}}$ is said to be \emph{balanced} if, for any $\gamma \in \mathcal{A}$, there are as many occurrences of $(+, \gamma)$ in $\sum_{p \in \mathcal{P}} \multi{x(p)}$ as occurrences of $(-, \gamma)$ in $\sum_{p \in \mathcal{P}} \multi{x(p)}$.
\end{defi}

\begin{fact}\label{fact: balanced}
Let $S$ be a differential in-PS. Then, for any $x \in \sm{S}$, $x$ is balanced.
\end{fact}

\begin{proof}
By induction on $(\depthof{S}, \Card{\boxes{S}}, \Card{\portsatzero{S}}, \Card{\cutsatzero{S}})$ lexicographically ordered.
\end{proof}

\begin{fact}\label{fact: up to renamings}
Let $\mathcal{P}$ be some finite set and let $\mathsf{T}$ be a function $\mathcal{P} \to \mathbb{T}$. Let $x_1, x_2 \in {(D_{\mathcal{A}})}^{\mathcal{P}}$ $\mathcal{A}$-injective and balanced such that $U \circ x_1 = U \circ x_2$ and, for any $p \in \mathcal{P}$, we have $U(x_1(p)) \in \sm{\mathsf{T}(p)}$. Then there exists $\sigma \in \mathfrak{R}(x_1)$ such that $\sigma \cdot x_1 = x_2$.
\end{fact}

\begin{proof}
By induction on $\sum_{p \in \mathcal{P}} \size{\mathsf{T}(p)}$.
\end{proof}

\begin{lem}\label{lem: renamings reflect atomicity}
Let $\mathcal{P}$ be some finite set. 
Let $\mathcal{D} \subseteq {(D_{\mathcal{A}})}^{\mathcal{P}}$. Let $x \in \mathcal{D}$ and let $\sigma \in \mathfrak{R}(x)$ such that $\sigma \cdot x \in {\mathcal{D}}_{\textit{At}}$. Then $x \in {\mathcal{D}}_{\textit{At}}$.
\end{lem}

\begin{proof}
Let $y \in \mathcal{D}$ and let $\tau: \mathcal{A} \to D_{\mathcal{A}}$ such that $\tau \cdot y = x$. Then, by Lemma~\ref{lem: substitutions are actions}, we have $(\sigma \cdot \tau) \cdot y = \sigma \cdot (\tau \cdot y) = \sigma \cdot x$. We have $\sigma \cdot \tau \in \mathfrak{R}(y)$, hence, by Fact~\ref{fact: st renaming => t renaming}, $\tau \in \mathfrak{R}(y)$. 
\end{proof}

Theorem~\ref{theorem: typed framework} is similar to Theorem~\ref{thm: injectivity in the untyped framework}:

\begin{thm}\label{theorem: typed framework}
Let $(R, \mathsf{T})$ be a cut-free typed PS. If, for any $X \in \propvar$, the set $\sm{X}$ is infinite, then there exists $D_0 \subseteq \sm{(R, \mathsf{T})}$ with $\Card{D_0} = 2$ such that, for any cut-free typed PS $(R', \mathsf{T'})$ with $\conclusions{R} = \conclusions{R'}$ and $\restriction{\mathsf{T}}{\conclusions{R}} = \restriction{\mathsf{T'}}{\conclusions{R'}}$, we have $(D_0 \subseteq \sm{(R', \mathsf{T'})} \Rightarrow (R, \mathsf{T}) \equiv (R', \mathsf{T'}))$.
\end{thm}

\begin{proof}
By Theorem~\ref{thm: injectivity in the untyped framework}, there exists an $\mathcal{A}$-injective subset $D'_0 = \{ y, z \}$ of $\sm{R}_{\textit{At}}$ with $\Card{D_0} = 2$ such that, for any cut-free PS $R'$, we have $(D'_0 \subseteq \sm{R'}_{\textit{At}} \Rightarrow R \equiv R')$. We can take $D_0 = \{ U \circ x ; x \in D'_0 \}$. Indeed:

Let $(R', \mathsf{T'})$ be a typed PS such that $\conclusions{R} = \conclusions{R'}$, $\restriction{\mathsf{T}}{\conclusions{R}} = \restriction{\mathsf{T'}}{\conclusions{R'}}$ and $D_0 \subseteq \sm{(R', \mathsf{T'})}$. 
By Lemma~\ref{lem: relating typed semantics and untyped semantics}, there exist $y', z' \in {\sm{R'}}_{\textit{At}}$ such that $U \circ y = U \circ y'$ and $U \circ z = U \circ z'$. By Fact~\ref{fact: balanced}, we can apply Fact~\ref{fact: up to renamings} to obtain that there exist $\sigma_y \in \mathfrak{R}(y)$ and $\sigma_z \in \mathfrak{R}(z)$ such that $\sigma_y \cdot y = y'$ and $\sigma_z \cdot z = z'$. By Lemma~\ref{lem: renamings reflect atomicity}, we thus have $D'_0 \subseteq {\sm{R'}}_{\textit{At}}$, hence $R \equiv R'$. By Fact~\ref{fact: iso of typed PS's}, we obtain $(R, \mathsf{T}) \equiv (R', \mathsf{T'})$.
\end{proof}

Now, we take advantage of the normalization property of the typed PS's:

\begin{cor}\label{cor: injectivity}
Let $(R_1, \mathsf{T}_1)$ and $(R_2, \mathsf{T}_2)$ be two typed PS's such that $\conclusions{R_1} = \conclusions{R_2}$ and $\restriction{\mathsf{T}_1}{\conclusions{R_1}} = \restriction{\mathsf{T}_2}{\conclusions{R_2}}$. If, for any $X \in \propvar$, the set $\sm{X}$ is infinite, then we have $((R_1, \mathsf{T}_1) \simeq_{\beta} (R_2, \mathsf{T}_2) \Leftrightarrow \sm{(R_1, \mathsf{T}_1)} = \sm{(R_2, \mathsf{T}_2)})$.
\end{cor}

\begin{proof}
Let us assume that $\sm{(R_1, \mathsf{T}_1)} = \sm{(R_2, \mathsf{T}_2)}$. There exist two cut-free typed PS's $(R'_1, \mathsf{T'}_1)$ and $(R'_2, \mathsf{T'}_2)$ such that $(R_1, \mathsf{T}_1) \simeq_{\beta} (R'_1, \mathsf{T'}_1)$ and $(R_2, \mathsf{T}_2) \simeq_{\beta} (R'_2, \mathsf{T'}_2)$. We have $\conclusions{R'_1} = \conclusions{R_1} = \conclusions{R_2} = \conclusions{R'_2}$ and $\sm{(R'_1, \mathsf{T'}_1)} = \sm{(R_1, \mathsf{T}_1)} = \sm{(R_2, \mathsf{T}_2)} = \sm{(R'_2, \mathsf{T'}_2)}$. By Theorem~\ref{theorem: typed framework}, we have $(R'_1, \mathsf{T'}_1) \equiv (R'_2, \mathsf{T'}_2)$, hence $(R_1, \mathsf{T}_1) \simeq_{\beta} (R_2, \mathsf{T}_2)$.
\end{proof}

\section*{Conclusion}

The aim of this work was to prove theorems that relate the syntax of the proof-nets with their relational semantics firstly and  with their Taylor expansion secondly. But incidentally we proved an interesting intrinsic semantical result: We showed that the entire relational semantics of any normalizable proof-net can be rebuilt from two well-chosen points (and that it is impossible to strengthen this result by rebuilding any proof-net from only one well-chosen point). So, in some way, these two points together could be seen as a principal typing of intersection types for the given proof-net. With the algorithm we described, we can first rebuild the normal form of the proof-net from this ``principal typing'' and then compute the semantics. Now, some technology could probably be developed to compute the semantics directly from these points without rebuilding the syntactic object, like in the case of untyped lambda-calculus (expansion, substitution...). 

\section*{Acknowledgment}
  \noindent One of the anonymous referees pointed out a mistake in a preliminary version of the proof of Lemma~\ref{lem: relating typed semantics and untyped semantics}. I wish to acknowledge him, also for his other helpful comments.

\bibliographystyle{plain}
\bibliography{ll_new}

\newpage

\section{Appendix: Proof of Lemma~\ref{lem: connected components of R_o}}\label{appendix}

\begin{proof}
With Lemma~\ref{lem: termofTaylor of R_o} in mind, we can describe $\termofTaylor{R_o}{e}{i}$:
\begin{itemize}
\item $\portsatzero{\termofTaylor{R_o}{e}{i}} = \portsatzero{\termofTaylor{R}{e}{i}} \cup \bigcup_{e_1\in e(o)} \{ (o, (e_1, p)) ; p \in \mathcal{Q}' \}$;
\item $\labelofport{\groundof{\termofTaylor{R_o}{e}{i}}}(p) = \left\lbrace \begin{array}{ll} \labelofport{\groundof{\termofTaylor{R}{e}{i}}}(p) & \textit{if $p \in \portsatzero{\termofTaylor{R}{e}{i}}$;} \\ \contr & \textit{otherwise;} \end{array} \right.$
\item $\isaCopyfrom{R_o}{e}{i}(p) = \left\lbrace \begin{array}{ll} \isaCopyfrom{R}{e}{i}(p) & \textit{if $p \in \portsatzero{\termofTaylor{R}{e}{i}}$;} \\ (o, \varphi_o(q')) & \textit{if $p = \varphi_{e_1}(q')$ for some $e_1 \in e(o)$, $q' \in \mathcal{Q}$;}\\
q & \textit{if $p = q$;} \end{array} \right.$
\item $\wiresatzero{\termofTaylor{R_o}{e}{i}} = \wiresatzero{\termofTaylor{R}{e}{i}} \cup \bigcup_{e_1\in e(o)} \{ (o, (e_1, p)) ; p \in \mathcal{Q}' \}$;
\item $\target{\groundof{\termofTaylor{R_o}{e}{i}}}(p)$ is the following port of $\groundof{\termofTaylor{R_o}{e}{i}}$:\\ $\left\lbrace \begin{array}{ll} 
(o, (e_1, \varphi_o(\target{\groundof{\termofTaylor{R}{e}{i}}}(p)))) & \textit{if $p = (o, (e_1, p'))$ with $e_1 \in e(o)$ and $\target{\groundof{\termofTaylor{R}{e}{i}}}(p) \in \mathcal{Q}$;}\\
\target{\groundof{\termofTaylor{R}{e}{i}}}(p) & \textit{otherwise.} \end{array} \right.$
\item $\cutsatzero{\termofTaylor{R_o}{e}{i}} = \cutsatzero{\termofTaylor{R}{e}{i}}$ and $\axiomsatzero{\termofTaylor{R_o}{e}{i}} = \axiomsatzero{\termofTaylor{R}{e}{i}}$
\item $\boxesatzero{\termofTaylor{R_o}{e}{i}} = \boxesatzero{\termofTaylor{R}{e}{i}}$
\item $B_{\termofTaylor{R_o}{e}{i}} = B_{\termofTaylor{R}{e}{i}}$
\item $\dom{\target{\termofTaylor{R_o}{e}{i}}} = \dom{\target{\termofTaylor{R}{e}{i}}}$ and $\target{\termofTaylor{R_o}{e}{i}}(o', p) = $ \\$\left\lbrace \begin{array}{ll} (o, (e_1, \varphi_o(\target{\termofTaylor{R}{e}{i}}(o', p))) & \begin{array}{l} \textit{if $\target{\termofTaylor{R}{e}{i}}(o', p) \in \mathcal{Q}$ and $o' = (o, (e_1, o''))$}\\ \textit{for some $e_1 \in e(o)$, $o'' \in \boxesatzero{\termofTaylor{B_R(o)}{e_1}{i}}$;}\end{array}\\ \target{\termofTaylor{R}{e}{i}}(o', p) & \textit{otherwise;} \end{array} \right.$
\end{itemize}
For any $p, p' \in \portsatzero{R \langle o, i, e_o \rangle} \cup \mathcal{Q} \cup \varphi_{e_o}[\mathcal{Q}]$ such that $\Card{\{ p, p' \} \cap \mathcal{Q}} \leq 1$, we have $p \coh_{\termofTaylor{R_o}{e}{i}} p'$ if, and only if, one of the following properties holds:
\begin{itemize}
\item $p, p' \in \portsatzero{\termofTaylor{R}{e}{i}} \setminus \mathcal{Q}$ and $p \coh_{\termofTaylor{R}{e}{i}} p'$
\item $\{ p, p' \} \subseteq \{ (o, (e_o, p'')), \varphi_{e_o}(\target{R}{(o, p'')}) \}$ for some $p'' \in \temporaryConclusions{R}{o}$
\item $\{ p, p' \}  \cap \mathcal{Q}$ is some singleton $\{ p'' \}$ with $p'' \in \mathcal{Q}$ and $\{ p, p' \} \subseteq \{ p'', \varphi_{e_o}(p'') \}$
\end{itemize}

1) Let $T \in \nontrivialconnected{\termofTaylor{R}{e}{i}}{\mathcal{P}}{k}$ such that $\portsatzero{T} \cap \portsatzero{R \langle o, i, e_o \rangle} \not= \emptyset$. We have $\portsatzero{T[\varphi_{e_o}]} = (\portsatzero{T} \setminus \mathcal{Q}) \cup \varphi_{e_o}[\mathcal{Q} \cap \portsatzero{T}]$, hence, by Lemma~\ref{lemma: ports of connected components}, $\portsatzero{T[\varphi_{e_o}]} \subseteq \portsatzero{R \langle o, i, e_o \rangle} \cup \varphi_{e_o}[\mathcal{Q} \cap \portsatzero{T}] \: \: \: (\ast)$. Again, by Lemma~\ref{lemma: ports of connected components}, we have $\portsatzero{T} \cap \mathcal{Q} \subseteq \mathcal{P} \: \: \: (\ast \ast)$. 

Let $w \in \wiresatzero{T}$: We have $w \notin \mathcal{P}$ (hence, by $(\ast \ast)$, $w \in \portsatzero{T[\varphi_{e_o}]} \setminus \mathcal{P}_{e_o}$). Again by Lemma~\ref{lemma: ports of connected components}, there exists $p \in \portsatzero{\termofTaylor{B_R(o)}{e_o}{i}}$ such that $w = (o, (e_o, p))$. We distinguish between two cases:
\begin{itemize}
\item $p \in \temporaryConclusions{R}{o}$: we have $\target{\groundof{\termofTaylor{R}{e}{i}}}(o, (e_o, p)) = \target{R}(o, p) \in \portsatzero{T} \cap \mathcal{Q}$, hence $\target{\groundof{\termofTaylor{R_o}{e}{i}}}(o, (e_o, p)) = (o, (e_o, \varphi_o(\target{R}(o, p)))) \in \portsatzero{T[\varphi_{e_o}]}$;
\item $p \notin \temporaryConclusions{R}{o}$: we have $\target{\groundof{\termofTaylor{R_o}{e}{i}}}(o, (e_o, p)) = \target{\groundof{\termofTaylor{R}{e}{i}}}(o, (e_o, p)) \in \portsatzero{T} \setminus \mathcal{Q} \subseteq \portsatzero{T[\varphi_{e_o}]}$.
\end{itemize}
We thus have
\begin{eqnarray*}
& & \wiresatzero{T}\\
& \subseteq & \{ w \in (\wiresatzero{\termofTaylor{R_o}{e}{i}} \cap \portsatzero{T[\varphi_{e_o}]}) \setminus \mathcal{P}_{e_o} ; \target{\groundof{\termofTaylor{R_o}{e}{i}}}(w) \in \portsatzero{T[\varphi_{e_o}]} \}
\end{eqnarray*}
Conversely, let $w \in (\wiresatzero{\termofTaylor{R_o}{e}{i}} \cap \portsatzero{T[\varphi_{e_o}]}) \setminus \mathcal{P}_{e_o}$ such that $\target{\groundof{\termofTaylor{R_o}{e}{i}}}(w) \in \portsatzero{T[\varphi_{e_o}]}$. By $(\ast)$, there exists $p \in \portsatzero{\termofTaylor{B_R(o)}{e_o}{i}}$ such that $w = (o, (e_o, p))$.  We distinguish between two cases:
\begin{itemize}
\item $p \in \temporaryConclusions{R}{o}$: we have $\target{\groundof{\termofTaylor{R_o}{e}{i}}}(o, (e_o, p)) = (o, (e_o, \varphi_o(\target{R}(o, p)))) \in \portsatzero{T[\varphi_{e_o}]}$, hence $(o, (e_o, \varphi_o(\target{R}(o, p)))) \in \varphi_{e_o}[\mathcal{Q} \cap \portsatzero{T}]$; so, $\target{\groundof{\termofTaylor{R}{e}{i}}}(o, (e_o, p)) = \target{R}(o, p) \in \portsatzero{T}$, which shows that $(o, (e_o, p)) \in \wiresatzero{T}$;
\item $p \notin \temporaryConclusions{R}{o}$: we have $\target{\groundof{\termofTaylor{R}{e}{i}}}(o, (e_o, p)) = \target{\groundof{\termofTaylor{R_o}{e}{i}}}(o, (e_o, p)) \in \portsatzero{T[\varphi_{e_o}]} \setminus \varphi_{e_o}[\mathcal{Q}] \subseteq \portsatzero{T}$.
\end{itemize}
We thus have
\begin{eqnarray*}
&  & \{ w \in (\wiresatzero{\termofTaylor{R_o}{e}{i}} \cap \portsatzero{T[\varphi_{e_o}]}) \setminus \mathcal{P}_{e_o} ; \target{\groundof{\termofTaylor{R_o}{e}{i}}}(w) \in \portsatzero{T[\varphi_{e_o}]} \}\\
& \subseteq & \wiresatzero{T}
\end{eqnarray*}
Moreover, $\wiresatzero{T[\varphi_{e_o}]} = \wiresatzero{T}$; we thus showed:
\begin{eqnarray*}
& & \wiresatzero{T[\varphi_{e_o}]}\\
& = & \{ w \in (\wiresatzero{\termofTaylor{R_o}{e}{i}} \cap \portsatzero{T[\varphi_{e_o}]}) \setminus \mathcal{P}_{e_o} ; \target{\groundof{\termofTaylor{R_o}{e}{i}}}(w) \in \portsatzero{T[\varphi_{e_o}]} \}
\end{eqnarray*}
So, $T[\varphi_{e_o}] \sqsubseteq_{\mathcal{P}_{e_o}} \termofTaylor{R_o}{e}{i}$. 

Since $T$ is connected through ports not in $\mathcal{P}$, $T[\varphi_{e_o}]$ is connected through ports not in $\varphi_{e_o}[\mathcal{P} \cap \dom{\varphi_{e_o}}] \cup (\mathcal{P} \setminus \dom{\varphi_{e_o}}) = \mathcal{P}_{e_o}$.

Finally, $\cosize{T[\varphi_{e_o}]} = \cosize{T}$ and $\conclusions{\groundof{T[\varphi_{e_o}]}} = \portsatzero{T[\varphi_{e_o}]} \setminus (\wiresatzero{T[\varphi_{e_o}]} \cup \bigcup \cutsatzero{T[\varphi_{e_o}]}) = ((\portsatzero{T} \setminus \mathcal{Q}) \cup \varphi_{e_o}[\mathcal{Q} \cap \portsatzero{T}]) \setminus (\wiresatzero{T} \cup \bigcup \cutsatzero{T}) = (\conclusions{\groundof{T}} \setminus \mathcal{Q}) \cup \varphi_{e_o}[\mathcal{Q} \cap \portsatzero{T}] \subseteq (\mathcal{P} \setminus \mathcal{Q}) \cup \varphi_{e_o}[\mathcal{Q} \cap \portsatzero{T}] = \mathcal{P}_{e_o}$. Let $p \in \portsatzero{T[\varphi_{e_o}]}$ and $p' \in \portsatzero{\termofTaylor{R_o}{e}{i}}$ such that $p \coh_{\termofTaylor{R_o}{e}{i}} p'$ and $p' \notin \portsatzero{T[\varphi_{e_o}]}$ (hence $p \notin \mathcal{Q}$, so $\Card{\{ p, p' \} \cap \mathcal{Q}} \leq 1$): We distinguish between three cases:
\begin{itemize}
\item $p, p' \in \portsatzero{\termofTaylor{R}{e}{i}} \setminus \mathcal{Q}$ and $p \coh_{\termofTaylor{R}{e}{i}} p'$: we have $p \in \mathcal{P} \setminus \mathcal{Q} \subseteq \mathcal{P}_{e_o}$;
\item $\{ p, p' \} \subseteq \{ (o, (e_o, p'')), \varphi_{e_o}(\target{R}{(o, p'')}) \}$ for some $p'' \in \temporaryConclusions{R}{o}$: either $p = \varphi_{e_o}(\target{R}(o, p'))$ (hence $p \in \varphi_{e_o}[\mathcal{Q}] \subseteq \mathcal{P}_{e_o})$ or $p = (o, (e_o, p''))$. From now on, let us assume that $p = (o, (e_o, p''))$. We have $\target{R}(o, p'') \in \mathcal{Q}$. If $\target{R}(o, p'') \in \portsatzero{T}$, then $p' = \varphi_{e_o}(\target{R}(o, p'')) \in \portsatzero{T}[\varphi_{e_o}]$, which contradicts $p' \notin \portsatzero{T}[\varphi_{e_o}]$; we thus have $\target{R}(o, p'') \notin \portsatzero{T}$, which entails $p \in \mathcal{P} \setminus \mathcal{Q} \subseteq \mathcal{P}_{e_o}$.
\item $\{ p, p' \}  \cap \mathcal{Q}$ is some singleton $\{ p'' \}$ with $p'' \in \mathcal{Q}$ and $\{ p, p' \} \subseteq \{ p'', \varphi_{e_o}(p'') \}$: since $p \notin \mathcal{Q}$, we have $p' = p''$ and $p = \varphi_{e_o}(p') \in \portsatzero{T[\varphi_{e_o}]}$, hence $p' \in \portsatzero{T}$ and, finally, $p = \varphi_{e_o}(p'') \in \varphi_{e_o}[\mathcal{P} \cap \mathcal{Q}] \subseteq \mathcal{P}_{e_o}$.
\end{itemize}
In the three cases, we have $p \in \mathcal{P}_{e_o}$.

So, we showed $T[\varphi_{e_o}] \in \nontrivialconnected{\termofTaylor{R_o}{e}{i}}{\mathcal{P}_{e_o}}{k}$.

2) Let $T \in \nontrivialconnected{\termofTaylor{R_o}{e}{i}}{\mathcal{P}_{e_o}}{k}$ such that $\portsatzero{T} \subseteq \portsatzero{R_o \langle o, i, e_o \rangle}$. We have $\portsatzero{T[{\varphi_{e_o}}^{-1}]} = (\portsatzero{T} \setminus \varphi_{e_o}[\mathcal{Q}]) \cup \{ q' \in \mathcal{Q} ; \varphi_{e_o}(q') \in \portsatzero{T} \}$.

Since $\target{\groundof{\termofTaylor{R_o}{e}{i}}}[\varphi_{e_o}[\mathcal{Q}]] = \mathcal{Q}$ and $\mathcal{Q} \cap \portsatzero{T} = \emptyset$, we have 
\begin{eqnarray*}
\wiresatzero{T} & = & \{ w \in (\wiresatzero{\termofTaylor{R_o}{e}{i}} \cap \portsatzero{T}) \setminus \mathcal{P}_{e_o} ; \target{\groundof{\termofTaylor{R_o}{e}{i}}}(w) \in \portsatzero{T} \}\\
& = & \{ w \in (\wiresatzero{\termofTaylor{R_o}{e}{i}} \cap \portsatzero{T}) \setminus \mathcal{P}_{e_o} ; \target{\groundof{\termofTaylor{R}{e}{i}}}(w) \in \portsatzero{T} \}\\
& = & \{ w \in (\wiresatzero{\termofTaylor{R}{e}{i}} \cap \portsatzero{T}) \setminus \mathcal{P} ; \target{\groundof{\termofTaylor{R}{e}{i}}}(w) \in \portsatzero{T} \}
\end{eqnarray*}
Moreover, we have $\wiresatzero{T[{\varphi_{e_o}}^{-1}]} = \wiresatzero{T}$. Hence
$$\wiresatzero{T[{\varphi_{e_o}}^{-1}]} = \{ w \in (\wiresatzero{\termofTaylor{R}{e}{i}} \cap \portsatzero{T}) \setminus \mathcal{P} ; \target{\groundof{\termofTaylor{R}{e}{i}}}(w) \in \portsatzero{T} \}$$

Since $T$ is connected modulo $\mathcal{P}_{e_o}$, $T[{\varphi_{e_o}}^{-1}]$ is connected modulo ${\varphi_{e_o}}^{-1}[\mathcal{P}_{e_o} \cap \dom{{\varphi_{e_o}}^{-1}}] \cup (\mathcal{P}_{e_o} \setminus \dom{{\varphi_{e_o}}^{-1}}) = \mathcal{P}$.

Finally, $\cosize{T[{\varphi_{e_o}}^{-1}]} = \cosize{T}$ and $\conclusions{\groundof{T[{\varphi_{e_o}}^{-1}]}} = \portsatzero{T[{\varphi_{e_o}}^{-1}]} \setminus (\wiresatzero{T[{\varphi_{e_o}}^{-1}]} \cup \bigcup \cutsatzero{T[{\varphi_{e_o}}^{-1}]}) = ((\portsatzero{T} \setminus \varphi_{e_o}[\mathcal{Q}]) \cup \{ q' \in \mathcal{Q} ; \varphi_{e_o}(q') \in \portsatzero{T} \}) \setminus (\wiresatzero{T} \cup \bigcup \cutsatzero{T}) = (\conclusions{\groundof{T}} \setminus \varphi_{e_o}[\mathcal{Q}]) \cup \{ q' \in \mathcal{Q} ; \varphi_{e_o}(q') \in \portsatzero{T} \} \subseteq (\mathcal{P}_{e_o} \setminus \varphi_{e_o}[\mathcal{Q}]) \cup \{ q' \in \mathcal{Q} ; \varphi_{e_o}(q') \in \portsatzero{T} \} = (((\mathcal{P} \setminus \mathcal{Q}) \cup \varphi_{e_o}[\mathcal{P} \cap \mathcal{Q}])  \setminus \varphi_{e_o}[\mathcal{Q}]) \cup \{ q' \in \mathcal{Q} ; \varphi_{e_o}(q') \in \portsatzero{T} \} \subseteq \mathcal{P}$. Let $p \in \portsatzero{T[{\varphi_{e_o}}^{-1}]}$ and $p' \in \portsatzero{\termofTaylor{R}{e}{i}}$ such that $p \coh_{\termofTaylor{R}{e}{i}} p'$ and $p' \notin \portsatzero{T[{\varphi_{e_o}}^{-1}]}$. We distinguish between three cases:
\begin{itemize}
\item $p, p' \notin \mathcal{Q}$: we have $p \in \portsatzero{T}$, $p' \notin \portsatzero{T}$ and $p \coh_{\termofTaylor{R_o}{e}{i}} p'$, hence, since $T \trianglelefteq_{\mathcal{P}_{e_o}} \termofTaylor{R_o}{e}{i}$, we have $p \in \mathcal{P}_{e_o}$; but $\mathcal{P}_{e_o} \cap \portsatzero{\termofTaylor{R}{e}{i}} \subseteq \mathcal{P}$;
\item $p \in \mathcal{Q}$: we have $\varphi_{e_o}(p) \in \conclusions{\groundof{T}} \subseteq \mathcal{P}_{e_o}$, hence $p \in \mathcal{P}$;
\item $p \notin \mathcal{Q}$ and $p' \in \mathcal{Q}$: we have $p \in \conclusions{\groundof{T}} \subseteq \mathcal{P}_{e_o}$; but $\mathcal{P}_{e_o} \cap \portsatzero{\termofTaylor{R}{e}{i}} \subseteq \mathcal{P}$.
\end{itemize}
In the three cases, we have $p \in \mathcal{P}$.

So, we showed $T[{\varphi_{e_o}}^{-1}] \in \nontrivialconnected{\termofTaylor{R}{e}{i}}{\mathcal{P}}{k}$.

3) Let $T \in \nontrivialconnected{\termofTaylor{R_o}{e}{i}}{\mathcal{P}}{k}$ such that $(\forall e_1 \in e(o)) \portsatzero{T} \cap \portsatzero{R_o \langle o, i, e_1 \rangle} = \emptyset$. Notice that $q \notin \portsatzero{T}$; indeed, since $q \in \conclusions{\groundof{\termofTaylor{R_o}{e}{i}}}$, we have $(q \in \portsatzero{T} \Rightarrow q \in \mathcal{P})$, but $q \notin \mathcal{P}$. 

Since $(\forall w \in \portsatzero{T}) w \notin \bigcup_{e_1 \in e(o)} \{ (o, (e_1, p)) ; p \in \mathcal{Q}' \}$, we have:
\begin{itemize}
\item $\wiresatzero{\termofTaylor{R}{e}{i}} \cap \portsatzero{T} = \wiresatzero{\termofTaylor{R_o}{e}{i}} \cap \portsatzero{T}$
\item and $(\forall w \in \wiresatzero{\termofTaylor{R}{e}{i}} \cap \portsatzero{T}) \target{\groundof{\termofTaylor{R}{e}{i}}}(w) = \target{\groundof{\termofTaylor{R_o}{e}{i}}}(w)$.
\end{itemize}
Hence 
\begin{eqnarray*}
\wiresatzero{T} & = & \{ w \in (\wiresatzero{\termofTaylor{R_o}{e}{i}} \cap \portsatzero{T}) \setminus \mathcal{P} ; \target{\groundof{\termofTaylor{R_o}{e}{i}}}(w) \in \portsatzero{T} \}\\
& = & \{ w \in (\wiresatzero{\termofTaylor{R}{e}{i}} \cap \portsatzero{T}) \setminus \mathcal{P} ; \target{\groundof{\termofTaylor{R}{e}{i}}}(w) \in \portsatzero{T} \}
\end{eqnarray*}
So, $T \sqsubseteq_{\mathcal{P}} \termofTaylor{R}{e}{i}$.

Let $p \in \portsatzero{T}$ and $q' \in \bigcup_{e_1 \in e(o)} \portsatzero{R \langle o, i, e_1 \rangle}$ such that $p \coh_{\termofTaylor{R}{e}{i}} q'$: We have $p \in \mathcal{Q}$. Since from $T \trianglelefteq_{\mathcal{P}} \termofTaylor{R_o}{e}{i}$ and $(\forall e_1 \in e(o)) \portsatzero{T} \cap \portsatzero{R_o \langle o, i, e_1 \rangle} = \emptyset$ we deduce $\portsatzero{T} \cap \mathcal{Q} \subseteq \mathcal{P}$, we obtain $p \in \mathcal{P}$. We thus have $T \trianglelefteq_{\mathcal{P}} \termofTaylor{R}{e}{i}$.

We showed $T \in \nontrivialconnected{\termofTaylor{R}{e}{i}}{\mathcal{P}}{k}$.

4) Let $T \in \nontrivialconnected{\termofTaylor{R}{e}{i}}{\mathcal{P}}{k}$ such that $(\forall e_1 \in e(o)) \portsatzero{T} \cap \portsatzero{R \langle o, i, e_1 \rangle} = \emptyset$. 

Let $p \in \portsatzero{T}$ and $q' \in \bigcup_{e_1 \in e(o)} \portsatzero{R \langle o, i, e_1 \rangle}$ such that $p \coh_{\termofTaylor{R_o}{e}{i}} q'$: We have $p \in \mathcal{Q}$. Since from $T \trianglelefteq_{\mathcal{P}} \termofTaylor{R}{e}{i}$ and $(\forall e_1 \in e(o)) \portsatzero{T} \cap \portsatzero{R \langle o, i, e_1 \rangle} = \emptyset$ we deduce $\portsatzero{T} \cap \mathcal{Q} \subseteq \mathcal{P}$, we obtain $p \in \mathcal{P}$. We thus have $T \trianglelefteq_{\mathcal{P}} \termofTaylor{R_o}{e}{i}$.

We showed $T \in \nontrivialconnected{\termofTaylor{R_o}{e}{i}}{\mathcal{P}}{k}$.
\end{proof}

\end{document}